\documentclass[11pt,letterpaper]{article}

\usepackage{etoolbox}

\usepackage[margin=1in]{geometry}

\usepackage{amsmath}
\usepackage{amssymb}
\usepackage{amsthm}
\usepackage{bbm}
\usepackage{thm-restate}
\newtheorem{theorem}{Theorem}
\newtheorem{lemma}[theorem]{Lemma}

\newtheorem{definition}[theorem]{Definition}

\newtheorem{corollary}[theorem]{Corollary}

\newtheorem{claim}[theorem]{Claim}

\usepackage{url}

\usepackage[table]{xcolor}
\usepackage{array}

\usepackage{setspace}

\usepackage{wrapfig}

\usepackage[pdfpagelabels,bookmarks=false]{hyperref}
\hypersetup{colorlinks, linkcolor=darkblue, citecolor=darkgreen, urlcolor=darkblue}

\makeatletter
\newcommand{\labeltarget}[1]{\Hy@raisedlink{\hypertarget{#1}{}}}
\makeatother

\usepackage{times}
\usepackage{upgreek}

\usepackage[inline]{enumitem}
\usepackage{moreenum}

\usepackage{mathtools}
\usepackage[super]{nth}

\usepackage{bm}

\usepackage{dsfont}

\usepackage{xspace}

\usepackage{sectsty}
\allsectionsfont{\boldmath}

\setlist[enumerate]{nosep,topsep=0em}
\setlist[enumerate,1]{label=(\roman*), leftmargin=2.2em}
\setlist[enumerate,2]{label=(\alph*)}
\setlist[itemize]{nosep,topsep=0.1em}

\usepackage[textsize=footnotesize, color=blue!30!white]{todonotes}

\usepackage{xfrac}

\usepackage{tikz}
\usetikzlibrary{calc}
\usetikzlibrary{shapes.geometric}
\usetikzlibrary{arrows}
\usetikzlibrary{decorations.pathreplacing, decorations.pathmorphing}

\usepackage{tcolorbox}
\tcbuselibrary{skins}

\usepackage{float}
\usepackage{graphicx}
\graphicspath{{graphics/}}
\makeatletter
\newcommand\appendtographicspath[1]{%
  \g@addto@macro\Ginput@path{#1}%
}
\makeatother

\usepackage[margin=10pt,font=small,labelfont=bf,skip=-5pt]{caption}
\usepackage{subcaption}

\usepackage[vlined,ruled,algo2e]{algorithm2e}
\SetAlCapHSkip{0.2em}
\SetCustomAlgoRuledWidth{0.95\textwidth}

\makeatletter
\patchcmd{\@algocf@start}{%
  \begin{lrbox}{\algocf@algobox}%
}{%
  \rule{0.025\textwidth}{\z@}%
  \begin{lrbox}{\algocf@algobox}%
  \begin{minipage}{0.95\textwidth}%
}{}{}
\patchcmd{\@algocf@finish}{%
  \end{lrbox}%
}{%
  \end{minipage}%
  \end{lrbox}%
}{}{}
\makeatother

\usepackage{mdframed}

\definecolor{darkblue}{rgb}{0,0,0.38}
\definecolor{darkred}{rgb}{0.6,0,0}
\definecolor{darkgreen}{rgb}{0.1,0.35,0}

\DeclareMathOperator{\poly}{poly}
\DeclareMathOperator{\conv}{conv}

\DeclareMathOperator{\supp}{supp}
\DeclareMathOperator{\topv}{top}
\DeclareMathOperator{\bottomv}{bottom}

\newcommand\OPT{\ensuremath{\mathrm{OPT}}}
\renewcommand{\epsilon}{\varepsilon}

\newcommand\apxfac{\ensuremath{1.393}\xspace}

\def\cupp{\stackrel{.}{\cup}}
\def\bigcupp{\stackrel{.}{\bigcup}}

\makeatletter
\let\@@pmod\pmod
\DeclareRobustCommand{\pmod}{\@ifstar\@pmods\@@pmod}
\def\@pmods#1{\mkern8mu({\operator@font mod}\mkern 6mu#1)}
\makeatother

\makeatletter
\let\@@mod\mod
\DeclareRobustCommand{\mod}{\@ifstar\@mods\@@mod}
\def\@mods#1{\mkern8mu{\operator@font mod}\mkern 6mu#1}
\makeatother

\def\Ascr{\mathcal{A}}
\def\Bscr{\mathcal{B}}
\def\Cscr{\mathcal{C}}

\def\Fscr{\mathcal{F}}

\def\Iscr{\mathcal{I}}
\def\Lscr{\mathcal{L}}
\def\Mscr{\mathcal{M}}

\def\Rscr{\mathcal{R}}
\def\Sscr{\mathcal{S}}

\def\Wscr{\mathcal{W}}

\makeatletter
\def\@fnsymbol#1{\ensuremath{\ifcase#1\or *\or %
\ddagger\or
    \mathsection\or \mathparagraph\or \|\or **\or \dagger\dagger
    \or \ddagger\ddagger \else\@ctrerr\fi}}
\makeatother

\title{Bridging the Gap Between Tree and Connectivity Augmentation: Unified and Stronger Approaches\thanks{
This project received funding from Swiss National Science Foundation grant 200021\_184622 and the European Research Council (ERC) under the European Union's Horizon 2020 research and innovation programme (grant agreement No 817750).

A short version of this article appeared at the 53rd ACM Symposium on Theory of Computing (STOC), 2021.
}} 
\author{
Federica Cecchetto\thanks{
Department of Mathematics, ETH Zurich, Zurich, Switzerland.
Email: \href{mailto:federica.cecchetto@ifor.math.ethz.ch}%
{federica.cecchetto@ifor.math.ethz.ch}.
}
\and
Vera Traub\thanks{
Department of Mathematics, ETH Zurich, Zurich, Switzerland.
Email: \href{mailto:vera.traub@ifor.math.ethz.ch}%
{vera.traub@ifor.math.ethz.ch}.
}
\and
Rico Zenklusen\thanks{
Department of Mathematics, ETH Zurich, Zurich, Switzerland.
Email: \href{mailto:ricoz@ethz.ch}%
{ricoz@ethz.ch}.}
}

\date{}

\begin{document}

\maketitle
\thispagestyle{empty}
\addtocounter{page}{-1}

\begin{tikzpicture}[overlay, remember picture, shift = {(current page.south east)}]
\begin{scope}[shift={(-1.5,1.8)}]
\def\hd{1.5}
\node at (-2*\hd,0) {\includegraphics[height=0.35cm]{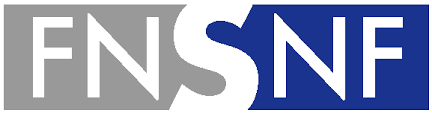}};
\node at (-\hd,0) {\includegraphics[height=0.7cm]{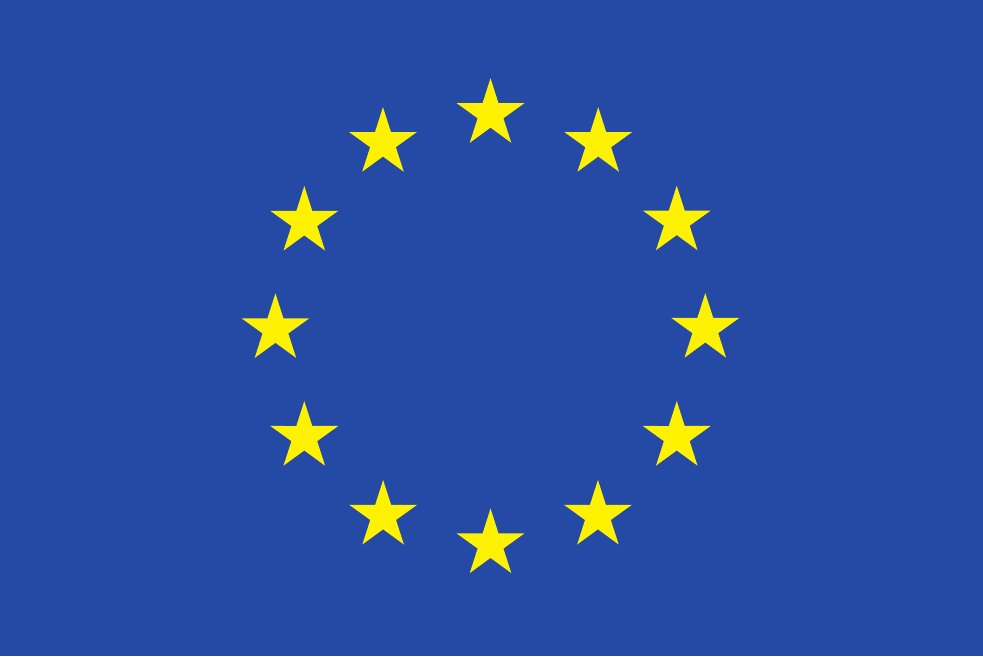}};
\node at (-0.2*\hd,0) {\includegraphics[height=0.8cm]{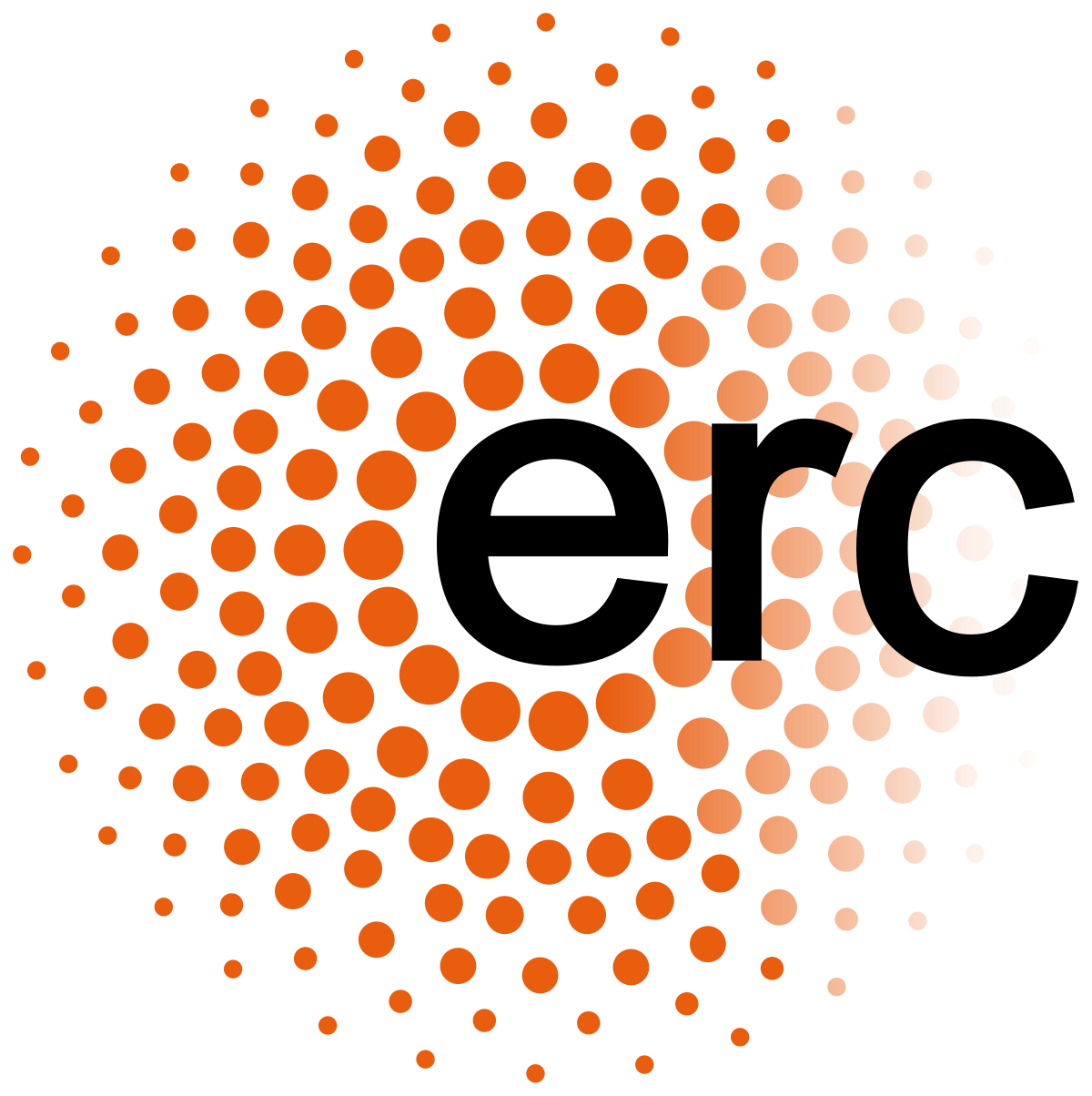}};
\end{scope}
\end{tikzpicture}

\begin{abstract}
We consider the Connectivity Augmentation Problem (CAP), a classical problem in the area of Survivable Network Design. It is about increasing the edge-connectivity of a graph by one unit in the cheapest possible way. More precisely, given a $k$-edge-connected graph $G=(V,E)$ and a set of extra edges, the task is to find a minimum cardinality subset of extra edges whose addition to $G$ makes the graph $(k+1)$-edge-connected. If $k$ is odd, the problem is known to reduce to the Tree Augmentation Problem (TAP)---i.e., $G$ is a spanning tree---for which significant progress has been achieved recently, leading to approximation factors below $1.5$ (the currently best factor is $1.458$). However, advances on TAP did not carry over to CAP so far. Indeed, only very recently, Byrka, Grandoni, and Ameli (STOC 2020) managed to obtain the first approximation factor below $2$ for CAP by presenting a $1.91$-approximation algorithm based on a method that is disjoint from recent advances for TAP.

We first bridge the gap between TAP and CAP, by presenting techniques that allow for leveraging insights and methods from TAP to approach CAP. We then introduce a new way to get approximation factors below $1.5$, based on a new analysis technique. Through these ingredients, we obtain a $\apxfac$-approximation algorithm for CAP, and therefore also for TAP. This leads to the currently best approximation result for both problems in a unified way, by significantly improving on the above-mentioned $1.91$-approximation for CAP and also the previously best approximation factor of $1.458$ for TAP by Grandoni, Kalaitzis, and Zenklusen (STOC 2018). Additionally, a feature we inherit from recent TAP advances is that our approach can deal with the weighted setting when the ratio between the largest to smallest cost on extra links is bounded, in which case we obtain approximation factors below $1.5$.
\end{abstract}

\section{Introduction}

The edge-connectivity of a graph is a key measure of its redundancy and reliability. Not surprisingly, there has been intensive work on various problems around how to increase the edge-connectivity of a graph in the most economical way. One of the arguably most elementary problems in this context is the Connectivity Augmentation Problem (CAP), which is about increasing the edge-connectivity of a graph by one unit through the addition of a smallest subset of extra edges out of a given set of possibilities.
Formally, we are given an undirected $k$-edge-connected graph $G=(V,E)$ together with a subset of potential extra edges $L\subseteq \binom{V}{2}$, which we call \emph{links} to distinguish them from the edges $E$, and the task is to find a minimum cardinality set $F\subseteq L$ such that $(V,E\cup F)$ is $(k+1)$-edge-connected.%
\footnote{We always treat $L$ and $E$ as being disjoint. In particular, if $F\subseteq L$ contains a link $\ell\in F$ between two endpoints of $G$ that are already connected by an edge, then there are two parallel edges between these endpoints in $(V,E\cup F)$.}

Clearly, for a set $F\subseteq L$ to increase $G$'s edge-connectivity from $k$ to $k+1$ when added to it, $F$ needs to cover all $k$-cuts of $G$, which are the minimum cuts of $G$ assuming that $G$ is not already $(k+1)$-edge-connected. In other words, for each min cut in $G$, there must be a link in $F$ crossing it. Classical results on the structure of minimum cuts then imply that the case of odd $k$ reduces to $k=1$ and the one for even $k$ reduces to $k=2$ (see~\cite{dinitz_1976_structure} and also~\cite{cheriyan_1992_2-coverings,khuller_1993_approximation}). Moreover, these reductions are approximation preserving. The $k=1$ case is a heavily studied special case known as the Tree Augmentation Problem (TAP). Significant progress has recently been achieved regarding the study of its approximability through the development of a rich set of techniques (see%
~\cite{%
adjiashvili_2018_beating,%
cheriyan_2018_approximating_a,%
cheriyan_2018_approximating_b,%
cheriyan_2008_integrality,%
cohen_2013_approximation,%
even_2009_approximation,%
fiorini_2018_approximating,%
frederickson_1981_approximating,%
grandoni_2018_improved,%
khuller_1993_approximation,%
kortsarz_2016_simplified,%
kortsarz_2018_lp-relaxations,%
nagamochi_2003_approximation,%
nutov_2017_tree%
}), leading to the currently best approximation factor of $1.458$~\cite{grandoni_2018_improved}.
If $k=1$, we can assume that $G$ is a (spanning) tree because all $2$-edge-connected components can be contracted without changing the problem, which explains the naming of TAP. Similarly, for $k=2$, one can contract any pair of vertices that are $3$-edge-connected, which leads to a cactus graph (or simply \emph{cactus}), i.e., a connected graph where each edge is in a unique cycle.%
\footnote{The typical reduction employed to reduce CacAC for a $k$-edge-connected graph $G$ with $k$ even to the case $k=2$ is based on the fact that min cuts can be represented by a cactus. This leads to an equivalent $k=2$ instance that is already a cactus and there is no need to contract vertex sets that are $3$-edge-connected.

Moreover, we highlight that, depending on the literature, a cactus is sometimes defined to be a connected graph where every edge is in \emph{at most} one cycle instead of being in \emph{exactly} one cycle.
} 
Due to this, the $k=2$ case is typically called the \emph{Cactus Augmentation Problem} (CacAP). See Figure~\ref{fig:CacAP_example} for an example instance of CacAP. CacAP includes TAP as a special case, because any TAP instance can be transformed into an equivalent CAP instance by adding for each tree edge one parallel edge to obtain a cactus.
Hence, connectivity augmentation is equivalent (also in terms of approximability) to CacAP.

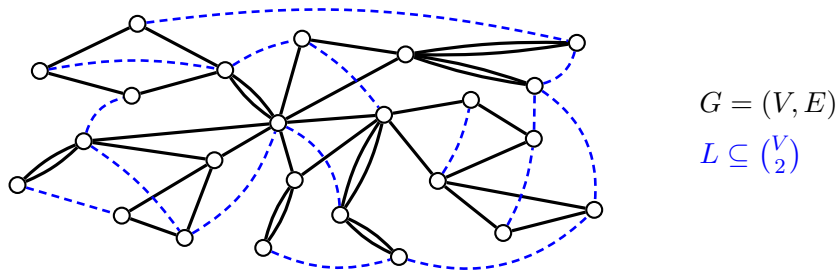
\begin{figure}[h!]
\begin{center}
\begin{tikzpicture}[scale=0.55]

\tikzset{
lks/.style={line width=1pt, blue, densely dashed},
ts/.style={every node/.append style={font=\scriptsize}}
}

\begin{scope}[every node/.style={thick,draw=black,fill=white,circle,minimum size=6, inner sep=2pt}]
\node  (1) at (17.06,-4.94) {};
\node  (2) at (17.46,-6.32) {};
\node  (3) at (19.62,-4.74) {};
\node  (4) at (16.70,-7.96) {};
\node  (5) at (18.56,-7.16) {};
\node  (6) at (19.98,-8.18) {};
\node  (7) at (20.92,-6.34) {};
\node  (8) at (23.24,-5.32) {};
\node  (9) at (21.72,-4.38) {};
\node (10) at (22.50,-7.60) {};
\node (11) at (24.70,-7.04) {};
\node (12) at (12.36,-5.38) {};
\node (13) at (15.52,-5.84) {};
\node (14) at (10.76,-6.44) {};
\node (15) at (13.28,-7.18) {};
\node (16) at (14.80,-7.72) {};
\node (17) at (20.14,-3.28) {};
\node (18) at (17.64,-2.90) {};
\node (19) at (23.26,-4.04) {};
\node (20) at (24.28,-3.00) {};
\node (21) at (15.78,-3.66) {};
\node (22) at (13.66,-2.54) {};
\node (23) at (11.30,-3.68) {};
\node (24) at (13.52,-4.28) {};
\end{scope}

\begin{scope}[very thick]
\draw[bend left=10] (21) to  (1);
\draw (18) --  (1);
\draw (13) --  (1);
\draw  (1) --  (2);
\draw  (2) --  (3);
\draw  (1) -- (12);
\draw  (1) -- (17);
\draw[bend left=10]  (1) to (21);
\draw  (3) --  (1);
\draw  (9) --  (3);

\draw  (2) to[bend left=10] (4);
\draw  (2) to[bend right=10] (4);

\draw  (3) to[bend left=10] (5);
\draw  (3) to[bend right=10] (5);

\draw  (5) to[bend left=10]  (6);
\draw  (5) to[bend right=10] (6);

\draw  (3) --  (7);
\draw  (7) --  (8);
\draw  (8) --  (9);
\draw  (7) -- (10);
\draw (10) -- (11);
\draw (11) --  (7);
\draw (14) to[bend left=10] (12);
\draw (12) to[bend left=10] (14);
\draw (16) -- (13);
\draw (12) -- (13);
\draw (13) -- (15);
\draw (15) -- (16);
\draw (17) -- (18);

\draw (17) to[bend left=5] (19);
\draw (19) to[bend left=5] (17);

\draw (17) to[bend left=5] (20);
\draw (20) to[bend left=5] (17);
\draw (24) -- (21);
\draw (21) -- (22);
\draw (22) -- (23);
\draw (23) -- (24);
\end{scope}

\begin{scope}[lks]
\draw (14) to (15);
\draw (16) to[bend right=20] (1);
\draw (4) to[bend right=20] (6);
\draw (12) to[bend left] (24);
\draw (22) to[bend left=10] (20);
\draw (23) to[bend left=10] (21);
\draw (7) to[bend right=10] (9);
\draw (10) to[bend right=10] (8);
\draw (19) to[bend right] (20);
\draw (19) to[bend left] (11);
\draw (8) to (19);
\draw (1) to[bend left] (5);
\draw (6) to[bend right] (11);
\draw (18) to[bend left=10] (3);
\draw (21) to[bend left=10] (18);
\draw (12) to[bend left=10] (16);

\end{scope}

\begin{scope}[shift={(27,-4.5)}]%
\def\ll{30mm} %
\def\vs{12mm} %

\node[right] at (0,0) {$G=(V,E)$};
\node[right,blue] at (0,-\vs) {$L\subseteq \binom{V}{2}$};
\end{scope}

\end{tikzpicture}
 \end{center}
\caption{Example of a CacAP instance. The links are drawn as dashed blue lines.}\label{fig:CacAP_example}
\end{figure}

Unfortunately, progress on CacAP (or equivalently CAP) has been very limited compared to TAP. Until very recently, the best approximation factor was $2$, achievable through a variety of techniques~\cite{goemans_1994_improved,jain_2001_factor,khuller_1993_approximation}. This factor is attainable even in the weighted setting where links have non-negative costs and the goal is to increase the edge connectivity by $1$ by adding a minimum cost set of links. However, even for (unweighted) CacAP , a better-than-2 approximation has only been discovered very recently by Byrka, Grandoni, and Ameli~\cite{byrka_2020_breaching}, who presented a $1.91$-approximation through a nice connection to Steiner Tree and by tailoring the analysis of the currently strongest approximation algorithm for Steiner Tree~\cite{byrka_2013_steiner} (see also~\cite{goemans_2012_matroids}) to the particular structure of the resulting Steiner Tree instances.\footnote{Prior to this, only for the much simpler subclass of CacAP on cycle instances, i.e., when the underlying graph is a single cycle, a better-than-2 approximation was known, namely a $(1.5+\varepsilon)$-approximation~\cite{galvez_2019_cycle}.}

A key hurdle why progress on CacAP was arguably significantly slower than progress on TAP, is that it was highly unclear how to leverage the rich set of recent TAP techniques to approach the more general CacAP, as emphasized in~\cite{byrka_2020_breaching}. Thus, the only progress below factor $2$ for CacAP follows a new approach disconnected from recent techniques developed for TAP. Nevertheless, this leaves the question open whether significantly stronger approximation guarantees for CacAP can be achieved, maybe even below $1.5$ as for TAP, by identifying new ways to leverage some of the ideas used in TAP to attack CacAP. The main contribution of this paper is to answer this question affirmatively by presenting several new techniques that allow for attacking CacAP with tools developed for TAP and with new ingredients that allow for stronger approximation factors for both CacAP and TAP through a unified approach. The strength of our approach is emphasized by the fact that we obtain even for the much more general CacAP, an approximation factor that beats the best prior approximation factor of $1.458$ for the more restrictive and heavily studied TAP.

\subsection{Our results}

The following theorem is our main result.
\begin{theorem}\label{thm:main}
There is a $\apxfac$-approximation algorithm for the Connectivity Augmentation Problem (CAP).
\end{theorem}
The above result simultaneously improves on several prior results in a unified way. First, we significantly improve the previously best approximation factor of $1.91$ for CacAP~\cite{byrka_2020_breaching}. Moreover, our results also improve on algorithms developed for special cases, including the previously best $1.458$-approximation algorithm for TAP~\cite{grandoni_2018_improved}, a $(\sfrac{3}{2}+\varepsilon)$-approximation algorithm for CacAP on cycles by G\'alvez, Grandoni, Ameli, and Sornat~\cite{galvez_2019_cycle}, and a very recent $\sfrac{5}{3}$-approximation algorithm of Nutov~\cite{nutov_2017_tree} for the special case of CAP where all links connect minimal deficient sets (these are also called leaf-to-leaf instances).
However, equally importantly, we believe that the techniques we introduce present new ways to attack these problems, which we hope will have further consequences in the future. We provide an overview of the techniques in Section~\ref{sec:overview}.

Analogous to some recent results for TAP~\cite{adjiashvili_2018_beating,fiorini_2018_approximating,nutov_2017_tree,grandoni_2018_improved}, being able to leverage TAP techniques for CAP allows us to extend our approach to weighted CAP. In particular, we obtain better-than-$1.5$ approximations for CAP whenever the ratio between largest to smallest cost is bounded by a constant. (An analogous result was previously known for TAP~\cite{grandoni_2018_improved_arxiv}.) To the best of our knowledge, this is the first progress on weighted CAP on a non-trivial class of instances that improves on classical $2$-approximations~\cite{khuller_1993_approximation,goemans_1994_improved,jain_2001_factor}. We expand on this extension to the weighted setting in Section~\ref{sec:weighted}.

\subsection{Further related results}
Recent work of Nutov~\cite{nutov_2020_approximation} exhibits the following black-box reduction linking approximation algorithms for Steiner Tree to CAP: Any $\alpha$-approximation algorithm for Steiner Tree allows for obtaining a $(1 + \ln(4-x) + \varepsilon)$-approximation for CAP, where $x$ is the unique solution to $1 + \ln(4-x) = \alpha + (\alpha - 1)x$. This leads to a $1.942$-approximation for CAP when using the currently best $(\ln(4)+\epsilon)$-approximation for Steiner Tree.

Moreover, a well-known problem class related to CAP are problems dealing with designing $2$-edge-connected subgraphs under different assumptions. The arguably most canonical problem of this type is known as the \emph{$2$-edge-connected spanning subgraph problem} ($2$-ECSS). In its unweighted version, one is given an undirected graph and the task is to find a $2$-edge-connected subgraph with a smallest number of edges that spans all vertices. The currently best known approximation factor for (unweighted) $2$-ECSS is $\sfrac{4}{3}$~\cite{sebo_2014_shorter,hunkenschroder_2019_approximation}. Interestingly, for weighted $2$-ECSS, where edges have non-negative costs and the goal is to find a minimum cost $2$-edge-connected spanning subgraph, no approximation factor better than $2$ is known. The factor 2 can be achieved through a variety of elegant techniques~\cite{khuller_1994_biconnectivity}, including some very general classical techniques like primal-dual methods~\cite{goemans_1995_general} and iterative rounding~\cite{jain_2001_factor} (see also~\cite{williamson_2011_design,lau_2011_iterative}). Intriguingly, even for the special case of $2$-ECSS with edge costs of only $0$ and $1$, which is known as the \emph{Forest Augmentation Problem} (FAP), no better-than-$2$ approximation is known. TAP can be interpreted as a special case of FAP, where the $0$-edges form a connected subgraph. Another interesting ``extreme'' case of FAP, somewhat antipodal to TAP, is when the $0$-edges form a matching. This problem is known as the~\emph{Matching Augmentation Problem}, and, very recently, better-than-$2$ approximations have been obtained for it~\cite{cheriyan_2020_matching,cheriyan_2020_improved} with the currently best approximation factor being \sfrac{5}{3}~\cite{cheriyan_2020_improved}.

Finally, we remark that for the weighted version of the Tree Augmentation Problem, a better-than-$2$ approximation has been discovered~\cite{traub_2021_better} very recently. However, for weighted CAP, no approximation algorithm with a factor better than $2$ is known.

\subsection{Organization of paper}

We start with an overview of our techniques in Section~\ref{sec:overview}. Section~\ref{sec:reduction} presents a key bridge between techniques introduced for TAP and CacAP (or equivalently CAP) that allows for reducing general CacAP instances to so-called $O(1)$-wide instances. One component to complete this boils down to solving a particular type of rectangle hitting problem, which we discuss in Section~\ref{sec:heavy_cut_covering}. Section~\ref{sec:cg} presents a linear programming relaxation and crucial rounding procedure for $O(1)$-wide CacAP instances which already leads to a $(\sfrac{3}{2}+\varepsilon)$-approximation algorithm for CacAP. In Section~\ref{sec:stack_analysis}, we present a novel way to obtain factors significantly below $1.5$, leading to our main result, Theorem~\ref{thm:main}. Finally, Sections~\ref{sec:weighted} and~\ref{sec:fpt} contain a discussion of  how our results extend to weighted CacAP with bounded weights, and an FPT algorithm for weighted CacAP that we need in our derivations, respectively.

\section{Overview of our approach}\label{sec:overview}

As mentioned, one of our key goals is to introduce new methods that allow for leveraging recent advances for TAP in the context of CacAP. Indeed, through the new bridges between CacAP and TAP, we are able to design CacAP algorithms that, on a high level, follow closely recent TAP approaches (see~\cite{adjiashvili_2018_beating,fiorini_2018_approximating,grandoni_2018_improved}), even though transferring prior TAP techniques to CacAP faces significant hurdles, as already emphasized in prior work.
We first briefly outline this high-level approach, which consists of three main parts. We then expand on these three main parts in Sections~\ref{sec:overviewKWide},~\ref{sec:overviewRounding}, and~\ref{sec:overviewBetterApprox}, respectively.
Throughout this paper, we denote a CacAP instance as a tuple $(G=(V,E),L)$, where $G=(V,E)$ is the underlying cactus and $L\subseteq \binom{V}{2}$ the set of links that can be added.

\medskip

The first main part is a black-box reduction to a particular type of structured instance, which we call \emph{$k$-wide}, using an approach for TAP inspired by Adjiashvili~\cite{adjiashvili_2018_beating}.
The idea is to reduce to instances consisting of many parts of small complexity that share a single vertex $r$. We measure the complexity of a CacAP instance in terms of the number of degree two vertices it contains, which we call \emph{terminals}.
\begin{definition}[terminal of CacAP instance]
A vertex of a CacAP instance is a \emph{terminal} if it has degree two.
\end{definition}
The number of terminals of a CacAP instance is indeed a parameter that impacts the complexity of the problem. More precisely, as discussed later,  the following statement, which even holds for weighted CacAP, can readily be derived from results by Basavaraju, Fomin, Golovach, Misra, Ramanujan, and Saurabh~\cite{basavaraju_2014_parameterized}. (For completeness, a proof can be found in Section~\ref{sec:fpt}.)
\begin{lemma}\label{lem:fpt_terminals}
Weighted CacAP can be solved in time $3^{|T|} \cdot \poly(|V|)$, where $|T|$ and $|V|$ are the number of terminals and vertices, respectively, of the instance.
\end{lemma}
Terminals naturally extend the notion of leaves in TAP; indeed when modeling TAP as CacAP, the leaves in TAP are precisely the terminals of the corresponding CacAP instance. Now, a $k$-wide instance is defined as follows, where, for a vertex $r\in V$, the graph $G-r$ is the subgraph of $G=(V,E)$ induced by $V\setminus \{r\}$.
\begin{definition}[$k$-wide CacAP]
Let $k\in \mathbb{Z}_{\geq 1}$. A CacAP instance $(G=(V,E),L)$ is \emph{$k$-wide} if there is a vertex $r$ such that each connected component of $G-r$ contains at most $k$ terminals of $G$. We call the vertex $r$ a \emph{$k$-wide center} of $G$, or simply a \emph{center}. Moreover, for each $W\subseteq V\setminus \{r\}$ that is the vertex set of a connected component of $G-r$, we call $G[W\cup \{r\}]$ an $r$-principal subcactus of $G$.\footnote{Our notion of $k$-wide CacAP instances maps to the $k$-wide notion used in~\cite{grandoni_2018_improved} for the special case of TAP instances. Moreover, we remark that the center may also be a terminal, which corresponds precisely to the case where there is a single principal subcactus.}
\end{definition}
See Figure~\ref{fig:k-wide} for an exemplification of the above notions. When there is no danger of ambiguity, we often simply talk about \emph{principal subcacti} of $G$ without explicitly specifying the center.

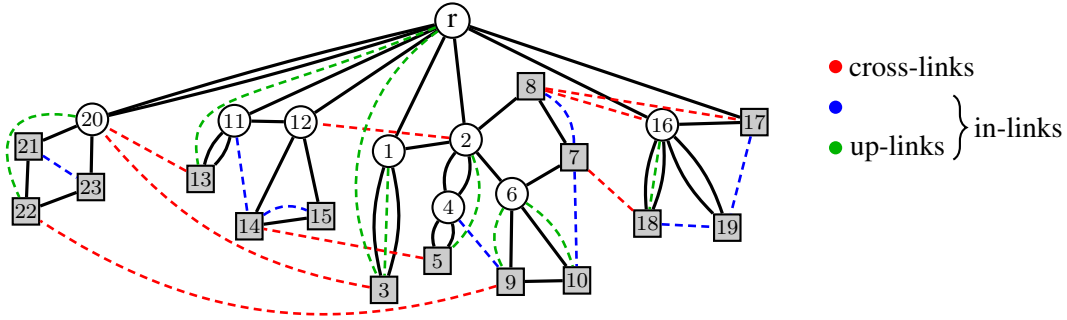
\begin{figure}[!ht]
\begin{center}
\begin{tikzpicture}[scale=0.45]

\pgfdeclarelayer{bg}
\pgfdeclarelayer{fg}
\pgfsetlayers{bg,main,fg}

\tikzset{
  prefix node name/.style={%
    /tikz/name/.append style={%
      /tikz/alias={#1##1}%
    }%
  }
}

\tikzset{root/.style={fill=white,minimum size=13}}

\tikzset{lks/.style={line width=1pt, blue, densely dashed}}
\tikzset{crossl/.append style={red}}
\tikzset{upl/.append style={green!70!black}}

\tikzset{q1/.append style={}}

\tikzset{q2/.append style={}}

\tikzset{q3/.append style={}}

\tikzset{q4/.append style={}}

\tikzset{
term/.style={fill=black!20, rectangle, minimum size=10},
termg/.style={fill=black!20, rectangle, minimum size=10},
tf/.append style={font=\scriptsize\color{black}},
}

\newcommand\cac[2][]{

\begin{scope}[prefix node name=#1]

\tikzset{npc/.style={#2}}

\begin{scope}[every node/.append style={thick,draw=black,fill=white,circle,minimum size=12, inner sep=2pt}]
\node[root]  (1) at (18.68,-1.64) {r};
\node  (2) at (16.80,-5.50) {};
\node  (3) at (19.08,-5.16) {};
\node[term]  (4) at (16.66,-9.60) {};
\node  (5) at (18.56,-7.16) {};
\node[term]  (6) at (18.24,-8.74) {};
\node  (7) at (20.44,-6.76) {};
\node[term]  (8) at (22.28,-5.68) {};
\node[term]  (9) at (21.00,-3.58) {};
\node[term] (10) at (20.40,-9.36) {};
\node[term] (11) at (22.36,-9.32) {};

\begin{scope}[npc]
\node (12) at (12.24,-4.60) {};
\node (13) at (14.20,-4.64) {};
\node[termg] (14) at (11.24,-6.32) {};
\node[termg] (15) at (12.68,-7.72) {};
\node[termg] (16) at (14.80,-7.42) {};
\node (17) at (24.88,-4.72) {};
\node[termg] (18) at (27.60,-4.64) {};
\node[termg] (19) at (24.44,-7.64) {};
\node[termg] (20) at (26.80,-7.76) {};
\node (21) at (8.04,-4.56) {};
\node[termg] (22) at (6.16,-5.36) {};
\node[termg] (23) at (6.08,-7.28) {};
\node[termg] (24) at (8.00,-6.56) {};
\end{scope}
\end{scope}

\begin{scope}[every node/.append style={font=\scriptsize}]

\foreach \i/\t in {2/,3/,4/tf,5/,6/tf,7/,8/tf,9/tf,10/tf,11/tf,12/,13/,14/tf,15/tf,16/tf,17/,18/tf,19/tf,20/tf,21/,22/tf,23/tf,24/tf} {
\pgfmathparse{int(\i-1)}
\node[\t] at (\i) {$\pgfmathresult$};
}
\end{scope}

\begin{scope}[very thick]

\draw  (1) --  (2);
\draw  (2) --  (3);
\draw  (3) --  (1);
\draw  (9) --  (3);
\draw  (2) to[bend left=15] (4);
\draw  (4) to[bend left=15] (2);
\draw  (3) to[bend left=25] (5);
\draw  (5) to[bend left=25] (3);
\draw  (5) to[bend left=25] (6);
\draw  (6) to[bend left=25] (5);
\draw  (3) --  (7);
\draw  (7) --  (8);
\draw  (8) --  (9);
\draw  (7) -- (10);
\draw (10) -- (11);
\draw (11) --  (7);

\begin{scope}[npc]
\draw  (1) to[bend left=2] (21);
\draw  (1) -- (12);
\draw  (1) -- (17);
\draw (21) to[bend left=2]  (1);
\draw (18) --  (1);
\draw (13) --  (1);
\draw (12) to[bend left=15] (14);
\draw (14) to[bend left=15] (12);
\draw (16) -- (13);
\draw (12) -- (13);
\draw (13) -- (15);
\draw (15) -- (16);
\draw (17) -- (18);
\draw (17) to[bend left=13] (19);
\draw (19) to[bend left=13] (17);
\draw (17) to[bend left=13] (20);
\draw (20) to[bend left=13] (17);
\draw (24) -- (21);
\draw (21) -- (22);
\draw (22) -- (23);
\draw (23) -- (24);
\end{scope}

\end{scope}

\end{scope}

}%

\begin{scope}

\tikzset{
grayout/.style={
every node/.append style={draw=black!30}, draw=black!30
},
inlink/.style={q1,blue},
crlink/.style={q1,red}
}

\begin{scope}[xshift=-1cm]
\cac[]{}
\end{scope}

\begin{scope}[lks]

\begin{scope}[crossl]

\begin{scope}[q1]
\draw (3) to (13);
\draw (4) to[bend left=20] (21);
\draw (6) to (15);
\draw (8) to (19);
\draw (9) to (17);
\draw (9) to (18);
\draw (14) to (21);
\end{scope}

\begin{scope}[q2]
\draw (10) to[bend left=25] (23);
\end{scope}

\end{scope}

\begin{scope}[blue]

\begin{scope}[q1]
\draw[upl] (1) to[bend right=40] (4);
\draw[upl] (1) to[out=202,in=100,out looseness=2, in looseness=0.7] (14);
\draw[upl] (2) to (4);
\draw (5) to (10);
\draw[upl] (7) to[bend right] (10);
\draw[upl] (7) to[bend left=20] (11);
\draw (8) to[bend right] (9);
\draw[upl] (17) to (19);
\draw (19) to (20);
\end{scope}

\begin{scope}[q2]
\draw[upl] (3) to[bend left=35] (6);
\draw (12) to (15);
\draw[upl] (21) to[out=170, in =115,looseness=1.7] (23);
\end{scope}

\begin{scope}[q3]
\draw (8) to (11);
\draw (18) to (20);
\end{scope}

\begin{scope}[q4]
\draw (15) to[bend left] (16);
\draw (22) to (24);
\end{scope}

\end{scope}

\end{scope}
\end{scope}%

\begin{scope}[shift={(29,-3)}]%
\def\ll{30mm} %
\def\vs{12mm} %

\begin{scope}[every node/.style={circle,inner sep=0pt,minimum size=5pt}]
\node[fill=red] (rd) at (0,0) {};
\node[fill=blue] (bd) at (0,-\vs) {};
\node[fill=green!70!black] (gd) at (0,-2*\vs) {};

\node at (rd.east)[right=2pt] {cross-links};
\node (ilt) at (gd.east)[right=2pt] {up-links};

\path (bd) -| coordinate (batop) (ilt.east);
\path (gd) -| coordinate (babot) (ilt.east);

\coordinate (btop) at ($(batop)+(1ex,2ex)$);
\coordinate (bbot) at ($(babot)+(1ex,-2ex)$);

\draw[decorate, decoration={brace,mirror,amplitude=5pt},line width=0.8pt] (bbot) -- node[right=2.5mm] {in-links} (btop);

\end{scope}

\end{scope}

\end{tikzpicture}
 \end{center}
\caption{A $6$-wide CacAP instance with center $r$ and $4$ principal subcacti. The terminals are drawn as gray squares. The dashed lines represent the links, which can be partitioned into cross-links and in-links (with up-links being a sub-class of the in-links) as highlighted by the different colors and explained before we state Lemma~\ref{lem:bundle_rounding}.}\label{fig:k-wide}
\end{figure}

The key result of the first step is the following, which shows that it suffices to consider $O(1)$-wide CacAP instances to obtain approximation algorithms for general CacAP, at the expense of an arbitrarily small constant loss in the approximation factor.\footnote{An analogous theorem for TAP was presented in~\cite{grandoni_2018_improved}.}
\begin{theorem}\label{thm:main_reduction}
Let $\alpha \ge 1$ and $\varepsilon > 0$.
Given an $\alpha$-approximation algorithm $\Ascr$ for $\frac{64(8+3\varepsilon)}{\varepsilon^2}$-wide CacAP instances, there is an $\alpha \cdot(1+\epsilon)$-approximation algorithm $\Bscr$ for (unrestricted) CacAP that calls $\Ascr$ at most polynomially many times and performs further operations taking polynomial time.
\end{theorem}
On a high level, we follow a recent approach of~\cite{adjiashvili_2018_beating} for TAP that successively splits a general instance into $O(1)$-wide ones. However, we face two additional challenges.
First, contrary to TAP, there are new hurdles to make sure that good splittings are always possible, which we address by introducing a well-chosen rectangle hitting problem, whose solution can be leveraged to preprocess the instance.
Second, contrary to TAP, merging solutions of split sub-instances does not immediately lead to a solution to the original instance. We therefore define a particular type of solutions to sub-instances that allows for cheap merging, and show that such solutions can always be found via the oracle $\mathcal{A}$.
See Section~\ref{sec:overviewKWide} for details on how we prove Theorem~\ref{thm:main_reduction}.

The reduction to $O(1)$-wide instances heavily exploits that the instance is unweighted. (Actually, it suffices that the largest to smallest cost ratio is small.) However, some of our results for $O(1)$-wide instances also apply to weighted instances without modification; we thus present them in the weighted setting.

\medskip

Again in the spirit of prior TAP techniques, the second main part of our approach is about designing two procedures to find solutions for weighted $O(1)$-wide CacAP instances with different guarantees depending on the type of links. We call these procedures \emph{backbone procedures}.
By returning the better solution of the two backbone procedures, we obtain a $1.5$-approximation for $O(1)$-wide CacAP, which, due to Theorem~\ref{thm:main_reduction}, immediately implies a $(1.5+\varepsilon)$-approximation for CacAP, and thus also for CAP. Note that this already significantly improves on the previously best $1.91$-approximation~\cite{byrka_2020_breaching}.
We later improve one of the backbone procedures to obtain the claimed $\apxfac$-approximation.

To express the guarantees of the backbone procedures in terms of different link types, we extend three well-known link types from TAP to CacAP in a canonical way, namely cross-links, in-links, and up-links. More precisely, let $k\in \mathbb{Z}_{\geq 1}$ and let $(G=(V,E),L)$ be a $k$-wide CacAP instance with center $r$. The links $L$ are partitioned into \emph{cross-links} and \emph{in-links}, denoted by $L_{\mathrm{cross}}$ and $L_{\mathrm{in}}$, respectively, as follows.
\emph{In-links} are all links with both endpoints in the same principal subcactus, and \emph{cross-links} are all remaining links.
Hence, cross-links have endpoints in different principal subcacti and do not have the center as one of their endpoints. (Note that the center appears in all principal subcacti.)
Hence, all links incident to the center are in-links. Finally, \emph{up-links}, denoted by $L_{\mathrm{up}}\subseteq L_{\mathrm{in}}$, are all in-links $\ell \in L_{\mathrm{in}}$ such that one endpoint of $\ell$ is an ancestor of the other one, where $u\in V$ being an \emph{ancestor} of $v\in V$ means that $u$ lies on every $v$-$r$ path in $G$ (in which case we call $v$ a \emph{descendant} of $u$). 
See Figure~\ref{fig:k-wide} for an illustration.
For simplicity, we often state results on $k$-wide instances that refer to cross-links, in-links, and up-links without specifying the center. In such cases these results hold for any $k$-wide center of the instance.

The first backbone procedure extends an approach introduced by Adjiashvili~\cite{adjiashvili_2018_beating} for TAP that is based on so-called bundle constraints. The extension readily follows from Lemma~\ref{lem:fpt_terminals}.
We denote by $\OPT\subseteq L$ an optimal solution to the problem under discussion, which is a weighted CacAP instance in this case.
\begin{lemma}\label{lem:bundle_rounding}
For any weighted $k$-wide CacAP instance $\mathcal{I}=(G=(V,E),L)$ with link costs $c\in\mathbb{R}_{\geq 0}^L$, we can compute in time $3^k \poly(|V|)$ a CacAP solution $F\subseteq L$ with $c(F) \leq c(\OPT) + c(\OPT\cap L_{\mathrm{cross}})$.
\end{lemma}
\begin{proof}
For each principal subcactus we compute an optimal CacAP solution through Lemma~\ref{lem:fpt_terminals} to the CacAP sub-instance corresponding to that subcactus only, i.e., only $2$-cuts within this principal subcactus need to be covered.
Notice that these sub-instances are indeed instances of CacAP: The only difference to a regular CacAP instance is that cross-links with one endpoint in the considered sub-cactus have the other endpoint outside of it. However, for such links we can simply replace the endpoint outside the considered sub-cactus by the center of the $k$-wide instance. (Such a replacement link is also called a \emph{shadow} of the original link; an analogous concept has been used in the context of TAP and we will expand on this later.)

Finally, we return the union $F\subseteq L$ of all solutions. $F$ is clearly a feasible CacAP solution and the running time bound follows from Lemma~\ref{lem:fpt_terminals}. Finally, the claimed cost guarantee on $F$ holds because $F$ is not more costly than returning for each subcactus the links of $\OPT$ with at least one endpoint in the subcactus (excluding the center), which leads to a cost of $c(\OPT) + c(\OPT \cap L_{\mathrm{cross}})$ because each in-link appears in one sub-instance and each cross-link in two.
\end{proof}
The second backbone procedure, which is stated below, is the CacAP-counterpart to an elegant LP-based technique introduced by Fiorini, Gro\ss, K\"onemann, and Sanit\`a~\cite{fiorini_2018_approximating}, based on Chv\'atal-Gomory cuts.
It introduces a polytope $P_{\mathrm{cross}}$ that is a relaxation of the convex hull $P_{\mathrm{CacAP}}(\mathcal{I})$ of all actual solutions, i.e.,
\begin{equation*}
P_{\mathrm{CacAP}}(\mathcal{I})\coloneqq \conv(\{\chi^F : F\subseteq L, (V,E\cup F) \text{ is $3$-edge-connected}\})\enspace.\footnote{$\chi^F\in \{0,1\}^L$ denotes the characteristic vector of $F$, i.e., $\chi^F_{\ell}=1$ for $\ell\in F$, and $\chi^F_\ell=0$ otherwise.}
\end{equation*}
The introduced polytope $P_{\mathrm{cross}}$ is \emph{solvable}, which means that we can efficiently separate over it. 
\begin{restatable}{lemma}{lemcrosslinkrounding}\label{lem:cross-link_rounding}
Given a weighted $k$-wide CacAP instance $\mathcal{I}=(G=(V,E),L)$ with link costs $c\in \mathbb{R}^L_{\geq 0}$, there is a solvable polytope $P_{\mathrm{cross}}\supseteq P_{\mathrm{CacAP}}(\mathcal{I})$ with facet complexity bounded by $\poly(|V|)$ such that for any $x\in P_{\mathrm{cross}}$ we can efficiently obtain a solution $F\subseteq L$ satisfying $c(F) \leq c^T x + \sum_{\ell\in L_{\mathrm{in}}\setminus L_{\mathrm{up}}} c_\ell x_\ell$.\footnote{%
We highlight that Lemma~\ref{lem:cross-link_rounding} has no requirements on the parameter $k$, and hence applies to any instance, even if it is not $O(1)$-wide. Nevertheless, we use the notion of $k$-wide instances in its statement because the cost guarantee on $F$ depends on the different link types, which we introduced in the context of $k$-wide instances, and because we typically use the statement in such contexts.
}
\end{restatable}
The \emph{facet complexity} of a polytope is an upper bound on the encoding length of each of its facets. A polynomially-bounded facet complexity together with a separation oracle allows for efficiently optimizing any linear function over the polytope via the Ellipsoid Method. (We refer the interested reader to~\cite{groetschel_1993_geometric} for additional information.)
Hence, the following is an immediate consequence of Lemma~\ref{lem:cross-link_rounding}.
\begin{corollary}\label{cor:cross-link_rounding}
For any weighted $k$-wide CacAP instance $(G=(V,E),L)$ with link costs $c\in \mathbb{R}_{\geq 0}^L$, we can efficiently compute a solution $F\subseteq L$ with $c(F)\leq c(\OPT)+c(\OPT\cap (L_{\mathrm{in}}\setminus L_{\mathrm{up}}))$.
\end{corollary}
\begin{proof}
By Lemma~\ref{lem:cross-link_rounding}, we can use the Ellipsoid Method to efficiently obtain an optimal solution $x^*$ of $\min\{c^T x + \sum_{\ell\in L_{\mathrm{in}}\setminus L_{\mathrm{up}}} c_\ell x_\ell : x\in P_{\mathrm{cross}}\}$ (see~\cite[Theorem~6.4.9]{groetschel_1993_geometric}) together with a CacAP solution $F\subseteq L$ satisfying $c(F) \leq c^T x^* + \sum_{\ell\in L_{in}\setminus L_{up}}c_\ell x^*_\ell$. The solution $F$ has the desired cost guarantee because $P_{\mathrm{cross}} \supseteq P_{\mathrm{CacAP}}(G)$ implies $c^T x^* + \sum_{\ell\in L_{\mathrm{in}}\setminus L_{\mathrm{up}}} c_\ell x^*_\ell \leq c(\OPT) + c(\OPT\cap (L_{\mathrm{in}}\setminus L_{\mathrm{up}}))$.
\end{proof}

Finally, the two backbone procedures presented above, together with Theorem~\ref{thm:main_reduction}, already imply a significant improvement in terms of approximation factor for CacAP, and thus also for CAP.

\begin{corollary}
For any $\epsilon >0$, there is a $(1.5+\epsilon)$-approximation algorithm for CacAP.
\end{corollary}
\begin{proof}
By Theorem~\ref{thm:main_reduction}, it suffices to provide a $1.5$-approximation algorithm for $O(1)$-wide CacAP instances. This is obtained by efficiently computing the CacAP solutions $F_1, F_2$ guaranteed by Lemma~\ref{lem:bundle_rounding} and Corollary~\ref{cor:cross-link_rounding}, respectively, and returning the better of the two. We recall that these guarantees are $c(F_1)\leq c(\OPT) + c(\OPT\cap L_{\mathrm{cross}})$ and $c(F_2)\leq c(\OPT) + c(\OPT\cap (L_{\mathrm{in
}}\setminus L_{\mathrm{up}})))$. This indeed leads to the desired guarantee because $c(F_1)+c(F_2)\leq 2c(\OPT) + c(\OPT\setminus L_{\mathrm{up}}) \leq 3c(\OPT)$.
\end{proof}

The key technical novelty in this part is our derivation of Lemma~\ref{lem:cross-link_rounding}. In particular, we develop a novel approach based on a dual-improvement argument.
More precisely, we start with a dual solution of a weak LP that has integrality gap $2$. This dual solution reveals a subset of constraints from which we derive stronger Chv\'atal-Gomory constraints. Moreover, we then show that those constraints allow us to select cross-links that lead to a residual problem with a reduced dual value allowing us to complete the selected cross-links cheaply.

\smallskip

In the last main part of our approach, we present a new technique to obtain approximation factors significantly below $1.5$ for CacAP and thus also for CAP. The only prior related technique is a method based on so-called rewirings presented in~\cite{grandoni_2018_improved} for TAP reaching an approximation factor of $1.458$. With our new technique, we obtain a substantial strengthening of Lemma~\ref{lem:bundle_rounding}, which enables achieving approximation factors below $1.4$ even for the much more general CacAP.
Whereas our method is inspired by the algorithm used in~\cite{grandoni_2018_improved} for TAP, our main contribution lies in a novel analysis approach based on assigning cross-links to terminals to build up stacks.
We then define a stack problem that allows us to obtain a strong lower bound on the improvement compared to Lemma~\ref{lem:bundle_rounding}. We expand on this in Section~\ref{sec:overviewBetterApprox}.

\subsection{Reduction to $k$-wide instances}\label{sec:overviewKWide}

We now provide an overview of some main ingredients used to prove our reduction theorem, Theorem~\ref{thm:main_reduction}. To obtain a black-box reduction, as stated in the theorem, we use a round-or-cut approach together with an extension of Adjiashvili's~\cite{adjiashvili_2018_beating} splitting idea for TAP. More precisely, given a point $x\in [0,1]^L$, we first use $x$ to break the original CacAP instance $\mathcal{I}=(G=(V,E),L)$ into independent CacAP sub-instances $\mathcal{I}_1,\ldots, \mathcal{I}_q$ that are all $k$-wide for $k=\frac{64(8+3\epsilon)}{\epsilon^2}$ and such that one of the following holds:
\begin{enumerate}
\item\label{item:reduction_round} Either we can, based on the $\alpha$-approximation $\mathcal{A}$ for $k$-wide instances, obtain for each sub-instance $\mathcal{I}_i$ with $i\in [q]$ a CacAP solution such that the union of all those solutions (together with further links fixed during the splitting process) leads to a $(1+\varepsilon) \alpha$-approximation for $\mathcal{I}$, as desired;

\item\label{item:reduction_sep} Or $\mathcal{A}$ can be used to certify that the restriction of $x$ to at least one of the sub-instances is not a convex combination of solutions. In this case we show that we can separate $x$ from the convex combination $P_{\mathrm{CacAP}}(\mathcal{I})$ of all solutions of the original instance $\mathcal{I}$.
\end{enumerate}
The above then implies Theorem~\ref{thm:main_reduction} by the round-or-cut framework. In short, assume that we know the cardinality $|\OPT|$ of an optimal solution $\OPT\subseteq L$ to instance $\mathcal{I}$, which can be guessed upfront.\footnote{One can also use a more efficient binary search procedure to obtain a good bound on the optimal value.} Then one can use the Ellipsoid Method to find a point in $P_{\mathrm{CacAP}}(\mathcal{I})\cap \{y\in [0,1]^L: y(L)=|\OPT|\}$ by using the above procedure as separation oracle. Indeed, given a vector $x\in[0,1]^L$ with $x(L)=|\OPT|$, either point~\ref{item:reduction_round} applies, in which case we obtain a solution with the desired guarantees and are done. Or point~\ref{item:reduction_sep} applies, in which case we return a valid separating hyperplane.
Because the Ellipsoid Method runs in polynomial time and $P_{\mathrm{CacAP}}(\mathcal{I})\cap \{y\in [0,1]^L: y(L)=|\OPT|\} \neq \emptyset$, point~\ref{item:reduction_round} must apply in some call to the separation oracle throughout the Ellipsoid Method.

\medskip

We now sketch how we break the instance into $k$-wide sub-instances such that always~\ref{item:reduction_round} or~\ref{item:reduction_sep} can be achieved. Let $\mathcal{C}_G\subseteq 2^V$ be all $2$-cuts of $G$. For convenience (and later use), we fix an arbitrary root $r\in V$ and define the cuts in $\mathcal{C}_G$ such that they do not contain $r$, i.e.,
\begin{equation*}
\mathcal{C}_G\coloneqq \left\{C\subseteq V\setminus \{r\} : |\delta_E(C)|=2\right\}\enspace.\footnote{We denote by \emph{cut} a vertex set $C\subseteq V$ with $\emptyset \neq C \neq V$. Depending on the literature, a cut sometimes refers to the edge set $\delta(C)$. The definition we use helps simplifying our exposition because we have two types of edges, the ones in $E$ and the links, and often use the same cut, i.e., vertex set, to discuss both types.}
\end{equation*}

We aim at splitting the instance at well-chosen cuts $C\in \mathcal{C}_G$. Splitting at $C$ leads to two sub-instances $\mathcal{I}_C$ and $\mathcal{I}_{V\setminus C}$, where $\mathcal{I}_C$ is the instance obtained from $\mathcal{I}$ by contracting all vertices except for $C$, and $\mathcal{I}_{V\setminus C}$ is the instance obtained from $\mathcal{I}$ by contracting $C$. See Figure~\ref{fig:splitting_example} for an example.
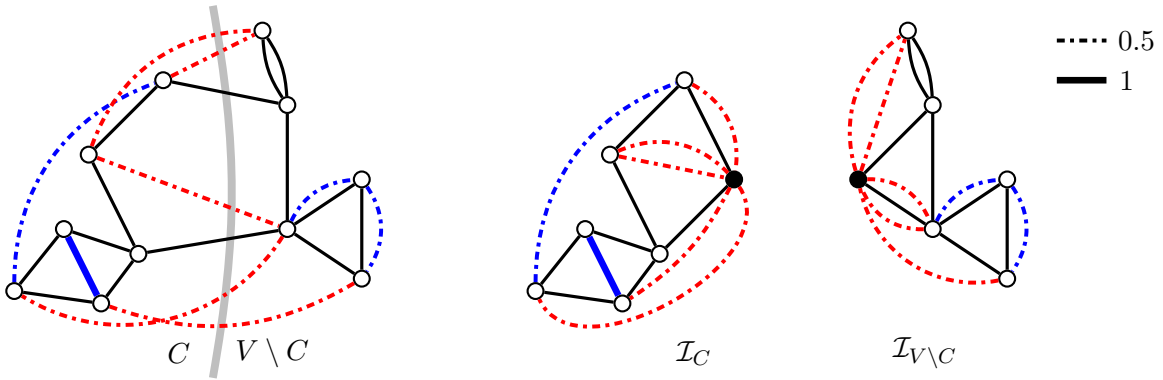
\begin{figure}[!ht]
\begin{center}
\begin{tikzpicture}[scale=0.33]

\tikzset{
  prefix node name/.style={%
    /tikz/name/.append style={%
      /tikz/alias={#1##1}%
    }%
  }
}

\tikzset{q2/.style={line width=1.5pt, dash pattern=on 3pt off 2pt on \the\pgflinewidth off 2pt}}

\tikzset{q4/.style={line width=2.5pt}}

\tikzset{node/.style={thick,draw=black,fill=white,circle,minimum size=6, inner sep=2pt}}

\newcommand\leftpart[2][]{
\begin{scope}[prefix node name=#1]

\begin{scope}[every node/.append style=node]
\node (1) at (13,13) {};
\node  (2) at (10,10) {};
\node  (3) at (12,6) {};
\node  (4) at (9,7) {};
\node  (5) at (7,4.5) {};
\node  (6) at (10.5,4) {};
\end{scope}

\begin{scope}[very thick]
\draw (1) --(2) --(3);
\draw (3) --(4) -- (5) -- (6) --(3);
\end{scope}

\begin{scope}[blue]
\draw[q4] (4) -- (6);
\draw[q2, bend left=35] (5) to (1);
\end{scope}

\end{scope}
}%

\newcommand\rightpart[2][]{
\begin{scope}[prefix node name=#1]

\begin{scope}[every node/.append style=node]
\node  (7) at (18,12) {};
\node  (8) at (18,7) {};
\node  (9) at (17,15) {};
\node (10) at (21,9) {};
\node (11) at (21,5) {};
\end{scope}

\begin{scope}[very thick]
\draw (7) -- (8);
\draw  (7) to[bend left=15] (9);
\draw  (9) to[bend left=15] (7);
\draw (8) -- (11) -- (10) -- (8);
\end{scope}

\begin{scope}[blue]
\draw[q2, bend left] (8) to (10);
\draw[q2, bend left=35] (10) to (11);
\end{scope}

\end{scope}
}%

\begin{scope}
\leftpart[o-]{}
\rightpart[o-]{}
\end{scope}

\draw[bend right=10,line width=3, gray, opacity=0.5] (15,1) to (15,16);
\node[left] (o-C) at (14.5,2) {$C$};
\node[right] (o-V-C) at (15.5,2) {$V\setminus C$};

\begin{scope}[very thick]
\draw (o-7) -- (o-1);
\draw (o-8) -- (o-3);
\end{scope}

\begin{scope}[q2,red]
\draw (o-2) -- (o-8);
\draw (o-1) -- (o-9);
\draw[bend left=35] (o-2) to (o-9);
\draw[bend left=45] (o-8) to (o-5);
\draw[bend left=25] (o-11) to (o-6);
\end{scope}

\begin{scope}[shift={(21,0)}]
\leftpart[s-]{}
\node[node,fill=black]  (cl) at (15,9) {};
\node[left] (o-C) at (14.5,2) {$\mathcal{I}_{C}$};
\end{scope}

\begin{scope}[shift={(26,0)}]
\rightpart[s-]{}
\node[node,fill=black]  (cr) at (15,9) {};
\node[right] (o-V-C) at (16,2) {$\mathcal{I}_{V\setminus C}$};
\end{scope}

\begin{scope}[very thick]
\draw (s-1) --(cl);
\draw (s-3) --(cl);
\end{scope}

\begin{scope}[very thick]
\draw (s-7) --(cr);
\draw (s-8) --(cr);
\end{scope}

\begin{scope}[q2,red]
\draw (s-2) -- (cl);
\draw[bend left=35] (s-2) to (cl);
\draw[bend left=95] (cl) to (s-5);
\draw[bend left=15] (cl) to (s-6);
\draw[bend left=35] (s-1) to (cl);
\end{scope}

\begin{scope}[q2,red]
\draw[bend left=35] (cr) to (s-8);
\draw[bend left=35] (cr) to (s-9);
\draw[bend left=35] (s-8) to (cr);
\draw[bend left=45] (s-11) to (cr);
\draw (cr) -- (s-9);
\end{scope}

\begin{scope}[yshift=+13cm]
\draw[q2] (49,1.6) -- +(2,0) node[right] {$0.5$};
\draw[q4] (49,0) -- +(2,0) node[right] {$1$};
\end{scope}

\end{tikzpicture} \end{center}
\caption{Example of splitting a CacAP instance $\mathcal{I}$ (together with a point $x\in [0,1]^L$) along a cut $C\in \mathcal{C}_G$ into the sub-instances $\mathcal{I}_C$ and $\mathcal{I}_{V\setminus C}$.
The red links are those crossing the cut $C$ and are present in both $\mathcal{I}_C$ and $\mathcal{I}_{V\setminus C}$.}\label{fig:splitting_example}
\end{figure}

Notice that, as exemplified in Figure~\ref{fig:splitting_example}, given a point $x\in [0,1]^L$, the splitting of the instance leads to sub-instances where $x$ naturally splits into $x_{C}, x_{V\setminus C}\in [0,1]^L$, obtained from $x$ by simply setting the values on all contracted links to $0$. Ideally, one would like to solve the sub-instances independently and then combine the solutions to one of the original instance while only suffering a negligible loss through the splitting in terms of the objective.
Even though this general plan is a quite direct generalization of Adjiashvili's approach for TAP, it comes with important additional challenges in CacAP which, in the special case of TAP, either do not even exist or are trivial to address.

Observe that the links in $\delta_L(C)$ appear in both sub-instances after splitting, whereas all other links only appear in one of the two sub-instances. Consequently, the total $x$-value on links of both sub-instances together after splitting is by $x(\delta_L(C))$ higher than before splitting. To make sure that this increase is small, we only split on so-called \emph{$x$-light} cuts, which are cuts $C\in \mathcal{C}_G$ with $x(\delta_L(C))\leq \sfrac{16}{\varepsilon}$. For this we want to reduce the instance $\mathcal{I}$ first to one with only $x$-light cuts. To this end, we compute a small link set covering all $x$-heavy $2$-cuts, i.e., the cuts in $\mathcal{C}_G$ that are not $x$-light.
We call such a link set a \emph{cheap $x$-heavy cut covering}.
\begin{definition}\label{def:x_heavy_cut_covering}
A link set $L_H\subseteq L$ is a \emph{cheap $x$-heavy cut covering} if
\begin{enumerate}
\item $\delta_L(C)\cap L_H\neq\emptyset \; \forall C\in \mathcal{C}_G$ with $x(\delta_L(C))>\frac{16}{\varepsilon}$\enspace, and
\item $|L_H| \leq \frac{\varepsilon}{2} \cdot x(L)$\enspace.
\end{enumerate}
\end{definition}
By including a cheap $x$-heavy cut covering $L_H$ in the solution, we obtain, as we explain formally in Section~\ref{sec:heavy_cut_covering}, a reduced CacAP instance without $x$-heavy cuts.
In case of TAP, finding such a link set $L_H$ is easy, because covering all heavy cuts is itself again a TAP instance in which each cut is heavily covered by the solution $x$. The existence of a set $L_H$ with the above properties then follows immediately by the fact that the integrality gap of the canonical LP relaxation for TAP (the cut-LP) is bounded by $2$. (See~\cite{adjiashvili_2018_beating} for details.) However, contrary to TAP, the problem of covering all $x$-heavy cuts in a CacAP instance does not reduce to another CacAP instance. Nevertheless, we show that it can be reduced to (constructively) bounding the integrality gap of the canonical relaxation of a natural rectangle hitting problem. We then provide a rounding method for the rectangle hitting problem that allows for efficiently constructing a cheap $x$-heavy cut covering.
Hence, we can assume from now on that all cuts in $\mathcal{C}_G$ are $x$-light.

Similar to TAP, to make sure that, after splitting, the (constant) increase in terms of total LP cost is small compared to $|\OPT|$, we only split at cuts $C\in \mathcal{C}_G$ that are \emph{big}, i.e., $|C\cap T| > \sfrac{k}{2}=\Theta(\sfrac{1}{\epsilon^2})$, where $T\subseteq V$ are all terminals of $G$. This allows us to amortize the increase in $x$-value due to splitting against the number of links needed to solve the sub-instance $\mathcal{I}_{C}$. This amortization works globally because we successively split on inclusion-wise minimal big $x$-light cuts.

Another key difference to TAP is how to combine solutions $F_{C}$, $F_{V\setminus C} \subseteq L$ to the sub-instances $\mathcal{I}_C$ and $\mathcal{I}_{V\setminus C}$, respectively, to obtain a solution for the original instance before splitting. In TAP, $F_C \cup F_{V\setminus C}$ is a solution of the original instance. However, this is in general not true for CacAP as highlighted in Figure~\ref{fig:combining_example}.

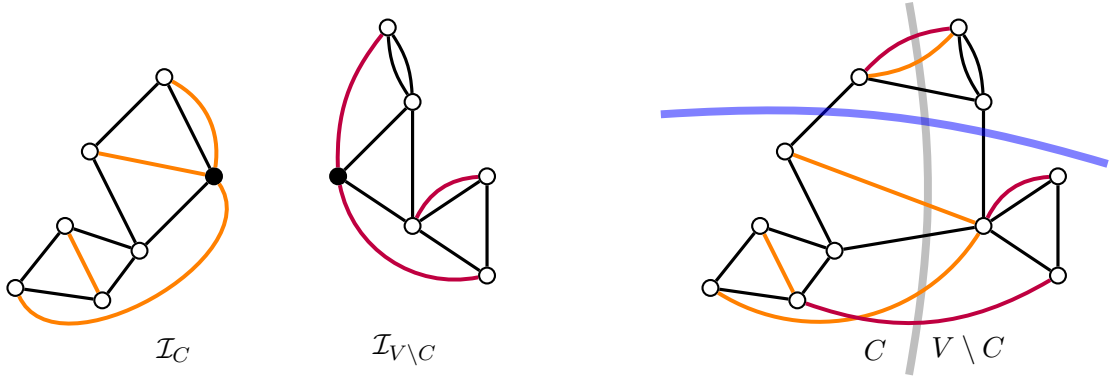
\begin{figure}[!ht]
\begin{center}
\begin{tikzpicture}[scale=0.33]

\tikzset{
  prefix node name/.style={%
    /tikz/name/.append style={%
      /tikz/alias={#1##1}%
    }%
  }
}

\tikzset{link/.style={line width=1.5pt}}

\tikzset{node/.style={thick,draw=black,fill=white,circle,minimum size=6, inner sep=2pt}}

\newcommand\leftpart[2][]{
\begin{scope}[prefix node name=#1]

\begin{scope}[every node/.append style=node]
\node (1) at (13,13) {};
\node  (2) at (10,10) {};
\node  (3) at (12,6) {};
\node  (4) at (9,7) {};
\node  (5) at (7,4.5) {};
\node  (6) at (10.5,4) {};
\end{scope}

\begin{scope}[very thick]
\draw (1) --(2) --(3);
\draw (3) --(4) -- (5) -- (6) --(3);
\end{scope}

\begin{scope}[orange, link]
\draw (4) -- (6);
\end{scope}

\end{scope}
}%

\newcommand\rightpart[2][]{
\begin{scope}[prefix node name=#1]

\begin{scope}[every node/.append style=node]
\node  (7) at (18,12) {};
\node  (8) at (18,7) {};
\node  (9) at (17,15) {};
\node (10) at (21,9) {};
\node (11) at (21,5) {};
\end{scope}

\begin{scope}[very thick]
\draw (7) -- (8);
\draw  (7) to[bend left=15] (9);
\draw  (9) to[bend left=15] (7);
\draw (8) -- (11) -- (10) -- (8);
\end{scope}

\begin{scope}[purple, link]
\draw[bend left] (8) to (10);
\end{scope}

\end{scope}
}%

\begin{scope}
\leftpart[o-]{}
\rightpart[o-]{}
\end{scope}

\draw[bend right=10,line width=3, gray, opacity=0.5] (15,1) to (15,16);
\node[left] (o-C) at (14.5,2) {$C$};
\node[right] (o-V-C) at (15.5,2) {$V\setminus C$};
\draw[bend left=10,line width=3, blue, opacity=0.5] (5,11.5) to (23,9.5);

\begin{scope}[very thick]
\draw (o-7) -- (o-1);
\draw (o-8) -- (o-3);
\end{scope}

\begin{scope}[link]
\draw[orange] (o-2) -- (o-8);
\draw[bend left=20, purple] (o-1) to (o-9);
\draw[bend right=20, orange] (o-1) to (o-9);
\draw[bend left=45,orange] (o-8) to (o-5);
\draw[bend left=25, purple] (o-11) to (o-6);
\end{scope}

\begin{scope}[shift={(-28,0)}]
\leftpart[s-]{}
\node[node,fill=black]  (cl) at (15,9) {};
\node[left] (o-C) at (14.5,2) {$\mathcal{I}_{C}$};
\end{scope}

\begin{scope}[shift={(-23,0)}]
\rightpart[s-]{}
\node[node,fill=black]  (cr) at (15,9) {};
\node[right] (o-V-C) at (16,2) {$\mathcal{I}_{V\setminus C}$};
\end{scope}

\begin{scope}[very thick]
\draw (s-1) --(cl);
\draw (s-3) --(cl);
\draw (s-7) --(cr);
\draw (s-8) --(cr);
\end{scope}

\begin{scope}[link,orange]
\draw (s-2) -- (cl);
\draw[bend left=35] (s-1) to (cl);
\draw[bend left=95] (cl) to (s-5);
\end{scope}

\begin{scope}[link,purple]
\draw[bend left=20] (cr) to (s-9);
\draw[bend left=45] (s-11) to (cr);
\end{scope}

\end{tikzpicture} \end{center}
\caption{Feasible solutions of the sub-instances $\mathcal{I}_C$ and $\mathcal{I}_{V\setminus C}$.
Their union (right picture) is not a feasible solution of the original instance $\Iscr$ because the 2-cut shown in blue is not covered. }\label{fig:combining_example}
\end{figure}

Nevertheless, we show that one can bound the number of links that need to be added to $F_C\cup F_{V\setminus C}$ to obtain a solution as follows.
\begin{restatable}{proposition}{propcombinesplitsol}\label{prop:combine_split_sol}
Given a feasible CacAP instance $\mathcal{I}=(G=(V,E),L)$, a $2$-cut $C\in \mathcal{C}_G$, and solutions $F_C, F_{V\setminus C} \subseteq L$ to $\mathcal{I}_C$ and $\mathcal{I}_{V\setminus C}$, respectively, one can efficiently compute a link set $F\subseteq L$ such that
\begin{enumerate}
\item $F_C \cup F_{V\setminus C} \cup F$ is a CacAP solution to $\mathcal{I}$, and
\item $|F|\leq |\delta_L(C)\cap F_C|-1$.
\end{enumerate}
\end{restatable}
To control the cost of merging solutions, we look for solutions in the sub-instances with constant-size $|\delta_L(C)\cap F_C|$. Finally, we show that once the successive splitting procedure stops, in which case we broke the original CacAP instance into sub-instances that we call \emph{unsplittable}, we are left with $k$-wide instances only. 

\begin{restatable}{lemma}{lemunsplittable}\label{lem:unsplittable}
Every unsplittable CacAP instance is $k$-wide.
\end{restatable}
Moreover, we do the successive splitting in a way assuring that the sum of the $x$-values of all sub-instances is close to $x(L)$. One can now use $\mathcal{A}$ to obtain an $\alpha$-approximation for each sub-instance. If for any sub-instance, the $\alpha$-approximate solution is more than a factor of $\alpha$ higher than the total $x$-value, then this implies that the total $x$-value for the sub-instance is strictly less than the value of an optimal solution to the sub-instance.
We can convert this into a separating hyperplane to separate $x$ from $P_{\mathrm{CacAP}}(\mathcal{I})$.
We use this idea to show that we can always either return a separating hyperplane, or simply return the union of solutions to the sub-instances (together with the links used to cover all $x$-heavy cuts).

Section~\ref{sec:reduction} formalizes and proves the above ideas and statements, and shows that they indeed allow for obtaining the desired black-box reduction.

\subsection{Proof plan for backbone approach leading to Lemma~\ref{lem:cross-link_rounding}}\label{sec:overviewRounding}

Fiorini, Gro\ss, K\"onemann, and Sanit\`a~\cite{fiorini_2018_approximating} were able to show Lemma~\ref{lem:cross-link_rounding} for the special case of TAP. They presented an elegant proof based on strengthening a canonical relaxation by adding all $\{0,\sfrac{1}{2}\}$-Chv\'atal-Gomory (CG) cuts and showing integrality of a resulting description through results on binet matrices. A $\{0,\sfrac{1}{2}\}$-CG cut of a polyhedron $\{x\in \mathbb{R}^n : Ax \geq b\}$ with $A\in \mathbb{R}^{m\times n}$ and $b\in \mathbb{R}^m$ is any constraint of type
\begin{equation*}
\lambda^T A x \geq \left\lceil \lambda^T b\right\rceil\enspace,
\end{equation*}
where $\lambda\in \{0,\sfrac{1}{2}\}^m$ such that $A^T \lambda \in \mathbb{Z}^{n}$. It is not hard to observe that any such constraint is valid for any integer point $x\in \mathbb{Z}^n$ fulfilling $Ax \geq b$.
Crucial steps of the approach of~\cite{fiorini_2018_approximating} are specific to TAP, which includes the design of an efficient separation procedure for the introduced CG cuts and also a proof that the resulting constraint matrix is a binet matrix. Indeed, it is unclear how/whether one could efficiently separate over all $\{0,\sfrac{1}{2}\}$-CG cuts of the canonical relaxation of a CacAP instance $\mathcal{I}=(G=(V,E),L)$, which is given by
\begin{equation*}
P_{\mathrm{cut}}\coloneqq \left\{x\in \mathbb{R}_{\geq 0}^L\colon x(\delta_L(C)) \geq 1 \; \forall C\in \mathcal{C}_G\right\}\enspace,
\end{equation*}
and which we call the \emph{cut relaxation} in analogy to common terminology used in the context of TAP. (We recall that $\mathcal{C}_G \coloneqq \{C\subseteq V\setminus \{r\} : |\delta_E(C)| = 2\}$ where $r\in V$ is the center according to which we define the different link types.)
Nevertheless, we show how a carefully selected subset of CG cuts for $P_{\mathrm{cut}}$ over which efficient separation is possible, together with an appropriately chosen rounding procedure, allows for obtaining the desired polytope $P_{\mathrm{cross}}$ with the properties described in Lemma~\ref{lem:cross-link_rounding}.
One way to obtain a well-structured subset of constraints of $P_{\mathrm{cut}}$ for which we can efficiently separate over all $\{0,\sfrac{1}{2}\}$-CG cuts, is by considering constraints of $P_{\mathrm{cut}}$ corresponding to a laminar sub-family of $\mathcal{C}_G$, as this reduces to TAP. More precisely, let $\mathcal{L}\subseteq \mathcal{C}_G$
a laminar family. Then the relaxation
\begin{equation*}
P_{\mathrm{cut}}^{\mathcal{L}}\coloneqq
 \left\{x\in \mathbb{R}_{\geq 0}^L \colon x(\delta_L(C))\geq 1 \;\forall C\in \mathcal{L}\right\}
\end{equation*}
of $P_{\mathrm{cut}}$ is the cut relaxation of a TAP problem. Therefore, we can separate over all $\{0,\sfrac{1}{2}\}$-CG cuts of $P_{\mathrm{cut}}^{\mathcal{L}}$ using the procedure presented in~\cite{fiorini_2018_approximating}. Let $P_{\mathrm{CG}}^{\mathcal{L}}\subseteq P_{\mathrm{cut}}^{\mathcal{L}}$ be the polytope obtained by adding to $P_{\mathrm{cut}}^{\mathcal{L}}$ all of its $\{0,\sfrac{1}{2}\}$-CG cuts.
Of course, there can be exponentially many possibilities to choose a laminar family $\mathcal{L}$, and there may be many  further $\{0,\sfrac{1}{2}\}$-CG cuts of $P_{\mathrm{cut}}$ that cannot be obtained from any such laminar family. However, we show that there exists one laminar family $\mathcal{L}$ leading to $\{0,\sfrac{1}{2}\}$-CG cuts that are strong enough to obtain a polytope $P_{\mathrm{cross}}$ with the desired properties.
\begin{restatable}{proposition}{propgoodlaminar}\label{prop:overview_good_laminar}
One can efficiently compute a laminar family $\mathcal{L}\subseteq \mathcal{C}_G$ such that $P_{\mathrm{cross}}\coloneqq P_{\mathrm{cut}}\cap P_{\mathrm{CG}}^{\mathcal{L}}$ fulfills the properties of Lemma~\ref{lem:cross-link_rounding}.
\end{restatable}
Because the polytope $P_{\mathrm{cross}}$ was obtained from $P_{\mathrm{cut}}$ and $\{0,\sfrac{1}{2}\}$-CG cuts of $P_{\mathrm{cut}}^{\mathcal{L}}$, it has facet complexity bounded by $\poly(|V|)$, and, as discussed, we can efficiently separate over it. 
Hence, it remains to show how to find the laminar family $\mathcal{L}$ and how to round any point $x\in P_{\mathrm{cross}}$ to a solution $F\subseteq L$ with $c(F)\leq c^T x + \sum_{\ell\in L_{\mathrm{in}}\setminus L_{\mathrm{up}}} c_\ell x_\ell$, where $c\in \mathbb{R}^L_{\geq 0}$ are the given costs on the links.
To this end, we introduce a new LP-based $2$-approximation algorithm for CacAP from which we can derive both a laminar family $\mathcal{L}$ as required by Proposition~\ref{prop:overview_good_laminar} and a way to round solutions in the resulting polytope $P_{\mathrm{cross}}$.

The LP we use is inspired by a classical LP-based $2$-approximation for weighted TAP, which replaces each link that is not an up-link by two up-links covering the same cuts. 
(This approach can be interpreted as a rephrasing, in terms of an LP, of a $2$-approximation of \cite{khuller_1993_approximation} based on splitting links into up-links.)
More precisely, a non-up-link $\ell=\{w,v\}$ is replaced by $\{w,u\}$ and $\{v,u\}$, where $u$ is the least common ancestor of $w$ and $v$.
This transformation will increase the cost of an optimal solution by a factor of at most two and, moreover, the canonical linear program for weighted TAP, known as the \emph{cut LP}, is integral if only up-links are present.
(This is the LP that minimizes $c^T x$ over points $x$ in the polytope $P_{\mathrm{cut}}$ as defined above.)
For CacAP, it is however not possible to replace each link by two up-links covering the same cuts.

We provide a different viewpoint on the above-described $2$-approximation for weighted TAP that extends to CacAP. Instead of splitting a non-up-link $\{w,v\}$ into two links, we replace it by two directed links $(w,v)$ and $(v,w)$, both having the same cost as $\{w,v\}$. Let
\begin{equation*}
\vec{L}\coloneqq \bigcup_{\{v,w\}\in L}\left\{(v,w),(w,v)\right\}
\end{equation*}
be the resulting set of directed links. We then consider the following linear program where a $2$-cut can only be covered by links entering it:
\begin{equation}\label{eq:dir_cut_lp}
\min \left\{ \sum_{\ell \in \vec{L}} c_\ell x_\ell \colon
 x(\delta^-_{\vec{L}}(C))  \geq 1 \;\forall C\in \mathcal{C}_G, 
x \in \mathbb{R}^{\vec{L}}_{\geq 0}\right\}\enspace.
\end{equation}
One can observe that in the special case of weighted TAP, the above directed LP indeed corresponds to splitting each link into two up-links, where each of the two directions of an original link corresponds to one of the two up-links.
Passing to directed settings to approximate undirected augmentation problems is a classical idea that has been used in related contexts (see, e.g.,  \cite{frederickson_1981_approximating,khuller_1993_approximation,khuller_1994_biconnectivity}).

Moreover, standard combinatorial uncrossing arguments that exploit that $\mathcal{C}_G$ is an intersecting family, imply integrality of the directed LP. (Contrary to TAP, the constraint matrix of the directed LP is not totally unimodular.)
\begin{restatable}{lemma}{lembidirectedintergal}\label{lem:bidirected_integral}
The linear program~\eqref{eq:dir_cut_lp} is integral for any weighted CacAP instance.
\end{restatable}
Even though the integrality gap of LP~\eqref{eq:dir_cut_lp} is~$2$, and a direct rounding of solutions of~\eqref{eq:dir_cut_lp} is thus of little interest, we show that LP~\eqref{eq:dir_cut_lp} allows for deriving a crucial upper bound on the cost needed to augment a given set of cross-links to a CacAP solution. We derive this result through a particular, well-structured solution to the dual of~\eqref{eq:dir_cut_lp}, which is stated below:
\begin{equation}\label{eq:dir_cut_dual}
\max \Bigg\{ \sum_{C\in \mathcal{C}_G} y_C \colon 
\sum_{\substack{{C\in \mathcal{C}_G}: \ell \in \delta^-_{\vec{L}}(C)}} y_C \leq c_\ell \;\forall \ell\in \vec{L}, 
y \in \mathbb{R}_{\geq 0}^{\mathcal{C}_G} \Bigg\}\enspace.
\end{equation}
In particular, standard combinatorial uncrossing techniques allow for deriving that there is an optimal solution $y\in \mathbb{R}_{\geq 0}^{\mathcal{C}_G}$ whose support $\supp(y)$ is laminar. However, this is not sufficient for our purposes, where we need an additional property that we call minimality. Minimality means that we cannot move dual weight from any set to a strictly smaller one; it is formally defined as follows.
\begin{definition}\label{def:minimal}
A solution $y\in \mathbb{R}^{\mathcal{C}_G}_{\geq 0}$ to~\eqref{eq:dir_cut_dual} is \emph{minimal} if for any $\epsilon > 0$ and any two sets $C_1,C_2 \in \mathcal{C}_G$ with $C_1\subsetneq C_2$, the point $y-\epsilon \cdot \chi^{\{C_2\}} + \epsilon \cdot \chi^{\{C_1\}}$ is not a feasible solution to the dual problem~\eqref{eq:dir_cut_dual}.
\end{definition}
Again, combinatorial uncrossing techniques imply the existence of optimal solutions to~\eqref{eq:dir_cut_dual} that are minimal and have laminar support. Even though we do not need this property later, our proof techniques show that there is a unique minimal optimal solution $y^*$ to~\eqref{eq:dir_cut_dual} with laminar support. The following result is key in our approach and shows how $y^*$ can be used to bound the cost to complete a set of cross-links to a CacAP solution.
\begin{restatable}{lemma}{lemcompletecrosslinks}\label{lem:complete_cross_links}
Let $R \subseteq L_{\mathrm{cross}}$. Then one can efficiently compute a set $F\subseteq L$ such that
\begin{enumerate}
\item $R \cup F$ is a CacAP solution, and
\item $c(F) \leq
\underset{\substack{C\in \mathcal{C}_G:\\\delta_L(C)\cap R=\emptyset}}{\sum} y^*_C$,
where $y^*$ is the optimal solution of~\eqref{eq:dir_cut_dual} that is minimal and has laminar support.
\end{enumerate}
\end{restatable}
In words, a set of cross-links $R\subseteq L_{\mathrm{cross}}$ can be completed to a CacAP solution at cost no higher than the total $y^*$-load on $2$-cuts not crossed by $R$. We reiterate that, for this result to hold, we crucially exploit that $y^*$ is minimal.
A key feature of Lemma~\ref{lem:complete_cross_links} is that it reduces the completion cost to how $R$ covers a single laminar family
\begin{equation*}
\mathcal{L} \coloneqq \supp(y^*)\enspace.
\end{equation*}
However, how cross-links can cover a single laminar family has been well-understood in the context of TAP---which is precisely the problem of covering a single laminar family---since the work of Fiorini, Gro\ss, K\"onemann, and Sanit\`a~\cite{fiorini_2018_approximating}.
We leverage this to prove Proposition~\ref{prop:overview_good_laminar}.

\subsection{Going below approximation guarantee $1.5$ through stack analysis}\label{sec:overviewBetterApprox}

To obtain approximation factors below $1.5$, we strengthen the first backbone procedure, i.e., the one based on bundle constraints and whose guarantee is stated in Lemma~\ref{lem:bundle_rounding}. We restrict ourselves to unweighted $k$-wide instances of the Cactus Augmentation Problem in this part. Let $q\in \mathbb{Z}_{\geq 1}$ be the number of principal subcacti. Whereas our algorithm is strongly inspired by the rewiring technique introduced in~\cite{grandoni_2018_improved} for TAP, our main contribution lies in a novel and significantly stronger analysis approach. We first sketch the algorithm.

Notice that the way we showed Lemma~\ref{lem:bundle_rounding}, a solution to each principal subcactus was obtained independently and the union of these solutions was returned. To improve on this, we use that after merging solutions for each principal subcactus, some cross-links $\ell$ originally used for one subcactus $i\in [q]\coloneqq \{1,\ldots, q\}$ may not be needed anymore because cross-links from solutions of principal subcacti other than $i$ may render $\ell$ redundant such that it can be deleted while maintaining feasibility.
To make sure that a significant fraction of cross-links can be deleted due to such synergies, we only consider a well-structured type of solution in each principal subcactus, which we call $L_{\mathrm{cross}}$-minimal, and which is inspired by the shadow-minimality used in~\cite{grandoni_2018_improved}. In particular,  an $L_{\mathrm{cross}}$-minimal link set $F\subseteq L$ has the property that there are no two vertices $u,v \in V$ with $u$ being an ancestor of $v$ and such that $u$ and $v$ are both endpoints of cross-links in $F$. In other words, the endpoints of $F\cap L_{\mathrm{cross}}$ do not have any ancestry relationships.%
\footnote{We provide the precise definition of $L_{\mathrm{cross}}$-minimality in Section~\ref{sec:stack_analysis}. In the context of TAP, $L_{\mathrm{cross}}$-minimality can be seen to be a weaker form of shadow-minimality, which requires that if any link in $F$ gets replaced by a shadow, i.e., a shorter link, then strictly fewer cuts are covered. In TAP, shadow-minimality is much easier to deal with than in CacAP, because, contrary to CacAP, it reduces to a pairwise notion of minimality. This is why we introduce $L_{\mathrm{cross}}$-minimality, which has crucial properties we need while allowing efficient optimization over $L_{\mathrm{cross}}$-minimal solutions in each principal subcactus in an $O(1)$-wide instance.}

Our approach is based on the polytope $P_{\mathrm{bundle}}^{\min} = P_{\mathrm{bundle}}^{\mathrm{min}}(G,L)\subseteq [0,1]^L$ consisting of all points $x\in [0,1]^L$ such that for each principal subcactus $i\in [q]$, the restriction $x_i$ of $x$ to only the links involved in subcactus $i$ is a convex combination of $L_{\mathrm{cross}}$-minimal solutions. See Figure~\ref{fig:bundle_cross_min_ex} for an example.
As we show, one can efficiently separate over $P_{\mathrm{bundle}}^{\mathrm{min}}$. Our overall approach then first obtains an optimal solution $x^*$ to the LP $\min\{1^T x : x \in P_{\mathrm{bundle}}^{\mathrm{min}} \cap P_{\mathrm{cross}}\}$, where $P_{\mathrm{cross}}$ is the polytope from Lemma~\ref{lem:cross-link_rounding}, which is needed for the other backbone procedure. We then return the better of two rounding procedures, either the one guaranteed by Lemma~\ref{lem:cross-link_rounding} or the procedure we now outline and which strengthens Lemma~\ref{lem:bundle_rounding}.

\begin{figure}[!ht]
\begin{center}
\begin{tikzpicture}[scale=0.35]

\tikzset{
  prefix node name/.style={%
    /tikz/name/.append style={%
      /tikz/alias={#1##1}%
    }%
  }
}

\tikzset{root/.style={fill=white,minimum size=13}}

\tikzset{q1/.style={line width=1pt, dash pattern=on \pgflinewidth off 1pt}}

\tikzset{q2/.style={line width=1.5pt, dash pattern=on 3pt off 2pt on \the\pgflinewidth off 2pt}}

\tikzset{q3/.style={line width=2pt, dash pattern=on 7pt off 2pt}}

\tikzset{q4/.style={line width=2.5pt}}

\newcommand\cac[2][]{

\begin{scope}[prefix node name=#1]

\tikzset{npc/.style={#2}}

\begin{scope}[every node/.append style={thick,draw=black,fill=white,circle,minimum size=6, inner sep=2pt}]
\node[root]  (1) at (18.68,-1.64) {r};
\node  (2) at (16.80,-5.50) {};
\node  (3) at (19.08,-5.16) {};
\node  (4) at (16.66,-9.60) {};
\node  (5) at (18.56,-7.16) {};
\node  (6) at (18.24,-8.74) {};
\node  (7) at (20.44,-6.76) {};
\node  (8) at (22.28,-5.68) {};
\node  (9) at (21.00,-3.58) {};
\node (10) at (20.40,-9.36) {};
\node (11) at (22.36,-9.32) {};

\begin{scope}[npc]
\node (12) at (12.24,-4.60) {};
\node (13) at (14.20,-4.64) {};
\node (14) at (11.24,-6.32) {};
\node (15) at (12.68,-7.72) {};
\node (16) at (14.80,-7.42) {};
\node (17) at (24.88,-4.72) {};
\node (18) at (27.60,-4.64) {};
\node (19) at (24.44,-7.64) {};
\node (20) at (26.80,-7.76) {};
\node (21) at (8.04,-4.56) {};
\node (22) at (6.16,-5.36) {};
\node (23) at (6.08,-7.28) {};
\node (24) at (8.00,-6.56) {};
\end{scope}
\end{scope}

\begin{scope}[very thick]

\draw  (1) --  (2);
\draw  (2) --  (3);
\draw  (3) --  (1);
\draw  (9) --  (3);
\draw  (2) to[bend left=15] (4);
\draw  (4) to[bend left=15] (2);
\draw  (3) to[bend left=25] (5);
\draw  (5) to[bend left=25] (3);
\draw  (5) to[bend left=25] (6);
\draw  (6) to[bend left=25] (5);
\draw  (3) --  (7);
\draw  (7) --  (8);
\draw  (8) --  (9);
\draw  (7) -- (10);
\draw (10) -- (11);
\draw (11) --  (7);

\begin{scope}[npc]
\draw  (1) to[bend left=2] (21);
\draw  (1) -- (12);
\draw  (1) -- (17);
\draw (21) to[bend left=2]  (1);
\draw (18) --  (1);
\draw (13) --  (1);
\draw (12) to[bend left=15] (14);
\draw (14) to[bend left=15] (12);
\draw (16) -- (13);
\draw (12) -- (13);
\draw (13) -- (15);
\draw (15) -- (16);
\draw (17) -- (18);
\draw (17) to[bend left=13] (19);
\draw (19) to[bend left=13] (17);
\draw (17) to[bend left=13] (20);
\draw (20) to[bend left=13] (17);
\draw (24) -- (21);
\draw (21) -- (22);
\draw (22) -- (23);
\draw (23) -- (24);
\end{scope}

\end{scope}

\end{scope}

}%

\begin{scope} %

\begin{scope}[scale=1.2,xshift=-1cm]
\cac[]{}
\end{scope}

\begin{scope}

\begin{scope}[red]

\begin{scope}[q1]
\draw (2) to (14);
\draw (3) to (13);
\draw (4) to (19);
\draw (4) to[bend left=20] (21);
\draw (5) to (16);
\draw (6) to (15);
\draw (7) to (14);
\draw (8) to (19);
\draw (9) to (13);
\draw (9) to (17);
\draw (9) to (18);
\draw (14) to (21);
\end{scope}

\begin{scope}[q2]
\draw (10) to[bend left=25] (23);
\end{scope}

\end{scope}

\begin{scope}[blue]

\begin{scope}[q1]
\draw (1) to[bend right=40] (4);
\draw (1) to[out=202,in=100,out looseness=2, in looseness=0.7] (14);
\draw (2) to (4);
\draw (5) to (6);
\draw (5) to (10);
\draw (7) to[bend right] (10);
\draw (7) to[bend left=20] (11);
\draw (8) to[bend right] (9);
\draw (17) to (19);
\draw (19) to (20);
\end{scope}

\begin{scope}[q2]
\draw (3) to[bend left=35] (6);
\draw (12) to (15);
\draw (21) to[out=170, in =115,looseness=1.7] (23);
\end{scope}

\begin{scope}[q3]
\draw (8) to (11);
\draw (18) to (20);
\end{scope}

\begin{scope}[q4]
\draw (15) to[bend left] (16);
\draw (22) to (24);
\end{scope}

\end{scope}

\end{scope}
\end{scope}%

\begin{scope}[shift={(35,-2)}]%
\def\ll{30mm} %
\def\vs{12mm} %

\begin{scope}[every node/.style={circle,inner sep=0pt,minimum size=5pt}]
\node[fill=red] (rd) at (0,0) {};
\node[fill=blue] (bd) at (0,-\vs) {};

\node at (rd.east)[right=2pt] {cross-links};
\node at (bd.east)[right=2pt] {in-links};
\end{scope}

\begin{scope}[yshift=-4cm]
\draw[q1] (0,0) -- +(\ll,0) node[right] {$0.25$};
\draw[q2,yshift=-\vs] (0,0) -- +(\ll,0) node[right] {$0.5$};
\draw[q3,yshift=-2*\vs] (0,0) -- +(\ll,0) node[right] {$0.75$};
\draw[q4,yshift=-3*\vs] (0,0) -- +(\ll,0) node[right] {$1$};
\end{scope}

\end{scope}

\begin{scope}[shift={(-10,-12)}] %

\tikzset{
grayout/.style={
every node/.append style={draw=black!30}, draw=black!30
},
inlink/.style={q1,blue},
crlink/.style={q1,red}
}

\begin{scope}[shift={(0,-1)}]
\cac[c1-]{grayout}
\end{scope}

\begin{scope}[crlink]
\draw (c1-3) to (c1-13);
\draw (c1-4) to[bend left=20] (c1-21);
\end{scope}

\begin{scope}[inlink]
\draw (c1-3) to[bend left=35] (c1-6);
\draw (c1-5) to (c1-10);
\draw (c1-8) to[bend right] (c1-9);
\draw (c1-8) -- (c1-11);
\end{scope}

\begin{scope}[shift={(24.5,-1)}]
\cac[c2-]{grayout}
\end{scope}

\begin{scope}[crlink]
\draw (c2-4) to (c2-19);
\draw (c2-7) to (c2-14);
\draw (c2-9) to (c2-13);
\end{scope}

\begin{scope}[inlink]
\draw (c2-3) to[bend left=35] (c2-6);
\draw (c2-8) -- (c2-11);
\draw (7) to[bend right] (10);
\end{scope}

\begin{scope}[shift={(0,-12)}]
\cac[c3-]{grayout}
\end{scope}

\begin{scope}[crlink]
\draw (c3-5) to (c3-16);
\draw (c3-8) to (c3-19);
\draw (c3-9) to (c3-18);
\draw (c3-10) to[bend left=25] (c3-23);
\end{scope}

\begin{scope}[inlink]
\draw (c3-1) to[bend right=40] (c3-4);
\draw (c3-5) to (c3-6);
\draw (c3-7) to[bend left=20] (c3-11);
\end{scope}

\begin{scope}[shift={(24.5,-12)}]
\cac[c4-]{grayout}
\end{scope}

\begin{scope}[crlink]
\draw (c4-2) to (c4-14);
\draw (c4-6) to (c4-15);
\draw (c4-9) to (c4-17);
\draw (c4-10) to[bend left=25] (c4-23);
\end{scope}

\begin{scope}[inlink]
\draw (c4-2) -- (c4-4);
\draw (c4-8) -- (c4-11);
\end{scope}

\end{scope}

\end{tikzpicture}
 \end{center}

\caption{The top picture depicts a point $x\in P_{\mathrm{bundle}}^{\min}(G,L)$. The four graphs below highlight how the restriction of $x$ to links involved in the third principal subcactus can be written as a convex combination of four $L_{\mathrm{cross}}$-minimal solutions for this principal subcactus, each with a coefficient of $\sfrac{1}{4}$.}\label{fig:bundle_cross_min_ex}
\end{figure}
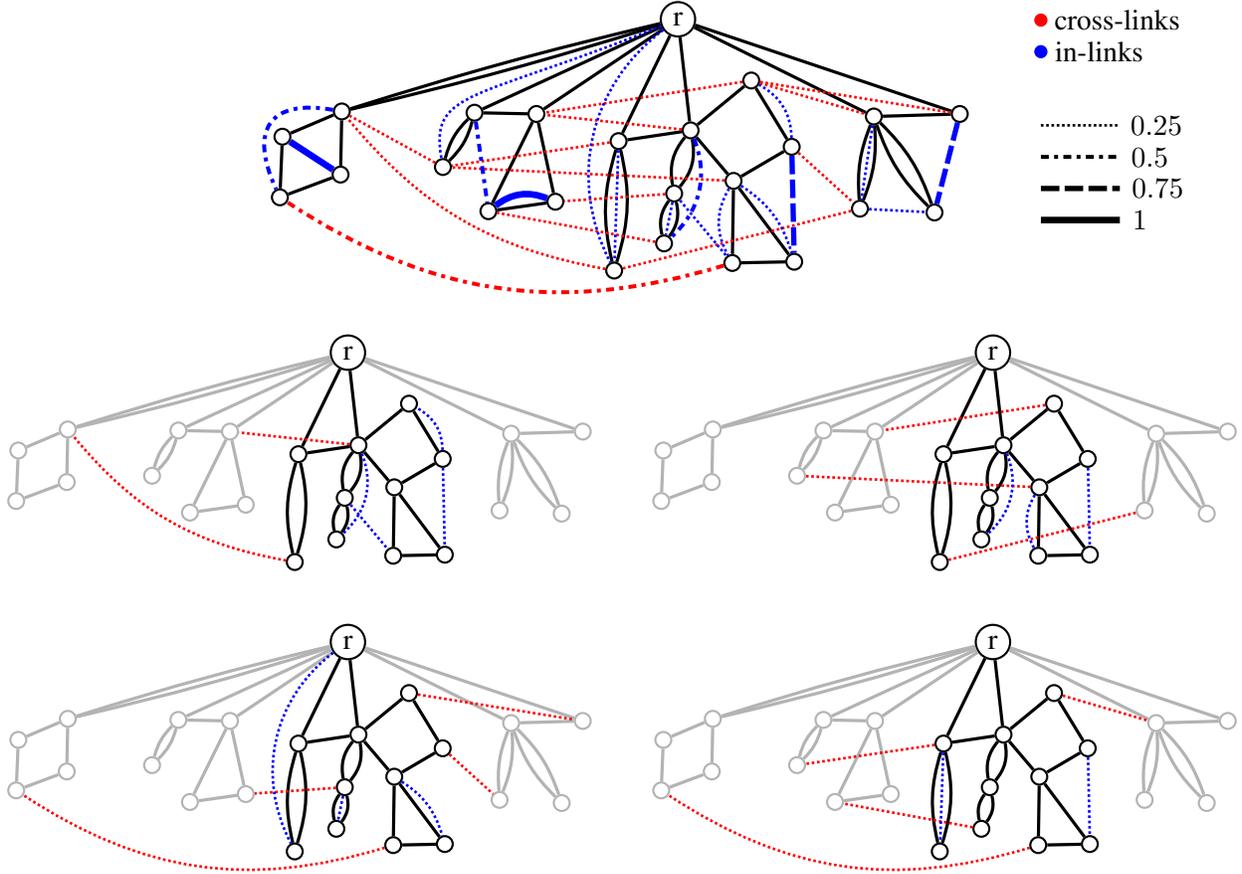

We describe a randomized rounding procedure for any point $x \in P_{\mathrm{bundle}}^{\mathrm{min}}$, which can be derandomized as we show later. First we obtain a solution for each principal subtcactus $i\in [q]$ independently as follows. We write the restriction of $x$ to the subcactus $i$ as a convex combination of $L_{\mathrm{cross}}$-minimal solutions and return one of these solutions at random with probability equal to its coefficient in the convex combination. Hence, for the third principal subcactus of the example in Figure~\ref{fig:bundle_cross_min_ex}, we can use the highlighted convex combination and pick uniformly at random one of the four illustrated solutions (because all coefficients in the convex combination are $\sfrac{1}{4}$).
By taking the (multi-)union of one randomly sampled solution $F_i\subseteq L$ for each principal subcactus $i\in [q]$, we obtain a solution $F\coloneqq \cup_{i=1}^q F_i$. Notice that because each cross-link contributes to two principal subcacti, it can be sampled by each one of them or even by both.
Figure~\ref{fig:bundle_rounding_sample} shows an example solution $F$ for the instance and the fractional point depicted in Figure~\ref{fig:bundle_cross_min_ex}. 

\begin{figure}[t]
\begin{center}
\begin{tikzpicture}[scale=0.5]

\tikzset{
  prefix node name/.style={%
    /tikz/name/.append style={%
      /tikz/alias={#1##1}%
    }%
  }
}

\tikzset{
grayout/.style={
every node/.append style={draw=black!30}, draw=black!30
}
}

\tikzset{root/.style={fill=white,minimum size=13, font=\normalsize}}

\tikzset{crlink/.style={red,stealth-,line width=1pt}}
\tikzset{inlink/.style={blue,line width=1pt}}

\newcommand\cac[2][]{

\begin{scope}[prefix node name=#1]

\tikzset{npc/.style={#2}}

\begin{scope}[every node/.append style={thick,draw=black,fill=white,circle,minimum size=12, inner sep=2pt, font=\scriptsize}]
\node[root]  (1) at (18.68,-1.64) {r};
\node  (2) at (16.80,-5.50) {};
\node  (3) at (19.08,-5.16) {};
\node  (4) at (16.66,-9.60) {};
\node  (5) at (18.56,-7.16) {};
\node  (6) at (18.24,-8.74) {};
\node  (7) at (20.44,-6.76) {};
\node  (8) at (22.28,-5.68) {};
\node  (9) at (21.00,-3.58) {};
\node (10) at (20.40,-9.36) {};
\node (11) at (22.36,-9.32) {};

\begin{scope}[npc]
\node (12) at (12.24,-4.60) {};
\node (13) at (14.20,-4.64) {};
\node (14) at (11.24,-6.32) {};
\node (15) at (12.68,-7.72) {};
\node (16) at (14.80,-7.42) {};
\node (17) at (24.88,-4.72) {};
\node (18) at (27.60,-4.64) {};
\node (19) at (24.44,-7.64) {};
\node (20) at (26.80,-7.76) {};
\node (21) at (8.04,-4.56) {};
\node (22) at (6.16,-5.36) {};
\node (23) at (6.08,-7.28) {};
\node (24) at (8.00,-6.56) {};
\end{scope}
\end{scope}

\begin{scope}[every node/.append style={font=\scriptsize}]
\foreach \i in {2,3,4,5,6,7,8,9,10,11,12,13,14,15,16,17,18,19,20,21,22,23,24} {
\pgfmathparse{int(\i-1)}
\node at (\i) {${\pgfmathresult}$};
}
\end{scope}

\begin{scope}[grayout, very thick]

\draw  (1) --  (2);
\draw  (2) --  (3);
\draw  (3) --  (1);
\draw  (9) --  (3);
\draw  (2) to[bend left=15] (4);
\draw  (4) to[bend left=15] (2);
\draw  (3) to[bend left=25] (5);
\draw  (5) to[bend left=25] (3);
\draw  (5) to[bend left=25] (6);
\draw  (6) to[bend left=25] (5);
\draw  (3) --  (7);
\draw  (7) --  (8);
\draw  (8) --  (9);
\draw  (7) -- (10);
\draw (10) -- (11);
\draw (11) --  (7);

\begin{scope}[npc]
\draw  (1) to[bend left=2] (21);
\draw  (1) -- (12);
\draw  (1) -- (17);
\draw (21) to[bend left=2]  (1);
\draw (18) --  (1);
\draw (13) --  (1);
\draw (12) to[bend left=15] (14);
\draw (14) to[bend left=15] (12);
\draw (16) -- (13);
\draw (12) -- (13);
\draw (13) -- (15);
\draw (15) -- (16);
\draw (17) -- (18);
\draw (17) to[bend left=13] (19);
\draw (19) to[bend left=13] (17);
\draw (17) to[bend left=13] (20);
\draw (20) to[bend left=13] (17);
\draw (24) -- (21);
\draw (21) -- (22);
\draw (22) -- (23);
\draw (23) -- (24);
\end{scope}

\end{scope}

\end{scope}

}%

\begin{scope} %

\begin{scope}[xshift=-1cm]
\cac[]{}
\end{scope}

\begin{scope}

\begin{scope}[crlink]
\draw (2) to (14);
\draw (6) to (15);
\draw (9) to[bend right=20] (17);
\draw (10) to[bend left=25] (23);

\draw (21) to[bend right=20] (4);

\draw (14) to (21);
\draw (13) to (9);

\draw (17) to[bend right=5] (9);
\end{scope}

\begin{scope}[inlink]

\draw (2) to (4);
\draw (8) to (11);

\draw (22) to (24);
\draw (21) to[out=170, in =115,looseness=1.7] (23);

\draw (12) to (15);
\draw (15) to[bend left] (16);

\draw (17) to (19);
\draw (18) to (20);

\end{scope}

\end{scope}
\end{scope}%

\end{tikzpicture}
 \end{center}

\caption{A possible solution obtained by taking the union of one sampled solution for each principal subcactus. In particular, for the third principal subcactus, the solution corresponding to the last picture in Figure~\ref{fig:bundle_cross_min_ex} was sampled. For illustrative purposes, cross-links are directed toward the principal subcactus which sampled them.}\label{fig:bundle_rounding_sample}
\end{figure}
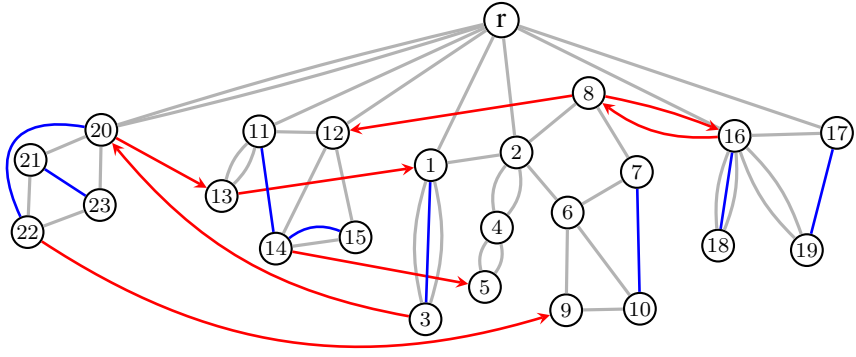

To improve a sampled solution as shown in Figure~\ref{fig:bundle_rounding_sample}, we focus on sampled cross-links and identify a maximum cardinality set of sampled cross-links $R\subseteq L_{\mathrm{cross}}$ that can be deleted without compromising feasibility. Our algorithm then returns all sampled links except for $R$. The problem of finding a largest removable set $R$ turns out to reduce to a minimum edge cover problem, which can be solved efficiently. The main difficulty lies in finding a good way to analyze the gain through this deletion, i.e., the (expected) cardinality of $R$.

We now sketch a simplified version of how we lower bound the expected cardinality of $R$, which, nevertheless, highlights key aspects. To this end, we exemplify the approach on the instance depicted in Figure~\ref{fig:bundle_cross_min_ex}, where we denote by $x\in P_{\mathrm{bundle}}^{\mathrm{min}}$ the fractional point before rounding.
We first assign each non-terminal to a terminal that is one of its descendants. Figure~\ref{fig:stack_principal_subcact_ex} shows, for one principal subcactus, an example of such an assignment. Consider now one terminal and the vertices assigned to it, say vertex~$5$ in Figure~\ref{fig:stack_principal_subcact_ex}, to which vertices $2$ and $4$ have been assigned. By the way we did the assignment, vertex $5$ together with $2$ and $4$ have a natural total order stemming from the ancestry relationship with terminal~$5$ at the bottom followed by vertex~$4$ and then vertex~$2$.

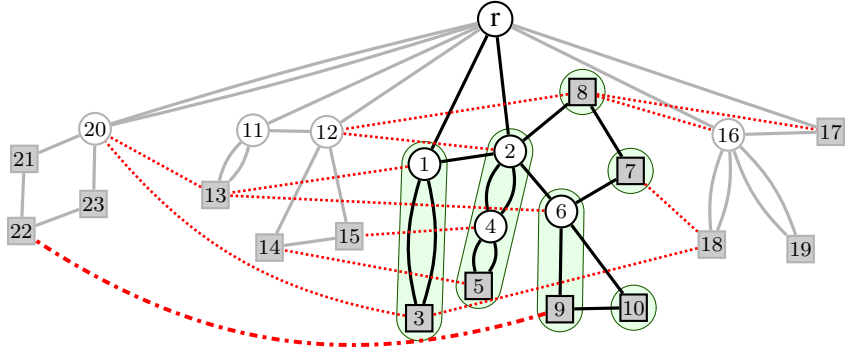
\begin{figure}[!ht]
\begin{center}
﻿\begin{tikzpicture}[scale=0.5]

\pgfdeclarelayer{bg}
\pgfdeclarelayer{fg}
\pgfsetlayers{bg,main,fg}

\tikzset{
  prefix node name/.style={%
    /tikz/name/.append style={%
      /tikz/alias={#1##1}%
    }%
  }
}

\tikzset{root/.style={fill=white,minimum size=13}}

\tikzset{q1/.style={line width=1pt, dash pattern=on \pgflinewidth off 1pt}}

\tikzset{q2/.style={line width=1.5pt, dash pattern=on 3pt off 2pt on \the\pgflinewidth off 2pt}}

\tikzset{q3/.style={line width=2pt, dash pattern=on 7pt off 2pt}}

\tikzset{q4/.style={line width=2.5pt}}

\tikzset{
term/.style={fill=black!20, rectangle, minimum size=10},
termg/.style={fill=black!20, rectangle, minimum size=10},
tf/.append style={font=\scriptsize\color{black}},
}

\newcommand\cac[2][]{

\begin{scope}[prefix node name=#1]

\tikzset{npc/.style={#2}}

\begin{scope}[every node/.append style={thick,draw=black,fill=white,circle,minimum size=12, inner sep=2pt}]
\node[root]  (1) at (18.68,-1.64) {r};
\node  (2) at (16.80,-5.50) {};
\node  (3) at (19.08,-5.16) {};
\node[term]  (4) at (16.66,-9.60) {};
\node  (5) at (18.56,-7.16) {};
\node[term]  (6) at (18.24,-8.74) {};
\node  (7) at (20.44,-6.76) {};
\node[term]  (8) at (22.28,-5.68) {};
\node[term]  (9) at (21.00,-3.58) {};
\node[term] (10) at (20.40,-9.36) {};
\node[term] (11) at (22.36,-9.32) {};

\begin{scope}[npc]
\node (12) at (12.24,-4.60) {};
\node (13) at (14.20,-4.64) {};
\node[termg] (14) at (11.24,-6.32) {};
\node[termg] (15) at (12.68,-7.72) {};
\node[termg] (16) at (14.80,-7.42) {};
\node (17) at (24.88,-4.72) {};
\node[termg] (18) at (27.60,-4.64) {};
\node[termg] (19) at (24.44,-7.64) {};
\node[termg] (20) at (26.80,-7.76) {};
\node (21) at (8.04,-4.56) {};
\node[termg] (22) at (6.16,-5.36) {};
\node[termg] (23) at (6.08,-7.28) {};
\node[termg] (24) at (8.00,-6.56) {};
\end{scope}
\end{scope}

\begin{scope}[every node/.append style={font=\scriptsize}]

\foreach \i/\t in {2/,3/,4/tf,5/,6/tf,7/,8/tf,9/tf,10/tf,11/tf,12/,13/,14/tf,15/tf,16/tf,17/,18/tf,19/tf,20/tf,21/,22/tf,23/tf,24/tf} {
\pgfmathparse{int(\i-1)}
\node[\t] at (\i) {$\pgfmathresult$};
}
\end{scope}

\begin{scope}[very thick]

\draw  (1) --  (2);
\draw  (2) --  (3);
\draw  (3) --  (1);
\draw  (9) --  (3);
\draw  (2) to[bend left=15] (4);
\draw  (4) to[bend left=15] (2);
\draw  (3) to[bend left=25] (5);
\draw  (5) to[bend left=25] (3);
\draw  (5) to[bend left=25] (6);
\draw  (6) to[bend left=25] (5);
\draw  (3) --  (7);
\draw  (7) --  (8);
\draw  (8) --  (9);
\draw  (7) -- (10);
\draw (10) -- (11);
\draw (11) --  (7);

\begin{scope}[npc]
\draw  (1) to[bend left=2] (21);
\draw  (1) -- (12);
\draw  (1) -- (17);
\draw (21) to[bend left=2]  (1);
\draw (18) --  (1);
\draw (13) --  (1);
\draw (12) to[bend left=15] (14);
\draw (14) to[bend left=15] (12);
\draw (16) -- (13);
\draw (12) -- (13);
\draw (13) -- (15);
\draw (15) -- (16);
\draw (17) -- (18);
\draw (17) to[bend left=13] (19);
\draw (19) to[bend left=13] (17);
\draw (17) to[bend left=13] (20);
\draw (20) to[bend left=13] (17);
\draw (24) -- (21);
\draw (21) -- (22);
\draw (22) -- (23);
\draw (23) -- (24);
\end{scope}

\end{scope}

\end{scope}

}%

\begin{scope}

\tikzset{
grayout/.style={
every node/.append style={draw=black!30}, draw=black!30
},
inlink/.style={q1,blue},
crlink/.style={q1,red}
}

\begin{scope}[xshift=-1cm]
\cac[]{grayout}
\end{scope}

\begin{scope}

\begin{scope}[red]

\begin{scope}[q1]
\draw (2) to (14);
\draw (3) to (13);
\draw (4) to (19);
\draw (4) to[bend left=20] (21);
\draw (5) to (16);
\draw (6) to (15);
\draw (7) to (14);
\draw (8) to (19);
\draw (9) to (13);
\draw (9) to (17);
\draw (9) to (18);
\draw (14) to (21);
\end{scope}

\begin{scope}[q2]
\draw (10) to[bend left=25] (23);
\end{scope}

\end{scope}

\end{scope}
\end{scope}%

\begin{pgfonlayer}{bg}

\def\d{6mm}

\begin{scope}[fill=green!10,draw=darkgreen]

\newcommand\blob[2]{
\pgfmathanglebetweenpoints%
{\pgfpointanchor{#1}{center}}%
{\pgfpointanchor{#2}{center}}

\edef\angle{\pgfmathresult}

\filldraw ($(#1)+(\angle+90:\d)$) arc (\angle+90:\angle+270:\d) --
($(#2)+(\angle+270:\d)$) arc (\angle+270:\angle+450:\d) -- cycle;
}

\blob{2}{4}
\blob{6}{3}
\blob{7}{10}
\blob{11}{11}
\blob{8}{8}
\blob{9}{9}

\end{scope}
\end{pgfonlayer}

\end{tikzpicture}
 \end{center}

\caption{Example of one option to define stacks in the third principal subcactus. The non-terminal vertices in each set are assigned to the single terminal in the same set, which is a descendant of all non-terminals in that set.%
}\label{fig:stack_principal_subcact_ex}
\end{figure}

We assign to each terminal what we call a \emph{stack} of cross-links. For example, the stack of terminal~$5$ contains all cross-links with one endpoint in the set $\{2,4,5\}$, i.e., these are the cross-links $\{14,5\}$, $\{15,4\}$, and $\{12,2\}$, ordered from bottom to top with respect to their endpoint in the vertex set $\{2,4,5\}$. We think of each of these links as participating in the stack in a weighted way, with weight equal to their $x$-load, which is $\sfrac{1}{4}$ for all three cross-links in this case.
We now first analyze each stack independently. To exemplify the approach, consider again the stack of terminal~$5$. Because the vertices $\{2,4,5\}$ in this stack have an ancestry relationship, and we only sample $L_{\mathrm{cross}}$-minimal solutions, the third principal subcactus will sample at most one cross-link with an endpoint in $\{2,4,5\}$. (This corresponds to having at most one incoming arc in this stack when representing sampled solutions as in Figure~\ref{fig:bundle_rounding_sample}.) We are interested in the event that the third principal subcactus samples a cross-link $\ell_1$ with an endpoint $v\in \{2,4,5\}$ in this stack and, moreover, another cross-link $\ell_2$ gets sampled (by a different subcactus) with an endpoint either at $v$ or a vertex further below in the stack of terminal~$5$. An example of such a constellation is obtained if the links $\ell_1=\{15,4\}$ and $\ell_2=\{5,14\}$ get sampled, by the third and second principal subcactus, respectively. In such a situation we say that the link $\ell_1$ is \emph{dominated} (by the link $\ell_2$). One can easily observe that any dominated cross-link can be removed from the sampled solution without compromising feasibility, because the link dominating it covers all $2$-cuts in the subcactus where the dominating link got sampled. However, when removing several dominated links simultaneously, the solution may become infeasible. Indeed, there are situations where only $\sfrac{1}{3}$ of all dominated links can be removed without violating feasibility.

Our analysis to show that, in expectation, many dominated cross-links can be removed while maintaining feasibility, can be broken down into two components.
First, we provide a lower bound for the expected total number of dominated cross-links. This can be done independently per stack and leads to a clean formula. Second, we show that, in expectation, a significant fraction of dominated cross-links can be removed while maintaining feasibility. Actually, a first easy-to-derive lower bound is that one can always remove at least $\sfrac{1}{3}$ of all dominated links in any outcome. This simple bound already leads to an approximation ratio of $1.459$ for CacAP. Later on, we show that, either, in expectation, a significantly larger fraction than $\sfrac{1}{3}$ of all dominated cross-links can be removed, or we can substantially improve the lower bound we use for the number of dominated links. This leads to the approximation factor of our main result, Theorem~\ref{thm:main}. The proof of this stronger bound is significantly more involved, as we need to analyze different events describing interactions between dominated cross-links.
Details of the above discussion can be found in Section~\ref{sec:stack_analysis}.

\section{Reducing to $\bm{k}$-wide instances}\label{sec:reduction}

In this section, we complete the proof of Theorem~\ref{thm:main_reduction}, which we sketched in Section~\ref{sec:overviewKWide}. Again, let $k=\frac{64(8+3\epsilon)}{\epsilon^2}$ be the parameter appearing in the theorem.

This section is organized as follows. We start in Section~\ref{sect:residual_instances} by providing details on $x$-heavy cut coverings. To this end, we discuss residual instances that are obtained by fixing a subset of links in a solution, like the links in an $x$-heavy cut covering. We then identify the underlying problem which later allows us to show that there exists a cheap $x$-heavy cut covering and that we can find one efficiently. This problem reduces to a particular rectangle hitting problem, the discussion of which we postpone to Section~\ref{sec:heavy_cut_covering}.
In Section~\ref{sec:reduction_splitting} we provide details  on how the splitting is performed. In particular, we prove Proposition~\ref{prop:combine_split_sol}, which bounds the number of extra links needed to obtain a feasible solution when combining solutions to two sub-instances obtained from splitting.

The results so far break the problem into $k$-wide sub-instances. Even though we assume that we are given an approximation algorithm $\mathcal{A}$ for $k$-wide instances, we are not done at this point. The reason is that we need to find solutions to the sub-instances with an additional property to be able to later control the extra cost incurred by first splitting and then merging solutions. Section~\ref{sec:reduction_algo_for_unsplittable} presents the underlying result we need to be able to control these extra costs.
Section~\ref{sec:reduction_algo_no_heavy_cuts} combines the above ingredients to show how they can be used to either get a solution of the desired quality for the original instance or to return a separating hyperplane, as needed in our round-or-cut procedure.
Finally, Section~\ref{sec:reduction_ellipsoid} presents details of the round-or-cut-procedure and completes the proof of Theorem~\ref{thm:main_reduction}.

We recall that, given a CacAP instance $\mathcal{I}=(G=(V,E),L)$ and a set $U\subseteq V$, we denote by $\delta_L(U)$ the links in $L$ with precisely one endpoint in $U$, and by $\delta_E(U)$ the edges in $E$ with precisely one endpoint in $U$.
Moreover, we denote by $E[U]$ the set of edges in $E$ that have both endpoints in $U$ and by $G[U]=(U,E[U])$ the subgraph of $G$ that is induced by $U$.

\subsection{Covering heavy cuts}\label{sect:residual_instances}

First, we discuss how we can use a cheap $x$-heavy cut covering $L_H\subseteq L$ to obtain an instance without heavy cuts.
We will fix $L_H$ to be part of the solution and consider the resulting residual instance.
For this we start by formally defining residual instances with respect to some link set $L'\subseteq L$.
This concept has been used in prior work (see~\cite{nutov_2020_approximation}), and we repeat it here and prove some basic properties of it for completeness.
We first define residual instances depending on a certain contraction order and later show that any order leads to the same outcome, which justifies to talk about ``the'' residual instance.
\begin{definition}[residual instance]
 Let $\Iscr=(G,L)$ be a CacAP instance and let $L'\subseteq L$. Let $L'=\{\ell_1, \ldots, \ell_h\}$ be a numbering (ordering) of the links in $L'$. 
 The \textit{residual instance} of $\mathcal{I}$ with respect to $L'$ and this numbering is the instance that arises by performing the following contraction operation sequentially for each link $\ell=\ell_1$ up to $\ell=\ell_h$:
contract all vertices that are on every $u$-$v$ path in the cactus, where $u$ and $v$ are the endpoints of $\ell$, into a single vertex. (See Figure~\ref{fig:residual_instance}.)
 \end{definition}
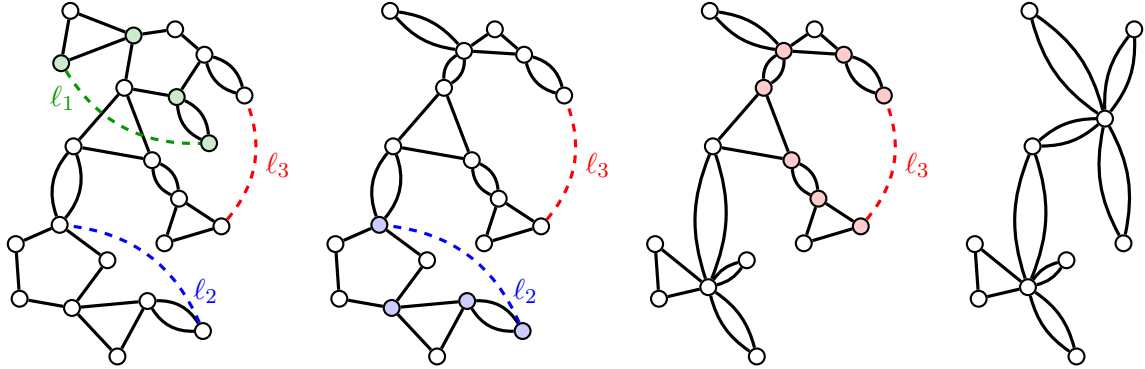
\begin{figure}[!ht]
\begin{center}
\begin{tikzpicture}[scale=0.25]

\tikzset{node/.style={thick,draw=black,fill=white,circle,minimum size=6, inner sep=2pt}}

\begin{scope}[xshift=-51cm,rotate=35]

\begin{scope}[every node/.append style=node]
\node (1) at (13,13) {};
\node (2) at (10,10) {};
\node (3) at (6,8) {};
\node (22) at (7.5,10.5) {};
\node (4) at (11,7) {};
\node (5) at (8,6) {};
\node (6) at (8.5,2.5) {};
\node (7) at (11.5,4) {};
\node (8) at (13,1) {};

\node (9) at (16,10) {};
\node (10) at (16,7.5) {};
\node (11) at (14,6) {};
\node (12) at (17,5) {};

\node (13) at (17,14) {};
\node[fill=green!60!black!20] (14) at (19,16) {};
\node (18) at (21,15) {};
\node[fill=green!60!black!20] (15) at (19,12) {};
\node (16) at (21.5,13) {};
\node (17) at (22,10) {};

\node (19) at (17,19) {};
\node[fill=green!60!black!20] (20) at (15,17) {};
\node[fill=green!60!black!20] (21) at (19,9) {};
\end{scope}

\begin{scope}[very thick]
\draw (1) to [bend right](2);
\draw (1) to [bend left](2);
\draw (2) to (22);
\draw (22) to (3);
\draw (2) to (4);
\draw (4) to (5);
\draw (3) to (5);
\draw (5) to (6);
\draw (5) to (7);
\draw (6) to (7);
\draw (7) to [bend left](8);
\draw (7) to [bend right](8);

\draw (1) to (9);
\draw (9) to[bend right] (10);
\draw (9) to[bend left] (10);
\draw (10) to (11);
\draw (11) to (12);
\draw (10) to (12);

\draw (1) to (13);
\draw (9) to (13);
\draw (13) to (14);
\draw (14) to (18);
\draw (18) to (16);
\draw (13) to (15);
\draw (15) to (16);
\draw (16) to [bend right](17);
\draw (16) to [bend left](17);

\draw (15) to[bend right] (21);
\draw (15) to[bend left] (21);
\draw (14) to (19);
\draw (19) to (20);
\draw (20) to (14);
\end{scope}

\begin{scope}[dashed, very thick]
\draw[blue] (8) to[bend right] node[pos=0.2,right]{$\ell_2$} (2);
\draw[red] (12) to [bend right] node[pos=0.5,right]{$\ell_3$} (17);
\draw[green!60!black] (20) to [bend right] node[pos=0.2,left]{$\ell_1$} (21);
\end{scope}

\end{scope}

\begin{scope}[xshift=-34cm,rotate=35]

\begin{scope}[every node/.append style=node]
\node (1) at (13,13) {};
\node[fill=blue!20] (2) at (10,10) {};
\node (3) at (6,8) {};
\node (22) at (7.5,10.5) {};
\node (4) at (11,7) {};
\node[fill=blue!20] (5) at (8,6) {};
\node (6) at (8.5,2.5) {};
\node[fill=blue!20] (7) at (11.5,4) {};
\node[fill=blue!20] (8) at (13,1) {};
\node (9) at (16,10) {};
\node (10) at (16,7.5) {};
\node (11) at (14,6) {};
\node (12) at (17,5) {};
\node (13) at (17,14) {};
\node (18) at (21,15) {};
\node (16) at (21.5,13) {};
\node (17) at (22,10) {};
\node (19) at (17,19) {};

\node (g) at (19,15) {};
\end{scope}

\begin{scope}[very thick]
\draw (1) to [bend right](2);
\draw (1) to [bend left](2);
\draw (2) to (22);
\draw (22) to (3);
\draw (2) to (4);
\draw (4) to (5);
\draw (3) to (5);
\draw (5) to (6);
\draw (5) to (7);
\draw (6) to (7);
\draw (7) to [bend left](8);
\draw (7) to [bend right](8);

\draw (1) to (9);
\draw (9) to[bend right] (10);
\draw (9) to[bend left] (10);
\draw (10) to (11);
\draw (11) to (12);
\draw (10) to (12);

\draw (1) to (13);
\draw (9) to (13);
\draw (13) to [bend right](g);
\draw (g) to (18);
\draw (18) to (16);
\draw (13) to [bend left](g);
\draw (g) to (16);
\draw (16) to [bend right](17);
\draw (16) to [bend left](17);

\draw (g) to [bend right=20](19);
\draw (19) to [bend right=20](g);

\end{scope}

\begin{scope}[dashed, very thick]
\draw[blue] (8) to[bend right] node[pos=0.2,right]{$\ell_2$} (2);
\draw[red] (12) to [bend right] node[pos=0.5,right]{$\ell_3$} (17);
\end{scope}

\end{scope}

\begin{scope}[xshift=-17cm,rotate=35]

\begin{scope}[every node/.append style=node]
\node (1) at (13,13) {};
\node (3) at (6,8) {};
\node (22) at (7.5,10.5) {};
\node (4) at (11,7) {};
\node (6) at (8.5,2.5) {};
\node[fill=red!20] (9) at (16,10) {};
\node[fill=red!20] (10) at (16,7.5) {};
\node (11) at (14,6) {};
\node[fill=red!20] (12) at (17,5) {};
\node[fill=red!20] (13) at (17,14) {};
\node (18) at (21,15) {};
\node[fill=red!20] (16) at (21.5,13) {};
\node[fill=red!20] (17) at (22,10) {};
\node (19) at (17,19) {};

\node[fill=red!20] (g) at (19,15) {};
\node (b) at (8.5,7) {};
\end{scope}

\begin{scope}[very thick]
\draw (1) to [bend right=20](b);
\draw (1) to [bend left=20](b);
\draw (b) to [bend right=20](4);
\draw (4) to [bend right=20](b);
\draw (3) to (b);
\draw (22) to (b);
\draw (3) to (22);
\draw (b) to [bend right=20](6);
\draw (6) to [bend right=20](b);

\draw (1) to (9);
\draw (9) to[bend right] (10);
\draw (9) to[bend left] (10);
\draw (10) to (11);
\draw (11) to (12);
\draw (10) to (12);

\draw (1) to (13);
\draw (9) to (13);
\draw (13) to [bend right](g);
\draw (g) to (18);
\draw (18) to (16);
\draw (13) to [bend left](g);
\draw (g) to (16);
\draw (16) to [bend right](17);
\draw (16) to [bend left](17);

\draw (g) to [bend right=20](19);
\draw (19) to [bend right=20](g);

\end{scope}

\begin{scope}[dashed, very thick]
\draw[red] (12) to [bend right] node[pos=0.5,right]{$\ell_3$} (17);
\end{scope}

\end{scope}

\begin{scope}[xshift=0cm,rotate=35]

\begin{scope}[every node/.append style=node]
\node (1) at (13,13) {};
\node (3) at (6,8) {};
\node (22) at (7.5,10.5) {};
\node (4) at (11,7) {};
\node (6) at (8.5,2.5) {};
\node (11) at (14,6) {};
\node (18) at (21,15) {};
\node (19) at (17,19) {};

\node (b) at (8.5,7) {};
\node (r) at (17,12){};
\end{scope}

\begin{scope}[very thick]
\draw (1) to [bend right=20](b);
\draw (1) to [bend left=20](b);
\draw (b) to [bend right=20](4);
\draw (4) to [bend right=20](b);
\draw (3) to (b);
\draw (22) to (b);
\draw (3) to (22);
\draw (b) to [bend right=20](6);
\draw (6) to [bend right=20](b);

\draw (1) to [bend right=20](r);
\draw (r) to [bend right=20](11);
\draw (11) to [bend right=20](r);

\draw (r) to [bend right=20](1);
\draw (r) to [bend right=20](18);
\draw (18) to [bend right=20](r);

\draw (r) to [bend right=20](19);
\draw (19) to [bend right=20](r);

\end{scope}

\begin{scope}[dashed, very thick]
\end{scope}

\end{scope}

\end{tikzpicture} \end{center}
\caption{On the leftmost graphic, a cactus and three links with a given numbering $\ell_1$, $\ell_2$, and $\ell_3$ are shown. In each of the first three graphics, the vertices that are contained in every path between the  endpoints of the link to be contracted are colored. The rightmost graphic depicts the cactus of the residual instance with respect to the three links. Notice that the ordering of the three links does not change the final result.}\label{fig:residual_instance}
\end{figure}
Notice that the graph arising from $G$ by the contractions described above is indeed again a cactus.
This justifies using the name \emph{instance} in the notion of residual instance, as we again obtain a CacAP problem after contraction. 

As mentioned, we now provide a formal proof that the contraction order is irrelevant.
\begin{lemma}\label{lem:resIndOrder}
Let $\Iscr=(G,L)$ be a CacAP instance and let $L'\subseteq L$. Then any order of the links in $L'$ leads to the same residual instance of $\Iscr$ with respect to $L'$.
\end{lemma}
\begin{proof}
Note that it suffices to prove the statement for $|L'|=2$ because we can transform a given order of $L'$ into any other order by a sequence of swaps of consecutive links in the order.
Hence, we may assume $L'=\{\{v_1,w_1\},\{v_2,w_2\}\}$.
We first observe that any residual instance of $\Iscr$ with respect to $L'$ arises from $\Iscr$ either by contracting all vertices that are contained in any $v_1$-$w_1$ path in $G$ into one vertex $\widetilde{v}_1$ and contracting all vertices that are contained in any $v_2$-$w_2$ path in $G$ into another vertex $\widetilde{v}_2$, or it arises by contracting all these vertices into a single vertex.
We will show that the first case applies if and only if there is a 2-cut $C\subseteq V$ in $G$ that separates $\{v_1,w_1\}$ from $\{v_2,w_2\}$.
This implies that the resulting residual instance is independent of the order.

Suppose the two contraction operations for the links $\{v_1,w_1\}$ and $\{v_2,w_2\}$ lead to two distinct vertices $\widetilde{v}_1$ and $\widetilde{v}_2$ in the residual instance of $\Iscr$ with respect to $L'$.
Then there is a 2-cut in the residual instance that separates $\widetilde{v}_1$ and $\widetilde{v}_2$.
This 2-cut $\widetilde{C}$ corresponds to a 2-cut $C$ in $\Iscr$ that separates $\{v_1,w_1\}$ and $\{v_2,w_2\}$.
(Here, $C\subseteq V$ is the set of vertices that are contained in $\widetilde{C}$ or were contracted into some vertex in $\widetilde{C}$.)

Now suppose there is a 2-cut $C$ in $G$ that separates $\{v_1,w_1\}$ from $\{v_2,w_2\}$, say $\{v_1,w_1\} \subseteq C \subseteq V\setminus \{v_2,w_2\}$.
Because $G$ is 2-edge-connected and $C$ is a 2-cut in $G$, the induced subgraphs $G[C]$ and $G[V\setminus C]$ are connected.
Hence, all the vertices that are contained in every $v_1$-$w_1$ path in $G$, are part of $C$.
Moreover, $G[V\setminus C]$ contains a $v_2$-$w_2$ path.
Because this path is still present after we contracted the vertices that are contained in every $v_1$-$w_1$ path in $G$, the vertex resulting from this contraction is not part of every $v_2$-$w_2$ path in the resulting cactus.
\end{proof}

We now show that a residual instance with respect to some link set $L'$ is indeed a CacAP instance whose $2$-cuts correspond precisely to the $2$-cuts in $G$ that have not been covered by $L'$.
\begin{lemma}\label{lem:residualCutCorr}
Let $\Iscr=(G=(V,E),L)$ be a CacAP instance, let $L'\subseteq L$, and let $\widetilde{\Iscr}=(\widetilde{G}=(\widetilde{V},\widetilde{E}),\widetilde{L})$ be the residual instance of $G$ with respect to $L'$.
Consider the correspondence that assigns to each $2$-cut $\widetilde{W}\subseteq \widetilde{V}$ of $\widetilde{G}$ the set of vertices $W\subseteq V$ that got contracted into some vertex of $\widetilde{W}$.
Then this correspondence is a bijection between the $2$-cuts in $\widetilde{G}$ and the $2$-cuts in $G$ that are not covered by a link in $L'$.
\end{lemma}
\begin{proof}
Note that it suffices to prove the statement for $|L'|=1$.
The general case then follows by contracting one link at a time and repeatedly invoking the result for the $|L'|=1$ case.
Hence, let $L'=\{\ell\}$ and we denote the endpoints of $\ell$ by $u$ and $v$.

Let $\widetilde{W}\subseteq \widetilde{V}$ be a $2$-cut of $\widetilde{G}$ and $W\subseteq V$ the vertices in $G$ that correspond to $\widetilde{W}$, i.e., all vertices in $V$ that got contracted into some vertex of $\widetilde{W}$.
Clearly, $\delta_E(W) = \delta_{\widetilde{E}}(\widetilde{W})$, and hence $W$ is a $2$-cut of $G$.
Moreover, both $u$ and $v$ lie on the same side of $W$ because, in $\widetilde{G}$, they have been contracted into the same vertex, which is either in $\widetilde{W}$, in which case $u,v\in W$, or in $\widetilde{V}\setminus \widetilde{W}$, in which case $u,v\not\in W$.
In either case, the cut defined by $W$ is not covered by $\ell$.

Conversely, let $W\subseteq V$ be a $2$-cut in the graph $G$ that is not covered by $\ell=\{u,v\}$, and let $Q\subseteq V$ be all vertices of $G$ that lie on each $u$-$v$ path in $G$.
Hence, this is the set of vertices that will be contracted to obtain the residual instance $\widetilde{G}$.
Because $G$ is $2$-edge-connected, the graph obtained from $G$ by removing the two edges in $\delta_E(W)$ consists of precisely two connected components, one with vertex set $W$ and one with vertex set $V\setminus W$.
Because $u$ and $v$ are either both in $W$ or $V\setminus W$, as $\ell$ does not cover the cut $W$, we either have $Q\subseteq W$ or $Q\subseteq V\setminus W$.
If $Q\subseteq W$, then the set $\widetilde{W}\subseteq\widetilde{V}$ consisting of all vertices of $W\setminus Q$ together with the contracted vertex corresponding to $Q$ is the $2$-cut $\widetilde{W}$ in $\widetilde{G}$ that corresponds to the $2$-cut $W$ in $G$.
Otherwise, if $Q\subseteq V\setminus W$, we choose $\widetilde{W} = W$ and again obtain that the $2$-cut $W$ corresponds to $\widetilde{W}$, as desired.
\end{proof}

Lemma~\ref{lem:residualCutCorr} immediately implies that the residual instance allows for precisely describing the link sets that lead to a feasible solution together with the links that have been fixed to obtain the residual instance.
\begin{corollary}\label{cor:residual_instance}
 Let $\mathcal{I}=(G=(V,E),L)$ be a CacAP instance and let $L' \subseteq L$.
 Then $F \subseteq L$ is a feasible solution to the residual instance of $\Iscr$ with respect to $L'$
 if and only if $L' \cup F$ is a feasible solution of $\Iscr$.
\end{corollary}

For convenience, we recall some terminology that we introduced in Section~\ref{sec:overviewKWide}. Namely, we fixed an arbitrary vertex $r\in V$, called the \emph{root}, and denote by 
\begin{equation*}
 \mathcal{C}_G=\{ C \subseteq V\setminus \{r\} : |\delta_E(C)| = 2\}
\end{equation*}
the set of $2$-cuts in $G$ not containing $r$.
Note that a set $F\subseteq L$ is a feasible solution if and only if $|\delta_{F}(C)| \geq 1$ for all $C\in \mathcal{C}_{G}$. We say that a link $\ell\in L$ covers a cut $C\in \mathcal{C}_G$ if $\ell \in \delta_L(C)$.
Furthermore, we recall the notion of $x$-light and $x$-heavy cuts.
\begin{definition}[light and heavy cuts]
 Let $x \in [0,1]^L$. A cut $C\in \mathcal{C}_G$ and the set $\delta_L(C)$ of links in the cut are called \emph{$x$-light} if $x(\delta_{L}(C)) \le \frac{16}{\epsilon}$. Otherwise, $C$ and $\delta_L(C)$ are called \emph{$x$-heavy}.
\end{definition}

Recall that we want to use a cheap $x$-heavy cut covering to get rid of $x$-heavy cuts.
We now show that any link set covering all $x$-heavy cuts indeed leads to a residual instance without $x$-heavy cuts.
\begin{lemma}\label{lem:all_light_after_heavy_cut_covering}
Let $L_H\subseteq L$ be a set of links such that $|\delta_{L_H}(C)| \ge 1$ for every $x$-heavy cut $C\in \mathcal{C}_G$.
Then the residual instance with respect to $L_H$ has no $x$-heavy cuts.
\end{lemma}
\begin{proof}
Consider a 2-cut $C'$ in the residual instance.
Let $C$ be the set of vertices from which $C'$ arose by contractions in the construction of the residual instance.
Then $C\in \Cscr_G$ and  $\delta_L(C) = \delta_L(C')$. 
Moreover, for every link $\ell \in L_H$, the endpoints of $\ell$ have been contracted and thus $\delta_L(C) \cap L_H = \delta_L(C')\cap L_H = \emptyset$. 
Because $L_H$ covers all heavy cuts, this implies that the cut $C$ is not heavy in the original instance.
Hence, $C'$ is not heavy in the residual instance.
\end{proof}

Finally, the following theorem implies the existence of a cheap $x$-heavy cut 
covering as introduced in Definition~\ref{def:x_heavy_cut_covering}, and also implies that it can be found efficiently. It will be proven in Section~\ref{sec:heavy_cut_covering} through a reduction to a rectangle hitting problem.
\begin{restatable}{theorem}{theoremcoveringheavy}\label{thm:covering_heavy}
Let $\Iscr = (G=(V,E), L)$ be a CacAP instance, let $r\in V$, and let $\mathcal{W}\subseteq \mathcal{C}_G$.
Then the LP
\begin{equation}\label{eq:heavyCutHitLP}
\renewcommand\arraystretch{1.5}
\begin{array}{r>{\displaystyle}rcll}
\min & x(L)  & & & \\
& x(\delta_L(W)) & \ge &1 &\forall\; W\in \Wscr \\
& x & \in & \mathbb{R}^L_{\geq 0}\makebox[0pt][l]{}  &
\end{array}
\end{equation}
has integrality gap at most 8.
Moreover, given a solution $x$ to~\eqref{eq:heavyCutHitLP}, we can compute an integral solution with objective value at most $8 \cdot x(L)$ in polynomial time.
\end{restatable}
To obtain a cheap $x$-heavy cut covering, we apply the above theorem with $\Wscr = \{ C \in \mathcal{C}_G\colon C\text{ is heavy}\}$. Then, by definition of heavy cuts, $\frac{\epsilon}{16}\cdot x$ is a feasible solution to LP~\eqref{eq:heavyCutHitLP}. Thus we get an integral solution whose support $L_H\subseteq L$ is an $x$-heavy cut covering satisfying $|L_H| \le \frac{\epsilon}{2} \cdot x(L)$, as desired.

\subsection{Splitting at light cuts}\label{sec:reduction_splitting}

As discussed, we now move to the splitting procedure. We start with a basic structural lemma about crossing $2$-cuts which we use later to prove Proposition~\ref{prop:combine_split_sol}. This statement follows in a straightforward way from known results on minimum cuts and very basic additional observations. We provide a proof just for completeness.
\begin{lemma}\label{lem:crossing_2_edge_cuts}
Let $A,B \subsetneq V$ be 2-cuts of $G$.
If $A$ and $B$ cross, i.e., $A\cap B, A\setminus B, B\setminus A$, and $V\setminus (A\cup B)$ are nonempty,
then the following holds:
\begin{enumerate}
\item $G - (\delta_E(A) \cup \delta_E(B))$ has exactly four connected components.
        The vertex sets of these connected components are $A\cap B, A\setminus B, B\setminus A$, and $V\setminus (A\cup B)$. Each of them is a 2-cut in $G$.
        \label{item:four_connected_comp}
\item $\delta_E(A)$ contains an edge in $E[B]$ and an edge in $E[V\setminus B]$.
        \label{item:cut_edges_separated}
\end{enumerate}
\end{lemma}
\begin{proof}
\begin{enumerate}
\item Because $G$ is 2-edge-connected and $|\delta_E(A) \cup \delta_E(B)| \le 4$, the graph  $G - (\delta_E(A) \cup \delta_E(B))$ has at most four connected components. Each of these connected components is a subset of  $A\cap B, A\setminus B, B\setminus A$, or $V\setminus (A\cup B)$.
If $A$ and $B$ cross, these four sets are all nonempty and hence $G - (\delta_E(A) \cup \delta_E(B))$ has exactly four connected components. Then the vertex sets of these connected components are  $A\cap B, A\setminus B, B\setminus A$, and $V\setminus (A\cup B)$.
Moreover, by submodularity and symmetry of the cut function $S \mapsto |\delta_E(S)|$ for $S\subseteq V$, we have
\[
2+2 \le |\delta_E(A\cap B)| + | \delta_E(V\setminus(A\cup B))| \le |\delta_E(A)| + |\delta_E(B)| = 4\enspace,
\]
implying $ |\delta_E(A\cap B)| = | \delta_E(V\setminus(A\cup B))| =2$.
Similarly,
\[
2+2\le |\delta_E(A \setminus B)| + |\delta_E(B\setminus A)|  \le |\delta_E(A)| + |\delta_E(B)| = 4\enspace,
\]
implying $ |\delta_E(A \setminus B)| = |\delta_E(B\setminus A)| =2$.
\item Because $G$ is 2-edge-connected and $|\delta_E(B)|=2$, the induced subgraphs $G[B]$ and $G[V\setminus B]$ are connected.
Therefore, if $A$ and $B$ cross, there is an edge from $A\cap B$ to $B\setminus A$.
Similarly, there is an edge from $V\setminus (A\cup B)$ to $A\setminus B$.
\end{enumerate}
\end{proof}

We are now ready to prove Proposition~\ref{prop:combine_split_sol}, which bounds the number of extra links needed to complete two solutions of split problems into one for the original instance. The statement is repeated below. We recall that $\mathcal{I}_C$ and $\mathcal{I}_{V\setminus C}$ are the sub-instances of $\mathcal{I}$ obtained by contracting $V\setminus C$ and $C$, respectively.
\propcombinesplitsol*
\begin{proof}
Let $\delta_E(C)=\{e_1,e_2\}$. We are interested in $2$-cuts $W\subseteq V$ of $G$ that are not covered by any link in $F_C \cup F_{V\setminus C}$. Notice that any such $2$-cut $W$ must cross $C$; for otherwise, either $W$ or $V\setminus W$ is fully contained in $C$ or $V\setminus C$, and is thus also a $2$-cut in one of the sub-instances $\mathcal{I}_C$ or $\mathcal{I}_{V\setminus C}$, which implies that it must be covered by at least one link in $F_C \cup F_{V\setminus C}$.
By applying Lemma~\ref{lem:crossing_2_edge_cuts}~\ref{item:cut_edges_separated} to $C$ and $W$, we obtain that one of the edges among $e_1, e_2$ has both endpoints in $W$ and the other one has no endpoint in $W$.
Let $\mathcal{W}\subseteq 2^V$ be the family of all $2$-cuts $W\subseteq V$ that are not covered by any link in $F_C \cup F_{V\setminus C}$ and contain both endpoints of $e_1$, i.e.,
\begin{equation*}
\mathcal{W} = \left\{W \subseteq V: e_1\in E[W], |\delta_E(W)| =2, \delta_L(W) \cap \left( F_C \cup F_{V\setminus C}\right) = \emptyset \right\}\enspace,
\end{equation*}
and, consequently, each set in $\mathcal{W}$ contains none of the endpoints of $e_2$.
Then for every 2-cut $W$ that is not covered by any link in $F_C \cup F_{V\setminus C}$, either $W$ or $V\setminus W$ is contained in $\Wscr$.
Hence, our goal is to show that we can efficiently compute a set $F\subseteq L$ of at most $|\delta_L(C)\cap F_C|-1$ links covering all cuts in $\mathcal{W}$, which implies that $F_C\cup F_{V\setminus C}\cup F$ is a CacAP solution for $\mathcal{I}$. We prove this by showing first that
\begin{enumerate}
\item $\mathcal{W}$ is a chain, which we then use to show
\item\label{item:W_small_card} $|\mathcal{W}|\leq |\delta_L(C)\cap F_C|-1$.
\end{enumerate}
Note that this indeed implies Proposition~\ref{prop:combine_split_sol} because we can choose for each cut $W\in \mathcal{W}$ a link $\ell_W\in L$ that covers $W$ and return $F = \{\ell_W : W\in \mathcal{W}\}$. 
Because $\mathcal{I}$ is a feasible CacAP instance, such links $\ell_W$ for $W\in \mathcal{W}$ exist.
Moreover, $F$ clearly covers all cuts in $\mathcal{W}$ and fulfills $|F|\leq |\mathcal{W}| \leq |\delta_L(C)\cap F_C|-1$ due to~\ref{item:W_small_card}.
Finally, $F$ can be computed efficiently, for example by starting with $F=\emptyset$ and successively adding a link to $F$ that covers a $2$-cut in $G$ not covered so far by $F_C\cup F_{V\setminus C}\cup F$.

We start by showing that $\mathcal{W}$ is a chain. Because each set $W\in \mathcal{W}$ contains both endpoints of $e_1$ and none of $e_2$ (and thus they have a common vertex and a vertex not contained in any of them), this is equivalent to showing that no two sets in $\mathcal{W}$ cross. By sake of contradiction, assume that $W,W'\in \mathcal{W}$ are crossing sets.
Then by Lemma~\ref{lem:crossing_2_edge_cuts}~\ref{item:four_connected_comp}, $G- (\delta_E(W)\cup\delta_E(W'))$ has four connected components and the vertex set of each of these components is a $2$-cut in $G$.
At least two of these connected components contain neither both endpoint of $e_1$ nor both endpoints of $e_2$.
Let $Z$ be the vertex set of such a component.
Then $Z\subseteq C$ or $Z\subseteq V\setminus C$.
Therefore, 
\begin{equation*}
\emptyset \ne (F_C \cup F_{V\setminus C}) \cap \delta_L(Z) \subseteq 
 (F_C \cup F_{V\setminus C}) \cap (\delta_L(W) \cup \delta_L(W'))\enspace,
\end{equation*}
contradicting $W, W' \in \mathcal{W}$.
Hence, $\mathcal{W}$ is a chain as claimed.

We now proceed with bounding the cardinality of $\mathcal{W}$. For this we first show a further structural result about the sets in $\mathcal{W}$. Let $\mathcal{W}=\{W_1, W_2, \ldots, W_m\}$ with $W_1 \subsetneq W_2 \subsetneq \dots \subsetneq W_m$. Moreover, we define $W_0 \coloneqq \emptyset$ and $W_{m+1} \coloneqq V$. An illustration of cuts $W_i$ is in Figure~\ref{fig:chain_cut_not_covered}.
\begin{figure}[!ht]
\begin{center}
\begin{tikzpicture}[scale=0.35]

\tikzset{
  prefix node name/.style={%
    /tikz/name/.append style={%
      /tikz/alias={#1##1}%
    }%
  }
}

\tikzset{link/.style={line width=1.5pt}}

\tikzset{node/.style={thick,draw=black,fill=white,circle,minimum size=6, inner sep=2pt}}

\newcommand\leftpart[2][]{
\begin{scope}[prefix node name=#1]

\begin{scope}[every node/.append style=node]
\node (1) at (12,19) {};
\node (12) at (9,13) {};
\node (13) at (10,16) {};
\node  (2) at (10,10) {};
\node  (3) at (12,6) {};
\node  (4) at (9,7) {};
\node  (5) at (7,4.5) {};
\node  (6) at (10.5,4) {};
\end{scope}

\begin{scope}[very thick]
\draw (1) --(13) -- (12) -- (2) --(3);
\draw (3) --(4) -- (5) -- (6) --(3);
\end{scope}

\begin{scope}[orange, link]
\draw (4) -- (6);
\end{scope}

\end{scope}
}%

\newcommand\rightpart[2][]{
\begin{scope}[prefix node name=#1]

\begin{scope}[every node/.append style=node]
\node  (7) at (18,18) {};
\node  (8) at (18,7) {};
\node  (9) at (17,21) {};
\node (10) at (20,3.5) {};
\node (11) at (17,4) {};
\node (14) at (19,10) {};
\node (15) at (19,14) {};
\end{scope}

\begin{scope}[very thick]
\draw (7) -- (15) -- (14) -- (8);
\draw  (7) to[bend left=15] (9);
\draw  (9) to[bend left=15] (7);
\draw (8) -- (11) -- (10) -- (8);
\end{scope}

\begin{scope}[purple, link]
\draw[bend left] (8) to (10);
\end{scope}

\end{scope}
}%

\begin{scope}
\leftpart[o-]{}
\rightpart[o-]{}
\end{scope}

\draw[bend right=10,line width=3, gray, opacity=0.5] (15,1) to (15,23);
\node[left] (o-C) at (15,1.5) {$C$};
\draw[bend left=10,line width=3, blue, opacity=0.5] (5,12) to (22.5,6.5);
\draw[bend left=10,line width=3, blue, opacity=0.5] (6,15) to (22,11);
\draw[bend left=10,line width=3, blue, opacity=0.5] (7,18) to (21.5,15);
\node[right] (o-W1) at (20,14) {$W_3$};
\node[right] (o-W2) at (20.5,10) {$W_2$};
\node[right] (o-W3) at (21,5.5) {$W_1$};

\begin{scope}[very thick]
\draw (o-7) -- (o-1);
\draw (o-8) -- (o-3);
\node[right] (e1) at (15,19) {$e_2$};
\node[right] (e2) at (13,7) {$e_1$};
\end{scope}

\begin{scope}[link]
\draw[orange] (o-2) -- (o-8);
\draw[bend left=20, purple] (o-1) to (o-9);
\draw[bend right=20, orange] (o-1) to (o-9);
\draw[bend left=12, purple] (o-13) to (o-15);
\draw[bend right=12, orange] (o-12) to (o-14);
\draw[bend left=10, purple] (o-12) to (o-14);
\draw[bend right=10, orange] (o-13) to (o-15);
\draw[bend left=45,orange] (o-8) to (o-5);
\draw[bend left=25, purple] (o-11) to (o-6);
\end{scope}

\begin{scope}[shift={(-28,0)}]
\leftpart[s-]{}
\node[node,fill=black]  (cl) at (15,12) {};
\node[left] (o-C) at (14.5,2) {$\mathcal{I}_{C}$};
\end{scope}

\begin{scope}[shift={(-23,0)}]
\rightpart[s-]{}
\node[node,fill=black]  (cr) at (15,12) {};
\node[right] (o-V-C) at (16,2) {$\mathcal{I}_{V\setminus C}$};
\end{scope}

\begin{scope}[very thick]
\draw (s-1) --(cl);
\draw (s-3) --(cl);
\draw (s-7) --(cr);
\draw (s-8) --(cr);
\end{scope}

\begin{scope}[link,orange]
\draw (s-2) -- (cl);
\draw (s-13) -- (cl);
\draw (s-12) -- (cl);
\draw[bend left=35] (s-1) to (cl);
\draw[bend left=100] (cl) to (s-5);
\end{scope}

\begin{scope}[link,purple]
\draw (s-14) -- (cr);
\draw (s-15) -- (cr);
\draw[bend left=20] (cr) to (s-9);
\draw[bend left=45] (s-11) to (cr);
\end{scope}

\end{tikzpicture} \end{center}
\caption{Feasible solutions of the sub-instances $\mathcal{I}_C$ and $\mathcal{I}_{V\setminus C}$, respectively. Their union (on the right) does not cover the cuts $W_1$, $W_2$, and $W_3$ (shown in blue).
These cuts form a chain. Moreover, each of the sets $L[W_1]$, $L[W_2\setminus W_1]$, $L[W_3 \setminus W_2]$, and $L[V\setminus W_3]$ contains a link in $F_C \cap \delta_L(C)$ (and a link in $F_{V\setminus C} \cap \delta_L(C)$).}\label{fig:chain_cut_not_covered}
\end{figure}
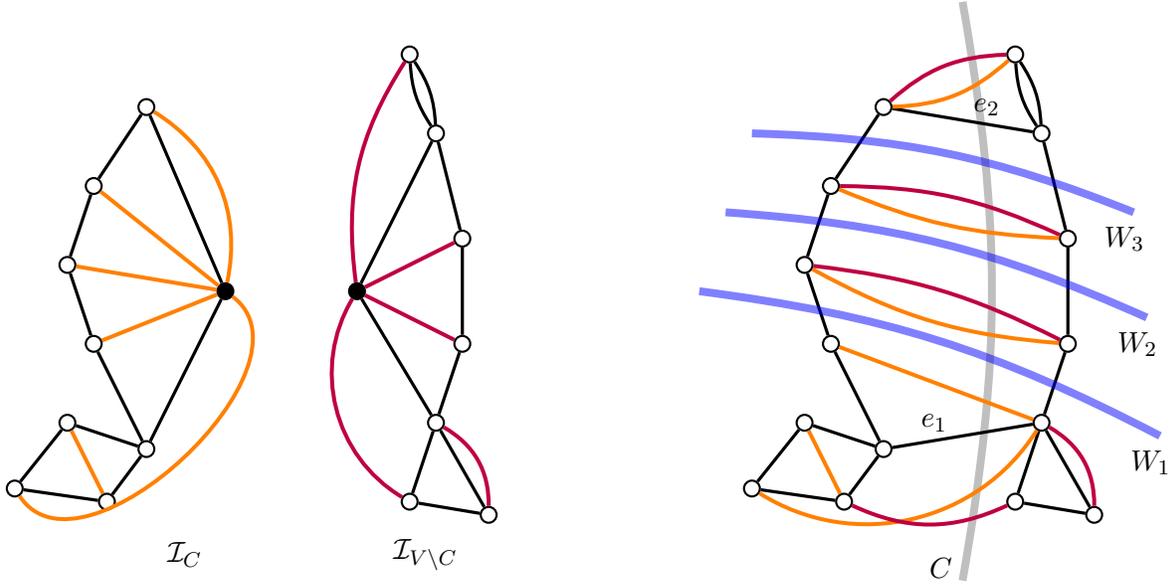
\enlargethispage{-2\baselineskip}
\begin{claim}\label{claim:sets_in_between}
Let $i\in \{0,\dots,m\}$.
\begin{itemize}
\item $(W_{i+1} \setminus W_i) \cap C$ is either empty or a 2-cut in $G$.
\item $(W_{i+1} \setminus W_i) \cap (V\setminus C)$ is either empty or a 2-cut in $G$.
\end{itemize}
\end{claim}
\begin{proof}
Note that the claim clearly holds for $i=0$ and $i=m$, in which case we have $W_1\setminus W_0 = W_1$ and $W_{m+1} \setminus W_m = V\setminus W_m$, and the result immediately follows from Lemma~\ref{lem:crossing_2_edge_cuts}~\ref{item:four_connected_comp} and the fact that $C$ crosses both $W_1$ and $W_m$.

Hence, let $i\in \{1,\ldots, m-1\}$.
Observe that $\delta_E(C) \cap E[W_{i+1}\setminus W_i] = \emptyset$, because $\delta_E(C)=\{e_1, e_2\}$ and each $W\in \mathcal{W}$ contains both endpoints of $e_1$ and none of the endpoints of $e_2$.
Thus,
\begin{equation*}
\left|\delta_E\big((W_{i+1}\setminus W_i)\cap C\big)\right|\ \le\ 
\big|\delta_E(W_{i+1})\setminus E[V\setminus C]\big| + \big|\delta_E(W_{i})\setminus E[V\setminus C]\big|\enspace.
\end{equation*}
Moreover, by Lemma~\ref{lem:crossing_2_edge_cuts}~\ref{item:cut_edges_separated} applied to $A=W_i$ and $B=C$, we have $|\delta_E(W_{i})\cap E[V\setminus C]| \ge 1$.
Thus,
\[
  \big|\delta_E(W_{i})\setminus E[V\setminus C]\big|\ = \big|\delta_E(W_{i})\big| -  \big|\delta_E(W_{i})\cap E[V\setminus C]\big| 
  \ \le\ 2 -1\ =\ 1\enspace.
\]
Analogously, we get $|\delta_E(W_{i+1})\setminus E[V\setminus C]| \le 1$.
Hence, $|\delta_E\left((W_{i+1}\setminus W_i)\cap C\right)| \le 2$.
By symmetry, also $|\delta_E\left((W_{i+1}\setminus W_i)\cap (V\setminus C)\right)| \le 2$.
Because $G$ is 2-edge-connected, this completes the proof.
\end{proof}

We now refine the above claim by showing that $(W_{i+1}\setminus W_i)\cap C \neq \emptyset$ for all $i\in \{0,\ldots, m\}$. (This holds analogously for $(W_{i+1}\setminus W_i)\cap (V\setminus C)$, though we do not need this fact.)
Let $i\in \{0,\ldots, m\}$ and assume for the sake of deriving a contradiction that $(W_{i+1}\setminus W_i)\subseteq V\setminus C$.
Then by Claim~\ref{claim:sets_in_between}, $W_{i+1}\setminus W_i$ induces a $2$-cut in $G$.
Therefore,
\[
\emptyset \ne F_{V\setminus C} \cap \delta_L(W_{i+1}\setminus W_i) \subseteq F_{V\setminus C} \cap ( \delta_L(W_i) \cup \delta_L(W_{i+1}))\enspace,
\]
which contradicts the fact that $\delta_L(W_j) \cap (F_C \cup F_{V\setminus C})=\emptyset$ for all $j\in \{0,\ldots, m\}$.
We thus obtain $(W_{i+1}\setminus W_i)\cap C \ne \emptyset$ as desired.

We now show that for all $i\in \{0,\dots,m\}$, the set $F_C$ contains a link in $\delta_L(C) \cap L[W_{i+1}\setminus W_i]$.
By Claim~\ref{claim:sets_in_between}, $(W_{i+1}\setminus W_i)\cap C$ induces a $2$-cut in $G$ and hence
$F_C \cap \delta_L( (W_{i+1}\setminus W_i)\cap C )$ contains a link $\ell$.
Because $\delta_L(W_i)$ and  $\delta_L(W_{i+1})$ do not contain any link of $F_C$, the link $\ell$ must be contained in
$L[W_{i+1}\setminus W_i] \cap \delta_L(C)$.
For different $i\in  \{0,\dots,m\}$ these links are distinct, implying $|F_C \cap \delta_L(C)| \ge m+1 = |\mathcal{W}| + 1$, and thus finishing the proof.
\end{proof}

As sketched in Section~\ref{sec:overviewKWide}, to control the extra cost incurred by splitting a CacAP problem, and merging independent solutions for each of the two sub-instances, we have to carefully choose the cuts $C\in \mathcal{C}_G$ on which we split. On the one hand we want to have $x$-light cuts because links crossing the cut appear in both sub-instances. Moreover, to amortize the cost incurred by doubling the links crossing $C$, we only perform splittings on big cuts, which, as we recall for convenience, are defined as follows, where $T=\{v\in V: |\delta_E(v)|=2\}$ are all terminals of $G$.
\begin{definition}[small and big sets]
 We call a set $C\subseteq V$ \emph{small} if $|C\cap T| \le \frac{k}{2}$.
 Otherwise $C$ is called \emph{big}.
\end{definition}
The idea is that any solution to a CacAP instance with $k$-many terminals needs at least $\sfrac{k}{2}$-many links, because each terminal must be the endpoint of at least one link. 
Hence, when splitting on a big cut $C\in \mathcal{C}_G$, the lower bound of $\frac{1}{2}|C\cap T|$ on any solution to $\mathcal{I}_C$ allows for amortizing the doubling of $x$-value on links crossing $C$ due to the splitting. Indeed, we have $x(\delta_L(C)) \le \sfrac{16}{\varepsilon} = O(\varepsilon |C\cap T|)$.

Before we discuss how to handle instances that cannot be split anymore, we formally define the notions of being \emph{splittable} at a certain $2$-cut and of an \emph{unsplittable instance}.

\begin{definition}[splittable instances]
Let $\mathcal{I}=(G=(V,E),\mathcal{I})$ be a CacAP instance with root $r\in V$ and $x\in [0,1]^L$. 
Then $\mathcal{I}$ is \emph{splittable} at $C\in \mathcal{C}_G$ if $C$ is $x$-light and big and $C \ne V\setminus \{r\}$.
\end{definition}

The reason why we disallow splitting at $C = V\setminus \{r\}$, even if  $V\setminus\{r\}$ is a big $x$-light 2-cut, is that
for $C=V\setminus\{r\}$, the instance $\Iscr_C$ is identical to $\Iscr$  (while $\Iscr_{V\setminus C}$ is a trivial instance where the cactus has two vertices only).
Hence, we do not make any progress by splitting at $C=V\setminus\{r\}$.

\begin{definition}
An instance $\mathcal{I}=(G=(V,E),L)$ of CacAP with root $r\in V$ is \emph{unsplittable} if all $2$-cuts in $\mathcal{C}_G$ are light and there is no $C\in \mathcal{C}_G$ such that $\mathcal{I}$ is splittable at $C$.
\end{definition}

\subsection{An algorithm for unsplittable instances}\label{sec:reduction_algo_for_unsplittable}

So far, we discussed key ingredients and results to split instances. We now focus on instances that are unsplittable.
In the next subsection we will then explain how we can use a splitting procedure to extend the procedure we discuss here to general instances.

To later prove Theorem~\ref{thm:main_reduction}, we will use the technique of round-or-cut. 
Recall that
\begin{equation*}
P_{\mathrm{CacAP}}(\mathcal{I}) \coloneqq \conv(\{\chi^F : F\subseteq L, (V,E\cup F) \text{ is $3$-edge-connected}\})\enspace
\end{equation*}
denotes the convex hull of incidence vectors of feasible solutions of the instance $\Iscr=(G=(V,E),L)$.
We will show that given an instance $\Iscr=(G,L)$ of CacAP and a vector $x\in [0,1]^L$, we can either compute a feasible solution $F$ with $|F|\le (\alpha + \epsilon) \cdot x(L)$ or
a vector $w\in \mathbb{R}^L$ such that $w^T x < w^T x'$ for all $x' \in P_{\mathrm{CacAP}}(\mathcal{I})$.
In this section we describe such an algorithm for unsplittable instances.

We start by proving Lemma~\ref{lem:unsplittable}, which we repeat below for convenience.
\lemunsplittable*
\begin{proof}
Let $\Iscr=(G=(V,E),L)$ with root $r\in V$ be an unsplittable CacAP instance.
Because $G$ is a cactus, either the vertex set of each connected component of $G-r$ is contained in $\Cscr_G \setminus \{V\setminus\{r\}\}$, or $G-r$ has exactly one connected component with vertex set $V\setminus \{r\}$.
In the first case, all the vertex sets of connected components of $G-r$ are small, because $\Iscr$ is unsplittable.
Thus, in this case $G$ is $k$-wide.

Now consider the case when $G-r$ is connected.
Then $r$ has degree $2$ in $G$ and, consequently, is contained in only one cycle of the cactus $G$.
We again distinguish two cases. 
Let us first consider the case where the cycle containing $r$ is of length at least three. Let $v_1$ and $v_2$ denote the two neighbors of $r$ in the cycle containing $r$ and let
\begin{align*}
V_1 :=& \{v\in V\setminus\{r,v_2\}: v\text{ is reachable from }v_1\text{ in }G- v_2\}\enspace, \text{ and}\\
V_2 :=& \{v\in V\setminus\{r,v_1\}: v\text{ is reachable from }v_2\text{ in }G- v_1\}\enspace.
\end{align*}
We claim that $|\delta_E(V_1)| = |\delta_E(V_2)|=2$ and $V_1 \cup V_2 = V\setminus \{r\}$.
Then $V_1$ and $V_2$ must be small and thus $|T \cap(V\setminus \{r\})| \le |T\cap V_1| + |T\cap V_2| \le k$.

To prove  $V_1 \cup V_2 = V\setminus \{r\}$, we consider a vertex $v\in V\setminus \{r\}$ and a $v$-$r$ path in $G$. 
Then $P$ visits $v_1$ or $v_2$ because these are the only neighbors of $r$.
If $v_1$ is the first vertex in $\{v_1,v_2\}$ visited by $P$,
the subpath of $P$ from $v$ to $v_1$ is contained in $G- v_2$ and hence $v\in V_1$.
Otherwise, an analogous argument shows $v\in V_2$.

Now we show $|\delta_E(V_1)|=2$, which, by symmetry, implies $|\delta_E(V_2)|=2$. 
Suppose $\delta_E(V_1)$ contains two distinct edges $\{a,v_2\},\{b,v_2\}$ incident with $v_2$.
Then $a,b\in V_1$ and $G-v_2$  contains a $v_1$-$a$ path $P_a$ and a $v_1$-$b$ path $P_b$.
This implies that the path $P_a$ together with $\{a,v_2\},\{v_2,r\},\{r,v_1\}$ and the path
$P_b$ together with $\{b,v_2\},\{v_2,r\},\{r,v_1\}$ form two distinct cycles in the cactus $G$, that both contain the edge $\{v_2,r\}$. This is a contradiction, implying $|\delta_E(V_1) \cap \delta_E(v_2)|\le 1$.
By the definition of $V_1$, all edges in $\delta_E(V_1)$ are incident to either $r$ or $v_2$.
Now $r$ has only degree two and one of its incident edges is $\{r,v_2\}\notin \delta_E(V_1)$.
Thus, $|\delta_E(V_1)|\le |\delta_E(V_1) \cap \delta_E(r)| + |\delta_E(V_1) \cap \delta_E(v_2)| \le 1 +1 =2$.
This completes the proof in the case where the cycle containing $r$ is of length at least three.

It remains to consider the case where the cycle containing $r$ is of length two.
Then $r$ has a unique neighbor $r'$. 
Because $G$ is a cactus, the vertex sets of the connected components of $G-r'$  are 2-cuts of $G$.
One of these vertex sets is $\{r\}$ and the others are contained in $\mathcal{C}_G \setminus \{V\setminus \{r\}\}$.
Hence, they are all small and $\Iscr$ is $k$-wide (with $r'$ as center).
\end{proof}

To be able to later combine solutions of $k$-wide sub-instances obtained through splitting, while controlling the cost of the merged solution, we cannot directly call the algorithm $\mathcal{A}$ for $k$-wide instances on each sub-instance and combine what we get. The reason why this approach does not work becomes apparent when considering how we merge two sub-instances stemming from a single split, as captured by Proposition~\ref{prop:combine_split_sol}. Indeed, to obtain a CacAP solution to the original instance from solutions $F_C$ and $F_{V\setminus C}$ to the sub-instances $\mathcal{I}_C$ and $\mathcal{I}_{V\setminus C}$, respectively, we need to add a link set $F$ of cardinality $|\delta_L(C)\cap F_C|-1$. To make sure that this is small, we are interested in a solution $F_C$ to the sub-instance $\mathcal{I}_C$ using only few links in $\delta_L(C)$. To achieve this, we will first enumerate over small link sets $S\subseteq \delta_L(C)$, and complement the links $S$ to a solution of the sub-instance by calling $\mathcal{A}$ on a residual instance obtained by fixing $S$. To be able to do so, we need to show that such a residual instance remains $k$-wide.
In order to prove this, we need the following observations.
(We recall that a vertex is a descendant of itself.)
\begin{lemma}\label{lem:descendants_define_2_cut}
Let $v\in V \setminus\{r\}$. Then the set of descendants of $v$ is a 2-cut of $G$.
\end{lemma}
\begin{proof}
For the sake of deriving a contradiction, suppose this is not the case. Then there are three distinct vertices $u_1,u_2,u_3$ in $G$ that are all neighbors of descendants of $v$ but not descendants of $v$ themselves. First observe that they must all be neighbors of $v$. Indeed, if, say $u_1$, were a neighbor of a descendant $w$ of $v$ with $w\neq v$, then there is a $u_1$-$r$ path not containing $v$ because one can first go from $u_1$ to $w$ and then to $r$ along an $w$-$r$ path not containing $v$, which exists because $w$ is not a descendant of $v$. However, this would contradict that $u_1$ is a descendant of $v$. Hence, assume that $u_1,u_2,u_3$ are all neighbors of $v$.
Because none of them is a descendant of $v$, there is, for each $j\in \{1,2,3\}$, an $r$-$u_j$ path $P_j$ in $G$ that does not contain $v$, which implies that we can extend it by the edge $\{u_j,v\}$ and maintain a path.
But then every pair $P_j,P_h$ of these three paths contains a cycle through the vertex $v$.
This cycle contains the edges $\{u_j,v\}$ and $\{u_h,v\}$ for the corresponding indices $j,h\in \{1,2,3\}$.
In particular, we obtain two distinct cycles containing the edge $\{u_1,v\}$, a contradiction to $G$ being a cactus.
\end{proof}

\begin{lemma}\label{lem:terminal_descendant_exists}
For every vertex $v\in V$, there is a descendant $t$ of $v$ in $G$ such that $t$ is a terminal.
\end{lemma}
\begin{proof}
For $w\in V$, let $U_w$ denote the set of descendants of $w$.
Let $t$ be a descendant of $v$ such that $|U_t|$ is minimal.
If $U_t =\{t\}$, the vertex $t$ is a terminal of $G$ by Lemma~\ref{lem:descendants_define_2_cut}.
Otherwise, there is a vertex $u\in U_t \setminus \{t\}$.
Then $u$ is a descendant of $t$ and hence also of $v$.
Moreover, every descendant of $u$ is also a descendant of $t$ and hence $U_u \subseteq U_t \setminus \{t\}$, a contradiction to the minimality of $|U_t|$.
\end{proof}

Now we are ready to show that residual instances of $k$-wide instances are again $k$-wide.

\begin{lemma}\label{lem:residual_instance_of_k_wide}
Let $\Iscr=(G,L)$ be a $k$-wide instance of the cactus augmentation problem and let $L'\subseteq L$.
Then the residual instance of $\Iscr$ with respect to $L'$ is $k$-wide.
\end{lemma}
\begin{proof}
Let $G=(V,E)$ and let $r\in V$ be a a $k$-wide center of $G$, i.e., every connected component of $G-r$ contains at most $k$ terminals.
Let $G'=(V',E')$ be the cactus in the residual instance of $\Iscr$ with respect to $L'$ and let $r'\in V'$ be the vertex that arose from the contraction of a vertex set containing $r$.
We claim that $G'$ is $k$-wide with center $r'$, i.e., every connected component of $G'-r'$ contains at most $k$ elements of $T'=\{v \in V': |\delta_{E'}(v)| = 2\}$.

If a link $\{v,w\}\in L'$ has endpoints in distinct connected components of $G-r$, then $r$ is contained in every $v$-$w$ path in $G$, and hence $v$, $w$, $r$, and possibly further vertices are contracted, resulting in the vertex $r'\in V'$.
Thus, every connected component $C'$ of $G'-r'$ corresponds to a connected component $C$ of $G-r$ in the sense that every vertex of $C'$ arose from the contraction of a subset of vertices of $C$.
It remains to show that such contraction operations cannot increase the number of vertices of degree~$2$. 
Every degree~$2$ vertex in $G'$ arose by contraction from a $2$-cut $C\subseteq V$ in $G$ (where it is possible that $C$ contains only a single vertex).
Thus, the terminals of $G'$ correspond to disjoint elements of $\Cscr_G$.
We claim that each of these 2-cuts contains a terminal, which implies that the number of terminals in $G'$ cannot be higher than in $G$. 

To see that each 2-cut $C\in \Cscr_G$ must contain a terminal, consider an arbitrary vertex $v\in C$.
Then by Lemma~\ref{lem:residual_instance_of_k_wide}, there is a terminal $t$ of $G$ that is a descendant of $v$.
This terminal $t$ must be contained in $C$ because $G$ contains two edge-disjoint $t$-$r$ paths in $G$ (because $G$ is 2-edge-connected) and each of them must visit $v$
because $t$ is a descendant of $v$. Using that $C$ is a 2-cut, we conclude $t\in C$.
\end{proof}

As discussed before the statement of Lemma~\ref{lem:descendants_define_2_cut}, we are now interested in finding solutions to $k$-wide sub-instances that allow for cheap merging. As stated in Proposition~\ref{prop:combine_split_sol}, we can keep the merging cost low by finding a solution $F\subseteq L$ to $\mathcal{I}_C$ using only few links of $\delta_L(s)$, where $s$ is the vertex that arose by contracting $V\setminus C$. This is why we want to minimize the objective $|F|+|\delta_{F}(s)|$ instead of simply the number of links $|F|$. The lemma below shows that approximately minimizing this objective is possible given an oracle $\mathcal{A}$ returning approximate solutions to $k$-wide CacAP instances, in the sense that one can either find a good solution efficiently or return a separating hyperplane that can be used in a round-or-cut framework.
\begin{lemma}\label{lem:algo_unsplittable}
 Let  $\epsilon' = \frac{\epsilon}{4}$.
 Suppose there is an $\alpha$-approximation algorithm $\Ascr$ for $k$-wide CacAP instances.
 Then there is a polynomial-time algorithm that, given an unsplittable CacAP instance $\mathcal{I}=(G=(V,E),L)$, a vector $x\in [0,1]^L$, and a vertex $s$ of $G$ with  $|\delta_E(s)|=2$ and $x(\delta_L(s))\le \frac{16}{\epsilon}$, either returns
 \begin{itemize}
  \item a CacAP solution $F\subseteq L$ with $|F| + |\delta_{F}(s)| \le (1+\epsilon')\cdot \alpha \cdot\left( x(L) + x(\delta_L(s))\right)$, or
  \item a vector $w\in \mathbb{R}^L$ such that $w^T x < w^T x'$ for all $x' \in P_{\mathrm{CacAP}}(\Iscr)$.
 \end{itemize}
\end{lemma}
\begin{proof}
By Lemma~\ref{lem:unsplittable}, the given unsplittable instance $\mathcal{I}=(G=(V,E),L)$ is $k$-wide.
If $x(\delta_L(s))<1$, we have $x(\delta_L(s)) < x'(\delta_L(s))$ for all $x'\in P_{\mathrm{CacAP}}(\Iscr)$ and can thus return $w=\chi^{\delta_L(s)}$.
Otherwise, we proceed as follows.
For every set $S \subseteq \delta_L(s)$ with $|S| \le \frac{1+\epsilon'}{\epsilon'}\cdot \frac{16}{\epsilon}$, we apply the given algorithm $\Ascr$ for $k$-wide instances to the residual instance of $(G, L\setminus \delta_L(s))$ with respect to $S$. This instance is $k$-wide by Lemma~\ref{lem:residual_instance_of_k_wide}.

If for some such set $S$, the algorithm $\Ascr$ returns a solution $F'$ with $|F'| + 2 |S| \le  (1+\epsilon')\cdot \alpha \cdot\left( x(L) + x(\delta_L(s))\right)$, we return $F=F' \cup S$.
Otherwise, we define $\mu:= 1 +  \frac{\epsilon' ( x(L) + x(\delta_L(s)) ) }{ x(\delta_L(s)) }$ and claim that
\begin{equation}
x'(L)+ \mu \cdot x'(\delta_L(s)) > x(L) + \mu \cdot x(\delta_L(s)) \quad \forall\; x'\in P_{\mathrm{CacAP}}(\Iscr)\enspace,\label{eq:sep_hyper_when_A_fails}
\end{equation}
which again leads to a vector $w\in \mathbb{R}^L$ as desired.
Suppose that~\eqref{eq:sep_hyper_when_A_fails} does not hold. 
Then there exists a solution $F^*$ of $\Iscr$ with
\begin{equation}\label{eq:solution_violating_cut}
|F^*| + \mu \cdot | \delta_{F^*}(s) | 
 \le x(L) + \mu \cdot x(\delta_L(s))  = (1+\epsilon') \cdot (x(L) + x(\delta_L(s))\enspace.
\end{equation}
For $S = \delta_{F^*}(s)$ this implies
\begin{equation*}
\frac{\epsilon' \left( x(L) + x(\delta_L(s)) \right) }{ x(\delta_L(s)) } \cdot |S|\ \le\ \mu \cdot  |S|\ \le\  (1+\epsilon') \cdot \left( x(L) + x(\delta_L(s)) \right)
\end{equation*}
and hence
\begin{equation*}
|S| \ \le\ \frac{1+\epsilon'}{\epsilon'} x(\delta_L(s))\ \le\ \frac{1+\epsilon'}{\epsilon'}\cdot \frac{16}{\epsilon}\enspace.
\end{equation*}
Thus, we considered the set $S$ in our algorithm described above.
Let $F'$ be the output of the $\alpha$-approximation algorithm $\Ascr$ applied to the residual instance of $(G, L\setminus \delta_L(s))$ with respect to $S$.
Then $|F'| \le \alpha \cdot |F^* \setminus S|$ because by Corollary~\ref{cor:residual_instance}, the set $F^*\setminus S$ is a feasible solution of this residual instance. 
Therefore, using $\alpha \ge 1$,  $S=\delta_{F^*}(s)$, and \eqref{eq:solution_violating_cut}, we get
\begin{align*}
|F' | + 2 |S| &\le \alpha \cdot |F^*\setminus S| + 2 |S| \le \alpha \cdot (|F^*| + |S|) \leq \alpha \cdot (|F^*| + \mu |\delta_{F^*}(s)|)\\
&\le (1+\epsilon')\cdot \alpha \cdot (x(L) + x(\delta_L(s))\enspace,
\end{align*}
contradicting the fact that we did not return $F' \cup S$.
\end{proof}

\subsection{An algorithm for instances without heavy cuts}\label{sec:reduction_algo_no_heavy_cuts}

In this section we present an algorithm for general CacAP instances without heavy cuts.
To this end we combine the results on splitting and on how to deal with unsplittable instances.
In particular, we use the splitting result given by Proposition~\ref{prop:combine_split_sol} together with Lemma~\ref{lem:algo_unsplittable}.
\begin{lemma}\label{lem:round_and_cut_without_heavy}
Suppose that there is an $\alpha$-approximation algorithm $\mathcal{A}$ for $k$-wide CacAP instances.
Then, for any CacAP instance $\mathcal{I}=(G=(V,E),L)$ and $x\in [0,1]^L$ such that no cut is $x$-heavy, there is a polynomial-time algorithm that computes either
\begin{itemize}
\item a CacAP solution $L'$ with $|L'|\le \alpha \cdot (1+\frac{\epsilon}{2}) \cdot x(L)$, or
\item a vector $w\in \mathbb{R}^L$ such that $w^T x < w^T x'$ for all $x' \in P_{\mathrm{CacAP}}(\mathcal{I})$.
\end{itemize}
\end{lemma}
\begin{proof}
We fix an arbitrary root $r\in V$.
We use induction on the number of vertices of $G$.
If $\Iscr$ is unsplittable, it is also $k$-wide by Lemma~\ref{lem:unsplittable}.
Then we apply the given $\alpha$-approximation algorithm $\Ascr$ to $\Iscr$ to obtain a feasible solution $L'$.
If $|L'| \le \alpha \cdot x(L)$, we return $L'$. 
Otherwise, we return $w = \chi^{L}$, which fulfills $w^T x < w^T x'$ for all $x' \in P_{\mathrm{CacAP}}(\mathcal{I})$ because $\Ascr$ is an $\alpha$-approximation algorithm.

Hence, we may assume that $ \Iscr$ is splittable.
Let $C\in \mathcal{C}_G$ be a minimal set such that $\Iscr$ is  splittable at $C$.
We apply our splitting result, Proposition~\ref{prop:combine_split_sol}, to $C$ and $\mathcal{I}$.

As in Proposition~\ref{prop:combine_split_sol}, we denote by $\mathcal{I}_C=(G_C, L_C)$ the instance arising from $\mathcal{I}$ by contracting $V\setminus C$ and we choose as new root the vertex $s$ arising from the contraction of $V\setminus C$. (Note that $\{r\} \subsetneq V\setminus C$.)
Because $C$ was chosen minimally, the instance $\mathcal{I}_C$ with root $s$ is unsplittable, and thus, by Lemma~\ref{lem:unsplittable}, it is $k$-wide.
We apply Lemma~\ref{lem:algo_unsplittable} to $\mathcal{I}_C$ to either obtain a CacAP solution $F_C$ for $\mathcal{I}_C$ with
\begin{equation}\label{eq:good_subsol_C}
|F_C| + |\delta_{F_C}(s)| \leq \left(1+\tfrac{\varepsilon}{4}\right)\cdot \alpha\cdot (x(L_C) + x(\delta_L(s)))\end{equation}
or a vector $w\in \mathbb{R}^{L_C}$ satisfying $w^T x < w^T x'$ for all $x' \in P_{\mathrm{CacAP}}(\mathcal{I}_C)$.

Moreover, we denote by $\mathcal{I}_{V\setminus C}=(G_{V\setminus C}, L_{V\setminus C})$ the CacAP instance arising from $\mathcal{I}$ by contracting $C$.
Because $\mathcal{I}$ is splittable at $C$, the set $C$ is big and therefore $G_{V\setminus C}$ has strictly fewer vertices than $G$.
By the induction hypothesis we either obtain a CacAP solution $F_{V\setminus C}$ for $\mathcal{I}_{V\setminus C}$ with
\begin{equation}\label{eq:good_subsoc_comp_C}
|F_{V\setminus C}|\leq \alpha \cdot\left(1+\tfrac{\varepsilon}{2}\right) \cdot x(L_{V\setminus C})
\end{equation}
 or a vector $\bar w\in \mathbb{R}^{L_{V\setminus C}}$ such that $\bar w^T x < \bar w^T x'$ for all $x' \in P_{\mathrm{CacAP}}(\mathcal{I}_{V\setminus C})$.

If we obtained a vector $w\in \mathbb{R}^{L_C}$ such that $w^T x < w^T x'$ for all $x' \in P_{\mathrm{CacAP}}(\mathcal{I}_C)$, we can extend it to a vector $w\in \mathbb{R}^L$ such that $w^T x < w^T x'$ for all $x' \in P_{\mathrm{CacAP}}(\mathcal{I})$ by setting $w_\ell =0$ for all $\ell \in L \setminus L_C$.
Here we use that restricting a CacAP solution $F$ of $\mathcal{I}$ to $F\cap L_C$ yields a CacAP solution for $\mathcal{I}_C$.
We can proceed analogously if we have a vector $\bar w\in \mathbb{R}^{L_{V\setminus C}}$ such that $\bar w^T x < \bar w^T x'$ for all $x' \in P_{\mathrm{CacAP}}(\Iscr_{V\setminus C})$.

Otherwise, we obtained solutions $F_C$ and $F_{V\setminus C}$ as discussed above and apply Proposition~\ref{prop:combine_split_sol}.
This yields a set $F\subseteq L$ such that $F \cup F_C \cup F_{V\setminus C}$ is a CacAP solution for $\Iscr$ and $|F| \le |F_C \cap \delta_L(C)| = |\delta_{F_C}(s)|$.
Thus, combining~\eqref{eq:good_subsol_C} and~\eqref{eq:good_subsoc_comp_C}, we obtain
\begin{equation}\label{eq:upper_bound_combining_subinstances}
\begin{aligned}
|F \cup F_C \cup F_{V\setminus C}|
 &\le  |F_{V\setminus C}| + |F_C| +  |\delta_{F_C}(s)|\\
 &\le  \alpha \left(1+\tfrac{\epsilon}{2}\right) x(L_{V\setminus C}) + \alpha \left(1+\tfrac{\epsilon}{4}\right) \left(x(L_C) + x(\delta_L(C))\right)\\
 &=   \alpha \left(1+\tfrac{\epsilon}{2}\right) x(L_{V\setminus C}) +\alpha \left(1+\tfrac{\epsilon}{2}\right) x(L_C) - \alpha \tfrac{\epsilon}{4} \cdot x(L_C) + \alpha \left(1+\tfrac{\epsilon}{4}\right) \cdot x(\delta_L(C)) \\
 &= \alpha \left(1+\tfrac{\epsilon}{2}\right) \cdot x(L) + \alpha \left(2+\tfrac{3\epsilon}{4}\right) \cdot x(\delta_L(C)) - \alpha \tfrac{\epsilon}{4} \cdot x(L_C)\enspace,
\end{aligned}
\end{equation}
where we used $L_C \cup L_{V\setminus C} = L$ and $L_C \cap L_{V\setminus C} = \delta_L(C)$.
Because $C$ is big, we have $x(L_C) \ge \frac{k}{4} = 16\cdot\frac{8+3\epsilon}{\epsilon^2}$.
Moreover, $x(\delta_L(C))\le \frac{16}{\epsilon}$ because $\delta_L(C)$ is $x$-light.
This implies
\begin{equation*}
 \left(2+\tfrac{3\epsilon}{4}\right) \cdot x(\delta_L(C))\ \le\  \left(2+\tfrac{3\epsilon}{4}\right) \cdot  \tfrac{16}{\epsilon}\  =\ \tfrac{\epsilon}{4} \cdot 16\cdot\tfrac{8+3\epsilon}{\epsilon^2}\ \le\ \tfrac{\epsilon}{4}\cdot x(L_C)\enspace.
\end{equation*}
Together with \eqref{eq:upper_bound_combining_subinstances}, we obtain
\begin{equation*}
|F \cup F_C \cup F_{V\setminus C}|\ \le\ \alpha \cdot \left(1+ \tfrac{\epsilon}{2}\right) \cdot x(L)\enspace,
\end{equation*}
as desired.
\end{proof}

\subsection{Proof of Theorem~\ref{thm:main_reduction}}\label{sec:reduction_ellipsoid}

We are now ready to put everything together to prove Theorem~\ref{thm:main_reduction} through the round-or-cut framework.
Given an arbitrary CacAP instance $\mathcal{I}=(G=(V,E),L)$, we first guess the cardinality $|\OPT|$ of an optimal solution $\OPT\subseteq L$ to instance $\mathcal{I}$.
(More precisely, we run the algorithm described below for all possible values of $|\OPT|$ and return the best of the resulting solutions.)
Note that we can check upfront whether the instance $\Iscr$ is feasible, so me may assume that this is the case.

We run the Ellipsoid Method to find a feasible point in the convex hull 
$P_{\mathrm{CacAP}}(\mathcal{I})\cap \{y\in [0,1]^L: y(L)=|\OPT|\}$ of all incidence vectors of optimal solutions using the following separation oracle.
Given a vector $x\in [0,1]^L$, we first check if $x(L) = |\OPT|$ and otherwise return $w= \chi^L$ or $w=-\chi^L$ as a separating hyperplane.
If $x(L) = |\OPT|$, we apply Theorem~\ref{thm:covering_heavy} with $\Wscr = \{C\in \mathcal{C}_G : C\text{ is $x$-heavy}\}$ to obtain a cheap heavy cut covering, i.e., a set $L_H\subseteq L$ of links covering all heavy cuts with $|L_H|\leq \frac{\epsilon}{2}\cdot x(L)$. (See the discussion right after Theorem~\ref{thm:covering_heavy} for more details.)
Let  $\bar{\mathcal{I}}=(\bar G, \bar L)$ denote the residual instance of $\mathcal{I}$ with respect to $L_H$ and let $\bar x$ be the restriction of $x$ to $\bar L$.
By Lemma~\ref{lem:all_light_after_heavy_cut_covering}, the residual instance has no heavy cuts and hence we can apply Lemma~\ref{lem:round_and_cut_without_heavy} to $\Iscr'$.
This yields either a solution $F$ of $\bar{\mathcal{I}}$ or 
 a vector $w \in \mathbb{R}^{\bar L}$ with $w^T x< w^T y$ for all $y \in P_{\mathrm{CacAP}}(\bar{\mathcal{I}})$.

In the latter case, we extend the vector $w\in \mathbb{R}^{\bar L}$ to a vector in $\mathbb{R}^L$ by setting $w(\ell)\coloneqq 0$ for $\ell\in L \setminus \bar L$.
Then we indeed have $w^T x < w^T x'$ for all $x' \in  P_{\mathrm{CacAP}}(\Iscr)$, because every feasible solution for $\Iscr$ is also a feasible solution for $\bar \Iscr$ (after dropping links in $L\setminus \bar L$).

Otherwise, we obtain a solution $F$ of the residual instance $\bar{\mathcal{I}}$.
Then the set $L_H \cup F$ is a feasible solution of $\mathcal{I}$. Moreover, due to the bound on the cardinality of $F$ guaranteed by Lemma~\ref{lem:round_and_cut_without_heavy}, we have
\begin{align*}
 |L_H \cup F| & \leq \frac{\epsilon}{2} \cdot x(L) + \alpha \cdot \left(1+\frac{\epsilon}{2}\right) \cdot x(L) \le \alpha \cdot (1+\epsilon) \cdot x(L)\enspace.
\end{align*}
Because $x(L) = |\OPT|$, the link set $L_H\cup F$ is a CacAP solution of cardinality at most $ (1+\epsilon) \cdot \alpha \cdot |\OPT|$. 
In this case, we can thus stop running the Ellipsoid Method because we have found the desired solution.

This completes the description of our separation oracle.
In some iteration of the Ellipsoid Method the vector $x$ given to the separation oracle must be contained in $P_{\mathrm{CacAP}}(\mathcal{I})\cap \{y\in [0,1]^L: y(L)=|\OPT|\}$.
In this iteration, the algorithm described above does not find any separating hyperplane, but the desired solution $L_H\cup F$. 
Polynomial-time termination of the Ellipsoid Method follows by classical results (see~\cite[Theorem~(6.4.1)]{groetschel_1993_geometric}) because 
the polytope $P_{\mathrm{CacAP}}(\mathcal{I})\cap \{y\in [0,1]^L: y(L)=|\OPT|\}$ over which we run the Ellipsoid Method, has facet complexity that is bounded polynomially in $|L|$, because it is a $\{0,1\}$-polytope in $[0,1]^L$.

\section{Covering heavy cuts}\label{sec:heavy_cut_covering}

In this section we show Theorem~\ref{thm:covering_heavy}.
We first prove it for the special case where the cactus $G$ is a single cycle and then reduce the general case to this case.

\begin{theorem}\label{thm:heavy_cut_covering_for_cycles}
Let $\mathcal{I} = (G=(V,E), L)$ be a CacAP instance with $G$ being a cycle, let $r\in V$, and let $\mathcal{W}\subseteq \mathcal{C}_G$. 
Then the LP
\begin{equation}\label{eq:heavyCutHitLPCycle}
\renewcommand\arraystretch{1.5}
\begin{array}{rrcll}
\min & x(L) & & & \\
& x(\delta_L(W)) & \ge & 1 & \forall\; W\in \Wscr \\
& x & \in & \mathbb{R}_{\geq 0}^L \makebox[0pt][l]{\enspace .} &
\end{array}
\end{equation}
has integrality gap at most 8. 
Moreover, given a solution $x$ to~\eqref{eq:heavyCutHitLPCycle}, we can compute an integral solution with objective value at most $8 \cdot x(L)$ in polynomial time.
\end{theorem}

To prove Theorem~\ref{thm:heavy_cut_covering_for_cycles}, we will reduce the problem to a certain rectangle hitting problem.
Let $v_0 = r,  v_1,\dots, v_n$ be the vertices numbered in a consecutive order with respect to how they appear on the cycle $G$, i.e., such that
$E= \{ \{v_i,v_{i+1}\} : i=0,\dots, n-1\} \cup \{ \{v_n,v_0\}\}$.
Then every set $W\in \Wscr$ contains a consecutive set of vertices, i.e., it is of the form $\{ v_i : a \le i \le b\}$ for some $a,b \in [n]$.
For vertices $v_i, v_j$ with $i< j$ we say that $v_i$ is \emph{left} of $v_j$ and $v_j$ is \emph{right} of $v_i$. Moreover, we say that a vertex $v$ is \emph{right/left of a set $W\in \Wscr$} if it is right/left of all vertices in $W$.

We associate every link $\ell=\{v_i,v_j\}$ with $i<j$ with the point $p_\ell=(i,j)$ in the Euclidean plane.
Let us now consider a set $W\in \Wscr$ with $W=\{ v_k : a \le k \le b\}$.
Then the link $\ell=\{v_i,v_j\}$ with $i<j$ covers the cut $W$ if and only if either
\begin{itemize}
    \item $v_i$ is contained in $W$ and $v_j$ is right of $W$, i.e.,
         \begin{equation*}
            p_\ell \in R_W^{\uparrow} \coloneqq \big\{ (i,j) : a \le i \le b,\ b +1 \le  j \le n\big\}\enspace, \text{ or}
         \end{equation*}
    \item $v_j$ is contained in $W$ and $v_i$ is left of $W$, i.e.,
          \begin{equation*}
            p_\ell \in R_W^{\leftarrow} \coloneqq \big\{ (i,j) : 0 \le i \le a-1,\ a\le j \le b \big\}\enspace.
         \end{equation*}
\end{itemize}
Viewing $R_W^{\uparrow}$ and $R_W^{\leftarrow}$ as subsets of the Euclidean plane, $R_W^{\uparrow}$ is a rectangle with topmost coordinate being $n$ and $R_W^{\leftarrow}$ is a rectangle with leftmost coordinate being $0$. An illustration is given in Figure~\ref{fig:rectangle_hitting}.
\begin{figure}[!ht]
\begin{center}
\begin{tikzpicture}[scale=0.23]

\pgfdeclarelayer{fg}
\pgfdeclarelayer{bg}
\pgfsetlayers{bg,main,fg}

\tikzset{
  prefix node name/.style={%
    /tikz/name/.append style={%
      /tikz/alias={#1##1}%
    }%
  }
}

\tikzset{link/.style={line width=1.5pt}}

\tikzset{node/.style={thick,draw=black,fill=white,circle,minimum size=6, inner sep=2pt}}

\colorlet{rwl}{orange!50!yellow}

\newcommand\cyc[2][]{
\begin{scope}[prefix node name=#1]

\begin{scope}[every node/.append style=node,shift={(2,10)}]
\tikzset{cv/.style={fill=blue!20}}
\foreach \i/\s in {0/,1/,2/,3/cv,4/cv,5/,6/} {
\node[\s] (\i) at (\i*360/7:10cm) {$v_{\i}$};
}
\end{scope}

\begin{scope}[very thick]
\draw (0) -- (1) -- (2) -- (3) -- (4) -- (5) -- (6) -- (0);
\draw[dashed,red] (2) to[bend right] node[above left] {$\ell_1$} (3); %
\draw[dashed,red] (1) to node[left] {$\ell_2$} (5); %
\draw[dashed,red] (4) to[out=-90-90/7,in=-90-90/7] node[below] {$\ell_3$} (6); %
\end{scope}

\end{scope}
}%

\newcommand\rect[2][]{
\begin{scope}[prefix node name=#1]

\tikzset{wb/.style={font=\bfseries,text=blue,draw=none}}

\begin{scope}
\coordinate[black, label={below left: 0}] (1) at (0,0) {};
\coordinate[black, label={below: 1}] (2) at (5,0) {};
\coordinate[black, label={below: 2}] (3) at (10,0) {};
\coordinate (4) at (15,0);
\node[wb] at (4)[below] {3};
\coordinate (5) at (20,0);
\node[wb] at (5)[below] {4};
\coordinate[black, label={below: 5}] (6) at (25,0) {};
\coordinate[black, label={below: 6}] (7) at (30,0) {};

\coordinate (8) at (0,5) {};
\node at (8)[left] {1};
\coordinate[black] (9) at (5,5) {};
\coordinate[black] (10) at (10,5) {};
\coordinate[black] (11) at (15,5) {};
\coordinate[black] (12) at (20,5) {};
\coordinate[black] (13) at (25,5) {};
\coordinate[black] (14) at (30,5) {};

\coordinate (15) at (0,10) {};
\node at (15)[left] {2};
\coordinate[black] (16) at (5,10) {};
\coordinate[black] (17) at (10,10) {};
\coordinate[black] (18) at (15,10) {};
\coordinate[black] (19) at (20,10) {};
\coordinate[black] (20) at (25,10) {};
\coordinate[black] (21) at (30,10) {};

\coordinate (22) at (0,15) {};
\node[wb] at (22)[left] {3};
\coordinate[black] (23) at (5,15) {};
\coordinate[black] (24) at (10,15) {};
\coordinate[black] (25) at (15,15) {};
\coordinate[black] (26) at (20,15) {};
\coordinate[black] (27) at (25,15) {};
\coordinate[black] (28) at (30,15) {};

\coordinate (29) at (0,20) {};
\node[wb] at (29)[left] {4};
\coordinate[black] (30) at (5,20) {};
\coordinate[black] (31) at (10,20) {};
\coordinate[black] (32) at (15,20) {};
\coordinate[black] (33) at (20,20) {};
\coordinate[black] (34) at (25,20) {};
\coordinate[black] (35) at (30,20) {};

\coordinate (36) at (0,25) {};
\node at (36)[left] {5};
\coordinate[black] (37) at (5,25) {};
\coordinate[black] (38) at (10,25) {};
\coordinate[black] (39) at (15,25) {};
\coordinate[black] (40) at (20,25) {};
\coordinate[black] (41) at (25,25) {};
\coordinate[black] (42) at (30,25) {};

\coordinate (43) at (0,30) {};
\node at (43)[left] {6};
\coordinate[black] (44) at (5,30) {};
\coordinate[black] (45) at (10,30) {};
\coordinate[black] (46) at (15,30) {};
\coordinate[black] (47) at (20,30) {};
\coordinate[black] (48) at (25,30) {};
\coordinate[black] (49) at (30,30) {};
\end{scope}

\begin{scope}[black!50]
\draw (1) -- (2) -- (3) -- (4) -- (5) -- (6) -- (7);
\draw (8) -- (9) -- (10) -- (11) -- (12) -- (13) -- (14);
\draw (15) -- (16) -- (17) -- (18) -- (19) -- (20) -- (21);
\draw (22) -- (23) -- (24) --(25) -- (26) -- (27) -- (28);
\draw (29) -- (30) -- (31) -- (32) -- (33) -- (34) -- (35);
\draw (36) -- (37) -- (38) -- (39) -- (40) -- (41) -- (42);
\draw (43) -- (44) -- (45) -- (46) -- (47) -- (48) -- (49);

\draw (1) -- (8) -- (15) -- (22) -- (29) -- (36) -- (43);
\draw (2) -- (9) -- (16) -- (23) -- (30) -- (37) -- (44);
\draw (3) -- (10) -- (17) -- (24) -- (31) -- (38) -- (45);
\draw (4) -- (11) -- (18) -- (25) -- (32) -- (39) -- (46);
\draw (5) -- (12) -- (19) -- (26) -- (33) -- (40) -- (47);
\draw (6) -- (13) -- (20) -- (27) -- (34) -- (41) -- (48);
\draw (7) -- (14) -- (21) -- (28) -- (35) -- (42) -- (49);
\end{scope}

\draw[fill, draw, ultra thick, green!20] (0,15) rectangle (10,20);
\draw[fill, draw, ultra thick, orange!20] (15,25) rectangle (20,30);

\begin{scope}[every node/.style={fill, circle, minimum size =4pt, inner sep=0.2pt, red}]
\node (p1) at (10,15) {};
\node (p2) at (5,25) {};
\node (p3) at (20,30) {};
\end{scope}
\begin{scope}[red]
\node[below right] (lp1) at (p1) {$p_1$};
\node[below right] (lp2) at (p2) {$p_2$};
\node[below right] (lp3) at (p3) {$p_3$};
\end{scope}

\end{scope}
}%

\begin{scope}[scale=0.8, shift={(-16,2)}]
\cyc[c-]{}
\end{scope}

\begin{scope}[scale=0.7,shift={(8,0)}]
\rect[r-]{}
\end{scope}

\begin{scope}[shift={(32,14)}]
\def\ll{60mm} %
\def\vs{48mm} %

\begin{scope}[scale=0.5]

\node[fill=blue!20,circle,draw=black, inner sep = 3pt] (b) at (0,0) {};
\node at (b.east)[right=2pt] {vertices in $W$};

\node[fill=orange!20, draw= black, inner sep = 4pt] (b) at (0,-\vs) {};
\node at (b.east)[right=2pt] {$R_W^{\uparrow}$};

\node[fill=green!20,draw=black,inner sep = 4pt] (b) at (0,-2*\vs) {};
\node at (b.east)[right=2pt] {$R_W^{\leftarrow}$};

\node[fill=red,circle,draw=black,inner sep = 3pt] (b) at (0,-3*\vs) {};
\node[align=left,yshift=-1.5ex] at (b.east)[right=2pt] {links and\\corresponding points};

\end{scope}

\end{scope}

\end{tikzpicture}
 \end{center}
\caption{Constructing the rectangle hitting problem for a cycle with $7$ vertices. Here, the set $W$ is given by $\{v_3,v_4\}$. The sets $R^{\uparrow}_W$ and $R^{\leftarrow}_W$ are shown in orange and green, respectively. Note that a link is contained in $\delta_L(W)$ if and only if its corresponding point is contained in one of the rectangles $R^{\uparrow}_W$ and $R^{\leftarrow}_W$.}\label{fig:rectangle_hitting}
\end{figure}
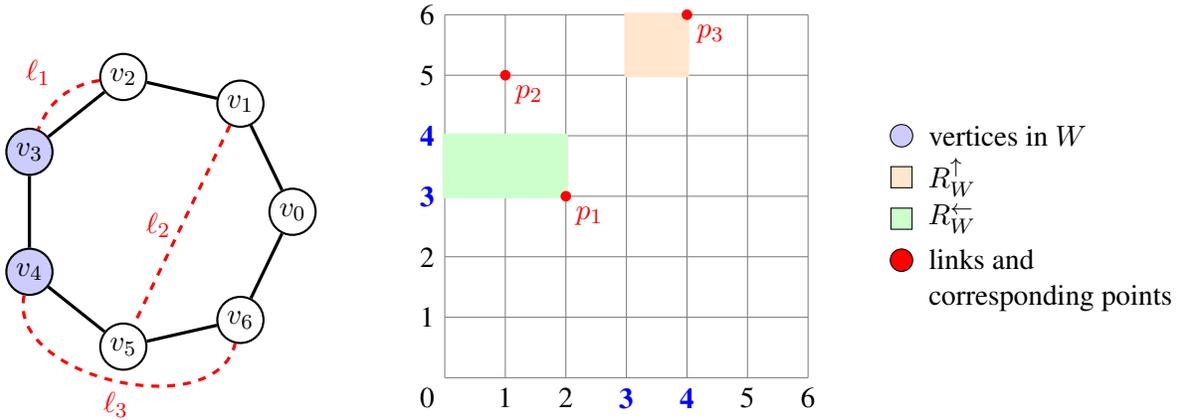

For every $W\in \Wscr$ we have $x(\delta_L(W)) \ge 1$ and thus $\sum_{\ell: p_\ell \in R_W^{\uparrow}} x_\ell \ge \frac{1}{2}$ or $\sum_{\ell: p_\ell \in R_W^{\leftarrow}} x_\ell \ge \frac{1}{2}$.
Hence we can partition $\Wscr$ into sets $\Wscr^{\uparrow}$ and $\Wscr^{\leftarrow}$ such that
\begin{align*}
    \sum_{\ell: p_\ell\in R^{\uparrow}_W} 2 \cdot x_\ell \ge&\ 1 \quad \forall\; W\in \Wscr^{\uparrow} \\
\intertext{and}
    \sum_{\ell: p_\ell \in R^\leftarrow_W} 2 \cdot x_\ell \ge&\ 1 \quad\forall\; W\in \Wscr^{\leftarrow}\enspace.
\end{align*}
Therefore, the below lemma applied to the point set $\{p_\ell : \ell\in L\}$, the vector with $x_{p_\ell} := 2x_\ell$, and the rectangle sets $\{ R^{\uparrow}_W : W \in \Wscr^{\uparrow}\}$
implies that we can find a set of at most $4 \cdot x(L)$ links covering all cuts in $\Wscr^{\uparrow}$.
Applying the lemma also to the  rectangle set $\{ R^{\leftarrow}_W: W \in \Wscr^{\leftarrow}\}$ (rotated such that the rectangles have a common topmost coordinate) yields Theorem~\ref{thm:heavy_cut_covering_for_cycles}.

\begin{lemma}\label{lem:rectangle_hitting}
Let $\Rscr$ be a set of closed rectangles in the Euclidean plane with identical topmost coordinate.
Let $P\subseteq \mathbb{R}^2$ be a finite set of points and $x: P \to [0,1]$ such that
\begin{equation*}
      \sum_{p\in R} x_p \ge 1  \quad \forall\; R\in \Rscr\enspace.
\end{equation*}
Then one can efficiently compute a set $H \subseteq P$ such that $R\cap H \ne \emptyset$ for all $R\in \Rscr$ and
$|H| \le 2 \cdot \sum_{p\in P} x_p$.
\end{lemma}
\begin{proof}
We assume $x_p > 0$ for all $p\in P$ as we can simply delete any point with an $x$-value of $0$. Moreover, we may also assume that no point lies above any rectangle. (If it does, it is not contained in any of the rectangles because they have a common topmost coordinate.)
First, we partition the point set $P$ into ``vertical stripes'' $S_1, \dots, S_k$ with $x(S_i)=\frac{1}{2}$ for $i\in [k-1]$ and $x(S_k) \le \frac{1}{2}$.
To this end we sort the points by their first coordinate.
Then the first stripe $S_1$ consists of the first points in this order such that $x(S_1)=\frac{1}{2}$, where we possibly split one point into two copies with appropriate $x$-values to make this possible.
Then we iteratively define the stripe $S_i$ to be the first points in the order that are not contained in $S_1 \cup \dots \cup S_{i-1}$ and such that $x(S_i)=\frac{1}{2}$.
Again, we might split one point to make this possible.
If the total $x$-value of the points not contained in any stripe is at most $\frac{1}{2}$, we define the last stripe $S_k$ to be the set of these points.

For the stripes $S_1, \dots, S_{k-1}$ (but not for $S_k$) we choose a topmost point in the stripe and include it in the set $H$.
Then $|H| \le k-1 = \sum_{i=1}^{k-1} 2 \cdot x(S_i) \le 2\cdot x(P)$.
We claim that $H$ hits every rectangle, i.e., $H\cap R\ne \emptyset$ for all $R\in \Rscr$.
To prove this, consider a fixed rectangle $R\in \Rscr$.
Let $S_a$ and $S_b$ be the leftmost and rightmost stripe intersecting $R$.
Then $1 \le x(R) =\sum_{j=a}^b x(S_j \cap R)$.
Using $x(S_j) \le \frac{1}{2}$ for all $j\in [k]$ and $x_p > 0$ for all $p\in P$, this implies that either
$S_a \subseteq R$ or $S_i \cap R \ne \emptyset$ for some $a<i<b$.
If $S_a \subseteq R$, the topmost point in $S_a$ hits $R$.
Otherwise, consider the topmost point in $S_i$ that we added to $H$.
Suppose this point $p$ is not contained in $R$.
Because $R$ contains at least one point in $S_a$, the point $p$ cannot be left of $R$
(meaning that its first coordinate is strictly smaller than the first coordinate of any point in $R$).
Similarly, $R$ contains at least one point in $S_b$ and hence it cannot be right of $R$.
Because no point is above any of the rectangles, $p$ must be below $R$.
But this contradicts $S_i \cap R \ne \emptyset$ since $p$ is a topmost point in $S_i$.
\end{proof}

Having completed the proof of Theorem~\ref{thm:heavy_cut_covering_for_cycles}, we now consider the general case where $G$ is not necessarily a cycle. We rely on the well-known observation (see~\cite{dinitz_1976_structure}) that for each cactus, there is a cycle on the same vertex set whose $2$-cuts contain the $2$-cuts of the cactus. This observation readily allows for reducing to the cycle case. We provide a proof for completeness.

\begin{lemma}[\cite{dinitz_1976_structure}]\label{lem:replace_cactus_by_cycle}
Let $G=(V,E)$ be a cactus and $r\in V$. Moreover, let
\begin{equation*}
\mathcal{W}\subseteq \mathcal{C}_G\coloneqq \{C\subseteq V\setminus \{r\}: |\delta_E(C)| = 2\}\enspace.
\end{equation*}
Then there exists a cycle $(V,\bar E)$ such that $|\delta_{\bar E}(W)|= 2$ for every $W\in \Wscr$.
\end{lemma}
\begin{proof}
We prove the theorem by induction on the number of cycles in the cactus $G$.
If $G$ contains only a single cycle, the statement holds for $\bar E=E$.
Otherwise, let $Q_1,Q_2$ be two distinct cycles in $G$ that share a common vertex $z$.
Let $v_1$ be a neighbor of $z$ in $Q_1$ and let $v_2$ be a neighbor of $z$ in $Q_2$.
We remove the edges $\{v_1,z\}$ and $\{v_2,z\}$ from $G$ and add a new edge $\{v_1,v_2\}$.
Let $G'$ be the resulting graph.
Then $G'$ is a cactus, where the two cycles $Q_1$ and $Q_2$ are replaced by a single cycle.
All other cycles of $G$ remain unchanged.

We have $|\delta_G(U)| \ge |\delta_{G'}(U)|$ for every set $U\subseteq V$.
In particular, $|\delta_{G'}(W)| \le 2$ for every $W\in \Wscr$.
Because $G'$ is 2-edge-connected, we must have $|\delta_{G'}(W)| = 2$ for every $W\in \Wscr$.
The statement now follows from the induction hypothesis applied to $G'$.
\end{proof}

As discussed, Theorem~\ref{thm:covering_heavy}, which we restate below for convenience, now readily follows by combining the above statements.

\theoremcoveringheavy*

\begin{proof}
We apply Lemma~\ref{lem:replace_cactus_by_cycle} to obtain a cycle $(V,\bar E)$ such that each element of $\Wscr$ is a 2-cut in $(V,\bar E)$.
Theorem~\ref{thm:heavy_cut_covering_for_cycles} applied to $((V,\bar E),L)$ and $\Wscr$ completes the proof.
\end{proof}

We remark that for our purpose the precise upper bound on the integrality gap is not important and any constant bound suffices. We therefore did not attempt to optimize the constant but rather strived for a clean analysis.

\section{Cross-link rounding using Chv\'atal-Gomory cuts}\label{sec:cg}

In this section we prove Lemma~\ref{lem:cross-link_rounding}.

\lemcrosslinkrounding*

This implies a $(1.5 + \epsilon)$-approximation for CAP as we have shown in Section~\ref{sec:overview}.
In Section~\ref{sec:stack_analysis} we will also use Lemma~\ref{lem:cross-link_rounding} to  prove our main result (Theorem~\ref{thm:main}).
Throughout this section we consider instances of weighted CacAP.
While considering the weighted version is not necessary in order to achieve our overall result for unweighted cactus augmentation,
 all proofs extend immediately to this setting.
 
 Recall that the high-level plan to prove Lemma~\ref{lem:cross-link_rounding} is as follows.
 We will choose $P_{\mathrm{cross}}$ to be  the polytope resulting from
 \begin{equation*}
P_{\mathrm{cut}}\coloneqq \left\{x\in \mathbb{R}_{\geq 0}^L\colon x(\delta_L(C)) \geq 1 \; \forall C\in \mathcal{C}_G\right\}\enspace,
\end{equation*}
by adding the $\{0,\sfrac{1}{2}\}$-Chv\'atal-Gomory (CG) cuts that can be obtained from a carefully chosen laminar family $\Lscr$.
We will obtain the laminar family $\Lscr$ from a particular dual solution of an integral linear program whose optimum solutions yield a 2-approximation for CacAP.
Finally, we develop a new dual-improvement argument and leverage results by Fiorini, Gro\ss, K\"onemann, and Sanit\`a~\cite{fiorini_2018_approximating} to prove that the laminar family $\Lscr$ has the desired properties.

\subsection{A simple 2-approximation by bidirecting links}\label{sec:bidirected_lp}

Let us start by recalling the integral linear program whose dual we use to obtain the laminar family $\Lscr$.
As we explained in Section~\ref{sec:overviewRounding}, this LP relaxation is inspired by a simple 2-approximation algorithm for TAP.
This 2-approximation algorithm is based on the following observations.
First, in an instance of TAP, we can replace a link $\ell = \{v,w\}$   that is not an up-link itself
by two up-links $\ell_1 =\{v,u\}$ and $\ell_2=\{w,u\}$ such that $\ell_1$ and $\ell_2$ together cover exactly the same cuts in $\Cscr_G$ that are covered by $\ell$.
(Here we choose $u$ to be the common ancestor of $v$ and $w$ that is farthest away from the center $r$.)
Second, if all links are up-links, then the polytope $P_{\mathrm{cut}}$ is integral.
Hence, replacing every link $\ell\in L \setminus L_{\mathrm{up}}$ by two up-links that cover the same $2$-cuts and computing a minimum cost extreme point of  the polytope $P_{\mathrm{cut}}$ for the resulting instance yields a 2-approximation for TAP.
We highlight that for CacAP, it is not possible to replace every link $L \setminus L_{\mathrm{up}}$ by two up-links that together cover exactly the same cuts in $\Cscr_G$.

In order to generalize the above 2-approximation from TAP to CacAP, we therefore provide a different view on the same algorithm.
Instead of splitting every link $\ell=\{v,w\}\in L \setminus L_{\mathrm{up}}$ into two up-links $\ell_1$ and $\ell_2$, 
we replace $\ell$ by two directed links $\vec{\ell}_1 = (w,v)$ and $\vec{\ell_2}=(v,w)$.
Clearly a link $\ell=\{v,w\}$ covers a cut $C\in \mathcal{C}_G$ if and only if either $(v,w)$ or $(w,v)$ is entering $C$. Moreover, if $\ell=\{v,w\}$ is an up-link, say with $v$ being an ancestor of $w$, then a cut $C\in\mathcal{C}_G$ is covered by $\ell$ if and only if $(v,w)$ enters the cut $C$.
This naturally leads to the use of the directed link set
\begin{equation*}
\vec{L}\coloneqq \bigcup_{\{v,w\}\in L}\left\{(v,w),(w,v)\right\}\enspace,
\end{equation*}
to define the linear program~\eqref{eq:dir_cut_lp}, which we recall here:
\begin{equation}\tag{\ref{eq:dir_cut_lp}}
\renewcommand\arraystretch{1.5}
\begin{array}{r>{\displaystyle}rcll}
\min & \sum_{\ell \in \vec{L}} c_\ell x_\ell & & & \\
 &x\big(\delta^-_{\vec{L}}(C)\big)  &\geq &1 &\forall\; C\in \mathcal{C}_G\\
 &x &\in &\mathbb{R}^{\vec{L}}_{\geq 0}\makebox[0pt][l]{\enspace .} &
\end{array}
\end{equation}

Every vector $x\in P_{\mathrm{cut}}$  can be turned into a feasible solution to~\eqref{eq:dir_cut_lp} of cost 
$c^T x + \sum_{\ell\in L \setminus L_{\mathrm{up}} } c_\ell x_\ell$ by setting $x_{(v,w)}= x_{(w,v)} = x_{\{v,w\}}$ for every link $\{v,w\}\in L \setminus L_{\mathrm{up}}$
and $x_{(v,w)} = x_{\{v,w\}}$ for every up-link $\{v,w\}\in L_{\mathrm{up}}$, where $v$ is the endpoint of $\{v,w\}$ that is closer to the center $r$.
Therefore, the below lemma yields a simple 2-approximation algorithm for CacAP.
But more importantly, we will use the integrality of the linear program~\eqref{eq:dir_cut_lp} in our proof of Lemma~\ref{lem:cross-link_rounding}.
To prove Lemma~\ref{lem:bidirected_integral}, we use standard combinatorial uncrossing arguments.

\lembidirectedintergal*
\begin{proof}
Let $x^*$ be an extreme point solution of~\eqref{eq:dir_cut_lp}. We will show that $x^*$ is integral.
Because $x^*$ is an extreme point solution of~\eqref{eq:dir_cut_lp}, there is a family $\mathcal{F} \subseteq \{C\in \mathcal{C}_G: x^*(\delta^-_{\vec{L}}(C)) = 1\}$ of $2$-cuts corresponding to $x^*$-tight constraints of~\eqref{eq:dir_cut_lp} such that $x^*$ is the unique solution to
\begin{equation}\label{eq:system_extreme_point}
\begin{aligned}
 x\big(\delta^-_{\vec{L}}(C)\big) =&\ 1 &  &\text{ for all }C\in \Fscr \\
 x_{\ell} =&\ 0 && \text{ for all } \ell\in \vec{L}\setminus \supp(x^*)\enspace.
 \end{aligned}
\end{equation}
We show that we can choose $\Fscr$ to be a laminar family, which will lead to a system~\eqref{eq:system_extreme_point} whose constraint matrix is totally unimodular. This in turn implies the desired integrality of $x^*$. Hence, it remains to show that $\mathcal{F}$ can be chosen to be laminar such that $x^*$ is the unique solution to~\eqref{eq:system_extreme_point}.

To this end, we consider a maximal laminar family $\Wscr \subseteq \Cscr_G$ of tight sets.
Let $\bar L := \supp(x^*) \subseteq \vec{L}$.
We have to show that the system~\eqref{eq:system_extreme_point} with $\mathcal{F}=\mathcal{W}$ has full column rank, which will imply that $x^*$ is its unique solution. Full column rank of~\eqref{eq:system_extreme_point} with $\mathcal{F}=\mathcal{W}$ is equivalent to the property that the vectors $\{\chi^{\delta^-_{\bar L}(C)} : C\in \Wscr\}$ span $Q=\{x\in \mathbb{R}^{\vec{L}}: \supp(x)\subseteq \supp(x^*)\}$. Suppose for the sake of deriving a contradiction that $\{\chi^{\delta^-_{\bar L}(C)} : C\in \Wscr\}$ does not span $Q$.
Then there is a tight set $U \in \Cscr_G \setminus \Wscr$ such that $\chi^{\delta^-_{\bar{L}}(U)}$ is not in the span of
$\{\chi^{\delta^-_{\bar{L}}(C)} : C\in \Wscr\}$. 
Among all such sets let $U$ be one that minimizes the number of sets $W\in \Wscr$ that are crossed by $U$,
i.e., for which $U\cap W$, $W\setminus U$, and $U\setminus W$ are nonempty. (Note that $V\setminus (U\cup W)$ is always nonempty because $r\in V\setminus (U\cup W)$.)

By the maximality of $\Wscr$, there is at least one set $W\in \Wscr$ that crosses $U$.
Then $U\cap W$ and $U\cup W$ are 2-cuts of $G$ by Lemma~\ref{lem:crossing_2_edge_cuts}~\ref{item:four_connected_comp}.
Moreover,
\[
2\ \le\ x^*\big(\delta^-_{\bar{L}}(U\cap W)\big) + x^*\big(\delta^-_{\bar{L}}(U\cup W)\big)\ \le\ x^*\big(\delta^-_{\bar{L}}(U)\big) + x^*\big(\delta^-_{\bar{L}}(W)\big)\  =\ 2\enspace,
\]
which implies that both $U\cup W$ and $U\cap W$ are tight sets and
\[
 \chi^{\delta^-_{\bar{L}}(U)}\ =\  \chi^{\delta^-_{\bar{L}}(U\cap W)} +  \chi^{\delta^-_{\bar{L}}(U\cup W)} -  \chi^{\delta^-_{\bar{L}}(W)}\enspace.
\]
Because $U\cap W$ and $U\cup W$ cross fewer cuts of $\Wscr$ than $U$, our choice of the set $U$ implies that $U\cap W$ and $U\cup W$ are both contained in $\Wscr$.
Hence $\chi^{\delta^-_{\bar{L}}(U)}$ is contained in the span of  $\big\{\chi^{\delta^-_{\bar{L}}(C)} : C\in \Wscr\big\}$, which is a contradiction, and thus shows that $\mathcal{F}=\mathcal{W}$ is a laminar family with the desired properties.
\end{proof}

\subsection{Obtaining the laminar family $\Lscr$}

Next, we define the laminar family $\Lscr$ for which we will impose Chv\'atal-Gomory constraints.
We obtain $\Lscr$ from the dual
\begin{equation}\tag{\ref{eq:dir_cut_dual}}
\renewcommand\arraystretch{1.5}
\begin{array}{r>{\displaystyle}rcll}
\max &\sum_{C\in \mathcal{C}_G} y_C & & & \\
&\sum_{\substack{{C\in \mathcal{C}_G}:\\ \ell \in \delta^-_{\vec{L}}(C)}} y_C &\leq &c_\ell &\forall \ell\in \vec{L}\\[-1.5ex]
&y &\in &\mathbb{R}_{\geq 0}^{\mathcal{C}_G}&
\end{array}
\end{equation}
of the bidirected LP~\eqref{eq:dir_cut_lp}.
More precisely, we choose $\Lscr$ to be the support $\supp(y^*)$ of an optimum dual solution to~\eqref{eq:dir_cut_dual} that is minimal and has laminar support.
Recall that $y\in \mathbb{R}^{\mathcal{C}_G}_{\geq 0}$ is \emph{minimal} if for any $\epsilon > 0$ and any two sets $C_1,C_2 \in \mathcal{C}_G$ with $C_1\subsetneq C_2$, the point $y-\epsilon \cdot \chi^{\{C_2\}} + \epsilon \cdot \chi^{\{C_1\}}$ is not a feasible solution to the dual LP. (See Definition~\ref{def:minimal}.)
An optimum dual solution $y^*$ that is minimal and has laminar support always exists and it can be computed in polynomial time, as we show next. (Actually, even though we do not need to use this fact, our proof techniques show that $y^*$ is unique; we therefore also talk about \emph{the} optimal solution of~\eqref{eq:dir_cut_dual} that is minimal and has laminar support.)

\begin{lemma}\label{lem:solving_simple_dual}
We can compute in polynomial time an optimum solution $y^*$ of the dual LP~\eqref{eq:dir_cut_dual} such that
$y^*$ is minimal and has laminar support.
\end{lemma}
\begin{proof}
First, we compute an optimum solution $\bar y$ of LP~\eqref{eq:dir_cut_dual}.
This is possible in polynomial time because $\Cscr_G$ contains at most one cut for every pair of edges in $G$ and hence the LP has only a polynomial number of constraints.
Then we compute a vector $y^*$ minimizing the linear objective $\sum_{C\in \Cscr_G} \sqrt{|C|} \cdot y_C$ over the face of optimum solutions of~\eqref{eq:dir_cut_dual}, i.e., over the set of feasible solutions $y$ of \eqref{eq:dir_cut_dual} with $\sum_{C\in \mathcal{C}_G}y_C = \sum_{C\in \mathcal{C}_G}\bar y_C$.
We claim that $y^*$ has the desired properties. 

First, $y^*$ is minimal because otherwise there were sets $C_1,C_2 \in \mathcal{C}_G$ with $C_1 \subsetneq C_2$ such that $y' := y^* - \epsilon \cdot \chi^{\{C_2\}} + \epsilon \cdot \chi^{\{C_1\}}$ is a feasible dual solution for some $\epsilon > 0$. 
But then $y'$ is also contained in the face of optimum solutions of~\eqref{eq:dir_cut_dual} and 
$\sum_{C\in \Cscr_G} \sqrt{|C|} \cdot y'_C < \sum_{C\in \Cscr_G}\sqrt{|C|} \cdot y^*_C$, a contradiction.
Second, $y^*$ has laminar support for the following reason.
Suppose there are sets $A,B \in \supp(y^*)$ such that $A$ and $B$ cross, i.e., $A\cap B$, $A\setminus B$, and $B\setminus A$ are nonempty. Then $A\cap B$ and $A\cup B$ are contained in $\Cscr_G$ (by Lemma~\ref{lem:crossing_2_edge_cuts}~\ref{item:four_connected_comp})
 and
\[
\chi^{\delta^-_{\vec{L}}(A\cup B)} +\chi^{\delta^-_{\vec{L}}(A\cap B)} \le \chi^{\delta^-_{\vec{L}}(A)} + \chi^{\delta^-_{\vec{L}}(B)}\enspace.
\]
Hence decreasing $y^*_{A}$ and $y^*_{B}$ by $\epsilon := \min\{ y^*_A, y^*_B\}$ while increasing $y^*_{A \cup B}$ and $y^*_{A \cap B}$ by $\epsilon$ maintains an optimum dual solution, but decreases $\sum_{C\in \Cscr_G}\sqrt{|C|} \cdot y^*_C$, a contradiction.
Therefore, $y^*$ has laminar support.
\end{proof}

We can use a dual solution $y^*$ as in Lemma~\ref{lem:solving_simple_dual} in order to bound the cost for completing a set $R$ of cross-links to a feasible CacAP solution.
This is formally stated by Lemma~\ref{lem:complete_cross_links}, which we recall below.
We emphasize that it is crucial that $y^*$ is minimal.

\lemcompletecrosslinks*
\begin{proof}
Let $y$ be the vector resulting from $y^*$ by restricting to sets $C \in \mathcal{C}_G$ with $R \cap \delta_L(C)=\emptyset$.
Then $y$ is a feasible solution to the LP~\eqref{eq:dir_cut_dual} for the residual instance $\Iscr'=(G',L')$ of $(G,L)$ with respect to $R$.
We claim that $y$ is an optimum solution. 
Before showing this, we observe that it readily implies the desired result, due to the following.
The point $y$ being an optimal solution implies that the value of~\eqref{eq:dir_cut_dual} for the residual instance $\Iscr'$ equals 
$\sum_{C\in \mathcal{C}_G:\\\delta_L(C)\cap R=\emptyset} y^*_C$. By strong duality, it is also equal to the optimum value of~\eqref{eq:dir_cut_lp}.
We compute an optimum extreme point solution to~\eqref{eq:dir_cut_lp}, which is integral by Lemma~\ref{lem:bidirected_integral}.
Dropping the orientation of the links in this integral solution yields the desired set $F\subseteq L$ of cost 
\[
c(F)\le \sum_{C\in \Cscr_G: R\cap \delta_L(C)=\emptyset} y^*_C \enspace.
\]

Hence it remains to prove that $y$ is an optimum solution of the LP~\eqref{eq:dir_cut_dual} for the residual instance $\Iscr'$.
Suppose this is not the case. 
By Lemma~\ref{lem:solving_simple_dual}, the LP~\eqref{eq:dir_cut_dual} for the residual instance $\Iscr'$  has an optimum
solution $z$ that is minimal and has laminar support.
Because we assumed that $y$ is not optimal, we have
\[
\sum_{C \in \Cscr_{G'}} z_C\ >\ \sum_{C\in \Cscr_{G'}} y_C
\]
and hence there exists a set $W\in \Cscr_{G'}$ with $z_W > y_W$.
Let $W$ be a minimal set with this property.
Then we have $y_C \ge z_C$ for all $C\in \Cscr_{G'}$ with $C\subsetneq W$.

Suppose we have $y_{U} > z_{U}$ for some $U\in \Cscr_{G'}$ with $U\subsetneq W$.
Let $\epsilon := \min\{z_W - y_W,\  y_{U} - z_{U}\} > 0$.
We show that decreasing $z_W$ by $\epsilon$ and increasing $z_{U}$ by $\epsilon$ maintains feasibility of $z$, contradicting the minimality of $z$.
First, we observe that we chose $\epsilon$ such that the variable $z_W$ remains non-negative.
The only constraints that could potentially get violated by the change of $z$ correspond to directed links $\ell\in \delta_{\vec{L}}^-(U)\setminus \delta_{\vec{L}}^-(W)$.
Because $U\subsetneq W$, this implies $\ell=(v,w)$ with $v \in W\setminus U$ and $u\in U$.
We have chosen $\epsilon$ small enough such that we maintain $z_C \le y_C$ for all $C\in \Cscr_{G'}$ with $C\subsetneq W$.
But then
\[
\sum_{C\in \Cscr_{G'}: \ell\in \delta_{\vec{L}}^-(C)} z_C\ \le\ \sum_{C\in \Cscr_{G'}: \ell\in \delta_{\vec{L}}^-(C)} y_C\ \le\ c_{\ell}
\] 
by the feasibility of $y$.
This shows that decreasing $z_W$ by $\epsilon$ and increasing $z_{U}$ by $\epsilon$ maintains feasibility of~$z$, a contradiction.
Therefore, we have shown that $y^*_C = y_C = z_C$ for all sets $C\in \Cscr_{G'}$ with $C\subsetneq W$.

Because $R\subseteq L_{\mathrm{cross}}$ is a set of cross-links, we have
\[
\Cscr_{G'} = \{ C \in \Cscr_G : R \cap \delta_L(C) =\emptyset \} = \{ C \in \Cscr_G : C \text{ does not contain an endpoint of a link in }R \} 
\]
and hence $W\in \Cscr_{G'}$ implies that $C\in \Cscr_{G'}$ for all $C\in \Cscr_G$ with $C\subsetneq W$.
Thus,
\begin{align}\label{eq:z_equals_y_inside_W}
y^*_C\ =&\ z_C \qquad \text{ for all } C\in \Cscr_G\text{ with }C\subsetneq W.
\end{align}
Now we distinguish two cases.

\vspace*{3mm}
\noindent \textbf{Case 1:} There exists a set $U\in \supp(y^*)$ that crosses $W$, i.e., $U\setminus W$, $W\setminus U$, and $W\cap U$ are nonempty.
\vspace*{3mm}

Let $U$ be a minimal set with this property. 
First, we observe that $U\cap W$ is a 2-cut of $G$, i.e., $U\cap W \in \Cscr_G$,  by Lemma~\ref{lem:crossing_2_edge_cuts}~\ref{item:four_connected_comp}.
Let $\epsilon =\min\{z_W - y^*_W,\ y^*_{U}\} > 0$ and let $\tilde y$ result from $y^*$ by decreasing $y^*_{U}$ by $\epsilon$ and increasing $y^*_{U\cap W}$ by $\epsilon$. 
We claim that $\tilde y$ is a feasible solution of~\eqref{eq:dir_cut_dual} for $\Iscr$.
Then it is clearly also an optimum solution, contradicting the minimality of $y^*$.

Let us now prove that $\tilde y$ is indeed a feasible solution of~\eqref{eq:dir_cut_dual} for $\Iscr$.
We chose $\epsilon$ small enough such that $\tilde y_{U}$ is non-negative.
Therefore, the only constraints of~\eqref{eq:dir_cut_dual} that could potentially be violated correspond to directed links $\ell\in \delta_{\vec{L}}^-(U\cap W) \setminus \delta_{\vec{L}}^-(U)$.
Fix such a link $\ell=(v,w)$. Then $v\in U\setminus W$ and $w\in U \cap W$.
Now consider a set $C\in \Cscr_G \setminus \{W\}$ with $\ell\in \delta_{\vec{L}}^-(C)$ and $y^*_C > 0$. 
Then $w \in C\cap W$ and hence $C \cap W \ne \emptyset$. (See the first picture in Figure~\ref{fig:intersecting_sets}.)
Moreover, because $v\in U \setminus C$ and the support of $y^*$, which contains both $C$ and $U$, is laminar, we have $C\subsetneq U$.
By the minimality of $U$, the set $C$ does not cross $W$ and thus $C\subseteq W\cap U$.
This implies $z_C =y^*_C$ by \eqref{eq:z_equals_y_inside_W}.
Using  $\epsilon \le z_W - y^*_W$ and the feasibility of $z$, we conclude
\begin{align*}
\sum_{C \in \Cscr_G: \ell \in \delta_{\vec{L}}^-(C)} \tilde y_{C} =&\ \tilde y_W+\sum_{C\in \Cscr_G\setminus\{W\}: \ell \in \delta_{\vec{L}}^-(C)} \tilde y_C \\
=&\  y^*_W   +\epsilon + \sum_{ C \in \Cscr_G\setminus \{W\}: \ell \in \delta_{\vec{L}}^-(C) } y^*_C \\
\le&\  z_W + \sum_{C \in \Cscr_G\setminus\{W\}: \ell \in \delta_{\vec{L}}^-(C)} z_C \\
\le&\ c_{\ell}\enspace.
\end{align*}
This shows that $\tilde y$ is indeed a feasible solution, which contradicts the minimality of $y^*$ as we have observed above.
\begin{figure}[t]
\begin{center}
\begin{tikzpicture}[scale=0.23]
\usetikzlibrary{fit}

\tikzset{
  prefix node name/.style={%
    /tikz/name/.append style={%
      /tikz/alias={#1##1}%
    }%
  }
}

\tikzset{link/.style={line width=1.5pt}}

\tikzset{node/.style={thick,draw=black,fill=black,circle,minimum size=4, inner sep=1pt}}

\newcommand\set[2][]{
\begin{scope}[prefix node name=#1]

\begin{scope}
\node(1) at (7,7) {};
\node(2) at (15,7) {};
\node(3) at (10,2){};
\node(4) at (23,7){};
\node(5) at (18,2){};
\node(6) at (15,4) {};

\node[fit={(1)(2)(3)},draw, ellipse,minimum width=3cm,label={above left:$U$}](left){};
\node[fit={(2)(4)(5)},draw, ellipse,minimum width=3cm,label={above right:$W$}](left){};
\node[fit={(6)},draw, ellipse,minimum width=1cm,label={above:$C$}](left){};
\end{scope}

\begin{scope}[every node/.append style=node]
\node[label={right:$w$}] (7) at (14.5,4) {};
\node[label={left:$v$}] (8) at (7.5,4) {};
\end{scope}

\begin{scope}[very thick,dashed,stealth-]
\draw (7) to[bend right] node[pos=0.7,above] {$\ell$} (8);
\end{scope}

\end{scope}
}%

\newcommand\cas[2][]{
\begin{scope}[prefix node name=#1]

\begin{scope}
\node(9) at (15,4.5){};
\node(10) at (15,5){};
\node(11) at (18,5){};
\node(12) at (15,7){};
\node(13) at (13,4){};
\node(14) at (10,3){};
\node(15) at (23,9){};

\node[fit={(9)(10)},draw, ellipse,minimum width=0.5cm,label={above:$C$}](right){};
\node[fit={(9)(10)(11)(12)(13)},draw, ellipse,minimum width=1.5cm,label={above right:$W$}](right){};
\node[fit={(14)(15)(11)(12)(13)},draw, ellipse,minimum width=2.5cm,label={above right:$U$}](right){};
\end{scope}

\begin{scope}[every node/.append style=node]
\node(16) at (15,5){};
\node(17) at (23,5){};

\end{scope}

\begin{scope}[very thick,dashed,stealth-]
\draw (16) to [bend left] (17);
\node() at (18.5,7.5) {$\ell$};
\end{scope}

\end{scope}
}%

\begin{scope}[scale=0.8, shift={(-25,4)}]
\set[c1-]{}
\end{scope}

\begin{scope}[scale=0.8,shift={(8,4)}]
\cas[c2-]{}
\end{scope}

\end{tikzpicture} \end{center}
\caption{Illustration of sets $C$, $U$, and $W$ for the cases~$1$ and~$2$ in the proof of Lemma~\ref{lem:complete_cross_links}.}\label{fig:intersecting_sets}
\end{figure}
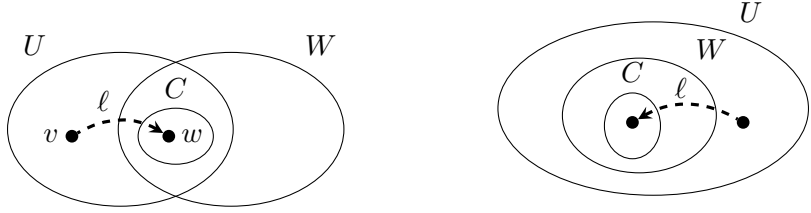

\vspace*{3mm}
\noindent\textbf{Case 2:} There is no set $U\in \supp(y^*)$ that crosses $W$. 
\vspace*{3mm}

If $W$ is a maximal set in the support of $y^*$, then, for every directed link $\ell\in \delta_{\vec{L}}^-(W)$, we have
\[
\sum_{C\in\Cscr_G: \ell\in \delta_{\vec{L}}^-(C)} y^*_C\ =\  y^*_W - z_W +\sum_{C\in\Cscr_G: \ell\in \delta_{\vec{L}}^-(C)} z_C\ \le\ c_{\ell} - (z_W -y^*_W)\enspace,
\]
because $y^*_C = z_C$ for all sets $C\subsetneq W$ by \eqref{eq:z_equals_y_inside_W} and because $z$ is a feasible solution of \eqref{eq:dir_cut_dual}.
Therefore, increasing $y^*_W$ by $z_W -y^*_W >0$ maintains feasibility, contradicting the optimality of $y^*$.

Now consider the remaining case where $W$ is not a maximal element of the support of $y^*$.
Then there exists a set $U\in \supp(y^*)$ with $W\subsetneq U$.
Let $U$ be a minimal set with this property.
For every link $\ell\in \delta_{\vec{L}}^-(W)\setminus \delta_{\vec{L}}^-(U)$, we have that all sets $C\in \supp(y^*)$ with $\ell\in \delta_{\vec{L}}^-(C)$ are (not necessarily strict) subsets of $W$ because we chose $U$ minimal (and we are in Case~2). See the second picture in Figure~\ref{fig:intersecting_sets}.
Recall that by  \eqref{eq:z_equals_y_inside_W}  we have $y^*_C = z_C$ for all sets $C\subsetneq W$.
Thus, when $\epsilon:= \min\{y^*_{U}, z_W - y^*_W\} > 0$, we have
\begin{equation*}
c_{\ell}\ \ge\ \sum_{C\in \Cscr_G: \ell\in \delta_{\vec{L}}^-(C)} z_C\ \ge\ \epsilon + \sum_{C\in \Cscr_G: \ell\in \delta_{\vec{L}}^-(C)} y^*_C \qquad \qquad \forall\; \ell\in \delta_{\vec{L}}^-(W)\setminus \delta_{\vec{L}}^-(U)\enspace.
\end{equation*}
Hence, increasing $y^*_W$ by $\epsilon$ and decreasing $y^*_{U}$ by $\epsilon$ maintains a feasible solution of~\eqref{eq:dir_cut_dual}.
This contradicts the minimality of $y^*$.

\vspace*{5mm}
We obtained a contradiction in all cases.
This shows that $y$ is indeed an optimum solution of~\eqref{eq:dir_cut_dual} for $\Iscr'$ and completes the proof.
\end{proof}

\subsection{Saving on cross-links via Chv\'atal-Gomory cuts}

Finally, recall that for a laminar family $\Lscr$, the polytope $P_{\mathrm{CG}}^{\mathcal{L}}$ arises from
\begin{equation*}
P_{\mathrm{cut}}^{\mathcal{L}}\coloneqq
 \left\{x\in \mathbb{R}_{\geq 0}^L \colon x(\delta_L(C))\geq 1 \;\forall C\in \mathcal{L}\right\}
\end{equation*}
by adding all its $\{0,\sfrac{1}{2}\}$-CG cuts.
We now prove that if we choose $\Lscr:= \supp(y^*)$ to be the support of the minimal optimum solution $y^*$ of \eqref{eq:dir_cut_dual} with laminar support,
the polytope $P_{\mathrm{cross}}\coloneqq P_{\mathrm{cut}}\cap P_{\mathrm{CG}}^{\mathcal{L}}$ has the desired properties.
The proof of Proposition~\ref{prop:overview_good_laminar}, which we restate below, completes the proof of Lemma~\ref{lem:cross-link_rounding}.
 
\propgoodlaminar*
\begin{proof}
We define $\Lscr:= \supp(y^*)$, where $y^*$ is the minimal optimum solution $y^*$ of \eqref{eq:dir_cut_dual} with laminar support.
This laminar family $\Lscr$ can be computed efficiently by Lemma~\ref{lem:solving_simple_dual}. 
Moreover, we can efficiently separate over $P_{\mathrm{cut}}$ because it has only a polynomial number of variables and constraints.
Using that $P_{\mathrm{cut}}^{\mathcal{L}}$ is the cut-relaxation of a TAP problem, we can conclude from the work of Fiorini et al.~\cite{fiorini_2018_approximating} that
we can also efficiently separate over $P_{\mathrm{CG}}^{\mathcal{L}}$.
Hence, because $P_{\mathrm{cross}}= P_{\mathrm{cut}}\cap P_{\mathrm{CG}}^{\mathcal{L}}$ has facet complexity bounded by $\poly(|V|)$, we can efficiently optimize over $P_{\mathrm{cross}}$.

Let $x\in P_{\mathrm{cross}}$.
To show that $P_{\mathrm{cross}}$ fulfills the properties stated in Lemma~\ref{lem:cross-link_rounding},
we need to show that we can efficiently obtain a solution $F\subseteq L$ satisfying $c(F) \leq c^T x + \sum_{\ell\in L_{\mathrm{in}}\setminus L_{\mathrm{up}}} c_\ell x_\ell$.
First notice that we can interpret $x$ as a fractional solution to the TAP instance over the family $\mathcal{L}$, i.e., it fulfills
\begin{equation}\label{eq:x_as_TAP_sol}
x(\delta_L(C)) \geq 1 \qquad \forall\; C\in \mathcal{L}\enspace.
\end{equation}
To round $x$, we first map $x$ to a solution $z$ where we replace in-links by directed links as we did in Section~\ref{sec:bidirected_lp}. More precisely, each in-link that is not an up-link will be replaced by two anti-parallel directed links, and all up-links are directed downwards. Formally, let
\begin{equation*}
\vec{L}_{\mathrm{in}} \coloneqq \bigcup_{\{u,w\}\in L_{\mathrm{in}}\setminus L_{\mathrm{up}}}\left\{ (u,w), (w,u) \right\} \cup \bigcup_{\ell \in L_{\mathrm{up}}} \left\{ (\topv(\ell), \bottomv(\ell) ) \right\}\enspace,
\end{equation*}
where for $\ell\in L_{\mathrm{up}}$, we denote by $\topv(\ell)$ the endpoint of $\ell$ closer to the center and by $\bottomv(\ell)$ the one farther away from the center. Let $z\in \mathbb{R}_{\geq 0}^{L_{\mathrm{cross}}\cup \vec{L}_{\mathrm{in}}}$ be defined by $z_\ell = x_\ell$ for $\ell\in L_{\mathrm{cross}}$ and $z_{(u,w)} = x_{\{u,w\}}$ for $(u,w)\in \vec{L}_{\mathrm{in}}$. We now interpret the links $L_{\mathrm{cross}}\cup \vec{L}_{\mathrm{in}}$ as links in the TAP problem that seeks to cover the cuts $\mathcal{L}$. More precisely, a cut $C\in \mathcal{L}$ is covered by a link $\ell\in L_{\mathrm{cross}}\cup \vec{L}_{\mathrm{in}}$ if either $\ell \in L_{\mathrm{cross}}$ and $\ell$ crosses $C$, or $\ell\in \vec{L}_{\mathrm{in}}$ enters $C$. As discussed in Section~\ref{sec:bidirected_lp}, for an instance of TAP this is analogous to splitting in-links that are not up-links into two up-links covering the same cuts. In particular, we have by~\eqref{eq:x_as_TAP_sol} that
\begin{equation}\label{eq:z_as_TAP_sol}
z(\delta_{L_{\mathrm{cross}}}(C)) + z(\delta_{\vec{L}_{\mathrm{in}}}^-(C)) \geq 1 \quad \forall C\in \mathcal{L}\enspace.
\end{equation}
In short, $z$ corresponds to a solution for the TAP problem over the cuts $\mathcal{L}$ without in-links that are not up-links.
Moreover, because $x\in P_{\mathrm{CG}}^{\mathcal{L}}$, also $z$ fulfills the $\{0,\sfrac{1}{2}\}$-CG constraints for the TAP problem over $\mathcal{L}$. This allows us to invoke a result by Fiorini et al.~\cite{fiorini_2018_approximating}, which shows that such TAP solutions can be rounded losslessly because when all in-links are up-links then $P^{\mathcal{L}}_{\mathrm{CG}}$ is integral. Formally speaking, one can efficiently obtain a set $F^{\Lscr}\subseteq L_{\mathrm{cross}}\cup \vec{L}_{\mathrm{in}}$ such that
\begin{enumerate}\itemsep3pt
\item\label{item:mixed_F_is_TAP_sol} $F^{\Lscr} \cap \left( \delta_{L_{\mathrm{cross}}}(C) \cup \delta_{\vec{L}_{\mathrm{in}}}^-(C) \right) \neq \emptyset \;\;\forall C\in \mathcal{L}$, and

\item\label{item:mixed_F_is_cheap} $c(F^{\Lscr}) \leq c^T z = c^T x + \sum_{\ell\in L_{\mathrm{in}}\setminus L_{\mathrm{up}}} c_\ell x_\ell$.
\end{enumerate}
\smallskip
Notice that although~\ref{item:mixed_F_is_TAP_sol} implies that $F^{\Lscr}$ covers every cut in $\Lscr$, the link set $F^{\Lscr}$ does not necessarily correspond to a feasible CacAP solution.
However, as we show next, we can obtain a CacAP solution from $F^{\Lscr}$ by replacing $\vec{F}_{\mathrm{in}}\coloneqq F^{\Lscr}\cap \vec{L}_{\mathrm{in}}$ by a suitable link set $S$ of cost no more than $c(\vec{F}_{\mathrm{in}})$.

We invoke Lemma~\ref{lem:complete_cross_links} with the set $R\coloneqq F^{\Lscr}\cap L_{\mathrm{cross}}$ to obtain $S\subseteq L$ satisfying
\begin{enumerate}[label=(\alph*)]
\item $R \cup S$ is a CacAP solution, and
\item\label{item:cost_U_dual_bound} $c(S) \leq \sum_{\substack{C\in \mathcal{C}_G}:\\ \delta_L(C)\cap R = \emptyset} y_C^*$ where $y^*$ is the optimal solution of~\eqref{eq:dir_cut_dual} that is minimal and has laminar support.
\end{enumerate}
\smallskip
The link set $R\cup S$ is the solution we return. To bound $c(R\cup S)$, we first observe that
\begin{equation}\label{eq:cost_bound_U}
c(S) \ \leq \sum_{\substack{C\in \mathcal{C}_G:\\ \delta_L(C)\cap R =\emptyset}} y_C^*
\ \leq\ \sum_{\vec{\ell}\in \vec{F}_{\mathrm{in}}}
     \sum_{\substack{C \in \mathcal{C}_G:\\ \vec{\ell}\in \delta_{\vec{L}}^-(C)}} y_C^*
\ \leq\ \sum_{\vec{\ell}\in \vec{F}_{\mathrm{in}}} c_{\vec{\ell}}
\ =\ c(\vec{F}_{\mathrm{in}})\enspace,
\end{equation}
where the first inequality follows from~\ref{item:cost_U_dual_bound}, the second one from~\ref{item:mixed_F_is_TAP_sol}, and the third one from $y^*$ being a solution to the dual LP~\eqref{eq:dir_cut_dual}. Finally, the desired cost bound on $R\cup S$ now follows from
\begin{align*}
c(R\cup S) &\ =\ c(R) + c(S)
\ \leq\ c(R) + c(\vec{F}_{\mathrm{in}})
\ =\ c(F^{\Lscr})
\ \leq\ c^T x + \sum_{\ell\in L_{\mathrm{in}}\setminus L_{\mathrm{up}}} c_\ell x_\ell \enspace,
\end{align*}
where the first inequality follows from~\eqref{eq:cost_bound_U} and the second one from~\ref{item:mixed_F_is_cheap}.

\end{proof}

\section{Going below $\bm{1.5}$ through stack analysis}\label{sec:stack_analysis}

In this section we describe a strengthening of the simple bundle rounding approach for $k$-wide instances.
Combining this with the cross-link rounding from the previous section and the black-box reduction to $k$-wide instances will yield the claimed \apxfac-approximation for CAP.
In contrast to the previous section, we now restrict ourselves again to the unweighted Cactus Augmentation Problem.
The algorithm we present is similar to the rewiring algorithm used in \cite{grandoni_2018_improved} for TAP.
We propose a novel analysis leading to a better approximation ratio.

\subsection{$L_{\mathrm{cross}}$-minimal solutions and an LP relaxation}

Our algorithm for $k$-wide instances that improves on the bundle rounding (Lemma~\ref{lem:bundle_rounding}) will first solve a linear program.
The goal of this subsection is to introduce this LP. 

We denote the vertex sets of the connected components of $G-r$ by $V_1,\dots,V_q$.
Then $G[V_i \cup \{r\}]$ for $i\in [q]$ are the principal subcacti of $G$.
Moreover, we define $L_i\subseteq L$ to be the set of links that have at least one endpoint in $V_i$.
Then every in-link is contained in precisely one of the sets $L_i$ and every cross-link is contained in precisely two sets $L_i$.
Now we consider partial solutions covering the 2-cuts of only one of the principal subcacti of $G$.
More precisely, we say that a set $F\subseteq L_i$ of links is a feasible solution for the principal subcactus $G_i$ if $F\cap \delta_L(C)\ne  \emptyset$ for every 2-cut $C$ with $C\subseteq V_i$.
Note that a set $F\subseteq L$ is a feasible solution for the instance $(G,L)$ if and only if $F\cap L_i$ is a feasible solution for the principal subcactus $G_i$ for all $i\in [q]$.

Our LP will have the property that, for every feasible solution $x\in [0,1]^L$, the restriction $x|_{L_i}$ of the vector $x$ to $L_i$
 is a convex combination of incidence vectors of feasible solutions for $G_i$ for all $i\in[q]$.
Moreover, we will ensure that $x|_{L_i}$  is not only a convex combination of incidence vectors of arbitrary  feasible solutions for $G_i$, but these solutions will have a property that we call \emph{$L_{\mathrm{cross}}$-minimal}.
The notion of $L_{\mathrm{cross}}$-minimality is inspired by the notion of shadow-minimality for TAP  used in \cite{grandoni_2018_improved}.

In order to define $L_{\mathrm{cross}}$-minimality, we need some preparation.
A link $\{\bar v, \bar w\}\ne \{v,w\}$ is a \emph{shadow} of a link $\{v,w\}$ if $\bar v$ and $\bar w$ are contained in every $v$-$w$ path in $G$.
(Note that $v=\bar v$ or $w=\bar w$ is allowed but $\{\bar v,\bar w\} =\{v,w\}$ is not.)
We say that an instance $(G,L)$ is \emph{shadow-complete} if for every link $\ell\in L$ all shadows of $\ell$ are contained in $L$.
\begin{lemma}\label{lemma:justify_shadow_notion}
Let $\bar \ell$ be a shadow of $\ell\in L$. Then every 2-cut  of $G$ that is covered by $\bar \ell$ is also covered by $\ell$.
\end{lemma}
\begin{proof}
Suppose there is a 2-cut $C$ of $G$ that is covered by $\bar\ell$ but not by $\ell=\{v,w\}$.
The induced subgraphs $G[C]$ and $G[V\setminus C]$ are connected because the cactus $G$ is 2-edge-connected and $C$ is a 2-cut in $G$.
Moreover, because $\ell$ does not cover the cut $C$, we have either $v,w\in C$ or $v,w\in V\setminus C$.
We conclude that there either is a $v$-$w$ path in $G$ that visits only vertices in $C$, or a $v$-$w$ path that visits only vertices in $V\setminus C$.
Because $\bar\ell$ is a shadow of $\ell=\{v,w\}$ this path must visit both endpoints of $\bar \ell$, contradicting the fact that $\bar\ell$ covers $C$.
\end{proof}

If $(G,L)$ is not shadow-complete, we add the shadows of all $\ell\in L$ to $L$ and solve the resulting instance.
Then, whenever we include a shadow of $\ell$ in the solution, we replace it by $\ell$.
This replacement maintains a feasible solution by Lemma~\ref{lemma:justify_shadow_notion}.
Therefore, we now assume w.l.o.g.~that $(G=(V,E),L)$ is shadow-complete.

Any shadow-complete instance has an optimal \emph{shadow-minimal} solution, i.e., a solution where none of the links can be omitted or replaced by one of its shadows. (See Lemma~\ref{lem:existence_optimum_minimal_solution} below.)
While this notion of minimality seems natural, we will instead work with a weaker form of minimality which we call $L_{\mathrm{cross}}$-minimality.
This will suffice for our purpose and it is easier to ensure that $x|_{L_i}$ is a convex combination of $L_{\mathrm{cross}}$-minimal solutions
than to ensure that it is a convex combination of shadow-minimal solutions.

Let us now define $L_{\mathrm{cross}}$-minimality or more generally $L'$-minimality for any set $L'$ of links. 
We say that a link $\ell_1\in L$ is \emph{minimal with respect to $\ell_2\in L$} if for any shadow $\bar{\ell_1}$ of $\ell_1$, the $2$-cuts covered by $\{\bar{\ell_1},\ell_2\}$ are a strict subset of those covered by $\{\ell_1,\ell_2\}$ and the 2-cuts covered by $\{\ell_2\}$ are a strict subset of those covered by $\{\ell_1,\ell_2\}$; or formally, for any shadow $\bar{\ell_1}$ of $\ell_1$,
\begin{align*}
\{ C\in \mathcal{C}_G : \{\bar\ell_1,\ell_2\} \cap \delta_L(C) \ne \emptyset\} \subsetneq&\ \{ C  \in \mathcal{C}_G : \{\ell_1, \ell_2\} \cap \delta_L(C) \ne \emptyset\}\enspace,\text{ and} \\
\{ C\in \mathcal{C}_G : \{\ell_2\} \cap \delta_L(C) \ne \emptyset\} \subsetneq&\ \{ C  \in \mathcal{C}_G : \{\ell_1, \ell_2\} \cap \delta_L(C) \ne \emptyset\}
\enspace.
\end{align*}
Figure~\ref{fig:shortening_example} illustrates the concept of minimality of a link $\ell_1$ with respect to another link $\ell_2$.
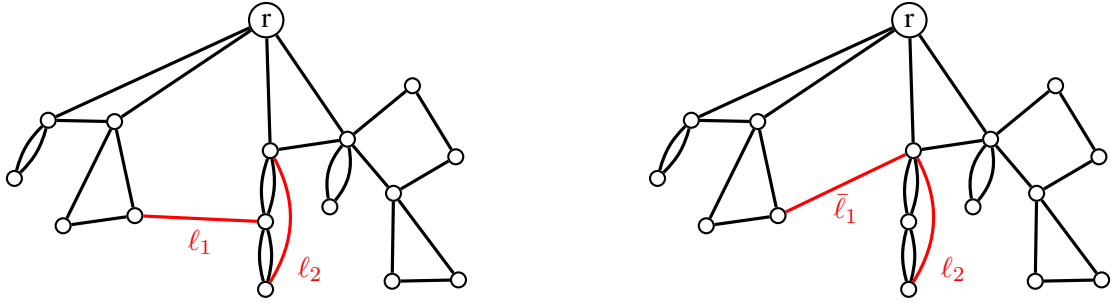
\begin{figure}[!ht]
\begin{center}
\begin{tikzpicture}[scale=0.45]

\pgfdeclarelayer{bg}
\pgfdeclarelayer{fg}
\pgfsetlayers{bg,main,fg}

\tikzset{
  prefix node name/.style={%
    /tikz/name/.append style={%
      /tikz/alias={#1##1}%
    }%
  }
}

\tikzset{root/.style={fill=white,minimum size=13}}

\tikzset{lks/.style={line width=1pt, blue, densely dashed}}

\tikzset{chosen/.append style={red,line width=1.5pt,very thick}}

\newcommand\cac[2][]{

\begin{scope}[prefix node name=#1]

\tikzset{npc/.style={#2}}

\begin{scope}[every node/.append style={thick,draw=black,fill=white,circle,minimum size=5, inner sep=2pt}]
\node[root]  (1) at (18.68,-1.64) {r};
\node  (2) at (18.80,-5.50) {};
\node  (3) at (21.08,-5.16) {};
\node  (4) at (18.66,-7.60) {};
\node  (5) at (20.56,-7.16) {};
\node  (7) at (22.44,-6.76) {};
\node  (8) at (24.28,-5.68) {};
\node  (9) at (23.00,-3.58) {};
\node (10) at (22.40,-9.36) {};
\node (11) at (24.36,-9.32) {};

\begin{scope}[npc]
\node (12) at (12.24,-4.60) {};
\node (13) at (14.20,-4.64) {};
\node (14) at (11.24,-6.32) {};
\node (15) at (12.68,-7.72) {};
\node (16) at (14.80,-7.42) {};
\node (6) at (18.66,-9.60) {};

\end{scope}
\end{scope}

\begin{scope}[every node/.append style={font=\scriptsize}]

\foreach \i/\t in {2/,3/,4/,5/,6/,7/,8/,9/,10/,11/,12/,13/,14/,15/,16/} {
\pgfmathparse{int(\i-1)}
}
\end{scope}

\begin{scope}[very thick]

\draw  (1) --  (2);
\draw  (2) --  (3);
\draw  (3) --  (1);
\draw  (9) --  (3);
\draw  (2) to[bend left=15] (4);
\draw  (4) to[bend left=15] (2);
\draw  (4) to[bend left=15] (6);
\draw  (6) to[bend left=15] (4);
\draw  (3) to[bend left=25] (5);
\draw  (5) to[bend left=25] (3);
\draw  (3) --  (7);
\draw  (7) --  (8);
\draw  (8) --  (9);
\draw  (7) -- (10);
\draw (10) -- (11);
\draw (11) --  (7);

\begin{scope}[npc]
\draw  (1) -- (12);
\draw (13) --  (1);
\draw (12) to[bend left=15] (14);
\draw (14) to[bend left=15] (12);
\draw (16) -- (13);
\draw (12) -- (13);
\draw (13) -- (15);
\draw (15) -- (16);
\end{scope}

\end{scope}

\end{scope}

}%

\begin{scope}[shift={(-10,-12)}]

\begin{scope}[shift={(0,0)}]
\cac[c1-]{}
\end{scope}

\begin{scope}[chosen]
\draw (16) --node[pos=0.5, below]{$\ell_1$} (4);
\draw (2)  to[bend left] node[pos=0.7,below right]{$\ell_2$} (6);
\end{scope}

\begin{scope}[shift={(19,0)}]
\cac[c2-]{}
\end{scope}

\begin{scope}[chosen]
\draw (16) --node[pos=0.5, below]{$\bar \ell_1$} (2);
\draw (2)  to[bend left] node[pos=0.7,below right]{$\ell_2$} (6);
\end{scope}

\end{scope}

\end{tikzpicture}
 \end{center}
\caption{In the picture on the left, $\ell_1$ is not minimal with respect to $\ell_2$. Indeed, the shadow $\bar{\ell}_1$ of $\ell_1$, as shown on the right, is such that  $\{\bar{\ell}_1,\ell_2\}$ and $\{\ell_1,\ell_2\}$ cover the same set of $2$-cuts.
Moreover, in our example, the shadow $\bar \ell_1$ of $\ell_1$ is minimal with respect to $\ell_2$.}\label{fig:shortening_example}
\end{figure}

\begin{definition}
Let $L'\subseteq L$.
Then a set $F\subseteq L$ is $L'$-minimal if for every two distinct links $\ell'\in F\cap L'$ and $\ell \in F$, the link $\ell'$ is minimal with respect to $\ell$.
\end{definition}
In particular, one can observe that if a set $F\subseteq L$ is $L_{\mathrm{cross}}$-minimal, then no two endpoints of cross-links in $F$ have an ancestry relationship.

Next, we show that every shadow-complete instance has an $L$-minimal optimum solution, which is therefore also $L'$-minimal for every $L'\subseteq L$.
\begin{lemma}\label{lem:existence_optimum_minimal_solution}
Let $(G,L)$ be a shadow-complete instance of cactus augmentation.
Then there is an optimum solution that is shadow-minimal and, hence, also $L$-minimal.
\end{lemma}
\begin{proof}
Let $F$ be an optimum solution. If $F$ is not shadow-minimal, there is a link in $F$ that can be omitted or replaced by one of its shadows while maintaining feasibility.
Iterating this yields a shadow-minimal solution.
The procedure terminates because this replacement operation decreases the following integral and non-negative potential function:
\[
    \sum_{\{v,w\}\in F} |\{u\in V: u\text{ is contained in every $v$-$w$ path in } G\}|\enspace,
\]
which finishes the proof.
\end{proof}

We define
\[
P^{\min}(G_i,L_i) \coloneqq \mathrm{conv}\left(\{ \chi^F : F \subseteq L_i\text{ is an $L_{\mathrm{cross}}$-minimal feasible solution for }G_i\}\right)\enspace
\]
to be the convex hull of incidence vectors of $L_{\mathrm{cross}}$-minimal feasible CacAP solutions for $G_i$.
Next, we prove that we can optimize over $P^{\min}(G_i)$ in polynomial time when $k$ is constant.
We need the following observation.

\begin{lemma}\label{lem:simple_observation_cross_link_domination}
Let $i\in [q]$ and let $\ell_1=\{v_1,w_1\}, \ell_2=\{v_2,w_2\} \in L_{\mathrm{cross}} \cap L_i$ be cross-links with $v_1,v_2\in V_i$.
If $v_1$ is an ancestor of $v_2$, then every 2-cut $C\in \mathcal{C}_{G_i}$ of $G_i$ that is covered by $\ell_1$ is also covered by $\ell_2$.
\end{lemma}
\begin{proof}
Let $C\in \mathcal{C}_{G_i}$ be a 2-cut of $G_i$ covered by $\ell_1$. 
Then $v_1\in C$ because $\ell_1$ is a cross-link and hence its other endpoint $w_1$ is not contained in $V_i\supseteq C$.
We need to show that $v_2 \in C$ because the other endpoint $w_2$ of $\ell_2$ is not contained in  $V_i\supseteq C$.
Suppose $v_2 \notin C$.
Because $G$ is 2-edge-connected and $C$ is a 2-cut of $G$, the subgraph $G[V\setminus C]$ is connected.
Therefore there is an $r$-$v_2$ path in $G[V\setminus C]$.
This contradicts $v_1 \in C$ because $v_1$ is an ancestor of $v_2$.
\end{proof}

\begin{lemma}\label{lem:optimize_single_subcactus}
Let $k$ be a constant.
For any principal subcactus $G_i$ of a $k$-wide instance $(G,L)$ we can optimize any linear objective $c^T x$ with $c\in \mathbb{R}^{L_i}$ over $P^{\min}(G_i,L_i)$ in polynomial time.
\end{lemma}
\begin{proof}
Let $F\subseteq L_i$ be an $L_{\mathrm{cross}}$-minimal feasible solution for $G_i$ with $c(F)$ minimum.
Then $\min \{ c^T x : x\in P^{\min}(G_i,L_i)\}$ is attained for $x = \chi^F$.

First, we show that $F$ contains at most $k$ cross-links.
For the sake of deriving a contradiction, suppose $F$ contains $k+1$ cross links $\ell_1,\dots,\ell_{k+1}$. 
Each of these cross-links $\ell_j$ with $j\in[k+1]$ has exactly one endpoint $v_j$ in $V_i$ and this endpoint is an ancestor of some terminal in $V_i$ (by Lemma~\ref{lem:terminal_descendant_exists}).
Because $G_i$ is $k$-wide, the set $V_i$ contains at most $k$ terminals.
Therefore, by the pigeonhole principal there are distinct indices $j_1,j_2\in [k+1]$ such that $v_{j_1}$  and $v_{j_2}$ are ancestors of the same terminal $t\in V_i$.
This means that both $v_{j_1}$ and $v_{j_2}$ are contained in every $r$-$t$ path in $G$ and therefore $v_{j_1}$ is an ancestor of $v_{j_2}$ or the other way around.
Without loss of generality, assume that $v_{j_1}$ is an ancestor of $v_{j_2}$.
But then by Lemma~\ref{lem:simple_observation_cross_link_domination} we can replace $\ell_{j_1}=\{v_{j_1},w_{j_1}\}$ by its shadow $\{r,w_{j_1}\}$ without changing the set of 2-cuts covered by $\{\ell_{j_1}, \ell_{j_2}\}$.
This contradicts the $L_{\mathrm{cross}}$-minimality of $F$.
Hence, $F$ contains at most $k$ cross-links.

Because $k$ is a constant, there are only polynomially many sets of at most $k$ cross-links.
Hence we can enumerate over all such sets $X\subseteq L_{\mathrm{cross}}$ with $|X|\le k$ in order to ``guess'' $F\cap L_{\mathrm{cross}}$.
If $X$ is not $L_{\mathrm{cross}}$-minimal, then $X$ cannot be $F\cap L_{\mathrm{cross}}$.
Otherwise, we delete all links $\ell'\in L_i \setminus X$ from $L_i$ for which there exist some $\ell\in X$ such that $\ell$ is not minimal with respect to $\ell'$.
If $X=F\cap L_{\mathrm{cross}}$, these links cannot be contained in $F$.
Then, we compute an optimum solution $F_X$ of the residual instance of $(G_i,L_i, c)$ with respect to $X$ (if this instance is feasible).
This can be done in polynomial time by Lemma~\ref{lem:fpt_terminals}. 
By Corollary~\ref{cor:residual_instance}, $X \cup F_X$ is a feasible solution of $G_i$.
Moreover, $X \cup F_X$ is $L_{\mathrm{cross}}$-minimal because we deleted all links that could lead to a violation of this condition.
The cheapest of the resulting solutions $X \cup F_X$ of $G_i$ (for all enumerated sets $X$) has cost at most $c(F)$
because we considered $X=F\cap L_{\mathrm{cross}}$ at some point.
Hence, the incidence vector of this solution $X \cup F_X$ is an optimum solution of $\min \{ c^T x : x\in P^{\min}(G_i,L_i)\}$.
\end{proof}

\begin{corollary}\label{cor:decompose_convex_combination}
Suppose that $k$ is a constant and let $x\in P^{\min}(G_i,L_i)$ for some $i\in [q]$.
Then we can in polynomial time decompose $x$ into $L_{\mathrm{cross}}$-minimal feasible solutions for $G_i$, i.e., we can find $L_{\mathrm{cross}}$-minimal feasible solutions $F^1, \dots, F^h$ for $G_i$ and coefficients $p_1,\dots,p_h >0$  with $\sum_{j=1}^h p_j=1$ such that
$x = \sum_{j=1}^h p_j \cdot \chi^{F^j}$.
\end{corollary}
\begin{proof}
This is an immediate consequence of the fact that we can efficiently optimize over $P^{\min}(G_i,L_i)$, as guaranteed by Lemma~\ref{lem:optimize_single_subcactus}, and a classic algorithmic version of Carath\'eodory's theorem (\cite[Theorem 6.5.11]{groetschel_1993_geometric}). More precisely, this latter statement guarantees that one can efficiently decompose $x$ into a convex combination of at most $|L|+1$ vertices of $P^{\min}(G_i, L_i)$. Because these vertices correspond to feasible solutions $F^j$ for $G_i$, the statement follows. The efficiency of the procedure follows by the fact that the facet-complexity of $P^{\min}(G_i,L_i)$ is bounded by $\poly(|L|)$ because $P^{\min}(G_i,L_i)$ has all vertices in $\{0,1\}^L$. 
\end{proof}

Now we are ready to formulate the relaxation that we will use.
We define 
\begin{equation*}
P^{\mathrm{min}}_{\mathrm{bundle}}(G,L) \coloneqq \left\{ x\in [0,1]^L : x|_{L_i} \in P^{\min}(G_i,L_i) \text{ for all }i\in [q] \right\}\enspace.
\end{equation*}
This is a relaxation of the convex hull
\begin{equation*}
P^{\mathrm{min}}_{\mathrm{CacAP}}(G,L)\coloneqq \conv(\{\chi^F : F\subseteq L, (V,E\cup F) \text{ is $3$-edge-connected and }F\text{ is }L_{\mathrm{cross}}\text{-minimal}\})\enspace
\end{equation*}
of $L_{\mathrm{cross}}$-minimal solutions. 
\begin{lemma}\label{lem:optimize_efficiently_min_bundle}
For any $O(1)$-wide CacAP instance, we can separate and optimize over $P^{\mathrm{min}}_{\mathrm{bundle}}(G,L)$ in polynomial time.
\end{lemma}
\begin{proof}
In order to efficiently separate over $P^{\mathrm{min}}_{\mathrm{bundle}}(G,L)$, it suffices to be able to efficiently separate over each $P^{\mathrm{min}}(G_i, L_i)$ for $i\in [q]$. This is possible because Lemma~\ref{lem:optimize_single_subcactus} guarantees that we can efficiently optimize over $P^{\mathrm{min}}(G_i, L_i)$. Hence, one can also efficiently separate over $P^{\mathrm{min}}(G_i, L_i)$ due to the equivalence of separation and optimization (see~\cite[Theorem 6.4.9]{groetschel_1993_geometric}). Applying the same equivalence between separation and optimization in the other direction to $P^{\mathrm{min}}_{\mathrm{bundle}}(G,L)$, implies that efficient optimization over $P^{\mathrm{min}}_{\mathrm{bundle}}(G,L)$ is possible.
\end{proof}

\subsection{Rounding algorithm}\label{sec:rounding_algo_stacks}

In this section we describe a randomized algorithm for obtaining an integral solution from a vector in the polytope $P^{\mathrm{min}}_{\mathrm{bundle}}(G,L)$.
Whereas our algorithm is similar to \cite{grandoni_2018_improved}, our contribution lies in a significantly stronger analysis approach.

For each $i\in[q]$ we first apply Corollary~\ref{cor:decompose_convex_combination}
to write $x|_{L_i} = \sum_{j=1}^h p_j \cdot \chi^{F_i^j}$ as a convex combination of $L_{\mathrm{cross}}$-minimal solutions $F_i^1, \dots, F_i^h$ for $G_i$.
Then we randomly choose one of the solutions $F_i\in \{ F_i^1, \dots, F_i^h\}$ where we choose $F_i^j$ with probability $p_j$.
We do this random sampling independently for different $i\in [q]$.
The union of the $F_i$ for $i\in [q]$ is a feasible solution for the CacAP instance $(G,L)$.
(This is analogous to the simple bundle rounding we used to prove Lemma~\ref{lem:bundle_rounding}.)

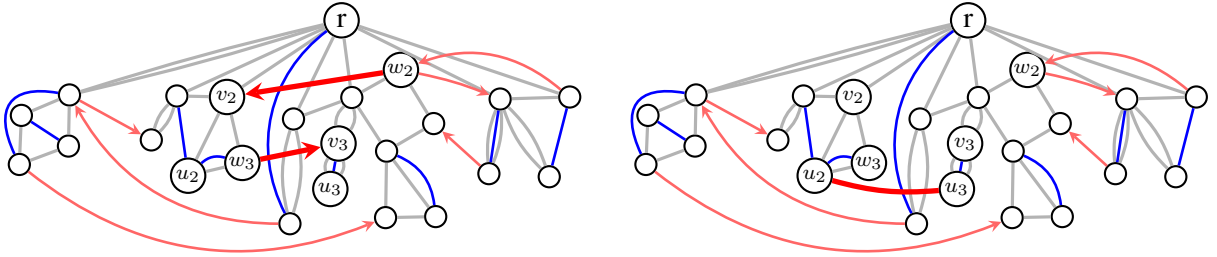
\begin{figure}[!ht]
\begin{center}
\begin{tikzpicture}[scale=0.34]

\tikzset{
  prefix node name/.style={%
    /tikz/name/.append style={%
      /tikz/alias={#1##1}%
    }%
  }
}

\tikzset{
grayout/.style={
every node/.append style={draw=black!30}, draw=black!30
}
}

\tikzset{root/.style={fill=white,minimum size=13, font=\normalsize}}

\tikzset{crlink/.style={red!60,stealth-,line width=1pt}}
\tikzset{inlink/.style={blue,line width=1pt}}

\tikzset{delcrlink/.style={red,stealth-,line width=2pt}}
\tikzset{addcrlink/.style={red,line width=2pt}}

\newcommand\cac[2][]{

\begin{scope}[prefix node name=#1]

\tikzset{npc/.style={#2}}

\begin{scope}[every node/.append style={thick,draw=black,fill=white,circle,minimum size=8, inner sep=0pt, font=\scriptsize}]
\node[root]  (1) at (18.68,-1.64) {r};
\node  (2) at (16.80,-5.50) {};
\node  (3) at (19.08,-4.66) {};
\node  (4) at (16.66,-9.60) {};
\node[minimum size=13pt]  (5) at (18.56,-6.46) {$v_3$};
\node[minimum size=13pt]  (6) at (18.24,-8.24) {$u_3$};
\node  (7) at (20.44,-6.76) {};
\node  (8) at (22.28,-5.68) {};
\node[minimum size=13pt]  (9) at (21.00,-3.58) {$w_2$};
\node (10) at (20.40,-9.36) {};
\node (11) at (22.36,-9.32) {};

\node (12) at (12.24,-4.60) {};
\node[minimum size=13pt] (13) at (14.20,-4.64) {$v_2$};
\node (14) at (11.24,-6.32) {};
\node[minimum size=13pt] (15) at (12.68,-7.72) {$u_2$};
\node[minimum size=13pt] (16) at (14.80,-7.22) {$w_3$};
\node (17) at (24.88,-4.72) {};
\node (18) at (27.60,-4.64) {};
\node (19) at (24.44,-7.64) {};
\node (20) at (26.80,-7.76) {};
\node (21) at (8.04,-4.56) {};
\node (22) at (6.16,-5.36) {};
\node (23) at (6.08,-7.28) {};
\node (24) at (8.00,-6.56) {};

\end{scope}

\begin{scope}[grayout, very thick]

\draw  (1) --  (2);
\draw  (2) --  (3);
\draw  (3) --  (1);
\draw  (9) --  (3);
\draw  (2) to[bend left=15] (4);
\draw  (4) to[bend left=15] (2);
\draw  (3) to[bend left=25] (5);
\draw  (5) to[bend left=25] (3);
\draw  (5) to[bend left=25] (6);
\draw  (6) to[bend left=25] (5);
\draw  (3) --  (7);
\draw  (7) --  (8);
\draw  (8) --  (9);
\draw  (7) -- (10);
\draw (10) -- (11);
\draw (11) --  (7);

\begin{scope}
\draw  (1) to[bend left=2] (21);
\draw (21) to[bend left=2]  (1);
\draw  (1) -- (12);
\draw  (1) -- (17);
\draw (18) --  (1);
\draw (13) --  (1);
\draw (12) to[bend left=15] (14);
\draw (14) to[bend left=15] (12);
\draw (16) -- (13);
\draw (12) -- (13);
\draw (13) -- (15);
\draw (15) -- (16);
\draw (17) -- (18);
\draw (17) to[bend left=13] (19);
\draw (19) to[bend left=13] (17);
\draw (17) to[bend left=13] (20);
\draw (20) to[bend left=13] (17);
\draw (24) -- (21);
\draw (21) -- (22);
\draw (22) -- (23);
\draw (23) -- (24);
\end{scope}

\end{scope}

\end{scope}

\begin{scope} %

\begin{scope}

\begin{scope}[crlink]

\draw (9) to [bend left](18);
\draw (8) to (19);
\draw (10) to [bend left](23);

\draw (17) to[bend right=5] (9);
\draw (21) to[bend right=30] (4);

\draw (14) to (21);

\end{scope}

\begin{scope}[inlink]

\draw (5) to (6);
\draw (7) to [bend left](11);
\draw (1) to [bend right=40](4);

\draw (22) to (24);
\draw (21) to[out=170, in =115,looseness=1.7] (23);

\draw (12) to (15);
\draw (15) to[bend left] (16);

\draw (17) to (19);
\draw (18) to (20);

\end{scope}

\end{scope}
\end{scope}%
}

\begin{scope}[shift={(-10,-12)}]

\begin{scope}[shift={(0,0)}]
\cac[c1-]{}
\end{scope}

\begin{scope}[delcrlink]
\draw (13) to (9);
\draw (5) to (16);

\end{scope}

\begin{scope}[shift={(24.5,0)}]
\cac[c2-]{}
\end{scope}

\begin{scope}[addcrlink]
\draw[bend left=10] (6) to (15);
\end{scope}

\end{scope}

\end{tikzpicture}
 \end{center}
\caption{In the picture on the left, a solution for a CacAP instance is shown with cross-links drawn in red and in-links in blue. It is given by the union of one solution for each of the four subcacti. The links $\{v_2,w_2\}$ and $\{v_3,w_3\}$, highlighted in bold in the left picture, appear in the solution for the second and the third subcactus, respectively. Thus we can replace those two links by $\{u_2,u_3\}$, highlighted in bold in the right picture, while maintaining feasibility, because $\{v_2,v_3\}$ is a shadow of $\{u_2,u_3\}$.}\label{fig:deleting_dominated_link}
\end{figure}

Then we improve this solution using the following observation which was used for TAP in~\cite{grandoni_2018_improved}.
If there are cross-links $\ell_i=\{v_i,w_i\}\in F_i \cap L_{\mathrm{cross}}$ and  $\ell_j=\{v_j,w_j\}\in  F_j \cap L_{\mathrm{cross}}$ with $v_i\in V_i$ and $v_j\in V_j$, we can replace those two cross-links by the link $\{v_i,v_j\}$ and maintain feasibility. 
More generally, by Lemma~\ref{lem:simple_observation_cross_link_domination} we can also replace $\ell_i$ and $\ell_j$ by
any link $\{u_i,u_j\}$ where $u_i$ is a descendant of $v_i$ and $u_j$ is a descendant of $v_j$.
Then $\{v_i,v_j\}$ is a shadow of $\{u_i,u_j\}$ (or we have $\{v_i,v_j\}=\{u_i,u_j\}$).
See Figure~\ref{fig:deleting_dominated_link} for an illustration of such a replacement.

Next, we show how we can compute in polynomial time an optimal set of such replacement operations that can be performed simultaneously. 
This reduces to an edge cover problem as we show next.
We apply the below lemma to the sets $F^i_{\mathrm{cross}} := F_i \cap L_{\mathrm{cross}}$.
Then replacing $F^1_{\mathrm{cross}} \cup \dots \cup F^q_{\mathrm{cross}}$ by the resulting set $F$ maintains a feasible solution.

\begin{lemma}\label{lem:edge_cover_reduction}
Given sets $F^i_{\mathrm{cross}}\subseteq L_{\mathrm{cross}}$ of cross links for all $i \in [q]$,
we can  compute in polynomial time a minimum cardinality set $F\subseteq L_{\mathrm{cross}}$ of cross links such that for all $i\in [q]$ the set 
 $F$ covers every 2-cut $C\in \mathcal{C}_{G_i}$ that is covered by $F^i_{\mathrm{cross}}$.
\end{lemma}
\begin{proof}
Let 
\[
A= \{ a\in V : a\in V_i\emph{ is an endpoint of a cross-link }\ell\in F^i_{\mathrm{cross}} \text{ for some }i \in [q] \}\enspace.
\]
Because every link $\ell \in F^i_{\mathrm{cross}}$ is a cross-link, such a link $\ell$ covers a 2-cut $C\in \mathcal{C}_{G_i}$
if and only if one of the endpoints of $\ell$ is contained in $C$. (Indeed, because $\ell$ is a cross link, the other endpoint must lie in a different principal subcactus than $G_i$ and is thus in $V\setminus C$.)
Hence, we are looking for a minimum-cardinality set $F$ of cross-links that cover all 2-cuts $C\in \mathcal{C}_G$ with $A\cap C\ne \emptyset$.

We claim that a set $F \subseteq L_{\mathrm{cross}}$ of cross-links has this property if and only if 
for every $a\in A$ there is a descendant of $a$ that is an endpoint of a link $\ell\in F$.
By Lemma~\ref{lem:simple_observation_cross_link_domination} this condition is sufficient.

To see that it is also necessary, we fix a vertex $a\in A$ and consider the set $U_a$ of descendants of $a$.
Let $i\in [q]$ be the index with $a\in V_i$.
By Lemma~\ref{lem:descendants_define_2_cut}, the set $U_a$ is a 2-cut, i.e., $U_a\in \mathcal{C}_{G_i}$.
By the definition of the set $A$, the vertex $a$ is the endpoint of some link $\ell \in F^i_{\mathrm{cross}}$.
Because this link $\ell$ covers $U_a$, also $F$ must cover $U_a$.
Hence, $F$ must contain a link with an endpoint in $U_a$, i.e., a link with an endpoint that is a descendant of $a$.

Therefore, it remains to show that we can in polynomial time compute a minimum cardinality set $F$ of cross-links 
such that, for every $a\in A$, some descendant of $a$ is an endpoint of a link in $F$.
If for two distinct vertices $a,a'\in A$, we have that $a$ is a descendant of $a'$, 
then every descendant of $a$ is also a descendant of $a'$.
Hence, we can remove $a'$ from $A$ without weakening our condition on $F$.
Thus, we may assume without loss of generality that $A$ does not contain distinct vertices $a,a'$, where one is a descendant of the other.
Then the sets $U_a$ of descendant of $a$ for different $a\in A$ are disjoint.

\begin{figure}[!ht]
\begin{center}
\begin{tikzpicture}[scale=0.5]

\tikzset{
  prefix node name/.style={%
    /tikz/name/.append style={%
      /tikz/alias={#1##1}%
    }%
  }
}

\tikzset{
grayout/.style={
every node/.append style={draw=black!30}, draw=black!30
}
}

\tikzset{root/.style={fill=white,minimum size=13, font=\normalsize}}

\tikzset{setA/.style={fill=orange!40,minimum size=13, font=\normalsize}}

\tikzset{crlink/.style={red,stealth-,line width=1pt}}
\tikzset{inlink/.style={blue!50,line width=1pt}}
\tikzset{crlinknotopt/.style={red,line width=1pt, dashed}}
\tikzset{crlinkopt/.style={red,line width=1pt}}

\newcommand\cac[2][]{

\begin{scope}[prefix node name=#1]

\tikzset{npc/.style={#2}}

\begin{scope}[every node/.append style={thick,draw=black,fill=white,circle,minimum size=12, inner sep=2pt, font=\scriptsize}]
\node[root]  (1) at (18.68,-1.64) {r};

\node  (2) at (16.80,-5.50) {};
\node  (3) at (19.08,-5.16) {};
\node  (4)[setA] at (16.66,-9.60) {};
\node  (5) at (18.56,-7.10) {};
\node  (6) at (18.24,-8.6) {};
\node  (7)[setA] at (20.44,-6.76) {};
\node  (8) at (22.28,-5.68) {};
\node  (9)[setA] at (21.00,-3.58) {};
\node (10) at (20.40,-9.36) {};
\node (11) at (22.36,-9.32) {};
\node (12) at (12.24,-4.60) {};
\node (13)[setA] at (14.20,-4.64) {};
\node (14)[setA] at (11.24,-6.32) {};
\node (15) at (12.68,-7.72) {};
\node (16) at (14.80,-7.42) {};
\node (17)[setA] at (24.88,-4.72) {};
\node (18) at (27.60,-4.64) {};
\node (19) at (24.44,-7.64) {};
\node (20) at (26.80,-7.76) {};
\node (21)[setA] at (8.04,-4.56) {};
\node (22) at (6.16,-5.36) {};
\node (23) at (6.08,-7.28) {};
\node (24) at (8.00,-6.56) {};

\end{scope}

\begin{scope}[every node/.append style={font=\scriptsize}]
\foreach \i in {2,3,4,5,6,7,8,9,10,11,12,13,14,15,16,17,18,19,20,21,22,23,24} {
\pgfmathparse{int(\i-1)}
\node at (\i) {${\pgfmathresult}$};
}
\end{scope}

\begin{scope}[grayout, very thick]

\draw  (1) --  (2);
\draw  (2) --  (3);
\draw  (3) --  (1);
\draw  (9) --  (3);
\draw  (2) to[bend left=15] (4);
\draw  (4) to[bend left=15] (2);
\draw  (3) to[bend left=25] (5);
\draw  (5) to[bend left=25] (3);
\draw  (5) to[bend left=25] (6);
\draw  (6) to[bend left=25] (5);
\draw  (3) --  (7);
\draw  (7) --  (8);
\draw  (8) --  (9);
\draw  (7) -- (10);
\draw (10) -- (11);
\draw (11) --  (7);

\begin{scope}[npc]
\draw  (1) to[bend left=2] (21);
\draw (21) to[bend left=2]  (1);
\draw  (1) -- (12);
\draw  (1) -- (17);
\draw (18) --  (1);
\draw (13) --  (1);
\draw (12) to[bend left=15] (14);
\draw (14) to[bend left=15] (12);
\draw (16) -- (13);
\draw (12) -- (13);
\draw (13) -- (15);
\draw (15) -- (16);
\draw (17) -- (18);
\draw (17) to[bend left=13] (19);
\draw (19) to[bend left=13] (17);
\draw (17) to[bend left=13] (20);
\draw (20) to[bend left=13] (17);
\draw (24) -- (21);
\draw (21) -- (22);
\draw (22) -- (23);
\draw (23) -- (24);
\end{scope}

\end{scope}

\end{scope}

}%

\newcommand\edgecover[2][]{

\begin{scope}[prefix node name=#1]

\tikzset{npc/.style={#2}}

\begin{scope}[every node/.append style={thick,draw=black,fill=white,circle,minimum size=12, inner sep=2pt, font=\scriptsize}]
\node  (4)[setA] at (16.66,-9.60) {};
\node  (7) [setA] at (20.44,-6.76) {};
\node  (9)[setA] at (21.00,-3.58) {};
\node (13)[setA] at (14.20,-4.64) {};
\node (14)[setA] at (11.24,-6.32) {};
\node (17)[setA] at (24.88,-4.72) {};
\node (21)[setA] at (8.04,-4.56) {};

\node (25) at (17.00,-12.00) {$\bar{r}$};
\node (26) at (21.00,-12.00) {v};

\end{scope}

\begin{scope}[every node/.append style={font=\scriptsize}]
\foreach \i in {4,7,9,13,14,17,21} {
\pgfmathparse{int(\i-1)}
\node at (\i) {${\pgfmathresult}$};
}
\end{scope}

\begin{scope}[gray, thick]

\draw (25) to (14);
\draw (25) to [bend left=10](13);
\draw (9) to (17);
\draw (4) to [bend left] (21);
\draw (25) to [bend right](17);
\draw (9) to [bend right=10](25);
\draw (14) to (21);

\begin{scope}[purple, very thick]
\draw (25) to (26);
\draw (7) to (14);
\draw (7) to (21);
\draw (9) to (13);
\draw (4) to[bend right=10] (17);

\end{scope}

\end{scope}

\end{scope}

}%

\begin{scope}[shift={(-10,-12)}]

\begin{scope}[shift={(0,13)},scale=0.8]
\edgecover[e-]{}
\end{scope}

\begin{scope}[shift={(0,22)},scale=0.8]
\cac[c1-]{}

\begin{scope}

\begin{scope}[crlink]
\draw (9) to [bend left=8](13);
\draw (7) to[bend right=8] (14);
\draw (4) to (19);

\draw (21) to[bend right=20] (4);

\draw (14) to (21);
\draw (13) to [bend left=8](9);

\draw (17) to (9);
\end{scope}

\begin{scope}[crlinknotopt]
\draw (2) to (14);
\draw (6) to (15);
\draw (10) to[bend left=30] (23);
\draw (3) to (13);

\draw (5) to (16);
\draw (8) to (19);
\draw (9) to (18);
\end{scope}

\begin{scope}[inlink]

\draw (3) to [bend left](6);
\draw (7) to [bend right](10);
\draw (8) to (11);

\draw (22) to (24);
\draw (21) to[out=170, in =115,looseness=1.7] (23);

\draw (12) to (15);
\draw (15) to[bend left] (16);

\draw (17) to (19);
\draw (18) to (20);

\end{scope}

\end{scope}

\end{scope}

\begin{scope}[shift={(0,0)},scale=0.8]
\cac[c2-]{}

\begin{scope}

\begin{scope}[crlinkopt]
\draw (4) to (19);
\draw (7) to[bend right=8] (14);
\draw (13) to [bend left=8](9);
\draw (10) to[bend left=30] (23);
\end{scope}

\begin{scope}[inlink]
\draw (3) to [bend left](6);
\draw (8) to (11);
\draw (22) to (24);
\draw (12) to (15);
\draw (15) to[bend left] (16);
\draw (18) to (20);
\end{scope}

\end{scope}

\end{scope}

\begin{scope}[shift={(25,20)}]%
\def\ll{30mm} %
\def\vs{18mm} %

\begin{scope}[scale=0.5]
\draw[black,line width=1pt] (0,0) -- +(\ll,0) node[right] {in-links}[blue!50];
\draw[black,line width=1pt,yshift=-\vs,dashed] (0,0) -- +(\ll,0) node[right] {cross-links}[red];
\draw[black,line width=1pt,yshift=-2*\vs,-stealth] (0,0) -- +(\ll,0) node[right] {sampled cross-links}[red];
\end{scope}

\begin{scope}[scale=0.5]

\node[fill=orange!40,circle,draw=black,inner sep=3pt] (b) at (0.45,-4*\vs) {};
\node at (b.east)[right=2pt] {vertices in $A$};

\end{scope}

\end{scope}%

\begin{scope}[shift={(25,10)}]%
\def\ll{30mm} %
\def\vs{18mm} %

\begin{scope}
\node[fill=orange!40,circle,draw=black,inner sep=3pt] (b) at (0.25,-1.5*\vs) {};
\node at (b.east)[right=2pt] {vertices in $A$};
\end{scope}

\begin{scope}[scale=0.5]
\draw[black,very thick, yshift=-\vs] (0,0) -- +(\ll,0) node[right] {optimal edge cover}[purple];
\end{scope}

\end{scope}%

\begin{scope}[shift={(25,-2)}]%
\def\ll{30mm} %
\def\vs{18mm} %

\begin{scope}
\node[fill=orange!40,circle,draw=black,inner sep=3pt] (b) at (0.3,-2*\vs) {};
\node at (b.east)[right=2pt] {vertices in $A$};
\end{scope}

\begin{scope}[scale=0.5]
\draw[black,line width=1pt, yshift=-\vs] (0,0) -- +(\ll,0) node[right] {in-links}[blue!50];
\draw[black,line width=1pt,yshift=-2*\vs] (0,0) -- +(\ll,0) node[right] {cross-links}[red];
\end{scope}

\end{scope}%

\end{scope}

\end{tikzpicture}
 \end{center}
\caption{The top picture shows a possible solution obtained by taking the union of one solution for each principal subcactus. As previously, the heads of the cross-links highlight from which principal subcactus a cross-link has been sampled. Cross-links not appearing in the solution are dashed. The vertices in the set $A$ as defined in the proof of Lemma~\ref{lem:edge_cover_reduction}, which is the vertex set of all arrowheads of cross-links, are drawn as orange vertices. The graph in the middle shows the edge cover instance we built to compute an optimal replacement of the sampled cross-links. A solution for the edge cover problem is highlighted with bold purple edges. In the cactus below, an optimal set of cross-links is shown, which corresponds to the highlighted edge cover solution in the middle picture.}\label{fig:edge_cover_construction}
\end{figure}
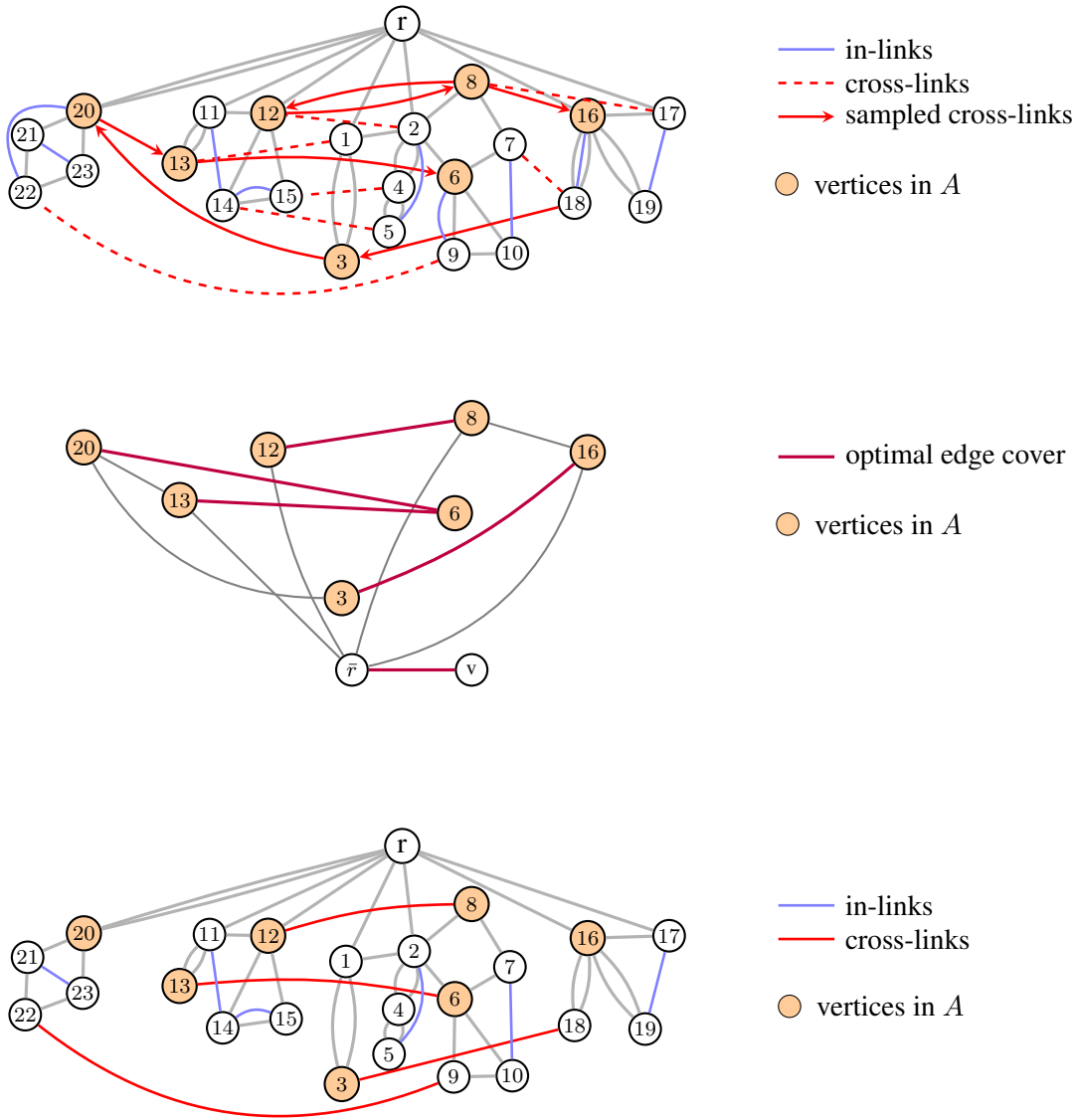

We now construct an instance of edge cover as follows.
The vertex set is $A \cupp \{\bar r, v\}$, where $\bar r$ and $v$ are two auxiliary vertices.
For a vertex $u\in V$ we define $a_u$ to be the unique vertex $a\in A$ with $u\in U_a$ if such a vertex $a$ exists, and we define $a_u:=\bar r$ otherwise.
The edge set of our edge cover instance consists of the edge $\{\bar r,v\}$ and 
the edges $\{a_u,a_w\}$ for all cross-links $\{u,w\}\in L_{\mathrm{cross}}$.
The vertex $v$ has only one incident edge and hence the edge $\{v,\bar r\}$ is contained in every feasible edge cover.
Moreover, a set $F$ of cross-links has an endpoint in each of the sets $U_a$ if and only if 
$\{\{a_u,a_w\} : \{u,w\}\in F\} \cupp \{\{v,\bar r\}\}$ is a feasible edge cover.
Therefore, we can compute an optimum set $F$ using a polynomial-time algorithm for the edge cover problem.
See Figure~\ref{fig:edge_cover_construction} for an illustration of this reduction.
 \end{proof}

Algorithm~\ref{algo:stack_rounding} summarizes the randomized algorithm that we analyze in the following.
It rounds a given vector $x\in P^{\mathrm{min}}_{\mathrm{bundle}}(G,L)$ to a feasible CacAP solution.

\smallskip
\begin{algorithm2e}[H]
\begin{enumerate}[label=(\arabic*)]\itemsep2pt
\item For each $i\in[q]$, apply Corollary~\ref{cor:decompose_convex_combination} to write $x|_{L_i} = \sum_{j=1}^h p_j \cdot \chi^{F_i^j}$ as a convex combination of $L_{\mathrm{cross}}$-minimal solutions $F_i^1, \dots, F_i^h$ for $G_i$.
\item For each $i\in[q]$ independently, choose one of the solutions $F_i\in \{ F_i^1, \dots, F_i^h\}$ at random such that $F_i^j$ is sampled with probability $p_j$. \label{item:algo_sampling}
\item Apply Lemma~\ref{lem:edge_cover_reduction} to the sets $F^i_{\mathrm{cross}} := F_i \cap L_{\mathrm{cross}}$ to obtain a set $F_{\mathrm{cross}}\subseteq L_{\mathrm{cross}}$.\label{item:compute_removable_set} 
\item Let $F^i_{\mathrm{in}} := F_i \cap L_{\mathrm{in}}$ for all $i\in[q]$ and return $F_{\mathrm{cross}} \cup F^1_{\mathrm{in}} \cup \dots \cup F^q_{\mathrm{in}}$.
\end{enumerate}
\caption{Randomized algorithm to round a vector $x\in P^{\mathrm{min}}_{\mathrm{bundle}}(G,L)$ to a solution of $(G,L)$}\label{algo:stack_rounding}
\end{algorithm2e}
\smallskip

\begin{lemma}
For any vector $x\in P^{\mathrm{min}}_{\mathrm{bundle}}(G,L)$, Algorithm~\ref{algo:stack_rounding} returns a feasible solution to the CacAP instance $(G,L)$.
\end{lemma}
\begin{proof}
Because $F_i$ is a feasible solution of $G_i$ for all $i\in [q]$, the union $F_1 \cup \dots \cup F_q$ is a feasible CacAP solution.
Lemma~\ref{lem:edge_cover_reduction} guarantees that then also the returned link set $F_{\mathrm{cross}} \cup F^1_{\mathrm{in}} \cup \dots \cup F^q_{\mathrm{in}}$ is a feasible solution.
\end{proof}

In our analysis of Step~\ref{item:compute_removable_set} 
we will restrict ourselves to sets $F_{\mathrm{cross}} \subseteq F^1_{\mathrm{cross}} \cup \dots \cup F^q_{\mathrm{cross}}$.
In other words, we analyze how many cross-links we can remove  from $F_1 \cupp \dots \cupp F_q$ while maintaining a feasible solution.
This yields an upper bound on the size of the set $F_{\mathrm{cross}}$ chosen in Step~\ref{item:compute_removable_set} of Algorithm~\ref{algo:stack_rounding}
and hence leads to an upper bound on the cardinality of the solution returned by Algorithm~\ref{algo:stack_rounding}.

Our overall $\apxfac$-approximation algorithm for $k$-wide CacAP will combine Lemma~\ref{lem:cross-link_rounding} and Algorithm~\ref{algo:stack_rounding}.
It starts by computing an optimum solution to $\min\{ 1^Tx : x \in P_{\mathrm{cross}} \cap P^{\mathrm{min}}_{\mathrm{bundle}}(G,L)\}$.
Then it obtains two CacAP solutions by applying Lemma~\ref{lem:cross-link_rounding} and Algorithm~\ref{algo:stack_rounding}.
Finally, it simply returns the better of these two solutions. We recall that we later discuss, in Section~\ref{sec:derandomizing}, a way to derandomize the algorithm.

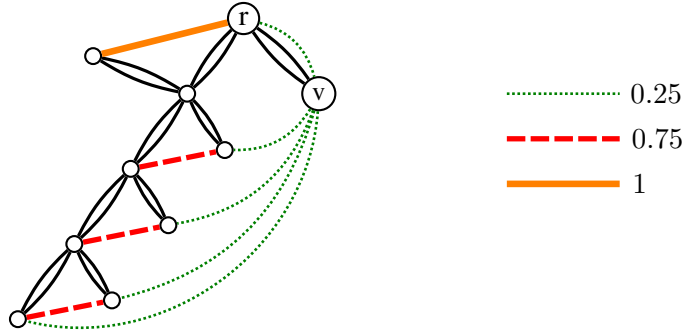
\begin{figure}[!ht]
\begin{center}
\begin{tikzpicture}[scale=0.5]

\tikzset{
  prefix node name/.style={%
    /tikz/name/.append style={%
      /tikz/alias={#1##1}%
    }%
  }
}

\tikzset{root/.style={fill=white,minimum size=6}}

\tikzset{q1/.style={line width=1pt, dash pattern=on \pgflinewidth off 1pt}}

\tikzset{q2/.style={line width=1.5pt, dash pattern=on 3pt off 2pt on \the\pgflinewidth off 2pt}}

\tikzset{q3/.style={line width=2pt, dash pattern=on 7pt off 2pt}}

\tikzset{q4/.style={line width=2.5pt}}

\tikzset{node/.style={thick,draw=black,fill=white,circle,minimum size=6, inner sep=2pt}}

\newcommand\cac[2][]{
\begin{scope}[prefix node name=#1]

\begin{scope}[every node/.append style=node]
\node (1) at (13,13) {r};
\node (2) at (11.5,11) {};
\node (3) at (15,11) {\small v};
\node (4) at (10,9) {};
\node (5) at (8.5,7) {};
\node (6) at (7,5) {};
\node (7) at (9,12) {};
\node (8) at (12.5,9.5) {};
\node (9) at (11,7.5) {};
\node (10) at (9.5,5.5) {};
\end{scope}

\begin{scope}[very thick, bend left=10]
\draw (1) to (2);
\draw (2) to (1);
\draw (2) to (4);
\draw (4) to (2);
\draw (4) to (5);
\draw (5) to (4);
\draw (5) to (6);
\draw (6) to (5);
\draw (2) to (7);
\draw (7) to (2);
\draw (1) to (3);
\draw (3) to (1);
\draw (2) to (8);
\draw (8) to (2);
\draw (4) to (9);
\draw (9) to (4);
\draw (5) to (10);
\draw (10) to (5);
\end{scope}

\begin{scope}[q4,orange]
\draw (1) -- (7);
\end{scope}

\begin{scope}[q1,green!50!black]
\draw (1) to [bend left] (3);
\draw (3) to [bend left] (8);
\draw (3) to [bend left] (9);
\draw (3) to [bend left] (10);
\draw (3) to [bend left=50] (6);
\end{scope}

\begin{scope}[q3,red]
\draw (4) -- (8);
\draw (5) -- (9);
\draw (6) -- (10);
\end{scope}

\end{scope}
}%

\begin{scope}
\cac[m-]{}

\end{scope}

\begin{scope}[shift={(20,15)}]%
\def\ll{30mm} %
\def\vs{12mm} %

\begin{scope}[yshift=-4cm]
\draw[q1] (0,0) -- +(\ll,0) node[right] {$0.25$}[green!50!black];
\draw[q3,yshift=-\vs] (0,0) -- +(\ll,0) node[right] {$0.75$}[red];
\draw[q4,yshift=-2*\vs] (0,0) -- +(\ll,0) node[right] {$1$}[orange];
\end{scope}

\end{scope}

\end{tikzpicture}
 \end{center}
\caption{Example showing that requiring $L_{\mathrm{cross}}$-minimality strengthens our LP (even for instances of TAP). The set $L$ consists of the links shown in the figure together with all their shadows.
The figure shows a vector $x\in [0,1]^L$ that would be feasible if we dropped the minimality requirement, i.e., we have $x\in P_{\mathrm{cross}}$ and for both principal subcacti $G_i$ we have that $x|_{L_i}$ is a convex combination of solutions for $G_i$. For the left principal subcactus $G_1$ such a convex combination of $x|_{L_1}$ takes the solution consisting of the orange and the red links with weight $\sfrac{3}{4}$ and the solution consisting of the orange and the green links in $L_1$ with weight $\sfrac{1}{4}$.
The latter of these solutions for $G_1$ is not $L_{\mathrm{cross}}$-minimal. For the shown vector $x$ we have $x(L)=4.5$.
\newline
In contrast, any vector $x \in P^{\mathrm{min}}_{\mathrm{bundle}}(G,L)$ fulfills $x(L) = 5 = |\OPT|$. 
To see this, we observe that any $L_{\mathrm{cross}}$-minimal solution of the left principal subcactus $G_1$ contains at most one link incident to $v$.
Moreover, every solution for $G_1$ has cardinality at least $4$ and every solution for $G_1$ that contains a link incident to $v$ has cardinality at least $5$.
Thus, we have $x(L_1) \ge 4 + x(\delta_{L_1}(v))$. Because $x|_{L_2}$ is a convex combination of solutions for the right principal subcactus $G_2$,
we have $x(\delta_L(v)) \ge 1$ and hence $x(L \setminus L_1) \ge x(\delta_{L\setminus L_1}(v)) \ge 1 - x(\delta_{L_1}(v))$.
This implies $x(L) = x(L_1) + x(L \setminus L_1) \ge 5$. }\label{fig:minimality_example}
\end{figure}

We remark that requiring that $x|_{L_i}$ is a convex combination of incidence vectors of $L_{\mathrm{cross}}$-minimal solutions (and not just of any solutions) for $G_i$
strengthens our LP $\min\{ 1^Tx : x \in P_{\mathrm{cross}} \cap P^{\mathrm{min}}_{\mathrm{bundle}}(G,L)\}$, i.e.,
it is not only useful in our analysis but it can lead to a larger LP value as the example in Figure~\ref{fig:minimality_example} shows.

\subsection{Simple stack analysis}\label{sec:simple_stack_problem}

Next, we present a simplified version of our analysis, which we will refine later.
This simplified analysis already leads to an approximation ratio below $1.5$ for connectivity augmentation, and allows us to highlight key ideas behind our approach in a simpler setting.

\subsubsection*{Stacks}
We assign every cross-link $\ell=\{v,w\}$ to two different terminals, corresponding to its two endpoints $v$ and $w$.
Recall that for every vertex $u\in V$, there is at least one terminal $t_u\in T$ that is a descendant of $u$.
We assign the link $\ell=\{v,w\}$ to the terminals $t_v$ and $t_w$.
If these terminals are not unique, we make an arbitrary choice.
For a terminal $t\in T$, we call the set $\Sscr_t$ of links assigned to it a \emph{stack}.

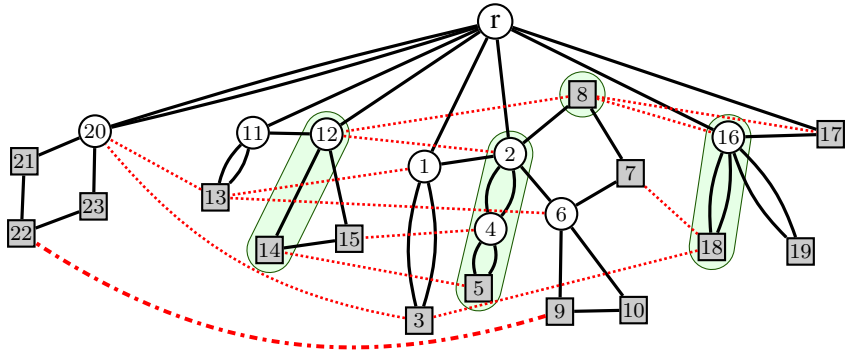
\begin{figure}[!ht]
\begin{center}
\begin{tikzpicture}[scale=0.5]

\pgfdeclarelayer{bg}
\pgfdeclarelayer{fg}
\pgfsetlayers{bg,main,fg}

\tikzset{
  prefix node name/.style={%
    /tikz/name/.append style={%
      /tikz/alias={#1##1}%
    }%
  }
}

\tikzset{root/.style={fill=white,minimum size=13}}

\tikzset{q1/.style={line width=1pt, dash pattern=on \pgflinewidth off 1pt}}

\tikzset{q2/.style={line width=1.5pt, dash pattern=on 3pt off 2pt on \the\pgflinewidth off 2pt}}

\tikzset{q3/.style={line width=2pt, dash pattern=on 7pt off 2pt}}

\tikzset{q4/.style={line width=2.5pt}}

\tikzset{
term/.style={fill=black!20, rectangle, minimum size=10},
termg/.style={fill=black!20, rectangle, minimum size=10},
tf/.append style={font=\scriptsize\color{black}},
}

\newcommand\cac[2][]{

\begin{scope}[prefix node name=#1]

\tikzset{npc/.style={#2}}

\begin{scope}[every node/.append style={thick,draw=black,fill=white,circle,minimum size=12, inner sep=2pt}]
\node[root]  (1) at (18.68,-1.64) {r};
\node  (2) at (16.80,-5.50) {};
\node  (3) at (19.08,-5.16) {};
\node[term]  (4) at (16.66,-9.60) {};
\node  (5) at (18.56,-7.16) {};
\node[term]  (6) at (18.24,-8.74) {};
\node  (7) at (20.44,-6.76) {};
\node[term]  (8) at (22.28,-5.68) {};
\node[term]  (9) at (21.00,-3.58) {};
\node[term] (10) at (20.40,-9.36) {};
\node[term] (11) at (22.36,-9.32) {};

\begin{scope}[npc]
\node (12) at (12.24,-4.60) {};
\node (13) at (14.20,-4.64) {};
\node[termg] (14) at (11.24,-6.32) {};
\node[termg] (15) at (12.68,-7.72) {};
\node[termg] (16) at (14.80,-7.42) {};
\node (17) at (24.88,-4.72) {};
\node[termg] (18) at (27.60,-4.64) {};
\node[termg] (19) at (24.44,-7.64) {};
\node[termg] (20) at (26.80,-7.76) {};
\node (21) at (8.04,-4.56) {};
\node[termg] (22) at (6.16,-5.36) {};
\node[termg] (23) at (6.08,-7.28) {};
\node[termg] (24) at (8.00,-6.56) {};
\end{scope}
\end{scope}

\begin{scope}[every node/.append style={font=\scriptsize}]

\foreach \i/\t in {2/,3/,4/tf,5/,6/tf,7/,8/tf,9/tf,10/tf,11/tf,12/,13/,14/tf,15/tf,16/tf,17/,18/tf,19/tf,20/tf,21/,22/tf,23/tf,24/tf} {
\pgfmathparse{int(\i-1)}
\node[\t] at (\i) {$\pgfmathresult$};
}
\end{scope}

\begin{scope}[very thick]

\draw  (1) --  (2);
\draw  (2) --  (3);
\draw  (3) --  (1);
\draw  (9) --  (3);
\draw  (2) to[bend left=15] (4);
\draw  (4) to[bend left=15] (2);
\draw  (3) to[bend left=25] (5);
\draw  (5) to[bend left=25] (3);
\draw  (5) to[bend left=25] (6);
\draw  (6) to[bend left=25] (5);
\draw  (3) --  (7);
\draw  (7) --  (8);
\draw  (8) --  (9);
\draw  (7) -- (10);
\draw (10) -- (11);
\draw (11) --  (7);

\begin{scope}[npc]
\draw  (1) to[bend left=2] (21);
\draw  (1) -- (12);
\draw  (1) -- (17);
\draw (21) to[bend left=2]  (1);
\draw (18) --  (1);
\draw (13) --  (1);
\draw (12) to[bend left=15] (14);
\draw (14) to[bend left=15] (12);
\draw (16) -- (13);
\draw (12) -- (13);
\draw (13) -- (15);
\draw (15) -- (16);
\draw (17) -- (18);
\draw (17) to[bend left=13] (19);
\draw (19) to[bend left=13] (17);
\draw (17) to[bend left=13] (20);
\draw (20) to[bend left=13] (17);
\draw (24) -- (21);
\draw (21) -- (22);
\draw (22) -- (23);
\draw (23) -- (24);
\end{scope}

\end{scope}

\end{scope}

}%

\begin{scope}

\tikzset{
grayout/.style={
every node/.append style={draw=black!30}, draw=black!30
},
inlink/.style={q1,blue},
crlink/.style={q1,red}
}

\begin{scope}[xshift=-1cm]
\cac[]{}
\end{scope}

\begin{scope}

\begin{scope}[red]

\begin{scope}[q1]
\draw (2) to (14);
\draw (3) to (13);
\draw (4) to (19);
\draw (4) to[bend left=20] (21);
\draw (5) to (16);
\draw (6) to (15);
\draw (7) to (14);
\draw (8) to (19);
\draw (9) to (13);
\draw (9) to (17);
\draw (9) to (18);
\draw (14) to (21);
\end{scope}

\begin{scope}[q2]
\draw (10) to[bend left=25] (23);
\end{scope}

\end{scope}

\end{scope}
\end{scope}%

\begin{pgfonlayer}{bg}

\def\d{6mm}

\begin{scope}[fill=green!10,draw=darkgreen]

\newcommand\blob[2]{
\pgfmathanglebetweenpoints%
{\pgfpointanchor{#1}{center}}%
{\pgfpointanchor{#2}{center}}

\edef\angle{\pgfmathresult}

\filldraw ($(#1)+(\angle+90:\d)$) arc (\angle+90:\angle+270:\d) --
($(#2)+(\angle+270:\d)$) arc (\angle+270:\angle+450:\d) -- cycle;
}

\blob{6}{3}
\blob{9}{9}
\blob{13}{15}
\blob{17}{19}

\end{scope}
\end{pgfonlayer}

\end{tikzpicture}
 \end{center}
\caption{Four different stacks are highlighted in green: $\Sscr_{14}, \Sscr_5, \Sscr_8$, and $\Sscr_{18}$. For example, stack $\Sscr_5$ contains the links $\{14,5\}$, $\{15,4\}$, and $\{12,2\}$. The total preorder $\preceq_{5}$ on the links in stack $\mathcal{S}_5$ fulfills the following relation for this example: $\{14,5\} \preceq_{5}\{15,4\} \preceq_{5} \{12,2\}$. Notice that in stack $\Sscr_{14}$ there exist distinct links $\{12,8\}$ and $\{12,2\}$ for which both $\{12,8\}\preceq_{14} \{12,2\}$ and $\{12,2\}\preceq_{14} \{12,8\}$ hold. Moreover, they are both above $\{14,5\}$.}\label{fig:example_stacks_and_order}
\end{figure}

For a stack $\Sscr_t$ we define a total preorder $\preceq_t$ on $\Sscr_t$ as follows.
(Recall that preorder means that $\preceq_t$ is transitive and each pair of links is comparable, but there might be distinct links $\ell,\ell'\in \Sscr_t$ for which both $\ell \preceq_t \ell'$ and $\ell' \preceq_t \ell$.)
First, we observe that there is a natural order on the ancestors of the terminal $t$.
The ancestors of $t$ appear in the same order on every $t$-$r$ path in $G$.
Defining $u\preceq v$ if and only if $v$ is an ancestor of $u$ yields  a total order on the ancestors of $t$.
This imposes a natural total preorder on the links in the stack $S_t$.
For cross-links $\ell_u, \ell_v \in \Sscr_t$ incident to the ancestors $u$ and $v$ of terminal $t$, respectively, 
we define $\ell_u \preceq_t \ell_v$ if and only if $u\preceq v$, i.e., if and only if $v$ is an ancestor of $u$.
See Figure~\ref{fig:example_stacks_and_order} for an illustration of stacks and the total preorder that we associate with links contained in a stack.
If $\ell_u \preceq_t \ell_v$, we also say that $\ell_u$ is \emph{below} $\ell_v$  and $\ell_v$ is \emph{above} $\ell_u$ in the stack $\mathcal{S}_t$.

\subsubsection*{Sampling and domination}

We now consider Algorithm~\ref{algo:stack_rounding}.
Recall that for every principle subcactus $G_i$ of $G$, we sample a feasible solution $F_i$ for $G_i$.
Because $F_i$ is $L_{\mathrm{cross}}$-minimal,  by Lemma~\ref{lem:simple_observation_cross_link_domination} it contains at most one link from every stack $\Sscr_t$ with $t\in T$.
We highlight that this holds not only for the stacks of terminals $t\in V_i \cap T$, but also for all other stacks.
This will become important in our analysis later on.

To lower bound the gain achieved by Algorithm~\ref{algo:stack_rounding} when reoptimizing cross-links in Step~\ref{item:compute_removable_set}, we define what we call \emph{removable} link sets, which have the property that the current solution $F_1 \cupp \dots \cupp F_q$ remains feasible even after removing such a link set.
To this end we first introduce the notion of domination, which describes situations where one single link can be removed due to another one.
\begin{definition}[domination]
Let $i,j\in [q]$ with $i\neq j$, and let $\ell_i\in F_i$ and $\ell_j\in F_j$.
We say that $\ell_i$ dominates $\ell_j$ if there is a terminal $t\in T\cap V_j$ satisfying
\begin{itemize}
    \item $\ell_i,\ell_j \in \Sscr_t$, and
    \item $\ell_i \preceq_t \ell_j$.
\end{itemize}
\end{definition}

\begin{definition}
A set $R \subseteq\ \bigcupp_{h=1}^q F^h_{\mathrm{cross}}$ is called \emph{removable} if
for every link $\ell \in R$ there is a link $\bar \ell \in  \left( \bigcupp_{h=1}^q F^h_{\mathrm{cross}}\right)\setminus R$ that dominates $\ell$.
\end{definition}

The next lemma shows that for any set $R$ of removable links, the expression $\sum_{h=1}^q |F^i_{\mathrm{cross}}| -|R|$ is an upper bound on the cardinality of the link set $F_{\mathrm{cross}}$ chosen in Step~\ref{item:compute_removable_set}
of Algorithm~\ref{algo:stack_rounding}.
It also implies that removing $R$ from the disjoint union $\cupp_{h=1}^q F_h$ maintains a feasible CacAP solution.

\begin{lemma}\label{lem:removing_maintains_feasible_solution}
Let $R \subseteq\ \bigcupp_{h=1}^q F^h_{\mathrm{cross}}$ be a removable set of cross-links.
Then $\bigcupp_{h=1}^q F^h_{\mathrm{cross}} \setminus R$ covers every 2-cut $C\in \mathcal{C}_{G_i}$ that is covered by $F^i_{\mathrm{cross}}$.
\end{lemma}
\begin{proof}
Let $C\in \mathcal{C}_{G_i}$ be a 2-cut of $G_i$ for some $i\in [q]$ such that
$F^i_{\mathrm{cross}}$ contains a link $\ell$ covering $C$.
If $\ell \in R$, there must be a link $\bar \ell \in\ \bigcupp_{h=1}^q F_h \setminus R $ that dominates $\ell$.
By the definition of domination and by Lemma~\ref{lem:simple_observation_cross_link_domination}, the link $\bar \ell$ covers the cut $C$.
\end{proof}

In the following we will give a lower bound on the expected maximum cardinality of  a removable set of cross links.

Consider a cross-link $\ell \in F_i$. 
Then $\ell$ is contained in two stacks $\Sscr_{t}$ and $\Sscr_{t'}$, and precisely one of the two terminals $t$ and $ t'$ is contained in $V_i$. 
Without loss of generality, assume that $t\in V_i$ and $ t' \notin V_i$.
Then we say that \emph{the terminal $t$ sampled the cross-link $\ell$}.
Note that every terminal samples at most one cross-link because $F_i$ contains at most one link from every stack.
Moreover, the sampling of the cross-links by terminals in the same principal subcactus is correlated, while the sampling of cross-links by terminals in different principal subcacti is independent.
We emphasize that for the notion of domination between sampled cross-links, it is important which terminals sampled the links.

To analyze the sampling process we consider the more abstract stack problem described below.
We call this problem the \emph{simple stack problem} because we will later describe a refined version.

\subsubsection*{Simple stack problem}

An instance of the simple stack problem consists of
\begin{itemize}
\item an undirected graph $H$ with vertex set $T$,
\item a total preorder $\preceq_t$ on the stack $S_t$, which is the set of edges incident to $t\in T$,
\item a coloring $c : T \to [q]$ such that for every link $\ell$ in $H$, the endpoints of $\ell$ have distinct colors, and
\item non-negative edge weights $x$ satisfying $x(S_t) \le 1$ for every stack $S_t$.
\end{itemize}
The edges of $H$ will correspond to cross-links and we will therefore refer to the edges of $H$ as \emph{links}. 
If $\ell \preceq_t \ell'$ for some $\ell,\ell' \in S_t$, we say that $\ell$ is \emph{below} $\ell'$ on the stack $S_t$ and that $\ell'$ is \emph{above} $\ell$ on $S_t$.
The colors of the terminals $t\in T$ correspond to the principal subcacti of $G$, i.e., a terminal $t\in V_i$ will have color $i$.\footnote{In case the center of the $k$-wide instance is also a terminal, one can assign to it an arbitrary color. The color of the center is irrelevant because the center will have an empty stack.}

For an instance of the simple stack problem, we consider a sampling procedure, where every terminal samples at most one link from its stack (but several links of the same stack can be sampled by different terminals). The sampling procedure has the following properties:
\begin{enumerate}
    \item Each terminal samples at most one link, where link $\ell\in S_t$ is sampled with probability $x_\ell$;\label{item:sampling_probability_link}
    \item if $\{t,v\}$ and $\{t',v'\}$ are sampled by two terminals $t$ and $t'$, respectively, with $t$ and $t'$ being of the same color, then $v\ne v'$;
    \label{item:correlation_of_colors}
    \item the sampling of links by terminals of different colors is independent. 
    \label{item:different_corlors_independent}
\end{enumerate}
We emphasize that the sampling of links by different terminals of the same color do not need to be independent.
(Some correlation might even be necessary for property~\ref{item:correlation_of_colors} to hold.)

We define a (random) directed graph $D=(T,F)$ where $T$ is the set of terminals and $F$ is defined as follows.
For each terminal $t$ and link $\{t,u\}$ sampled by $t$, the set $F$ contains an arc $(u,t)$.
Notice that every vertex in $D$ has at most one incoming link by~\ref{item:sampling_probability_link}.
Let $(u,v),(v,w)$ be two links in $F$. We say that $(v,w)$ \emph{dominates} $(u,v)$ if and only if $\{v,w\}$ is below $\{u,v\}$ in the stack $S_v$.
We also call the graph $D=(T,F)$ the \emph{domination graph} even though, as just explained, to know whether a link leaving a vertex dominates one that enters it, we also need to know their order on the stack corresponding to the vertex.

See Figure~\ref{fig:example_H_and_D} for an example of the graphs $H$ and $D$. (See second and third graph of the figure; the first graph highlights the link from CAP to the stack problem, which we will formally discuss later on.)

\begin{figure}[!ht]
\begin{center}
\begin{tikzpicture}[scale=0.5]

\pgfdeclarelayer{bg}
\pgfdeclarelayer{fg}
\pgfsetlayers{bg,main,fg}

\tikzset{
  prefix node name/.style={%
    /tikz/name/.append style={%
      /tikz/alias={#1##1}%
    }%
  }
}

\tikzset{root/.style={fill=white,minimum size=13}}

\tikzset{q1/.style={line width=1pt, dash pattern=on \pgflinewidth off 1pt}}

\tikzset{q2/.style={line width=1.5pt, dash pattern=on 3pt off 2pt on \the\pgflinewidth off 2pt}}

\tikzset{q3/.style={line width=2pt, dash pattern=on 7pt off 2pt}}

\tikzset{q4/.style={line width=2.5pt}}

\tikzset{s/.style={stealth-, line width=1pt, red}}

\tikzset{
term1/.style={fill=blue!20, rectangle, minimum size=10},
term2/.style={fill=orange!20, rectangle, minimum size=10},
term3/.style={fill=violet!20, rectangle, minimum size=10},
term4/.style={fill=yellow!20, rectangle, minimum size=10},
tf/.append style={font=\scriptsize\color{black}},
}

\newcommand\blob[2]{
\pgfmathanglebetweenpoints%
{\pgfpointanchor{#1}{center}}%
{\pgfpointanchor{#2}{center}}

\edef\angle{\pgfmathresult}

\filldraw ($(#1)+(\angle+90:\d)$) arc (\angle+90:\angle+270:\d) --
($(#2)+(\angle+270:\d)$) arc (\angle+270:\angle+450:\d) -- cycle;
}

\newcommand\cac[2][]{

\begin{scope}[prefix node name=#1]

\tikzset{npc/.style={#2}}

\begin{scope}[every node/.append style={thick,draw=black,fill=white,circle,minimum size=12, inner sep=2pt}]

\begin{scope} %
\node[term1] (22) at (6.16,-5.36) {};
\node[term1] (23) at (6.08,-7.28) {};
\node[term1] (24) at (8.00,-6.56) {};
\end{scope}

\begin{scope} %
\node[term2] (14) at (11.24,-6.32) {};
\node[term2] (15) at (12.68,-7.72) {};
\node[term2] (16) at (14.80,-7.42) {};
\end{scope}

\begin{scope} %
\node[term3]  (4) at (16.66,-9.60) {};
\node[term3]  (6) at (18.24,-8.74) {};
\node[term3]  (8) at (22.28,-5.68) {};
\node[term3]  (9) at (21.00,-3.58) {};
\node[term3] (10) at (20.40,-9.36) {};
\node[term3] (11) at (22.36,-9.32) {};
\end{scope}

\begin{scope} %
\node[term4] (18) at (27.60,-4.64) {};
\node[term4] (19) at (24.44,-7.64) {};
\node[term4] (20) at (26.80,-7.76) {};
\end{scope}

\node[root]  (1) at (18.68,-1.64) {r};
\node  (2) at (16.80,-5.50) {};
\node  (3) at (19.08,-5.16) {};
\node  (5) at (18.56,-7.16) {};
\node  (7) at (20.44,-6.76) {};
\node (12) at (12.24,-4.60) {};
\node (13) at (14.20,-4.64) {};
\node (17) at (24.88,-4.72) {};
\node (21) at (8.04,-4.56) {};

\end{scope}%

\begin{scope}[every node/.append style={font=\scriptsize}]

\foreach \i/\t in {2/,3/,4/tf,5/,6/tf,7/,8/tf,9/tf,10/tf,11/tf,12/,13/,14/tf,15/tf,16/tf,17/,18/tf,19/tf,20/tf,21/,22/tf,23/tf,24/tf} {
\pgfmathparse{int(\i-1)}
\node[\t] at (\i) {$\pgfmathresult$};
}
\end{scope}

\begin{scope}[very thick]

\draw  (1) --  (2);
\draw  (2) --  (3);
\draw  (3) --  (1);
\draw  (9) --  (3);
\draw  (2) to[bend left=15] (4);
\draw  (4) to[bend left=15] (2);
\draw  (3) to[bend left=25] (5);
\draw  (5) to[bend left=25] (3);
\draw  (5) to[bend left=25] (6);
\draw  (6) to[bend left=25] (5);
\draw  (3) --  (7);
\draw  (7) --  (8);
\draw  (8) --  (9);
\draw  (7) -- (10);
\draw (10) -- (11);
\draw (11) --  (7);

\draw  (1) to[bend left=2] (21);
\draw  (1) -- (12);
\draw  (1) -- (17);
\draw (21) to[bend left=2]  (1);
\draw (18) --  (1);
\draw (13) --  (1);
\draw (12) to[bend left=15] (14);
\draw (14) to[bend left=15] (12);
\draw (16) -- (13);
\draw (12) -- (13);
\draw (13) -- (15);
\draw (15) -- (16);
\draw (17) -- (18);
\draw (17) to[bend left=13] (19);
\draw (19) to[bend left=13] (17);
\draw (17) to[bend left=13] (20);
\draw (20) to[bend left=13] (17);
\draw (24) -- (21);
\draw (21) -- (22);
\draw (22) -- (23);
\draw (23) -- (24);

\end{scope}%

\begin{scope}

\begin{scope}[red]

\begin{scope}[q1]
\draw (2) to (14);
\draw (3) to (13);
\draw (4) to (19);
\draw (4) to[bend left=20] (21);
\draw (5) to (16);
\draw (6) to (15);
\draw (7) to (14);
\draw (8) to (19);
\draw (9) to (13);
\draw (9) to (17);
\draw (9) to (18);
\draw (14) to (21);
\end{scope}

\begin{scope}[q2]
\draw (10) to[bend left=25] (23);
\end{scope}

\end{scope}%
\end{scope}%

\begin{pgfonlayer}{bg}

\def\d{6mm}

\begin{scope}[fill=green!10,draw=darkgreen]

\blob{2}{4}
\blob{6}{3}
\blob{7}{10}
\blob{11}{11}
\blob{8}{8}
\blob{9}{9}
\blob{17}{19}
\blob{18}{18}
\blob{20}{20}
\blob{12}{14}
\blob{13}{15}
\blob{16}{16}
\blob{21}{24}
\blob{22}{22}
\blob{23}{23}

\end{scope}
\end{pgfonlayer}

\end{scope}%

}%

\newcommand\h[2][]{

\begin{scope}[prefix node name=#1]

\tikzset{npc/.style={#2}}

\begin{scope}[every node/.append style={thick,draw=black,fill=white,circle,minimum size=12, inner sep=2pt}]

\begin{scope} %
\node[term1] (22) at (6.16,-5.36) {};
\node[term1] (23) at (6.08,-7.28) {};
\node[term1] (24) at (8.00,-6.56) {};
\end{scope}

\begin{scope} %
\node[term2] (14) at (11.24,-6.32) {};
\node[term2] (15) at (12.68,-7.72) {};
\node[term2] (16) at (14.80,-7.42) {};
\end{scope}

\begin{scope} %
\node[term3]  (4) at (16.66,-9.60) {};
\node[term3]  (6) at (18.24,-8.74) {};
\node[term3]  (8) at (22.28,-5.68) {};
\node[term3]  (9) at (21.00,-3.58) {};
\node[term3] (10) at (20.40,-9.36) {};
\node[term3] (11) at (22.36,-9.32) {};
\end{scope}

\begin{scope} %
\node[term4] (18) at (27.60,-4.64) {};
\node[term4] (19) at (24.44,-7.64) {};
\node[term4] (20) at (26.80,-7.76) {};
\end{scope}

\end{scope}%

\begin{scope}[every node/.append style={font=\scriptsize}]

\foreach \i/\t in {2/,3/,4/tf,5/,6/tf,7/,8/tf,9/tf,10/tf,11/tf,12/,13/,14/tf,15/tf,16/tf,17/,18/tf,19/tf,20/tf,21/,22/tf,23/tf,24/tf} {
\pgfmathparse{int(\i-1)}
\node[\t] at (\i) {$\pgfmathresult$};
}
\end{scope}

\begin{scope}

\begin{scope}[q1]

\draw (4) to [bend left=40](14);
\draw (4) to (19);
\draw (4) to [bend left](24);
\draw (6) to [bend left=5] (15);
\draw (6) to [bend right=5] (15);
\draw (6) to [bend right](16);
\draw (8) to (19);
\draw (9) to [bend right](15);
\draw (9) to (18);
\draw (9) to [bend left](19);
\draw (14) to (24);

\end{scope}

\begin{scope}[q2]

\draw (10) to [bend left](23);

\end{scope}

\end{scope}

\end{scope}%

}%

\newcommand\D[2][]{

\begin{scope}[prefix node name=#1]

\tikzset{npc/.style={#2}}

\begin{scope}[every node/.append style={thick,draw=black,fill=white,circle,minimum size=12, inner sep=2pt}]

\begin{scope} %
\node[term1] (22) at (6.16,-5.36) {};
\node[term1] (23) at (6.08,-7.28) {};
\node[term1] (24) at (8.00,-6.56) {};
\end{scope}

\begin{scope} %
\node[term2] (14) at (11.24,-6.32) {};
\node[term2] (15) at (12.68,-7.72) {};
\node[term2] (16) at (14.80,-7.42) {};
\end{scope}

\begin{scope} %
\node[term3]  (4) at (16.66,-9.60) {};
\node[term3]  (6) at (18.24,-8.74) {};
\node[term3]  (8) at (22.28,-5.68) {};
\node[term3]  (9) at (21.00,-3.58) {};
\node[term3] (10) at (20.40,-9.36) {};
\node[term3] (11) at (22.36,-9.32) {};
\end{scope}

\begin{scope} %
\node[term4] (18) at (27.60,-4.64) {};
\node[term4] (19) at (24.44,-7.64) {};
\node[term4] (20) at (26.80,-7.76) {};
\end{scope}

\end{scope}%

\begin{scope}[every node/.append style={font=\scriptsize}]

\foreach \i/\t in {2/,3/,4/tf,5/,6/tf,7/,8/tf,9/tf,10/tf,11/tf,12/,13/,14/tf,15/tf,16/tf,17/,18/tf,19/tf,20/tf,21/,22/tf,23/tf,24/tf} {
\pgfmathparse{int(\i-1)}
\node[\t] at (\i) {$\pgfmathresult$};
}
\end{scope}

\begin{scope}

\begin{scope}[s]

\draw (9) to [bend left](19);
\draw (19) to [bend left=40](9);
\draw (4) to [out=170,in=-80](14);
\draw (15) to [bend left](9);
\draw (6) to [bend left=5] (15);
\draw (24) to [bend right](4);
\draw (10) to [bend left](23);
\draw (14) to (24);

\end{scope}

\end{scope}

\end{scope}%

}%

\begin{scope}

\begin{scope}[shift={(0,20)}]
\cac[]{}
\end{scope}

\begin{scope}[shift={(0,10)}]
\h[]{}
\end{scope}

\begin{scope}[shift={(0,0)}]
\D[]{}
\end{scope}

\end{scope}

\begin{scope}[shift={(30,18.5)}]%
\def\ll{30mm} %
\def\vs{40mm} %

\begin{scope}[scale=0.5]
\draw[q1,black] (0,0) -- +(\ll,0) node[right, align=left,yshift=-1.2ex] {cross-links with\\ $x$-value $0.25$}[red];
\draw[q2,black, yshift=-\vs] (0,0) -- +(\ll,0) node[right, align=left, yshift=-1.2ex] {cross-links with\\ $x$-value $0.5$}[red];
\end{scope}

\begin{scope}[scale=0.5]

\begin{pgfonlayer}{bg}
\def\d{6mm}
\end{pgfonlayer}

\end{scope}

\end{scope}%

\begin{scope}[shift={(30,6)}]%
\def\ll{30mm} %
\def\vs{40mm} %

\begin{scope}
\node (c) at (0,0) {};
\node at (c.west)[right=0pt] {Graph $H$};
\end{scope}

\begin{scope}[scale=0.5]
\draw[q1,black, yshift=-\vs] (0,0) -- +(\ll,0) node[right, align=left, yshift=-1.2ex] {cross-links with\\ $x$-value $0.25$};
\draw[q2,black, yshift=-2*\vs] (0,0) -- +(\ll,0) node[right, align=left, yshift=-1.2ex] {cross-links with\\ $x$-value $0.5$};
\end{scope}

\end{scope}%

\begin{scope}[shift={(30,-4)}]%
\def\ll{30mm} %
\def\vs{40mm} %

\begin{scope}
\node (d) at (0,0) {};
\node at (d.west)[right=0pt] {Graph $D=(T,F)$};
\end{scope}

\begin{scope}[scale=0.5]
\node (e) at (0,-\vs) {};
\node (f) at (\ll,-\vs) {};
\draw[s,red] (f) to (e);
\node at (f.east)[right=0pt] {sampled links};
\end{scope}

\end{scope}%

\end{tikzpicture}
 \end{center}
\caption{The first picture shows an instance in which we partition the set of cross-links into stacks, as highlighted in green. The second picture shows the graph $H$ that we can construct from the first instance. The third picture shows the graph $D$ and a possible solution of sampled cross-links. Notice that in all pictures, terminals in the same principal subcactus have the same color.}\label{fig:example_H_and_D}
\end{figure}
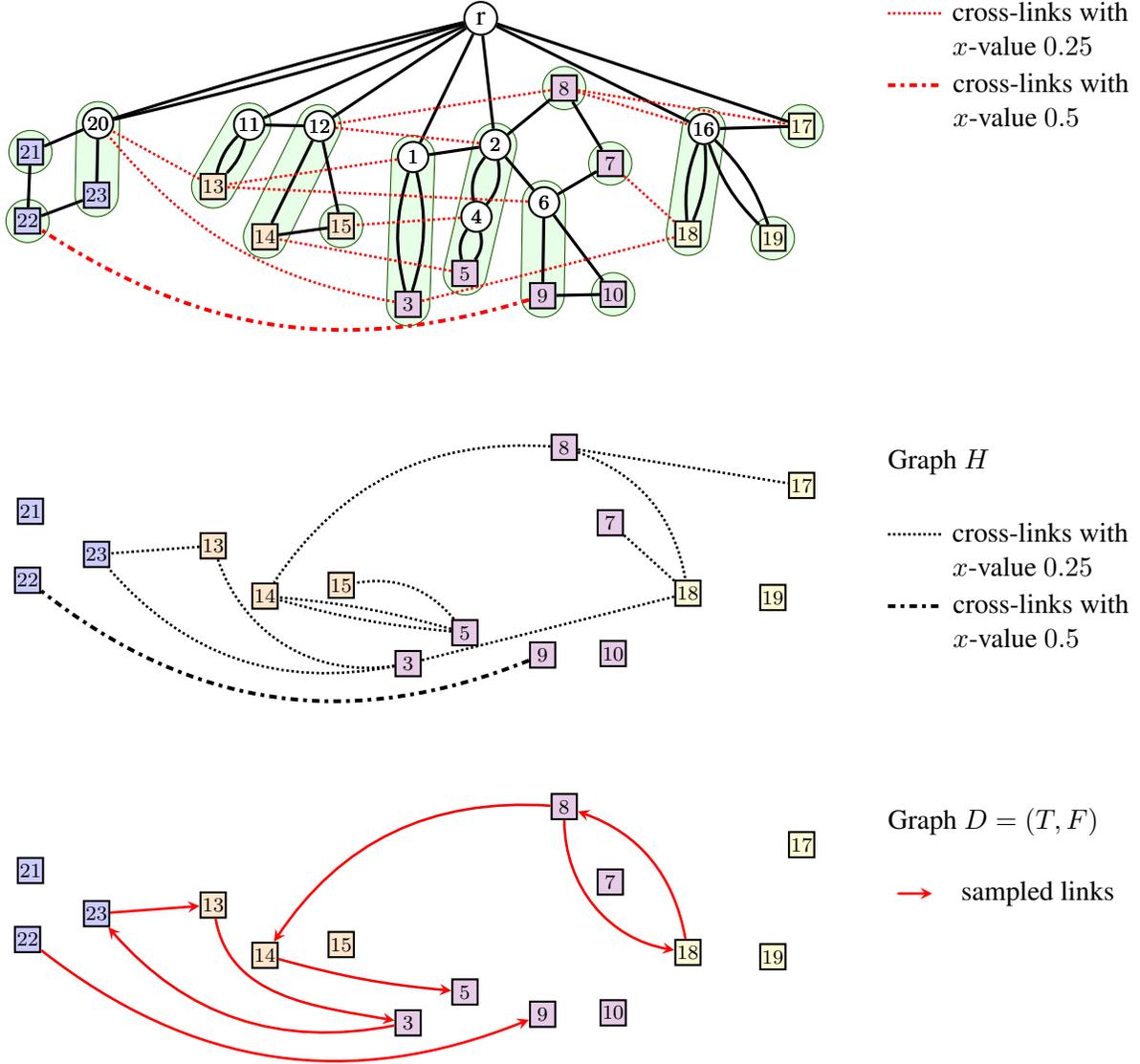

We call a set $R\subseteq F$ \emph{removable} if for every link $\ell \in R$ there exists a link in $F\setminus R$ that dominates $\ell$.
The task in the simple stack problem is to compute a maximum cardinality removable subset of $F$.

We will prove the following lower bound on the expected size of an optimum solution $R$ of the simple stack problem.

\begin{lemma}\label{lem:bound_simple_stack_problem}
For any instance of the simple stack problem and any sampling process fulfilling \ref{item:sampling_probability_link}, \ref{item:correlation_of_colors}, and \ref{item:different_corlors_independent},
we have
\[
 \mathbb{E}[|R|] \ge \frac{1}{3} \cdot \sum_{t\in T} f(x(S_t))\enspace,
\]
where $R$ denotes a maximum cardinality set of removable edges and $f(z) \coloneqq z + e^{-z} -1$.
\end{lemma}

We will apply this to the simple stack problem defined from an instance $(G,L)$ of connectivity augmentation as follows. 
Let $x \in P^{\min}_{\mathrm{bundle}}(G,L)\subseteq [0,1]^L$.
For every cross-link $\ell \in L_{\mathrm{cross}}$ contained in the stacks $\Sscr_{t_1}$ and $\Sscr_{t_2}$, the graph $H$ contains a link $\{t_1,t_2\}$ with $x_{\{t_1,t_2\}} \coloneqq x_{\ell}$.
The order $\preceq_t$ on a stack $S_t$ is the same as the order $\preceq_t$ on the corresponding links in $\mathcal{S}_t \subseteq L_{\mathrm{cross}}$.
Moreover, a terminal $t\in T$ has color $i$ if and only if $t\in V_i$.

We then consider the sampling process that arises from the randomized rounding procedure in Step~\ref{item:algo_sampling} of Algorithm~\ref{algo:stack_rounding}
 (sampling $F_1, \dots F_q$ from $x \in P^{\min}_{\mathrm{bundle}}(G,L)$)
by identifying the links of $H$ with the corresponding cross-links in $L_{\mathrm{cross}}\subseteq L$.
We have already observed that $L_{\mathrm{cross}}$-minimality of the sets $F_i$ implies that every terminal $t$ samples at most one link.
Because every link $\ell$ is sampled with probability $x_{\ell}$, this implies in particular that $x(S_t) \le 1$ for all $t\in T$.
Moreover, also property~\ref{item:correlation_of_colors} follows from the $L_{\mathrm{cross}}$-minimality of the sets $F_i$.
Here, we use that $|F_i \cap \mathcal{S}_t| \le 1$ holds also for all terminals $t\in T$ that are not contained in $V_i$.
Property~\ref{item:different_corlors_independent} holds because the colors correspond to the different principal subcacti $G_i$ and the sampling of the sets $F_i$ for different $i\in [q]$ is independent.

Let $R$ be an optimum solution to the simple stack problem, i.e., $R\subseteq F$ is a maximum cardinality removable set of links in the domination graph $D$.
 Note that the definition of \emph{domination} in the stack problem is such that a set $R\subseteq F$ of links in $D$ is removable if and only if its corresponding subset of $\bigcupp_{i=1}^q F_i$ is removable.
By Lemma~\ref{lem:removing_maintains_feasible_solution}, we have $|F_{\mathrm{cross}}| \le \sum_{h=1}^q |F^i_{\mathrm{cross}}| -|R|$ 
 and hence Algorithm~\ref{algo:stack_rounding} yields a solution with at most $\sum_{i=1}^q |F_i| - |R|$ many links.
Because every in-link is contained in precisely one of the sets $L_i$ and every cross-link is contained in two sets $L_i$, we have
\begin{equation}\label{eq:bound_in_terms_of_R_simple}
\sum_{i=1}^q \mathbb{E}[|F_i|] - \mathbb{E}[|R|] = \sum_{i=1}^q \mathbb{E}[x(L_i)] - \mathbb{E}[|R|] =
 x(L_{\mathrm{in}}) + 2 \cdot x(L_{\mathrm{cross}}) - \mathbb{E}[|R|]\enspace .
\end{equation}

\subsubsection*{Analyzing the simple stack problem}

Next, we prove Lemma~\ref{lem:bound_simple_stack_problem}.
To this end we will show that the expected number of dominated links in $F$ is at least $\sum_{t\in T} f(x(S_t))$ and
 that it is always possible to remove at least a third of the dominated links.

\begin{lemma}\label{lem:delte_a_third}
Let $F$ be a fixed outcome of the sampling in the simple stack problem and let $R \subseteq F$ be a maximum cardinality set of removable links.
Then
\[
|R| \ge \frac{1}{3} \cdot |\{ \ell \in F: \ell \text{ is dominated}\}|\enspace.
\]
\end{lemma}
\begin{proof}
Let $R\subseteq F$ be an (inclusionwise) maximal removable set of links.
Then for every dominated link $\ell \in F\setminus R$, the set $R\cup \{\ell\}$ is not removable.
Therefore, one of the following conditions must hold for each dominated link $\ell\in F \setminus R$:
\begin{enumerate}
\item $\ell$ is dominated only by links $\bar \ell\in R$; \label{item:dominating_link_deleted}
\item there is a link $\bar \ell \in R$ such that $\ell$ is the only link in $F\setminus R$ that dominates $\bar \ell$.
        \label{item:only_dominating_link}
\end{enumerate}
Consider a link $\bar \ell =(v,w) \in R$. 
Because $v$ has at most one incoming link in $F$, \ref{item:dominating_link_deleted} holds for at most one dominated link $\ell \in F\setminus R$.
Moreover, \ref{item:only_dominating_link} holds for at most one dominated link $\ell \in F\setminus R$.
This shows that $R$ contains at least a third of the dominated links in $F$.
\end{proof}

We remark that the bound in Lemma~\ref{lem:delte_a_third} is tight when $D=(T,F)$ is a triangle in which every link is dominated.

We now provide a lower bound on the number of dominated links in $F$, which corresponds to the left-hand side of the inequality in the lemma below.
\begin{lemma}\label{lem:simple_domination_prob}
\[
\sum_{t\in T} \mathbb{P}[t\text{ has an incoming dominated link in }F]\geq \sum_{t\in T} f(x(S_t))\enspace,
\]
where $f(z) \coloneqq z + e^{-z} -1$.
\end{lemma}
\begin{proof}
Consider a fixed terminal $t\in T$. 
We observe that the sampling of outgoing links of $t$ in $D$ is independent of the sampling of an incoming link.
The reason is that because of the way we defined the orientation of the links, the outgoing links of the terminal $t$ are not sampled by $t$ but by its neighbors.
All neighbors have a different color than $t$, implying by property~\ref{item:different_corlors_independent} independence with respect to the sampling of incoming links at $t$.
For a link $\ell \in S_t$ we define $B(\ell) = \{ \bar \ell \in S_t : \bar \ell \preceq_t \ell \}$ to be the set of links that are below $\ell$ on the stack $S_t$.
Then, using property~\ref{item:sampling_probability_link}, we conclude
\[
 \mathbb{P} [t\text{ has an incoming dominated link in }F] = \sum_{\ell \in S_t} x_{\ell} \cdot \mathbb{P}[t\text{ has an outgoing link in }F \cap B(\ell)]\enspace. 
\]
For every fixed color $c$, the probability that $F$ contains a link $\{t,v\}\in B(\ell)$ such that $v$ has color $c$,
is \[ \gamma^c(\ell) \coloneqq \sum_{\substack{\{t,v\}\in B(\ell):\\ v\text{ has color }c}} x_{\{t,v\}}\enspace.\]
This follows from properties~\ref{item:sampling_probability_link} and~\ref{item:correlation_of_colors} of the sampling process.
The sampling of outgoing links of $t$ that end in vertices of different colors is independent by~\ref{item:different_corlors_independent}.
Hence,
\begin{align*}
\mathbb{P}[t\text{ has an outgoing link in }F \cap B(\ell)] =&\ 1 - \prod_{\text{color }c} \left( 1- \gamma^c(\ell) \right) \notag\\
\ge&\ 1 - e^{-\sum_{\text{color }c} \gamma^c(\ell)} \\
=&\ 1- e^{- x(B(\ell))}\enspace, \notag
\end{align*}
where we used that $e^y \ge 1 +y $ for all $y \in \mathbb{R}$.
Therefore,  by the definition of $B(\ell)$, we get
\begin{align*}
\mathbb{P} [t\text{ has an incoming dominated link in }F] \ge&\ \sum_{\ell \in S_t} x_{\ell} \cdot \left(1- e^{- x(B(\ell))}\right) \\
\ge&\ \int_{0}^{x(S_t)} 1- e^{- z}\  dz \\
=&\ x(S_t) + e^{- x(S_t)} - 1 \\
=&\ f(x(S_t))\enspace. \qedhere
\end{align*}
\end{proof}

Combining Lemma~\ref{lem:delte_a_third} and Lemma~\ref{lem:simple_domination_prob} yields Lemma~\ref{lem:bound_simple_stack_problem}.

\subsubsection*{Obtaining an approximation ratio below $1.5$}

Next, we use Lemma~\ref{lem:bound_simple_stack_problem} to give an approximation algorithm for $k$-wide CacAP with approximation ratio less than $1.5$.
By Theorem~\ref{thm:main_reduction} this implies a better than $1.5$ approximation ratio for CAP.

First, we compute an optimum solution $x$ to $\min\{ 1^Ty : y \in P_{\mathrm{cross}} \cap P^{\mathrm{min}}_{\mathrm{bundle}}(G,L)\}$.
This is possible in polynomial time by Lemma~\ref{lem:cross-link_rounding} and Lemma~\ref{lem:optimize_efficiently_min_bundle}.
Lemma~\ref{lem:existence_optimum_minimal_solution} implies that $x(L)$ is a lower bound on the size $|\OPT|$ of an optimum solution.

Let $\alpha := \frac{x(L_{\mathrm{cross}})}{x(L)}$.
Then we have $2\alpha \cdot x(L) =  \sum_{t\in T} x(\mathcal{S}_t)$ because every cross link is contained in precisely two stacks.
Moreover, by the definition of terminals, the set $\{t\}$ is a 2-cut for every $t\in T$.
Hence, $x(\delta_L(t)) \ge 1 $ for all $t\in T$ and we have $2 \cdot x(L) \ge |T|$, implying
\begin{equation}\label{eq:simple_lower_bound_stack_loads}
\sum_{t\in T} x(\mathcal{S}_t) = 2 \cdot \alpha\cdot x(L) \ge \alpha \cdot  |T|\enspace.
\end{equation}

The rounding procedure discussed in Section~\ref{sec:cg} (using Chv\'atal-Gomory cuts) yields a solution with at most $(2-\alpha)\cdot x(L)$ many links 
by Lemma~\ref{lem:cross-link_rounding}.
Combining the bound \eqref{eq:bound_in_terms_of_R_simple} and Lemma~\ref{lem:bound_simple_stack_problem},  
yields that the solution returned by Algorithm~\ref{algo:stack_rounding}
has at most
\begin{align*}
x(L) + x(L_{\mathrm{cross}}) - \mathbb{E}[|R|] \le&\ \left(1 + \alpha - \frac{1}{3 \cdot x(L)} \sum_{t\in T} f(x(\mathcal{S}_t))  \right) \cdot x(L)
\end{align*}
many links in expectation.
Using  \eqref{eq:simple_lower_bound_stack_loads} we conclude that the algorithm that returns the better of the two solutions has an approximation ratio of at most
\begin{align*}
\max \Big\{ \min\Big\{ 2-\alpha,\  1 + \alpha - \frac{2 \alpha}{3 \cdot   \sum_{t\in T} y_t} \sum_{t\in T} f(y_t) \Big\} :&\  0\le y_t \le 1 \text{ for } t\in T,\  \sum_{t\in T} y_t \ge \alpha\cdot |T|,  \\ &\ \  0\le \alpha \le 1 \Big\}\enspace.
\end{align*}
Because the function $f$ defined by $f(z) \coloneqq z + e^{-z} -1$ is convex, the (inner) minimum is achieved when all $y_t$ are identical, which leads to 
\begin{align*}
\max \Big\{& \min\big\{ 2-\alpha,\  1 + \alpha - \frac{2 \alpha}{3 |T|\cdot y} \sum_{t\in T} f(y) \big\}  : 
\alpha\le y \le 1,\ 0\le \alpha \le 1 \Big\} \\
=\ \max \Big\{& \min\big\{ 2-\alpha,\  1 + \alpha - \frac{2\alpha}{3y} f(y)\big\} :\ 0\le \alpha \le y \le 1 \Big\} \\
<\ 1.459 \enspace.
\end{align*}
This approximation ratio is slightly higher than the previously best known approximation ratio of $1.458$ for TAP \cite{grandoni_2018_improved},
but it already is a significant improvement for CAP.
In the following we present a refinement of our simple stack analysis, leading to an approximation ratio below $1.4$.

\subsection{Types of links on stacks}\label{sec:link_types_stacks}

In this section we outline some of the observations that we use to improve on the simple stack analysis.
First, we observe that our lower bound~\eqref{eq:simple_lower_bound_stack_loads} on the average weight of the links contained in a stack, can only be tight if all links $\ell$ with $x_{\ell}>0$ are incident to terminals. 
More precisely, we can strengthen the bound~\eqref{eq:simple_lower_bound_stack_loads} as follows.
For a terminal $t\in T$, let $\Uplambda_t^0 \subseteq \Sscr_t$ be the set of cross-links that are incident to the terminal $t$ in~$G$ and
let $\uplambda^0_t := x(\Uplambda^0_t) = \sum_{\ell\in \Uplambda^0_t} x_{\ell}$ be the total weight of these cross-links.
Because $x(\delta_L(t))\ge 1$ for all $t\in T$, we have
\begin{equation*}
2 \cdot x(L) \ge \sum_{t\in T} (x(\delta_L(t)) + x(\Sscr_t \setminus \Uplambda^0_t)) \ge \sum_{t\in T} (1+ x(\Sscr_t \setminus \Uplambda^0_t)) = \sum_{t\in T} (1+ x(\Sscr_t) - \uplambda^0_t)\enspace.
\end{equation*}
It is possible that $\uplambda^0_t = x(\Sscr_t)$ for all terminals $t$, in which case our new bound is not stronger than~\eqref{eq:simple_lower_bound_stack_loads}.
In this case however, we can obtain a better lower bound on the number of dominated links.
The reason is that we then have $\ell \preceq_t \bar\ell$ for all $\ell, \bar \ell \in \mathcal{S}_t$.
Therefore, whenever  the terminal $t$ in the stack problem has both an incoming and an outgoing link in $D=(T,F)$,
the incoming link will be dominated. 
This allows for strengthening Lemma~\ref{lem:simple_domination_prob}.

It turns out that when exploiting the above observations, in the worst case there will be quite many cross-links that are not incident to terminals. 
In order to improve our approximation ratio further, we will show how we can profit from having such links, i.e., links in $\Sscr_t\setminus \Uplambda^0_t$.
We will partition $\Sscr_t\setminus \Uplambda^0_t$ into three sets.
For two of these sets we will show that a high $x$-weight of these sets allows us to improve our analysis further.
For the third set of links, we don't gain anything but we can give an upper bound on the total weight these links can have.

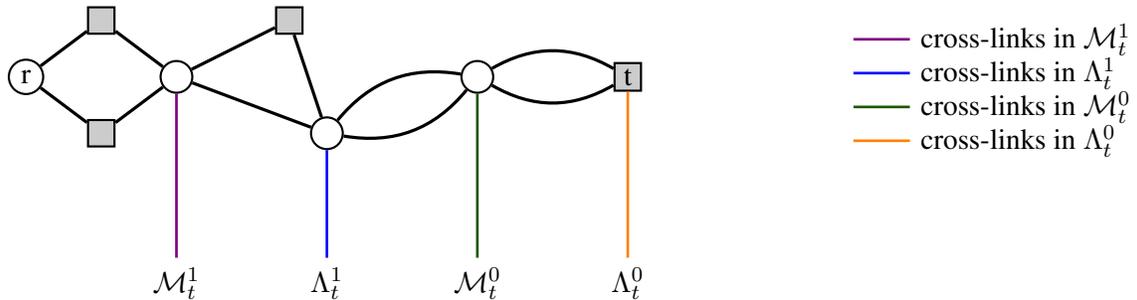
\begin{figure}[!ht]
\begin{center}
\begin{tikzpicture}[scale=0.5]

\tikzset{
  prefix node name/.style={%
    /tikz/name/.append style={%
      /tikz/alias={#1##1}%
    }%
  }
}

\tikzset{root/.style={fill=white,minimum size=13}}

\tikzset{q1/.style={line width=1pt, dash pattern=on \pgflinewidth off 1pt}}

\tikzset{q2/.style={line width=1.5pt, dash pattern=on 3pt off 2pt on \the\pgflinewidth off 2pt}}

\tikzset{q3/.style={line width=2pt, dash pattern=on 7pt off 2pt}}

\tikzset{q4/.style={line width=2.5pt}}

\tikzset{s/.style={line width=1pt}}

\tikzset{
term/.style={fill=black!20, rectangle, minimum size=10},
term1/.style={fill=blue!20, rectangle, minimum size=10},
term2/.style={fill=orange!20, rectangle, minimum size=10},
term3/.style={fill=violet!20, rectangle, minimum size=10},
term4/.style={fill=yellow!20, rectangle, minimum size=10},
tf/.append style={font=\scriptsize\color{black}},
}

\newcommand\cac[2][]{

\begin{scope}[prefix node name=#1]

\tikzset{npc/.style={#2}}

\begin{scope}[every node/.append style={thick,draw=black,fill=white,circle,minimum size=12, inner sep=2pt}]

\node[root]  (1) at (8,-2.5) {r};

\node[term] (2) at (10,-1) {};
\node (3) at (12,-2.5) {};
\node[term] (4) at (10,-4) {};
\node[term] (5) at (15,-1) {};
\node (6) at (16,-4) {};
\node (7) at (20,-2.5) {};
\node[term] (8) at (24,-2.5) {t};

\end{scope}%

\begin{scope}
\node (9) at (12,-8) {$\mathcal{M}_t^1$};
\node (10) at (16,-8) {$\Lambda_t^1$};
\node (11) at (20,-8) {$\mathcal{M}_t^0$};
\node (12) at (24,-8) {$\Lambda_t^0$};
\end{scope}

\begin{scope}[every node/.append style={font=\scriptsize}]

\foreach \i/\t in {2/,3/,4/,5/tf,6/,7/,8/tf,9/,10/,11/,12/} {
\pgfmathparse{int(\i-1)}
}
\end{scope}

\begin{scope}[very thick]

\draw (1) to (2);
\draw (1) to (4);
\draw (2) to (3);
\draw (4) to (3);
\draw (3) to (5);
\draw (3) to (6);
\draw (5) to (6);
\draw (6) to [bend right](7);
\draw (6) to [bend left](7);
\draw (7) to [bend right](8);
\draw (7) to [bend left](8);

\end{scope}%

\begin{scope}[s]
\draw[violet] (3) to (9);
\draw[blue] (6) to (10);
\draw[darkgreen] (7) to (11);
\draw[orange] (8) to (12);

\end{scope}%

\end{scope}%

}%

\begin{scope}

\begin{scope}[shift={(0,20)}]
\cac[]{}
\end{scope}

\end{scope}

\begin{scope}[shift={(30,18.5)}]%
\def\ll{30mm} %
\def\vs{18mm} %

\begin{scope}[scale=0.5]
\draw[s,black] (0,0) -- +(\ll,0) node[right] {cross-links in $\mathcal{M}_t^1$}[violet];
\draw[s,black, yshift=-\vs] (0,0) -- +(\ll,0) node[right] {cross-links in $\Lambda_t^1$}[blue];
\draw[s,black, yshift=-2*\vs] (0,0) -- +(\ll,0) node[right] {cross-links in $\mathcal{M}_t^0$}[darkgreen];
\draw[s,black, yshift=-3*\vs] (0,0) -- +(\ll,0) node[right] {cross-links in $\Lambda_t^0$}[orange];
\end{scope}

\end{scope}%

\end{tikzpicture}
 \end{center}
\caption{Different types of cross-links.}\label{fig:mu_lambda_types}
\end{figure}

We partition $\Sscr_t\setminus \Uplambda^0_t$ into $\Mscr^0_t$, $\Uplambda^1_t$, and $\Mscr^1_t$ as follows.
Consider a link $\ell=\{v,w\}$ of $\Sscr_t\setminus \Uplambda^0_t$, where $v$ is an ancestor of $t$.
See Figure~\ref{fig:mu_lambda_types}.
Then
\begin{itemize}
\item $\ell$ is contained in  $\Mscr^0_t$ if $v$ is not contained in any cycle of length at least 3 in $G$ 
         and $t$ is the only terminal that is a descendant of $v$; 
\item $\ell$ is contained in $\Uplambda^1_t$ if $v$ is contained in a cycle of length at least 3 in $G$
         and $t$ is the only terminal that is a descendant of $v$; (As we will see later, for any terminal $t$ there can be at most one vertex $v$ with these properties.)
\item $\ell$ is contained in $\Mscr^1_t$, otherwise; 
         then $t$ is not the only terminal that is a descendant of $v$.
\end{itemize}
We have $\mathcal{S}_t = \Uplambda^0_t \cupp \Mscr^0_t \cupp \Uplambda^1_t \cupp \Mscr^1_t$.
We define
\begin{align*}
\upmu_t^0     &\coloneqq x(\Mscr_t^0)\enspace,\\
\uplambda_t^1 &\coloneqq x(\Uplambda_t^1)\enspace,\\
\upmu^1_t     &\coloneqq x(\Mscr^1_t)\enspace, \text{ and we recall that}\\
\uplambda_t^0 &\coloneqq x(\Uplambda_t^0)\enspace.
\end{align*}
Then $x(\mathcal{S}_t) = \uplambda^0_t + \upmu^0_t + \uplambda^1_t + \upmu^1_t$.

Let us now explain how we can gain from links in $\Mscr^0_t$.
The lemma below states that a high LP value on  these links implies a high LP value on up-links. 
This allows us to gain in the rounding procedure from Section~\ref{sec:cg}, which is lossless on up-links. (See Lemma~\ref{lem:cross-link_rounding}.)
\begin{lemma}\label{lem:generating_up_links}
Let $x\in P^{\mathrm{min}}_{\mathrm{cross}}(G,L)$.
Then \[ x(L_\mathrm{up})\ \ge\ \sum_{t\in T} \upmu^0_t \enspace. \]
\end{lemma}
\begin{proof}
Let $t\in T$ and let $\ell\in \Mscr_t^0$.
We denote the endpoint of $\ell$ that is an ancestor of $t$ by $v$.
Then $v$ is not contained in any cycle of length at least three in $G$ and $t$ is the only terminal that is a descendant of $v$.

Because $\ell \notin \Uplambda^0_t$, the vertex $v$ is not a terminal.
Thus $v$ has degree at least $4$, i.e., it is contained in at least two cycles in $G$.
In fact, $v$ is contained in precisely two cycles in $G$ because $v$ is an ancestor of only one terminal.
Using that $v$ is not contained in any cycle of length at least three, we conclude that $v$ has exactly 2 neighbors $u,w$ in $G$ and is connected to each of them by a pair of parallel edges.
One of the two neighbors, say $u$, is an ancestor of $v$. The other neighbor $w$ is a descendant of $v$.
See Figure~\ref{fig:up_link_generation}.

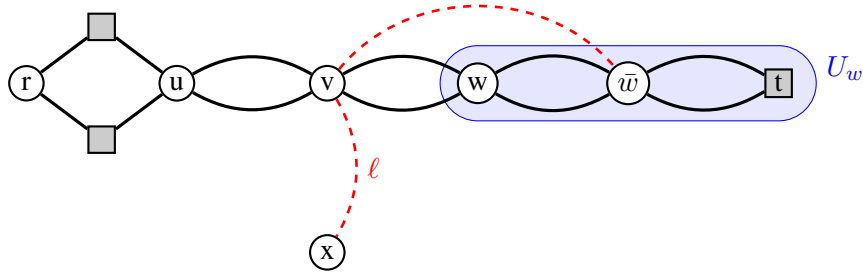
\begin{figure}[!ht]
\begin{center}
\begin{tikzpicture}[scale=0.5]

\pgfdeclarelayer{bg}
\pgfdeclarelayer{fg}
\pgfsetlayers{bg,main,fg}

\tikzset{root/.style={fill=white,minimum size=13}}

\tikzset{s/.style={dashed,red,line width=1pt}}

\tikzset{
term/.style={fill=black!20, rectangle, minimum size=10},
tf/.append style={font=\scriptsize\color{black}},
}

\newcommand\blob[2]{
\pgfmathanglebetweenpoints%
{\pgfpointanchor{#1}{center}}%
{\pgfpointanchor{#2}{center}}

\edef\angle{\pgfmathresult}

\filldraw ($(#1)+(\angle+90:\d)$) arc (\angle+90:\angle+270:\d) --
($(#2)+(\angle+270:\d)$) arc (\angle+270:\angle+450:\d) -- cycle;
}

\newcommand\cac[2][]{

\begin{scope}

\begin{scope}[every node/.append style={thick,draw=black,fill=white,circle,minimum size=12, inner sep=2pt}]

\node[root]  (1) at (8,-2.5) {r};

\node[term] (2) at (10,-1) {};
\node (3) at (12,-2.5) {u};
\node[term] (4) at (10,-4) {};
\node (5) at (16,-2.5) {v};
\node (6) at (20,-2.5) {w};
\node (7) at (24,-2.5) {$\bar{w}$};
\node[term] (8) at (28,-2.5) {t};
\node (9) at (16,-7) {x};

\end{scope}%

\begin{scope}[very thick]

\draw (1) to (2);
\draw (1) to (4);
\draw (2) to (3);
\draw (4) to (3);
\draw (3) to [bend right](5);
\draw (3) to [bend left](5);
\draw (5) to [bend right](6);
\draw (5) to [bend left](6);
\draw (6) to [bend right](7);
\draw (6) to [bend left](7);
\draw (7) to [bend right](8);
\draw (7) to [bend left](8);

\end{scope}%

\begin{scope}[s]
\draw (5) to [bend left] node[right]{$\ell$} (9);
\draw (5) to [bend left=50](7);
\end{scope}%

\begin{pgfonlayer}{bg}

\def\d{10mm}

\begin{scope}[fill=blue!10,draw=blue]

\blob{6}{8}

\end{scope}
\end{pgfonlayer}

\node () at (8.north)[right=0.5,blue]{$U_w$};

\end{scope}%

}%

\begin{scope}

\begin{scope}[shift={(0,20)}]
\cac[]{}
\end{scope}

\end{scope}

\end{tikzpicture} \end{center}
\caption{Illustration of the proof of Lemma~\ref{lem:generating_up_links}.}\label{fig:up_link_generation}
\end{figure}

Let $i\in[q]$ such that $t\in V_i$.  Consider any $L_{\mathrm{cross}}$-minimal solution $F_i$ for $G_i$.
By Lemma~\ref{lem:descendants_define_2_cut}, the set $U_w$ of descendants of $w$ (including $w$ itself) is a 2-cut of $G$.
Since this 2-cut is not covered by $\ell$, there must be a link $\{\bar v,\bar w\}\in F_i$ covering it.
Without loss of generality we have $\bar w\in U_w$, i.e., $\bar w$ is a descendant of $w$, and $\bar v\notin U_w$.
Now suppose $\bar v\ne v$.
Then the link $\{\bar v, \bar w\}$ covers the 2-cut $U_v$ of descendants of $v$.
But then, replacing $\ell=\{v,z\}$ by its shadow $\{u,z\}$ does not change the set of 2-cuts covered by $\ell$ and $\{\bar v, \bar w\}$,
 contradicting the $L_{\mathrm{cross}}$-minimality of $F_i$.
This shows $\bar v=v$.
Because $v$ is an ancestor of $w$ and $w$ is an ancestor of $\bar w$, the link $\{v,\bar w\} =\{\bar v,\bar w\}$ is an up-link.
Both endpoints of this up-link are ancestors of the terminal $t$, but not of any other terminal.

Recall that any $L_{\mathrm{cross}}$-minimal solution $F_i$ for $G_i$ contains at most one link $\ell \in \mathcal{S}_t$.
We have shown that if such an $F_i$  contains a link $\ell \in \Mscr^0_t$ for some terminal $t\in V_i$,
 then it also contains an up-link such that both endpoints of this up-link are ancestors of $t$, but not of any other terminal.
This implies that $F_i$ contains a distinct up-link for every cross link $\ell \in \bigcup_{t\in T\cap V_i} F_i \cap \Mscr^0_t$.
Because $x|_{L_i}$ is a convex combination of incidence vectors of $L_{\mathrm{cross}}$-minimal solutions for $G_i$, this concludes the proof.
\end{proof}

Next, we consider $\Uplambda^1_t$. We show that all links in $\Uplambda^1_t \subseteq L_{\mathrm{cross}}$ have a common endpoint~$v$.
Using this one can give a better bound on the expected number of dominated links, similar to the improvement we explained for $\Uplambda^0_t$ before.
\begin{lemma}\label{lem:unique_endpoint_lambda_1}
Let $t\in T$ be a terminal of the cactus $G$. Then there is at most one vertex $v\in V$ such that
\begin{enumerate}
\item $t$ is a descendant of $v$, \label{item:property_1_of_lambda_1}
\item $v$ is contained in a cycle of length at least 3 in $G$, and  \label{item:property_2_of_lambda_1}
\item $t$ is the only terminal below $v$. \label{item:property_3_of_lambda_1}
\end{enumerate}
If such a vertex $v$ exists, every ancestor $a$ of $v$ with $a\ne v$ is an ancestor of at least two terminals.
\end{lemma}
\begin{proof}
Let $v$ be a vertex satisfying \ref{item:property_1_of_lambda_1} -- \ref{item:property_3_of_lambda_1} and let $a\ne v$ be an ancestor of $v$.
By \ref{item:property_2_of_lambda_1}, the vertex $v$ is contained in a cycle $Q$ of length at least $3$ in $G$.
Because $v$ is descendant of $a$, also all other vertices of the cycle $Q$ are descendants of $a$.

Note that for every cycle in $G$, there can be at most one vertex that is an ancestor of some other vertex of the cycle.
(In fact, there is always exactly on such vertex, but we will not need this.) 
Because the cycle $Q$ has length at least $3$, there exist two distinct vertices $u,w$ in $Q$ such that neither $u$ is an ancestor of $w$ nor $w$ is an ancestor of $u$.
Since by Lemma~\ref{lem:terminal_descendant_exists} every vertex of $G$ is an ancestor of some terminal, there are terminals $t_u$ and $t_w$ such that $u$ is an ancestor of $t_u$ and $w$ is an ancestor of $t_w$.
The terminals $t_u$ and $t_w$ must be distinct because otherwise both $u$ and $w$ were ancestors of $t_u=t_w$, implying that $u$ is an ancestor of $w$ or vice versa.
Now, recall that the vertex $a$ is an ancestor of all vertices of the cycle $Q$.
In particular, $a$ is an ancestor of $u$ and $w$.
This implies that $a$ is also an ancestor of the two distinct terminals $t_u$ and $t_w$.

It remains to show that there is at most one vertex $v$ with \ref{item:property_1_of_lambda_1} -- \ref{item:property_3_of_lambda_1}.
Suppose this is not the case and there exist two distinct vertices $v$ and $\bar v$ with \ref{item:property_1_of_lambda_1} -- \ref{item:property_3_of_lambda_1}. 
Then by \ref{item:property_1_of_lambda_1}, one of these vertices is an ancestor of the other, say $\bar v$ is an ancestor of $v$.
But then $\bar v$ is an ancestor of at least two different terminals as we have shown above and this contradicts~\ref{item:property_3_of_lambda_1}.
Hence, there can be at most one vertex $v$ with \ref{item:property_1_of_lambda_1} -- \ref{item:property_3_of_lambda_1}.
\end{proof}

Finally, we give an upper bound on the total weight of links contained in sets $\Mscr^1_t$ for $t\in T$. 
Recall that $x(\mathcal{S}_t) \le 1$ for all $t\in T$ because every $L_{\mathrm{cross}}$-minimal solution $F_i$ for a principal subcatus $G_i$ contains at most one cross-link from $\mathcal{S}_t$.
This implies 
\[
\upmu^1_t\ =\ x(\Mscr^1_t) \ \le\ 1 - x(\Sscr_t \setminus \Mscr^1_t)\ =\ 1 - \uplambda^0_t - \upmu^0_t - \uplambda^1_t \enspace .
\]
We can strengthen the resulting bound on $\sum_{t\in T} \upmu^1_t$ as follows.
\begin{lemma}\label{lem:upper_bound_mu1_links}
We have
\begin{equation*}
\sum_{t\in T} \upmu^1_t\ \le\  \frac{1}{2} \cdot  \sum_{t\in T}\left( 1 - \uplambda^0_t - \upmu^0_t - \uplambda^1_t\right) \enspace.
\end{equation*}
\end{lemma}
\begin{proof}
We claim that for every terminal $t\in T$, the total $x$-weight of cross-links that are incident to an ancestor of $t$ is at most $1$.
Let $t\in T \cap V_i$ with $i\in [q]$ and consider an $L_{\mathrm{cross}}$-minimal solution $F_i$ of $G_i$.
Then by Lemma~\ref{lem:simple_observation_cross_link_domination}, the $L_{\mathrm{cross}}$-minimality of $F_i$ implies that 
$F_i$ contains at most one cross-link incident to an ancestor of $t$.
Because $x\in P^{\min}_{\mathrm{bundle}}$, this implies that indeed 
\[
  x\big(\big\{\{v,w\} \in L_{\mathrm{cross}} : v\text{ is an ancestor of }t \big\}\big)\ \le\ 1
\]
for every terminal $t\in T$.
Now recall that every link in $\mathcal{S}_t$ is a cross-link with an endpoint that is an ancestor of $t$.
Moreover, every cross-link $\ell \in \Mscr^1_t$ has an endpoint that is an ancestor of both the terminal $t$ and another terminal $\bar t\in T \setminus\{t\}$
with $\ell\notin \mathcal{S}_{\bar t}$.
Therefore, we get 
\[
 \sum_{t\in T}\left(x(\mathcal{S}_t) + \upmu^1_t\right) \le \sum_{t\in T}  x\big(\big\{\{v,w\} \in L_{\mathrm{cross}} : v\text{ is an ancestor of }t \big\}\big) \le |T|\enspace,
\]
implying
\[
\sum_{t\in T} \upmu^1_t\ \le\ \frac{1}{2}\cdot \Big(|T| - \sum_{t\in T}( x(\mathcal{S}_t) - \upmu^1_t)\Big)
\ =\  \frac{1}{2}\cdot \sum_{t\in T} \left( 1 - \uplambda^0_t - \upmu^0_t - \uplambda^1_t\right) \enspace.
\]
\end{proof}

In order to obtain an approximation ratio below $1.4$, we will use the improvements we described in this section together with a 
more careful analysis of the stack problem.
This analysis exploits that our previous lower bounds on the domination probabilities and on the fraction of dominated cross-links that we can remove
cannot be tight simultaneously.
Before explaining this in more detail, let us first define the \emph{stack problem}, which is a refined version of the simple stack problem we considered before.

\subsection{The stack problem}

An instance of the stack problem consists of
\begin{itemize}
\item an undirected graph $H$ with vertex set $T$,
\item a total preorder $\preceq_t$ on the stack $S_t$, which is the set of edges incident to $t\in T$,
\item non-negative edge weights $x$ satisfying $x(S_t)\leq 1$ for every stack $S_t$,
\item a partition of $S_t$ into $\Lambda^0_t$, $M^0_t$, $\Lambda^1_t$, and $M^1_t$ for every $t\in T$ such that
\begin{align*}
   \ell\ &\preceq_t\ \bar \ell & &\text{ for all }\ell, \bar \ell\in \Lambda_t^0 \\
\ell\ &\preceq_t\ \bar \ell & &\text{ for all }\ell, \bar \ell\in \Lambda_t^1 \\
\intertext{and}
\ell^0\ &\prec_t \ell & &\text{ for all }\ell^0 \in \Lambda_t^0,\ \ell \in S_t \setminus \Lambda_t^0\\
 m^0\ &\prec_t \ell^1 & &\text{ for all }m^0 \in M^0_t,\ \ell^1 \in \Lambda_t^1\\
\ell\ &\prec_t m^1 & &\text{ for all }\ell \in S_t \setminus M^1_t,\ m^1 \in  M_t^1\enspace,
\end{align*}
where we write $\ell_1 \prec \ell_2$ if $\ell_1 \preceq \ell_2$ and $\ell_2 \not\preceq \ell_1$, and
\item a coloring $c : T \to [q]$ such that for every edge $\ell$ in $H$, the endpoints of $\ell$ have distinct colors.
\end{itemize}
We define for every $t\in T$:
\begin{equation*}
   \lambda_t^0 := \sum_{l \in \Lambda_t^0} x_l\enspace, \qquad 
   \mu_t^0 := \sum_{l \in M_t^0} x_l\enspace, \qquad 
    \lambda_t^1 := \sum_{l \in \Lambda_t^1} x_l\enspace, \qquad
    \mu_t^1 := \sum_{l \in M_t^1} x_l\enspace.
\end{equation*}

As for the simple stack problem, we consider a sampling procedure where every terminal samples at most one link from its stack and
properties \ref{item:sampling_probability_link}, \ref{item:correlation_of_colors}, and \ref{item:different_corlors_independent} hold.
We consider the resulting (random) directed domination graph $D=(T,F)$ where $F$ is the set of sampled links and the orientation of the links is chosen 
such that $(u,v) \in F$ if and only if the link $\{u,v\}$ is sampled by the terminal $v$. 
The task in the stack problem is to compute a maximum cardinality removable subset of $F$.
In particular, the graph $H$ of the stack problem is the domination graph $D=(T,F)$. Due to this we also refer to the edges of $H$ as \emph{links}.

Before going into further technical details, we briefly sketch our goal for what comes next.
We will define a function $g : [0,1]^4 \to \mathbb{R}_{\ge 0}$ and fix $b=0.42$.
(See Section~\ref{sec:compute_domination_prob} for the precise definition of $g$.)
Then we will prove the following lower bound on the expected size of an optimum solution $R$ of the stack problem.

\begin{lemma}\label{lem:bound_stack_problem}
For any instance of the stack problem and any sampling process fulfilling \ref{item:sampling_probability_link}, \ref{item:correlation_of_colors}, and \ref{item:different_corlors_independent},
we have
\[
 \mathbb{E}[|R|] \ge b \cdot \sum_{t\in T} g(\lambda^0_t, \mu^0_t, \lambda^1_t, \mu^1_t)\enspace,
\]
where $R$ denotes a removable set of edges of maximum cardinality.
\end{lemma}
This lemma will be a strengthening of Lemma~\ref{lem:bound_simple_stack_problem} because we have $b=0.42 > \frac{1}{3}$ and we choose $g$ such that
\[
g(\lambda^0_t, \mu^0_t, \lambda^1_t, \mu^1_t)\ \ge\ f(\lambda^0_t + \mu^0_t + \lambda^1_t + \mu^1_t)\enspace .
\]
In what follows, we set $b=0.42$ and use $g$ for this fixed function which we formally define in Section~\ref{sec:compute_domination_prob}.

\subsection{Applying the stack problem}\label{sec:application_stack_game}

Next, we explain how we apply Lemma~\ref{lem:bound_stack_problem}.
First, we define an instance of the stack problem.
The graph $H$ is defined as for the simple stack problem, i.e., for every cross-link $\ell$ that is contained in the stacks $\Sscr_{t_1}$ and $\Sscr_{t_2}$, 
the graph $H$ contains a link $\{t_1,t_2\}$ with $x_{\{t_1,t_2\}} \coloneqq x_{\ell}$.
Also the order $\preceq_t$ on a stack $S_t$ is the same as for the simple stack problem, i.e., 
it is given by the order $\preceq_t$ on the corresponding links in $\mathcal{S}_t \subseteq L_{\mathrm{cross}}$.
Then we define $\Lambda^0_t$, $M^0_t$, $\Lambda^1_t$, and $M^1_t$ to be the subsets of $S_t$ that correspond to the 
subsets $\Uplambda^0_t$, $\Mscr^0_t$, $\Uplambda^1_t$, and $\Mscr^1_t$ of $\mathcal{S}_t$, respectively.
This partition of $S_t$ has indeed the properties required  in the definition of the stack problem as the following lemmas show.
\begin{lemma}
We have
\begin{itemize}
    \item $\ell \preceq_t \bar \ell$ for all $\ell, \bar \ell \in \Lambda_t^0$; 
    \item $\ell \preceq_t \bar \ell$ for all $\ell, \bar \ell \in \Lambda_t^1$.
\end{itemize}
\end{lemma}
\begin{proof}
By the definition of $\Uplambda_t^0$, all links contained in this set are incident to the terminal $t$ in the cactus $G$.
Hence,  $\ell \preceq_t \bar \ell$ for all $\ell,\bar \ell\in \Uplambda_t^0$.
Moreover, by Lemma~\ref{lem:unique_endpoint_lambda_1}, there is a vertex $v\in V$ such that all links in $\Uplambda_t^1$ are incident to $v$ in $G$.
Therefore, $\ell \preceq_t \bar \ell$ for all $\ell, \bar \ell \in \Uplambda_t^1$.
\end{proof}

\begin{lemma}
We have 
\begin{itemize}
    \item $\ell^0 \prec_t \ell$ for all  $\ell^0 \in \Lambda_t^0$  and $\ell \in S_t \setminus \Lambda_t^0$;
    \item $ m^0 \prec_t \ell^1$ for all  $m^0 \in M^0_t$ and  $\ell^1 \in \Lambda_t^1$;
    \item $\ell \prec_t m^1$ for all $\ell \in S_t \setminus M^1_t$ and  $m^1 \in  M_t^1$.
\end{itemize}
\end{lemma}
\begin{proof}
Let $\ell^0 \in \Uplambda_t^0$ and $\ell \in \mathcal{S}_t \setminus \Lambda_t^0$.
Then $\ell^0 \prec_t \ell$ because $\ell^0$ is incident to the terminal $t$ in the cactus $G$ but $\ell$ is not.

Next we consider links $m^0\in \Mscr^0_t$ and  $\ell^1\in \Uplambda^1_t $.
Suppose for the sake of deriving a contradiction that $\ell^1 \preceq_t m^0$.
By the definition of $\Uplambda^1_t$, the link $\ell^1$  is incident to the unique vertex $v$ defined in Lemma~\ref{lem:unique_endpoint_lambda_1}.
Moreover, $\ell^1 \preceq_t m^0$ implies that $m^0$  is incident to an ancestor $a$ of $v$ in $G$.
By Lemma~\ref{lem:unique_endpoint_lambda_1}, we either have $a=v$ or the vertex $a$ is an ancestor of at least two distinct terminals.
This contradicts $m^0 \in \Mscr^0_t$.
Therefore, we must have $m^0 \prec_t \ell^1$.

Finally, we consider links $\ell\in \mathcal{S}_t\setminus \Mscr_t^1$ and $m^1\in  \Mscr_t^1$.
Suppose for the sake of deriving a contradiction that $m^1 \preceq_t \ell$.
By the definition of $\Mscr^1_t$, the link $m^1$  is incident to a vertex $u$ that is an ancestor of both $t$ and some other terminal $\bar t \ne t$.
The relation $m^1 \preceq_t \ell$ implies that $\ell$ is incident to an ancestor $a$ of $u$ in $G$.
Then $a$ is an ancestor of both $t$ and $\bar t$.
This contradicts $\ell \in \mathcal{S}_t\setminus \Mscr_t^1 = \Uplambda^0_t \cup \Mscr^0_t \cup \Uplambda^1_t$.
(Here, we used that the terminal $t$ cannot be an ancestor of the terminal $\bar t$ because otherwise every $\bar t$-$r$ path in $G$ would visit $t$ and thus use both edges incident to $t$,
which contradicts the 2-edge-connectivity of $G$.)
Hence, $\ell \prec_t m^1$.
\end{proof}

We complete the definition of our instance of the stack problem by
defining that  $t\in T$ has color $i$ if and only if $t\in V_i$.

Then we consider the same sampling procedure as for the simple stack problem.
We have shown in Section~\ref{sec:simple_stack_problem} that this procedure fulfills the properties  \ref{item:sampling_probability_link}, \ref{item:correlation_of_colors}, and \ref{item:different_corlors_independent}.

Let now $R$ be an optimum solution to the stack problem, i.e., $R\subseteq F$ is a maximum cardinality removable set of links in the random domination graph $D$.
As we have argued already for the simple stack problem in Section~\ref{sec:simple_stack_problem},  Algorithm~\ref{algo:stack_rounding} yields a solution with at most
\begin{equation}\label{eq:bound_in_terms_of_R}
\sum_{i=1}^q \mathbb{E}[|F_i|] - \mathbb{E}[|R|] = x(L_{\mathrm{in}}) + 2 \cdot x(L_{\mathrm{cross}}) - \mathbb{E}[|R|]
\end{equation}
many links in expectation.
Together with Lemma~\ref{lem:bound_stack_problem}, this implies the following.

\begin{lemma}\label{lem:stack_algo}
Algorithm~\ref{algo:stack_rounding} returns a solution with at most
\[
   x(L_{\mathrm{in}}) + 2 \cdot x(L_{\mathrm{cross}}) - b \cdot \sum_{t\in T} g(\uplambda^0_t, \upmu^0_t, \uplambda^1_t, \upmu^1_t)
\]
many links in expectation.
\end{lemma}

\subsection{Bounding the number of removable links for a fixed outcome}

In the following we give a lower bound on the maximum cardinality of a removable set of links in the stack problem.
We start by analyzing the number of removable links for a fixed outcome $F$ of the sampling procedure.
First, we consider a simple special case, which will be a useful building block in the following.

\begin{definition}
Let $D'=(T',F')$ be a subgraph of $D$.
We say that $D'$ is \emph{maximally dominating} if for every $(v,w)\in F'$, every outgoing link of $w$ in $F'$ dominates the link $(v,w)$.
\end{definition}

\begin{lemma}\label{lem:simple_bound_deletable}
If $D=(T,F)$ is maximally dominating and has an Eulerian walk, then there is a removable set $R\subseteq F$ containing
\begin{itemize}\itemsep2pt
\item  at least $\frac{1}{2} \cdot |F|$ many links if $|F|$ is even, and
\item  at least $\frac{1}{2} \cdot (|F|-1)$  many links if $|F|$ is odd.
\end{itemize}
\end{lemma}
\begin{proof}
Let $\ell_1, \dots, \ell_{|F|}$ be the links in $F$ in the order in which they appear in an Eulerian walk.
We can remove all $\ell_i$ with odd index $i < |F|$. Then $\ell_i$ is dominated by $\ell_{i+1}$, which is not removed.
\end{proof}

In order to reduce a more general case to the special case in the above lemma, we need the following simple and purely graph-theoretical observation.

\begin{lemma}\label{lem:partition_into_walks}
Let $G'=(V',E')$ be a connected directed graph where every vertex has at most one incoming arc.
Then $E'$ can be partitioned into the arc sets of $\max\{1,\ |\{v\in V' : \delta^+(v) =\emptyset\}| \}$ many walks in $G'$.
\end{lemma}
\begin{proof}
Because every vertex has in-degree at most one, the vertices $v$ with $|\delta^-(v)| > |\delta^+(v)|$ are precisely those with out-degree zero and in-degree one. 
Thus the total excess in terms of in-degree is equal to  
$\sum_{v\in V'} \max\{|\delta^-(v)|- |\delta^+(v)|,0\} = h := |\{v\in V' : \delta^+(v) =\emptyset\}|$.
Therefore, we can add a set $X$ of $h$ auxiliary arcs to obtain an Eulerian graph $(V, E'\cupp X)$, i.e., a connected graph where $|\delta^-(v)|= |\delta^+(v)|$ for all $v\in V'$.
This graph has a closed Eulerian walk. Removing the $h$ arcs in $X$ results in $\max\{1,h\}$ many walks, whose arc sets partition $E'$.
\end{proof}

\begin{lemma}\label{lem:delete_half}
Suppose $D=(T,F)$ is connected and maximally dominating, but is not a cycle of odd length.
Then there is a removable set $R\subseteq F$ containing  at least half of the dominated links in $F$.
\end{lemma}
\begin{proof}
Let us first consider the case where $|\{v\in T : \delta^+(v) =\emptyset\}| =0$.
Then every vertex has exactly one incoming and one outgoing link 
because $|T| \le \sum_{t\in T} |\delta^+_F(t)| =\sum_{t\in T} |\delta^-_F(t)| $ and every terminal has in-degree at most one in $F$.
Using that $D$ is connected, we conclude that $D$ is a cycle and hence by assumption it is an even cycle.
Therefore, Lemma~\ref{lem:simple_bound_deletable} implies that there is a removable set containing at least half of the dominated links.

Now we consider the remaining case where $h:=|\{v\in T : \delta^+(v) =\emptyset\}| \ge 1$.
We may assume $|F|>0$, otherwise the statement is trivial.
Then by Lemma~\ref{lem:partition_into_walks}, we can partition $F$ into the edge sets of $h$ walks.
Applying Lemma~\ref{lem:simple_bound_deletable} to each of these walks (or more precisely to the graph consisting of the vertex set and the arc set of the walk), implies that there is a removable set containing at least $\frac{1}{2}(|F| -h)$ many links.
Because $D$ is connected, each of the vertices with out-degree zero, has an incoming link.
This link is not dominated. Hence, $F$ contains at most $|F|-h$ many dominated links.
\end{proof}

If $D$ is an odd cycle, Lemma~\ref{lem:simple_bound_deletable} yields a removable set containing at least
\[ \frac{1}{2} \cdot |\{ \ell\in F : \ell\text{ is dominated}\}| -  \frac{1}{2}\] many links and this bound is tight.
Because we cannot delete at least half of the dominated links in this case, odd cycles will play a special role in our analysis.
In some other cases, we are able to obtain a stronger bound than the one from Lemma~\ref{lem:delete_half}.

\begin{lemma}\label{lem:special_bound_deletable}
Suppose $D=(T,F)$ is connected and maximally dominating, but is not a cycle of odd length.
If there is an even length $s$-$t$ path in $D$ such that $s$ has in-degree zero and $t$ has out-degree zero and $F\ne \emptyset$,
then there is a removable set $R\subseteq F$ containing at least
$ \frac{1}{2} \cdot |\{ \ell\in F : \ell\text{ is dominated}\}| +  \frac{1}{2}$ many links.
\end{lemma}
\begin{proof}
Let $P$ be an even length $s$-$t$ path in $D$ such that $s$ has in-degree zero in $D$ and $t$ has out-degree zero in $D$.
Because every vertex has in-degree at most one in $D$, none of the arcs of $P$ dominates any link in $F\setminus E(P)$.
(Here, $E(P)$ denotes the arc set of $P$.)
Hence, the number of dominated links in $(T,F\setminus E(P))$ is exactly $|\{\ell \in F\setminus E(P) : \ell\text{ is dominated in }(T,F)\}|$.
Therefore, applying Lemma~\ref{lem:delete_half} to every connected component of $(T, F\setminus E(P))$ and applying Lemma~\ref{lem:simple_bound_deletable} to $P$ implies that there is a removable set $R\subseteq F$ containing at least
\[
 \frac{1}{2} \cdot |\{\ell \in F\setminus E(P) : \ell\text{ is dominated in }(T,F)\}| + \frac{1}{2} \cdot |E(P)|
\]
many links. 
This completes the proof because $ |\{\ell \in E(P) : \ell\text{ is dominated in }(T,F)\}|= |E(P)| -1$.
\end{proof}

In order to analyze the maximum size of a removable set of links in the general case, where the domination graph $(T,F)$ is not necessarily maximally dominating, we partition $F$ into several \emph{components}.
\begin{definition}
A \emph{component} of $F$ is a minimal nonempty set $\bar F \subseteq F$ such that
\begin{itemize}
\item if $\bar \ell\in \bar F$ dominates $\ell \in F$, then $\ell\in \bar F$, and
\item if $\bar \ell \in \bar F$ is dominated by $\ell \in F$, then $\ell \in \bar F$.
\end{itemize}
\end{definition}
In the following, we denote by $\Fscr$ the set of all components, which is a partition of $F$.
When analyzing the maximum size of a removable set $R\subseteq F$, we can consider the different components of $F$ independently.
We call a connected component of a graph \emph{trivial} if it consists only of a single vertex.
For a component $\bar F\in \Fscr$, the graph $(T,\bar F)$ has only one non-trivial connected component by the minimality of the component $\bar F$.

\begin{lemma}\label{lem:num_deletable_for_components}
Let $F'\in \Fscr$ be a component of $F$ and let $(T',F')$ be the unique non-trivial connected component of $(T,F')$.
Then there is a removable set $R'\subseteq F'$ containing at least
\begin{itemize}\itemsep3pt
\item $\frac{1}{2}\cdot |\{\ell\in F' : \ell\text{ dominated}\}| - \frac{1}{2}$ links if $(T',F')$ is a maximally dominating odd cycle;
\item $\frac{1}{2}\cdot |\{\ell\in F' : \ell\text{ dominated}\}| + \frac{1}{2}$ links if $(T',F')$ is an arborescence containing an even length path from the root to one of its leaves; and
\item $\frac{1}{2}\cdot |\{\ell\in F' : \ell\text{ dominated}\}|$ links, otherwise.
\end{itemize}
\end{lemma}
\begin{proof}
If $(T',F')$ is maximally dominating, the statement follows from
 Lemmas~\ref{lem:simple_bound_deletable},~\ref{lem:special_bound_deletable}, and~\ref{lem:delete_half}.
Let us therefore assume that $(T',F')$ is not maximally dominating.

Because every vertex has at most one incoming link in $F'$, we have $|F'| \le |T'|$.
Hence, $(T',F')$ contains at most one cycle.
By the minimality of $F'$, every link $(v,w)\in F'$ with $\delta_F^-(v)\ne \emptyset$ that is not part of this cycle must 
dominate the incoming link of $v$.
Moreover, there is at most one link $(v,w)\in F'$ in a cycle that does not dominate the incoming link of the vertex $v$.
Because $(T',F')$ is not maximally dominating, there is exactly one such link $(v,w)$.
We add a new auxiliary vertex $a$ to $T'$ and replace the link $(v,w)$ by a link $(a,w)$.
This modification does not change whether a link $\ell\in F'$ dominates another link $\ell'\in F'$.
The resulting component is maximally dominating and hence Lemma~\ref{lem:delete_half} implies that 
there is a removable set $R'\subseteq F'$ containing at least half of the dominated links.
\end{proof}

\subsection{Dangerous events: correlation vs.~independence}

As mentioned before, components that are maximally dominating odd cycles will play a special role in our analysis because 
these are the only reason why we cannot always remove at least half of the dominated links.
Especially short odd cycles only allow us to delete a small fraction of the dominated links.
Let us start with an informal explanation of how we avoid the worst-case assumption that we can always remove only a third of the dominating links.

The idea is to prove that one of the following is true:
\begin{itemize}
\item It is unlikely that there are many short odd cycles.
\item We can improve our previously used lower bound on the number of dominated links.
\end{itemize}
To this end we define the value $\eta^c_t$  for every terminal $t\in T$ and every color $c$ to be the total weight
\begin{equation*}
\eta^c_t := \sum_{\{v,t\}\in S_t : v\text{ has color }c} x_{\{v,t\}}
\end{equation*}
of links between the vertex $t$ and vertices of the color $c$ in the graph $H$.
If these values are close to $0$, it is unlikely that there are many short odd cycles.
If these values are large, we can improve our lower bound on the number of dominated links.
The reason why we can improve in the latter case is that the sampling of the links entering vertices of color $c$ is correlated. 
More precisely, by property~\ref{item:correlation_of_colors} of the sampling procedure, there is always at most one link $(t,v)\in F$ for which $v$ has color $c$.
We will use this to prove a better lower bound on the probability that an incoming link of $t$ is dominated.

Let us now outline how we obtain a good lower bound on the maximum size of a removable set of links when the numbers $\eta^c_t$ are small, i.e., close to zero.
We will not directly bound the probability of odd cycle components.
Instead, we will consider certain events $\Delta^{vw}$ for pairs $v,w$ of terminals.
These events have two important properties:
\begin{itemize}
\item Whenever there is a short odd cycle component, some of the $\Delta^{vw}$ events must happen. (See Lemma~\ref{lem:odd_cycles_imply_delta_events} for the precise statement.)
\item We can give a simple bound on the  probability of short odd cycle components, conditioned on $\Delta^{vw}$, and this bound is good for small $\eta^c_t$.
(See Lemma~\ref{lem:bound_conditional_prob}.)
\end{itemize}
These two properties allow us to show that for small values of $\eta^c_t$, the probability of short odd cycle components is small,
no matter whether the probability of the events $\Delta^{vw}$ is small or large.

Next, we formally define the events $\Delta^{vw}$.
First we will define a path $P_v$ in the graph $(T,F)$ for each vertex $v\in T$. 
Let $c_v$ be the color of $v$.
Then we consider the graph
$(T, F \setminus F_{c_v})$, where $F_{c_v} := \delta_F^-(\{u\in T: u\text{ has color }c_v\})$ is the set of links entering a vertex of color $c_v$.
In terms of our random sampling procedure, i.e., Algorithm~\ref{algo:stack_rounding}, which acts independently on each subcactus, $F_{c_v}$ are all cross-links sampled by the principal subcactus containing $v$.
If $v$ has no outgoing link in $F \setminus F_{c_v}$, the path $P_v$ has length 0. 
(In this case $v$ has no outgoing link in $F$ at all, because there can be no link from $v$ to a vertex of same color.
Later, it will become important that we consider $F\setminus F_{c_v}$.)
Otherwise, $v$ has an outgoing link in $F \setminus F_{c_v}$ and we follow an outgoing link in $F \setminus F_{c_v}$ which is lowest on the stack $S_v$. 
We break ties in an arbitrary deterministic way, independent of the sampled links.
For example by fixing a total order on the stack upfront and using this to break ties. 
(Hence, the only randomness involved in the events $\Delta^{vw}$ is the sampling of the links.)
Then we iteratively, do the following.
Let $u$ be the current vertex and let $\ell$ be the unique incoming link of $u$ in $F$.
If $u$ has at least one outgoing link in $F \setminus F_{c_v}$ that dominates $\ell$, we follow a link with that property.
We choose that link to be lowest possible on the stack $S_v$.
Again, we break ties in an arbitrary deterministic way.
Otherwise, i.e., if $u$ has no outgoing link in $F\setminus F_{c_v}$ dominating $\ell$, then the path $P_v$ ends in $u$.
This procedure terminates and produces a path because $v$ has no incoming link in $F \setminus F_{c_v}$ and all other vertices have at most one incoming link.

We recall that $b\coloneqq 0.42$. (This is the constant $b$ appearing in Lemma~\ref{lem:bound_stack_problem}.)
\begin{definition}
For $v,w\in T$ the event $\Delta^{vw}$ is the event in which the path $P_v$ ends in $w$ and the length of $P_v$ is strictly positive, even, and at most $\frac{1}{1-2b}$.
\end{definition}
Notice that the event $\Delta^{vw}$ is empty if $v$ and $w$ have the same color $c_v=c_w$ because in that case $w$ has no incoming link in $F \setminus F_{c_v}$.
Moreover, the events $\Delta^{vw}$ and $\Delta^{vu}$ for $w\ne u$ are disjoint.
We remark that the upper bound $\frac{1}{1-2b}$ on the length of $P_v$ would not be needed for our randomized algorithm, but it simplifies derandomization. (See Section~\ref{sec:derandomizing}.)

The next lemma provides a connection between odd cycles and the events $\Delta^{vw}$ for $v,w\in T$.
\begin{lemma}\label{lem:odd_cycles_imply_delta_events}
Let $F'$ be a component of $F$ such that $F'$ is the arc set of a maximally dominating odd cycle $(T',F')$ and $|F'| \le \frac{1}{1-2b}$.
Then for every color $c$, the following holds.
If $T'$ contains any vertex of color $c$, then there is a vertex $v\in T'$ of color $c$ such that $F\in \Delta^{vw}$ for some $w\in T'$.
\end{lemma}

\begin{proof}
Let $c$ be a color such that at least one vertex in $T'$ has color $c$. 
For every vertex $v\in T'$ of color $c$, we consider the walk
from $v$ to the next vertex on $(T',F')$ of color $c$ (following the cycle). 
These walks partition $F'$, which has odd cardinality, and hence at least one of these walks must have odd length.
Let $v$ be a vertex for which this walk has odd length.
Let $(w,v')$ be the last link of this walk.
Then $v$ and $v'$ have color $c$, but none of the other vertices in the $v$-$v'$ walk in $(T',F')$ has color $c$.
Because we always chose a lowest outgoing arc in the construction of the path $P_v$, the arc set of this $P_v$ must be contained in the component $F$.
This shows that $P_v$ is the even length $v$-$w$ path in $(T',F')$. 
Because $P_v$ has length at most $|F'| \le \frac{1}{1-2b}$, this implies $F\in \Delta^{vw}$.
\end{proof}

\begin{definition}
We say that a link $\ell=(u,v)$ is \emph{dominating} if $u$ has an incoming link that is dominated by~$\ell$.
(We recall that each vertex has at most one incoming link.)
\end{definition}

\begin{figure}[!hb]
\begin{center}
\begin{tikzpicture}[scale=0.5]

\pgfdeclarelayer{bg}
\pgfdeclarelayer{fg}
\pgfsetlayers{bg,main,fg}

\tikzset{
  prefix node name/.style={%
    /tikz/name/.append style={%
      /tikz/alias={#1##1}%
    }
  }
}

\tikzset{ 
    cross/.pic = {
    \draw[rotate = 45] (-#1,0) -- (#1,0);
    \draw[rotate = 45] (0,-#1) -- (0, #1);
    }
}

\tikzset{root/.style={fill=white,minimum size=13}}

\tikzset{s/.style={dashed,red,line width=1pt}}

\tikzset{
term/.style={fill=black!20, rectangle, minimum size=10},
tf/.append style={font=\scriptsize\color{black}},
}

\newcommand\blob[2]{
\pgfmathanglebetweenpoints%
{\pgfpointanchor{#1}{center}}%
{\pgfpointanchor{#2}{center}}

\edef\angle{\pgfmathresult}

\filldraw ($(#1)+(\angle+90:\d)$) arc (\angle+90:\angle+270:\d) --
($(#2)+(\angle+270:\d)$) arc (\angle+270:\angle+450:\d) -- cycle;
}

\newcommand\cac[2][]{

\begin{scope}

\begin{scope}[every node/.append style={thick,draw=black,fill=white,circle,minimum size=12, inner sep=2pt}]

\node[fill=red!20] (1) at (6,-6) {v};
\node[fill=blue!20] (2) at (4,-3) {};
\node[fill=green!20] (3) at (7,-0) {};
\node[fill=blue!20] (4) at (11,-1.5) {};
\node[fill=yellow!20] (5) at (10,-5) {w};

\end{scope}%

\begin{scope}[-stealth,very thick]
\draw[dashed] (1) to (2);
\draw (2) to (3);
\draw (3) to (4);
\draw (4) to (5);

\end{scope}%

\begin{scope}[s]

\end{scope}%

\end{scope}%

}%

\begin{scope}

\begin{scope}[shift={(0,20)}]
\node() at (3,0) {$\Delta^{vw}$};
\cac[delta0-]{}
\end{scope}

\begin{scope}[shift={(15,20)}]%
\def\ll{30mm} %
\def\vs{18mm} %

\begin{scope}[scale=0.5]
\draw[-stealth,black, very thick, yshift=-\vs] (0,0) -- +(\ll,0) node[right] {dominating links};
\draw[-stealth,black,dashed, yshift=-2*\vs] (0,0) -- +(\ll,0) node[right] {not known if dominating or not};
\end{scope}

\begin{scope}[scale=0.5,yshift=-5*\vs]

\begin{scope}[every node/.style={circle,inner sep=0pt,minimum size=5pt}]
\node[fill=red] (rd) at (0,0) {};
\node[fill=black] (bd) at (0,-\vs) {};

\node at (rd.east)[right=4pt] {links sampled by vertices of color $c_v$};
\node at (bd.east)[right=4pt] {other sampled links};
\end{scope}

\end{scope}

\end{scope}%

\begin{scope}[shift={(0,10)}]%
\node() at (3,0) {$\Delta^{vw}_1$};
\cac[delta1-]{}

\begin{scope}[-stealth,very thick]%
\draw[red] (5) to (1);
\end{scope}

\end{scope}%

\begin{scope}[shift={(18,10)}]%
\node() at (3,0) {$\Delta^{vw}_2$};
\cac[delta2-]{}
\node[fill=red!20,thick,draw=black,circle,minimum size=12, inner sep=2pt] (6) at (14,-6) {};
\begin{scope}[-stealth,very thick]%
\draw[red] (5) to (6);
\end{scope}
\end{scope}%

\begin{scope}[shift={(0,0)}]%
\node() at (3,0) {$\Delta^{vw}_3$};
\cac[delta3-]{}
\node[fill=red!20,thick,draw=black,circle,minimum size=12, inner sep=2pt] (6) at (14,-6) {};
\node[fill=violet!20,thick,draw=black,circle,minimum size=12, inner sep=2pt] (7) at (3,-8) {};
\begin{scope}[-stealth,very thick]
\draw[red] (5) to (6);
\end{scope}
\begin{scope}[-stealth,dashed,very thick]%
\draw[red] (7) to (1);
\end{scope}

\path (12,-5.5) pic[rotate = 90,very thick] {cross=20pt};

\end{scope}%

\begin{scope}[shift={(18,0)}]%
\node() at (3,0) {$\Delta^{vw}_4$};
\cac[delta4-]{}

\node[fill=red!20,thick,draw=black,circle,minimum size=12, inner sep=2pt] (6) at (14,-6) {};
\node[fill=violet!20,thick,draw=black,circle,minimum size=12, inner sep=2pt] (7) at (3,-8) {};
\begin{scope}[-stealth,very thick]
\draw[red] (5) to (6);
\end{scope}
\begin{scope}[-stealth,dashed,very thick]%
\draw[red] (7) to (1);
\end{scope}

\path (12,-5.5) pic[rotate = 90,very thick] {cross=20pt};
\path (4.5,-7) pic[rotate = 15,very thick] {cross=20pt};

\end{scope}%

\end{scope}

\end{tikzpicture} \end{center}
\caption{Illustration of different types of events.}\label{fig:different_events}
\end{figure}
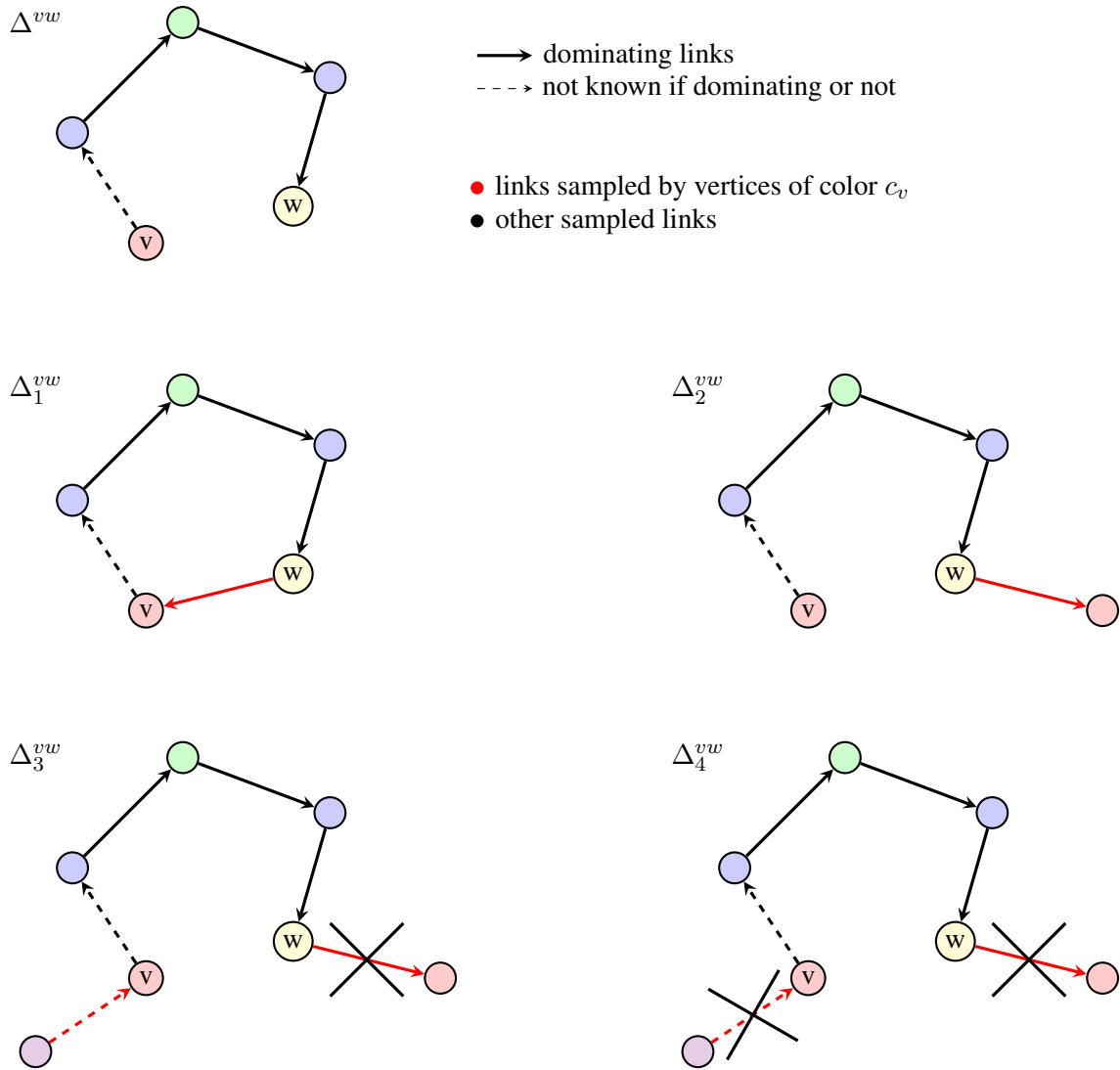

We consider the following four disjoint subevents of $\Delta^{vw}$:
\begin{itemize}
    \item $\Delta^{vw}_1$ is the subevent of $\Delta^{vw}$ in which a dominating link from $w$ to $v$ is sampled, i.e., such a link is contained in $F$;
    \item $\Delta^{vw}_2$ is the subevent of $\Delta^{vw}$ in which a dominating link from $w$ to a vertex $v'\neq v$ sampled and no dominating link from $w$ to $v$ is sampled. Notice that we must have $c_{v'}=c_v$ in this case;
    \item $\Delta^{vw}_3$ is the subevent of $\Delta^{vw}$ in which no dominating link in $F$ leaves $w$, but there exists a link in $F$ entering $v$;
    \item $\Delta^{vw}_4$ is the subevent of $\Delta^{vw}$ in which no dominating link in $F$ leaves $w$, and no link in $F$ enters $v$.
\end{itemize}
See Figure~\ref{fig:different_events}.
In the event $\Delta^{vw}_1$, the vertex $v$ is contained in an odd cycle component and $v$ is the only vertex of its color in this cycle.
In the event $\Delta^{vw}_2$, the vertex $v$ might be contained in an odd cycle component, but this does not have to be the case. If it is, then $v$ is not the only vertex of its color in this cycle.
The other two events are good events: in these cases $v$ is not contained in an odd cycle component.
We now prove bounds on the probabilities of the four events, conditioned on $\Delta^{vw}$.
For the events $\Delta^{vw}_1$ and $\Delta^{vw}_2$, where $v$ might be part of an odd cycle, we will give upper bounds on their respective probabilities. For the other two events we will give lower bounds.
Given $v,w \in T$, we define 
\[ 
    s_{vw}:= \sum_{\ell\in S_v \cap S_w} x_{\ell}\enspace.
\]
Moreover, for a terminal $v\in T$, we denote its color by $c_v$. Then
\begin{equation*}
 \eta_w^{c_v}= \sum_{\substack{u\in T:\\ c_u=c_v}} \sum_{\ell\in S_u\cap S_w} x_{\ell}\enspace.
\end{equation*}

\begin{lemma}\label{lem:bound_conditional_prob}
Let $v,w\in T$ such that $\Delta^{vw}$ is nonempty. Then
\begin{itemize}
\item $\mathbb{P}[\Delta_1^{vw} | \Delta^{vw} ] \le s_{vw}$\enspace.
\item $\mathbb{P}[\Delta_2^{vw} | \Delta^{vw} ] \le \eta_w^{c_v} - s_{vw}$\enspace.
\item $\mathbb{P}[\Delta_3^{vw} | \Delta^{vw} ] \ge x(S_v) - \eta_w^{c_v}$\enspace.
\item $\mathbb{P}[\Delta_4^{vw} | \Delta^{vw} ] \ge 1 - x(S_v) -\eta_w^{c_v} +s_{vw}$\enspace.
\end{itemize}
\end{lemma}
\begin{proof}
We observe that the event $\Delta^{vw}$ depends only on the links sampled by terminals that have a different color than $v$. 
If we are in the event $\Delta^{vw}$, every dominating link that leaves $w$ is a link sampled by a terminal with color $c_v$ (and every link entering $v$ is sampled by $v$).
Recall that the sampling of the links by terminals of color $c_v$ is independent of the sampling of links by terminals of other colors.
Thus, the first two inequalities follow directly from the definitions of $s_{vw}$ and $\eta^{c_v}_w$ because every terminal samples an incoming link $\ell$ 
from its stack with probability $x_{\ell}$.

Moreover, we have
\begin{align*}
    \mathbb{P}[\Delta_3^{vw} | \Delta^{vw} ] \ge&\   \mathbb{P}[F\text{ contains a link entering }v\ | \Delta^{vw}] \\
    &\ - \mathbb{P}[F\text{ contains a dominating link leaving }w\ | \Delta^{vw}] \\
    \ge&\ x(S_v) - \eta_w^{c_v}\enspace.
\intertext{and}
    \mathbb{P}[\Delta_4^{vw} | \Delta^{vw} ] \ge&\ \mathbb{P}[F\text{ contains no link entering }v\ | \Delta^{vw}] \\ &\ - \mathbb{P}[F\text{ contains a link leaving }w\ | \Delta^{vw}] \\ &\ + \mathbb{P}[F\text{ contains a link leaving }w\text{ and entering }v\ | \Delta^{vw}]\\
    \ge&\ (1 - x(S_v)) -\eta_w^{c_v} +s_{vw}\enspace.\qedhere
\end{align*}
\end{proof}

\subsection{Bounding the expected size of a removable link set}

We now prove a lower bound on the expected size of a maximum cardinality removable set of links.
To this end, we will define a random variable $X_v$ for every $v\in T$. 
We will show that there is a removable set of size at least $\sum_{v\in T} X_v$ and then compute a lower bound on the expected value of $\sum_{v\in T} X_v$.

Let
\begin{equation*}
    X_v \coloneqq b \cdot \mathbbm{1}[v\text{ has an entering dominated link in }F] + Z_v\enspace.
\end{equation*}
where the random variable $Z_v$ is defined as follows.
Recall that the events $\Delta^{vw}_i$ for $w\in T$ and $i\in \{1,\dots,4\}$ are disjoint.
If we are not in any of these events, we set $Z_v :=0$.
Otherwise, let $w\in T$ and $i\in\{1,\dots,4\}$ such that we are in the event $\Delta_i^{vw}$. Then we define
\begin{equation}
    Z_v := \begin{cases}
             z_1 \coloneqq  -\left(b-\frac{1}{3}\right)  & \text{ if }i=1 \\
             z_2 \coloneqq -2\left(b-\frac{2}{5}\right)+\frac{1}{30} & \text{ if }i=2 \\
             z_3 \coloneqq \frac{1}{2}-b  & \text{ if }i=3 \\
             z_4 \coloneqq 1-b   & \text{ if }i=4\enspace. \\
           \end{cases}
\end{equation}
Notice that $z_1 \le z_2 \le 0$ and $z_4 \ge z_3 \ge 0$.

In order to prove that there is a removable set $R\subseteq F$ with $|R| \ge \sum_{v\in T} X_v$, we need the following lemma.
\begin{lemma}\label{lem:structure_of_components}
Let $F'\in \Fscr$ be a component of $F$ and let $(T',F')$ be the unique non-trivial connected component of $(T,F')$.
Then $F'$ contains at most one link that is not dominating.
If $F'$ contains such a nondominating link $\ell$, then either
\begin{itemize}
\item  $(T',F')$ is an arborescence with root $r$ and $\delta^+(r)=\{\ell\}$,  or
\item  $(T',F'\setminus\{\ell\})$ is an arborescence.
\end{itemize}
\end{lemma}
\begin{proof}
Recall that $|\delta_F^-(v)| \le 1$ for all $v\in T$.
Because $F'$ is a component of $F$, either $(T',F')$ is an arborescence or it contains precisely one cycle.

If $(T',F')$ is an arborescence, its root is the only vertex that has an outgoing link which is not dominating.
This follows from the minimality of the component $F'$.
This minimality also implies that the root has only one outgoing link in $F'$ in this case.

Now suppose $(T',F')$ contains a cycle. 
Then, by the minimality of the component $F'$, all links that are not contained in this cycle must be dominating.
Moreover, the cycle can contain at most one nondominating link.
If such a link $\ell$ exists, $(T',F'\setminus\{\ell\})$ is an arborescence because $|\delta_F^-(v)| \le 1$ for all $v\in T$.
\end{proof}

\begin{lemma}\label{lem:deletion_componentwise}
Let $F'\in \Fscr$ be a component of $F$ and let $(T',F')$ be the unique non-trivial connected component of $(T,F')$.
Then there is a removable set $R'\subseteq F'$ with $|R'| \ge \sum_{v\in T'} \max\{ X_v, 0\}$.
\end{lemma}
\begin{proof}
We distinguish between two cases.
\smallskip

\noindent\textbf{Case 1:} $(T',F')$ is a maximally dominating odd cycle. \\[2mm]
Note that for every vertex $v\in T'$, the path $P_v$ is contained
in $(T',F')$ and thus the endpoint of $P_v$ has an outgoing dominating link.
Therefore, the events $\Delta^{vw}_3$ and $\Delta^{vw}_4$ cannot happen for any $v\in T'$.
Hence, $Z_v \le 0$.
Because every link in $F'$ is dominated, we have $X_v =b + Z_v \ge 0$ for all $v\in T'$.
By Lemma~\ref{lem:num_deletable_for_components}, there is a removable set containing at least $\frac{1}{2}(|F'|-1)$ many links.

If the odd cycle $(T',F')$ has length at least $\frac{1}{1-2b}$, we have
\[
\sum_{v\in T'} X_v\ \le\ b \cdot |F'|\ =\ \frac{1}{2} |F'| - \left(\frac{1}{2} - b\right)\cdot |F'|\ \le\ \frac{1}{2}\left(|F'|-1\right)\enspace.
\]
Hence, we may assume $|F'| \le \frac{1}{1-2b}$ now.

Because $(T',F')$ is an odd cycle, there are at least three different colors appearing as a color of a vertex in $T'$.
Now consider such a color $c$.
By Lemma~\ref{lem:odd_cycles_imply_delta_events}, there are vertices $v,w\in T'$ such that
$v$ has color $c$ and $F\in \Delta^{vw}$.
Because $(F',T')$ is a maximally dominating cycle, the vertex $w$ has an outgoing dominating link $(w,v')$.
By the definition of the path $P_v$, the vertex $v'$ must have color $c$.
(Otherwise we would have extended $P_v$ by the edge $(w,v')$.)
Therefore, $F\in \Delta^{vw}_1 \cup \Delta^{vw}_2$.
Moreover, we have $F\in \Delta^{vw}_1$ if and only if $v$ is the only vertex of color $c$ in $T'$.

Let $2\cdot l+1 = |F'|= |T'|$ be the length of the odd cycle $(T',F')$. 
We show that $\sum_{v\in T'} X_v \le l = \frac{1}{2}(|F'|-1)$.
We distinguish the cases $l=1$, $l=2$, and $l\ge 3$.

 First, suppose $l=1$. 
    Then $|T'|=3$ and every vertex in $T'$ has a different color. Thus, $F\in\Delta_1^{vw}$ for every vertex $v\in T'$ (for some $w$) and
    \begin{equation*}
        \sum_{v\in T'} X_v = 3 \cdot b + \sum_{v\in T'} Z_v = 3b - 3 \left(b-\frac{1}{3}\right) = 1 = l\enspace.
    \end{equation*}

Now consider the case $l=2$.
    Then $|T'|=5$ and $F\in \Delta_1^{vw} \cup \Delta_2^{vw}$ for at least three vertices $v\in T'$ (for some $w$). 
    Moreover, for at least one vertex $v$ we have $F\in \Delta_1^{vw}$ for some $v\in T'$ (because there are at least three colors and hence one vertex must be the only one of its color).
    Thus,
    \begin{align*}
        \sum_{v\in T'} X_v &= 5 b + \sum_{v\in T'} Z_v \\ &\le 5b - \left(b-\frac{1}{3}\right) +2 \cdot \left(- 2\left(b-\frac{2}{5}\right)+\frac{1}{30} \right) \\
        &= 5b - b + \frac{1}{3} - 4b + \frac{8}{5}+ \frac{1}{15} = \frac{30}{15} = 2 = l\enspace.
    \end{align*}

 Finally, we consider the remaining case $l\ge 3$.
   Again, $F\in \Delta_1^{vw} \cup \Delta_2^{vw}$ for at least three vertices $v\in T'$ (for some $w$). Hence,
    \begin{align*}
        \sum_{v\in T'} X_v &= (2l+1) \cdot b + \sum_{v\in T'} Z_v \\
        &\le (2l+1)\cdot b + 3 \cdot \left( -2\left(b-\frac{2}{5}\right)+\frac{1}{30} \right) \\
        &= (2l+1)\cdot b - 6b + \frac{12}{5} + \frac{1}{10} \\
        &= (2l-5)\cdot b + \frac{5}{2} \\
        &= 2lb +\frac{5}{2} (1-2b ) \\
        &< 2lb + l \cdot (1-2b)\ =\ l\enspace.
    \end{align*}
Actually, for our choice of $b=0.42$, the case $l\geq 3$ cannot happen because we would have $|F'|\leq \frac{1}{1-2b} = 6.25 < 2\cdot l + 1 = |F'|$ for any $l\geq 3$. Nevertheless, we also included this case because it may be useful for potential future improvements of the factor $b$.
\smallskip

\noindent\textbf{Case 2:} $(T,F')$ is not a maximally dominating odd cycle.\\[2mm]
First, we observe that for every vertex $v$ with an incoming dominated link, 
we have $F\notin \Delta^{vw}_4$ and hence $X_v = b + Z_v  \le b + \max\{z_1,z_2,z_3\} = \frac{1}{2}$.

If all links in $F'$ are dominating, then every vertex that does not have an incoming dominated link has out-degree zero in $F'$.
This implies that  for every such vertex $v$, the path $P_v$ has length zero and hence $X_v=Z_v=0$.
Therefore, in this case $\sum_{v\in T'} \max\{X_v,0\}$ is at most half of the number of dominated links.
Lemma~\ref{lem:num_deletable_for_components} implies that there indeed exists a removable set of this cardinality.

Thus, we may assume that there is a link $(t,s)\in F'$ that is not dominating.
By Lemma~\ref{lem:structure_of_components}, this is the only nondominating link in $F'$.
Hence, every vertex $v \ne t$ that does not have an incoming dominated link, has out-degree zero, implying $X_v=Z_v=0$.
If $t$ has an incoming dominated link or if $X_t = Z_t = 0$,  we again have that
$\sum_{v\in T'} \max\{X_v,0\}$ is at most half of the number of dominated links and by
Lemma~\ref{lem:num_deletable_for_components} there exists a removable set of this cardinality.

Hence, it remains to consider the case where $t$ has no incoming dominated link and $X_t = Z_t \ne 0$.
By Lemma~\ref{lem:structure_of_components}, $(T',F'\setminus \{(t,s)\})$ does not contain any cycle.
See Figure~\ref{fig:proof_vertex_construction}.
Moreover, Lemma~\ref{lem:structure_of_components} implies that either $(T',F')$ is an arborescence with root $t$ and $(t,s)$ is the only outgoing link of $t$, or
$t$ has an incoming link. Also in the latter case, $(t,s)$ must be the only outgoing link of $t$ in $F'$
because any link in $\delta_{F'}^+(t)\setminus\{(t,s)\}$ would dominate the incoming link of $t$.
(Here, we used that all links in $F' \setminus \{(t,s)\}$ are dominating.)

We claim that there is a vertex $u\in T'$ such that $u$ has an incoming dominated link and $Z_u=0$.
Because $Z_t \ne 0$, the path $P_t$ has length at least two.
The first arc of this path must be the link $(t,s)$.
Let $P$ be a maximum length directed path in $(T',F' \setminus \{(t,s)\})$ that starts in $s$ and ends in some vertex $v$.
Because $P_t$ has length at least two, the path $P$ has length at least one.
Let $(u,v)$ be the last arc of $P$.
Then $(u,v) \ne (t,s)$ and hence $u\ne t$.

\begin{figure}[!ht]
\begin{center}
\begin{tikzpicture}[scale=0.5]

\pgfdeclarelayer{bg}
\pgfdeclarelayer{fg}
\pgfsetlayers{bg,main,fg}

\tikzset{
  prefix node name/.style={%
    /tikz/name/.append style={%
      /tikz/alias={#1##1}%
    }
  }
}

\tikzset{ 
    cross/.pic = {
    \draw[rotate = 45] (-#1,0) -- (#1,0);
    \draw[rotate = 45] (0,-#1) -- (0, #1);
    }
}

\tikzset{root/.style={fill=white,minimum size=13}}

\tikzset{s/.style={dashed,red,line width=1pt}}

\tikzset{
term/.style={fill=black!20, rectangle, minimum size=10},
tf/.append style={font=\scriptsize\color{black}},
}

\newcommand\blob[2]{
\pgfmathanglebetweenpoints%
{\pgfpointanchor{#1}{center}}%
{\pgfpointanchor{#2}{center}}

\edef\angle{\pgfmathresult}

\filldraw ($(#1)+(\angle+90:\d)$) arc (\angle+90:\angle+270:\d) --
($(#2)+(\angle+270:\d)$) arc (\angle+270:\angle+450:\d) -- cycle;
}

\newcommand\gra[2][]{

\begin{scope}

\begin{scope}[every node/.append style={thick,draw=black,fill=white,circle,minimum size=12, inner sep=2pt}]

\node (1) at (1,1) {t};
\node (2) at (4,-1) {s};
\node (3) at (7.5,0) {};
\node (4) at (9.5,2) {};
\node (5) at (11,0.75) {};
\node (6) at (10.5,-1) {};
\node (7) at (6,-4) {};
\node (8) at (9,-3) {u};
\node (9) at (12,-4) {v};
\node (10) at (3,-5) {};
\node (11) at (5.5,-7) {};
\node (12) at (8.5,-6) {};

\end{scope}%

\begin{scope}[-stealth,very thick]

\draw[dotted,black!75] (1) to (2);
\draw (2) to (3);
\draw (3) to (4);
\draw (3) to (5);
\draw (3) to (6);
\draw[red] (2) to (7);
\draw[red] (7) to (8);
\draw[red] (8) to (9);
\draw (7) to (10);
\draw (7) to (11);
\draw (7) to (12);

\end{scope}%

\end{scope}%

}%

\newcommand\cac[2][]{

\begin{scope}

\begin{scope}[every node/.append style={thick,draw=black,fill=white,circle,minimum size=12, inner sep=2pt}]

\node (1) at (1,-6) {t};
\node (2) at (4,-8) {s};
\node (3) at (8,-7) {};
\node (4) at (7,-3) {};
\node (5) at (3,-2) {};
\node (6) at (3,1) {u};
\node (7) at (6.5,1.5) {v};
\node (8) at (8,0) {};

\end{scope}%

\begin{scope}[-stealth,very thick]

\draw[dotted,black!75] (1) to (2);
\draw[red] (2) to (3);
\draw[red] (3) to (4);
\draw[red] (4) to (5);
\draw (5) to (1);
\draw[red] (5) to (6);
\draw[red] (6) to (7);
\draw (4) to (8);

\end{scope}%

\end{scope}%

}%

\begin{scope}

\begin{scope}[shift={(12,0)},scale=0.8]
\cac[]{}
\end{scope}

\begin{scope}[shift={(0,0)},scale=0.8]
\gra[]{}
\end{scope}

\begin{scope}[shift={(22,0)}]%
\def\ll{30mm} %
\def\vs{18mm} %

\begin{scope}[scale=0.5]
\draw[-stealth,black, very thick] (0,0) -- +(\ll,0) node[right] {dominating links};
\draw[-stealth,black, very thick, dotted, yshift=-\vs] (0,0) -- +(\ll,0) node[right] {nondominating links}[black!75];
\end{scope}

\begin{scope}[scale=0.5,yshift=-4*\vs]

\begin{scope}[every node/.style={circle,inner sep=0pt,minimum size=5pt}]
\node[fill=red] (rd) at (0,0) {};

\node at (rd.east)[right=4pt] {links in path $P$};
\end{scope}

\end{scope}

\end{scope}%

\end{scope}

\end{tikzpicture} \end{center}
\caption{Illustration of the proof of case~2 of Lemma \ref{lem:deletion_componentwise}.}\label{fig:proof_vertex_construction}
\end{figure}
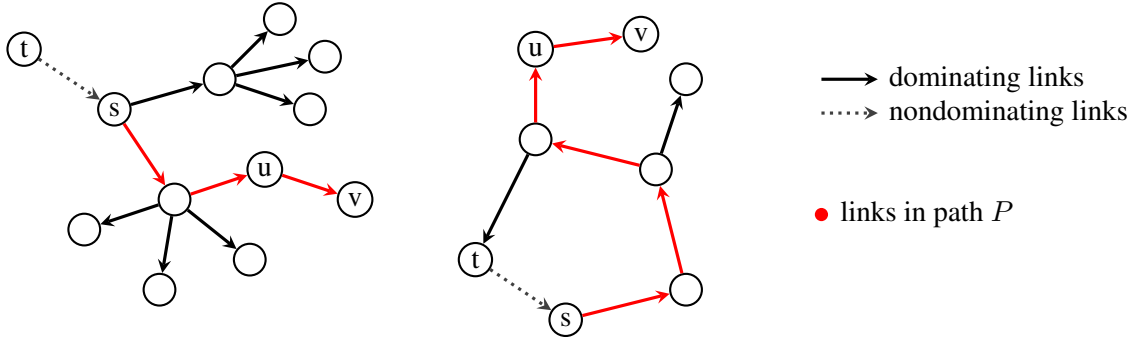

We consider the path $P_u$. By the definition of $P_u$, all its arcs except for possibly the first one are dominating links.
Because $u\ne t$ and because $(t,s)$ is not dominating, the path $P_u$ is a path in $(T',F' \setminus\{(t,s)\})$.
Suppose $Z_u \ne 0$. Then $P_u$ has length at least two and hence extending the $s$-$u$ subpath of $P$ by the path $P_u$ yields a directed walk in $(T',F' \setminus \{(t,s)\})$
that is strictly longer than $P$.
Because $(T',F'\setminus  \{(t,s)\})$ does not contain a cycle, this walk must be a path, contradicting the choice of $P$.
Hence $Z_u = 0$. 
Moreover, because all links in $F'\setminus \{(t,s)\}$ are dominating, the link $(u,v)$ is dominating.
We conclude that
\begin{itemize}
\item $u$ has an incoming dominated link and hence $X_u = b + Z_u = b$;
\item $X_t = Z_t$;
\item $X_v \le \frac{1}{2}$ for all vertices $v\ne u$ with an incoming dominated link; and
\item $X_v = 0$ for all other vertices $v$.
\end{itemize}
Therefore, if $Z_t \le \frac{1}{2}-b$, the sum $\sum_{v\in T'} \max\{X_v, 0\}$ is at most half of the number of dominated links.
By Lemma~\ref{lem:num_deletable_for_components} there exists a removable set of links of this cardinality.

The only remaining case is when $Z_t=1-b$. Then 
\[
\sum_{v\in T'} \max\{0,X_v\}\le \tfrac{1}{2}\cdot|\{\ell\in F': \ell\text{ dominated}\}| + \tfrac{1}{2}
\] 
 and we have $F \in \Delta^{tw}_4$ for some $w\in T'$.
By the definition of the event $\Delta^{tw}_4$, the vertex $t$ has in-degree zero and hence $(T',F')$ is a arborescence. Moreover, there is an even length path from $t$ to a vertex $w$ that does not have an outgoing dominating link.
Then either $w$ does not have any outgoing link at all or this link is contained in a different component, implying that $w$ is a leaf of $(T',F')$.
We conclude that $(T',F')$ is an arborescence containing an even length path from its root $t$ to the leaf $w$.
By Lemma~\ref{lem:num_deletable_for_components}, there is a removable set containing at least $ \frac{1}{2}\cdot|\{\ell\in F': \ell\text{ dominated}\}| + \frac{1}{2}$  many links.
\end{proof}

\begin{lemma}
There exists a removable subset of $F$ that has cardinality at least $\sum_{v\in T} X_v$.
\end{lemma}
\begin{proof}
Recall that $\Fscr$ is a partition of $F$.
By the definition of the components of $F$, the union of removable sets $R'\subseteq F'$ for all components $F'\in \Fscr$ is removable. 
For every $F'\in \Fscr$ we denote by  $(T_{F'},F')$ the unique non-trivial connected component of $(T,F')$.
We apply Lemma~\ref{lem:deletion_componentwise} to every component $F'\in \Fscr$ and conclude that we can delete at least
\[\sum_{F'\in \Fscr} \sum_{v\in T_{F'}} \max\{ X_v, 0\}\] links.
Every vertex $v\in T$ with $X_v > 0$ has an incoming dominated link or is the starting point of a postive length path $P_v$ in $(T,F)$.
Therefore, every vertex $v\in T$ with $X_v >0$ is contained in $T_{F'}$ for some $F'\in \Fscr$.
Hence, we have
\[ \sum_{F'\in \Fscr} \sum_{v\in T_{F'}} \max\{ X_v, 0\}\ \ge\ \sum_{v\in T}  \max\{ X_v, 0\} \ \ge\  \sum_{v\in T}  X_v\enspace. \]
\end{proof}

In the following we give a lower bound on $\sum_{v\in T} \mathbb{E}[X_v]$. 
 We need the following observation.
\begin{lemma}\label{lem:amount_of_receivers_of_profit}
For every color $c$ and every vertex $w$
\begin{equation*}
 \sum_{v\text{ with color }c} \mathbb{P}[\Delta^{vw}] \le x(S_w) -\eta_w^c\enspace.
\end{equation*}
\end{lemma}
\begin{proof}
We fix a vertex $w$ and a color $c$.
Recall that for every vertex $v$ of color $c$, the path $P_v$ is a path in the graph
$(T, F \setminus F_{c})$, where $F_{c} := \delta_F^-(\{u\in T: u\text{ has color }c\})$ is the set of links entering a vertex of color $c$.
In this graph, every vertex has in-degree at most one and every vertex of color $c$ has in-degree zero.
Therefore, there can be at most one vertex $v$ of color $c$ such that $(T, F \setminus F_{c})$ contains a $v$-$w$ path.
This implies that the events $\Delta^{vw}$ for different vertices $v$ of color $c$ are disjoint.
We conclude $\sum_{v\text{ with color }c} \mathbb{P}[\Delta^{vw}] = \mathbb{P}[\,\bigcup_{v\text{ with color }c} \Delta^{vw}] $.

To prove that this probability is at most $x(S_w) -\eta_w^c$, we observe that in any event $\Delta^{vw}$ the vertex $w$ has an incoming link $(u,w)$ in $F$. 
This happens with probability $x(S_w)$. 
With probability $\eta^c_w$, the vertex $w$ has an incoming link $(u,w)$ where $u$ has color~$c$.
In this case, the graph $(T, F \setminus F_{c})$ does not contains any path of length at least two that ends in $w$.
This implies that in this case the event $\Delta^{vw}$ cannot happen for any vertex $v$ of color~$c$.
\end{proof}

We will fix functions $g\colon [0,1]^{4} \to \mathbb{R}$ and $\mathrm{gain}\colon [0,1]^{5} \to \mathbb{R}$ such that for all $v\in T$:
\begin{align}\label{eq:distributing_eta_profit}
\mathbb{P}[v\text{ has an entering dominated link in }F] \ge g(\lambda_v^0, \mu^0_v, \lambda_v^1, \mu^1_v) 
+ \sum_{c=1}^q \mathrm{gain}(\lambda_v^0, \mu^0_v, \lambda_v^1, \mu^1_v, \eta^c_v)\enspace.
\end{align}
Both $g(\lambda_v^0, \mu^0_v, \lambda_v^1, \mu^1_v) $ and $ \mathrm{gain}(\lambda_v^0, \mu^0_v, \lambda_v^1, \mu^1_v, \eta^c_v)$ will be non-negative for every terminal $v\in T$.
(We fix the precise definition of these functions and prove the above inequality in Section~\ref{sec:compute_domination_prob}.)
The function $g$ gives a lower bound on the probability of an incoming dominated link in $S_v$ and this lower bound is independent of the values of $\eta^c_v$.
The function $\mathrm{gain}$ reflects the increase of this lower bound when $\eta^c_v$ is large.
We will prove that the expected size of the maximum removable set $R\subseteq F$ is at least $ b \cdot \sum_{v\in T} g(\lambda_v^0, \mu^0_v,\lambda_v^1, \mu^1_v)$.

Now we give a lower bound on $\sum_{v\in T} \mathbb{E}[X_v]$, and hence on the 
expected size of $R$.
Using \eqref{eq:distributing_eta_profit} and Lemma~\ref{lem:amount_of_receivers_of_profit}, we have
\begin{align*}
    \mathbb{E}\left[\sum_{v\in T} X_v\right] =&\ b \cdot \sum_{v\in T}\mathbb{P}[v\text{ has an entering dominated link in }F]
 + \sum_{v\in T}\sum_{w\in T} \mathbb{P}[\Delta^{vw}] \cdot \sum_{i=1}^4 \mathbb{P}[\Delta^{vw}_i | \Delta^{vw}] \cdot z_i
 \\[4mm]
\ge&\  b \cdot \sum_{v\in T} \left(g(\lambda_v^0, \mu^0_v, \lambda_v^1, \mu^1_v)
        + \sum_{c=1}^q \mathrm{gain}(\lambda_v^0, \mu^0_v, \lambda_v^1, \mu^1_v, \eta^c_v)\right)
 \\
 &\ + \sum_{v\in T}\sum_{w\in T} \mathbb{P}[\Delta^{vw}] \cdot \sum_{i=1}^4 \mathbb{P}[\Delta^{vw}_i | \Delta^{vw}] \cdot z_i 
 \\[4mm]
\ge&\  b \cdot \sum_{v\in T} g(\lambda_v^0, \mu^0_v, \lambda_v^1, \mu^1_v) \\
&\  +  \sum_{w\in T} \sum_{c=1}^q  b \cdot \mathrm{gain}(\lambda_w^0, \mu^0_w, \lambda_w^1, \mu^1_w, \eta^c_w) \cdot  \frac{1}{ x(S_w) -\eta_w^c} \cdot \sum_{v\text{ with color }c} \mathbb{P}[\Delta^{vw}]
 \\
 &\ + \sum_{v\in T}\sum_{w\in T} \mathbb{P}[\Delta^{vw}] \cdot \sum_{i=1}^4 \mathbb{P}[\Delta^{vw}_i | \Delta^{vw}] \cdot z_i 
\\[4mm]
\ge&\  b \cdot \sum_{v\in T} g(\lambda_v^0, \mu^0_v, \lambda_v^1, \mu^1_v) \\
&\ + \sum_{v\in T} \sum_{w\in T} \mathbb{P}[\Delta^{vw}]  \left(  \frac{b}{x(S_w) -\eta_w^{c_v}} \cdot  \mathrm{gain}(\lambda_w^0,\mu^0_w, \lambda_w^1, \mu^1_w, \eta^{c_v}_w)  + \sum_{i=1}^4 \mathbb{P}[\Delta^{vw}_i | \Delta^{vw}] \cdot z_i  \right)\enspace,
\end{align*}
where $c_v$ denotes the color of the vertex $v\in T$.

We will choose the function $\mathrm{gain}$ such that
\begin{equation}\label{eq:worst_case_no_short_odd_cycles}
  \frac{b}{x(S_w) -\eta_w^{c_v}} \cdot  \mathrm{gain}(\lambda_w^0,\mu^0_w, \lambda_w^1, \mu^1_w, \eta^{c_v}_w)  + \sum_{i=1}^4 \mathbb{P}[\Delta^{vw}_i | \Delta^{vw}] \cdot z_i  \ge 0
\end{equation}
for all $v,w\in T$.
Once we have defined the functions $g$ and $\mathrm{gain}$ and proved that they indeed fulfill \eqref{eq:distributing_eta_profit} and \eqref{eq:worst_case_no_short_odd_cycles},
we can conclude that the expected size of the maximum removable set $R\subseteq F$ is 
\begin{equation*}
\mathbb{E}[|R|]\ \ge\ b \cdot \sum_{v\in T} g(\lambda_v^0, \mu^0_v, \lambda_v^1, \mu^1_v)\enspace,
\end{equation*}
completing the proof of Lemma~\ref{lem:bound_stack_problem} and thus of Lemma~\ref{lem:stack_algo}.

\subsection{Bounding the expected number of dominated links}\label{sec:compute_domination_prob}

We now define the functions $g$ and  $\mathrm{gain}$.
We explain the arguments leading to this particular choice below, in the proof of  Lemma~\ref{lem:bound_probability_using_gain}.
We define
\begin{equation}\label{eq:defg}
\begin{aligned}
g(\lambda^0, \mu^0, \lambda^1, \mu^1) \coloneqq&\  e^{-\lambda^0 - \mu^0 -\lambda^1 - \mu^1} \cdot
                                                               \Big( 1+ e^{\mu^1 + \lambda^1} - e^{\mu^1 + \lambda^1 + \mu^0} \cdot (1+\lambda^0) \\
                                                                   &\        - e^{\mu^1} \cdot (1 + \lambda^1) 
                                                                      + e^{\lambda^0 +\mu^0 +\lambda^1 + \mu^1} \cdot (\lambda^0 + \mu^0 +\lambda^1 + \mu^1)\Big)\enspace.
\end{aligned}
\end{equation}
Moreover, we set
\begin{equation*}
 \mathrm{gain}(\lambda^0, \mu^0, \lambda^1, \mu^1, \eta) \coloneqq h_1(\lambda^0, \mu^0, \lambda^1, \mu^1, \eta) \cdot h_2(\lambda^0, \mu^0, \lambda^1, \mu^1, \eta)\enspace,
\end{equation*}
where
\begin{align*}
 h_1(\lambda^0, \mu^0, \lambda^1, \mu^1, \eta) \coloneqq&
 \begin{cases}
   e^{- (\lambda^0 + \mu^0 + \lambda^1 + \mu^1) + \eta} &\text{ if }\eta > \frac{1}{2}\left(\lambda^0 + \mu^0 + \lambda^1 + \mu^1\right) \\ 
   1- (\lambda^0 + \mu^0 + \lambda^1 + \mu^1) + \eta & \text{ otherwise}
\end{cases} 
\\[3mm]
 h_2(\lambda^0, \mu^0, \lambda^1, \mu^1, \eta) \coloneqq&
 \begin{cases}
  h^a(\lambda^0, \mu^0, \lambda^1, \mu^1, \eta) & \text{ for } \eta \le \mu^1  \\
  h^b(\lambda^0, \mu^0, \lambda^1, \mu^1, \eta) &\text{ for }  \mu^1  < \eta \le \mu^1 + \lambda^1   \\
  h^c(\lambda^0, \mu^0, \lambda^1, \mu^1, \eta) & \text{ for }   \mu^1 + \lambda^1 < \eta \le \mu^1 + \lambda^1 + \mu^0  \\
  h^d(\lambda^0, \mu^0, \lambda^1, \mu^1, \eta) &  \text{ for }   \mu^1 + \lambda^1 + \mu^0  < \eta \\
\end{cases}
\end{align*}
and the functions $h^a$, $h^b$, $h^c$, and $h^d$ are defined as follows:
\begin{align*}
 h^a(\lambda^0, \mu^0, \lambda^1, \mu^1, \eta) \coloneqq\ &    1 - \eta + \frac{\eta \cdot \eta}{2} - e^{-\eta} \\
  h^b(\lambda^0, \mu^0, \lambda^1, \mu^1, \eta)\coloneqq\  &   - e^{-\eta} +  e^{\mu^1 -\eta} \cdot \left(1 +\lambda^1\right) - \mu^1 \cdot \left(1 - \eta + \tfrac{\mu^1}{2} \right)
                                                                                          -\lambda^1 \cdot \left(1 + \mu^1 - \eta \right)\\
  h^c(\lambda^0, \mu^0, \lambda^1, \mu^1, \eta)\coloneqq\  &  - e^{-\eta} - e^{\lambda^1 + \mu^1 - \eta} + e^{\mu^1 - \eta} \cdot \left(1 + \lambda^1\right)
                                                                                          + 1  + \frac{\lambda^1 \cdot \lambda^1}{2} - \eta + \frac{\eta \cdot \eta}{2}   \\
  h^d(\lambda^0, \mu^0, \lambda^1, \mu^1, \eta) \coloneqq\ &    - e^{-\eta} - e^{\lambda^1 + \mu^1 - \eta}    + e^{\mu^1 + \lambda^1 + \mu^0 - \eta} \cdot \left( 1 + \lambda^0\right)
                                                                                           + e^{\mu^1 - \eta} \cdot \left( 1 + \lambda^1\right) \\
                                                                                       &    - \left( \mu^0 + \mu^1\right) \cdot \left( 1 + \tfrac{\mu^0}{2} + \tfrac{\mu^1}{2} - \eta \right) 
                                                                                          - \lambda^1 \cdot \left( 1 + \mu^0 + \mu^1 - \eta\right) \\
                                                                                        &  - \lambda^0 \cdot \left( 1 + \mu^1 + \lambda^1 + \mu^0 - \eta\right)\enspace.
\end{align*}

Next, we show that with these definitions inequality~\eqref{eq:distributing_eta_profit} is indeed fulfilled.
\begin{lemma}\label{lem:bound_probability_using_gain}
For every terminal $t\in T$, we have
\begin{align*}
\mathbb{P}[t\text{ has an entering dominated link in }F] \ge g(\lambda_t^0, \mu^0_t, \lambda_t^1, \mu^1_t) 
+ \sum_{c=1}^q \mathrm{gain}(\lambda_t^0, \mu^0_t, \lambda_t^1, \mu^1_t, \eta^c_t)\enspace.
\end{align*}
Moreover, $g(\lambda_t^0, \mu^0_t, \lambda_t^1, \mu^1_t) $ and $ \mathrm{gain}(\lambda_t^0, \mu^0_t, \lambda_t^1, \mu^1_t, \eta^c_t)$ 
are non-negative for every terminal $t\in T$.
\end{lemma}
\begin{proof}
Consider a fixed terminal $t\in T$. 
For a link $\ell \in S_t$ we define $B(\ell) = \{ \bar \ell \in S_t : \bar \ell \preceq_t \ell \}$ to be the set of links that are below $\ell$ on the stack $S_t$.
Because the  outgoing links of $t$ in $D=(T,F)$ are sampled independently from the incoming link and because $t$ samples a link $\ell\in S_t$ with probability $x_{\ell}$,
we have
\begin{equation}\label{eq:simple_formula_domination_prob}
 \mathbb{P} [t\text{ has an incoming dominated link in }F] = \sum_{\ell \in S_t} x_{\ell} \cdot \mathbb{P}[t\text{ has an outgoing link in }F \cap B(\ell)]\enspace. 
\end{equation}
By properties~\ref{item:sampling_probability_link} and~\ref{item:correlation_of_colors},
 the probability that $F$ contains a link $(t,v)\in B(\ell)$ such that $v$ has color $c$,
is \[ \gamma^c(\ell)\ \coloneqq \sum_{\substack{\{t,v\}\in B(\ell):\\ v\text{ has color }c}} x_{\{t,v\}}\ \le \sum_{\substack{\{t,v\}\in S_t:\\ v\text{ has color }c}} x_{\{t,v\}} \ =\ \eta^c_t\enspace \]
for every color $c\in [q]$.
Because the sampling of outgoing links of $t$ that end in vertices of different colors is independent by~\ref{item:different_corlors_independent}, we have
\begin{align*}
\mathbb{P}[t\text{ has an outgoing link in }F \cap B(\ell)]\ =\ 1 - \prod_{c=1}^q \left( 1- \gamma^c(\ell) \right)  \enspace.
\end{align*}
Let us assume without loss of generality that $\eta^{1}_t \ge \eta^{2}_t \ge \dots \ge \eta^{q}_t$.
For $i\in \{0,1,\dots,q\}$, we define 
\[
p_t^i (\ell)\ :=\ 1 - \prod_{c=1}^i \left( 1- \gamma^{c}(\ell) \right) \cdot \prod_{c=i+1}^q  e^{- \gamma^{c}(\ell)}\enspace.
\]
Because  $e^{-\gamma^{c}} \ge 1 - \gamma^{c}$, we have
\[
p_t^0(\ell)\ \le\ p_t^1(\ell)\ \le\ \dots\ \le\ p_t^q(\ell)\ =\ \mathbb{P}[t\text{ has an outgoing link in }F \cap B(\ell)] \enspace.
\]
We will show that
\begin{equation}\label{eq:bound_to_show_g}
 \sum_{\ell \in S_t} x_{\ell} \cdot p_t^0(\ell)  \ \ge\  g(\lambda_t^0, \mu^0_t, \lambda_t^1, \mu^1_t) 
\end{equation}
and
\begin{equation}\label{eq:bound_to_show_profit}
  \sum_{\ell \in S_t} x_{\ell} \cdot  \left(  p_t^i(\ell) - p_t^{i-1}(\ell)   \right) \ \ge\ \mathrm{gain}(\lambda_t^0, \mu^0_t, \lambda_t^1, \mu^1_t, \eta^i_t)
\end{equation}
for all $i\in [q]$.
By \eqref{eq:simple_formula_domination_prob}, summing up these inequalities will complete the proof.

To prove \eqref{eq:bound_to_show_g} and \eqref{eq:bound_to_show_profit}, it will be useful to index the links in $S_t$ as follows.
First, we number the links in $S_t=\{\ell_1,\dots,\ell_l\}$ according to the order $\preceq_t$, i.e., such that we have
$\ell_1 \preceq_t \ell_2 \preceq_t \dots \preceq_t \ell_l$.
For a value $z$ in $[0,x(S_t)]$, we define $\ell_z$ to be the link $\ell_j$ for which $\sum_{i=1}^{j-1} x_{\ell_i} < z \le \sum_{i=1}^{j} x_{\ell_i}$.
By the properties of the partition $S_t = \Lambda^0_t \cup M^0_t \cup \Lambda^1_t \cup M^1_t$ in the definition of the stack problem, we have
\begin{equation}\label{eq:stack_partition}
\begin{aligned}
\ell_z \in&\ \Lambda^0_t & &\text{ if  }\ 0 \le z \le \lambda^0_t\enspace, \\
\ell_z \in&\ M^0_t & &\text{ if }\lambda^0_t < z \le \lambda^0_t + \mu^0_t\enspace, \\
\ell_z \in&\ \Lambda^1_t & &\text{ if }\lambda^0_t + \mu^0_t < z \le \lambda^0_t + \mu^0_t + \lambda^1_t\enspace,\text{ and }\\
\ell_z \in&\ M^1_t & &\text{  if }\lambda^0_t + \mu^0_t + \lambda^1_t < z \le \lambda^0_t + \mu^0_t + \lambda^1_t + \mu^1_t = x(S_t)\enspace.
\end{aligned}
\end{equation}
Moreover,
\begin{equation}\label{eq:size_below_sets}
x(B(\ell_z))\ \ge\ \beta(\ell_z) \coloneqq \begin{cases}
                              \lambda^0_t &\text{ if }\ell_z \in \Lambda^0_t \\
                              z &\text{ if } \ell_z \in M^0_t \\
                              \lambda^0_t + \mu^0_t + \lambda^1_t &\text{ if }\ell_z \in \Lambda^1_t \\
                              z &\text{ if } \ell_z \in M^1_t 
                           \end{cases}
\end{equation}
for all $z\in [0, x(S_t)]$.

Now we prove \eqref{eq:bound_to_show_g}. Because $x(B(\ell)) = \sum_{c=1}^q \gamma^c(\ell)$ for all $\ell \in S_t$, we have
\begin{align*}
 \sum_{\ell \in S_t} x_{\ell} \cdot p_t^0(\ell)\ =\ \int_{0}^{x(S_t)} p_t^0(\ell_z)\ dz\
                                                                    =\ \int_{0}^{x(S_t)} 1- \prod_{c=1}^q  e^{- \gamma^{c}(\ell_z)}\ dz \
                                                                   =\ \int_{0}^{x(S_t)} 1 -  e^{-x(B(\ell_z))}\ dz\enspace.
\end{align*}
Using \eqref{eq:stack_partition} and \eqref{eq:size_below_sets}, this implies
\begin{align*}
 \sum_{\ell \in S_t} x_{\ell} \cdot p_t^0(\ell)\ \ge&\ 
         \int_{0}^{\lambda^0_t}  1 - e^{-\lambda^0_t}\ dz 
         + \int_{\lambda^0_t}^{\lambda^0_t + \mu^0_t}  1 - e^{-z}\ dz  \\
         &+ \int_{\lambda^0_t + \mu^0_t}^{\lambda^0_t + \mu^0_t + \lambda^1_t}  1 - e^{-(\lambda^0_t + \mu^0_t + \lambda^1_t)}\ dz 
         + \int_{\lambda^0_t + \mu^0_t + \lambda^1_t}^{x(S_t)}  1 - e^{-z}\ dz \\[4mm]
     =&\ g(\lambda_t^0, \mu^0_t, \lambda_t^1, \mu^1_t) \enspace .
\end{align*}
This completes the proof of \eqref{eq:bound_to_show_g}. 
Moreover, the above different way of  expressing $g(\lambda_t^0, \mu^0_t, \lambda_t^1, \mu^1_t)$ readily implies that it is non-negative.
It remains to prove \eqref{eq:bound_to_show_profit} and prove non-negativity of $\mathrm{gain}(\lambda_t^0, \mu^0_t, \lambda_t^1, \mu^1_t, \eta^i_t)$ for all $t\in T$ and $i\in [q]$.

In order to prove \eqref{eq:bound_to_show_profit}, we first consider the case where $\eta^i_t > \frac{1}{2} \cdot x(S_t)$. 
Then we have $i=1$ because $\eta^{1}_t \ge \eta^{2}_t \ge \dots \ge \eta^{q}_t$ and $\sum_{c=1}^q \eta^c_t = x(S_t)$.
Therefore,
\begin{align*}
   p_t^1(\ell) - p_t^{0}(\ell)  \ =&\  \left(  e^{-\gamma^1(\ell)} - 1+ \gamma^{1}(\ell) \right) \cdot \prod_{c=2}^q  e^{- \gamma^{c}(\ell)} \\
 \ =&\  \left( e^{-\gamma^1(\ell)} -1 + \gamma^{1}(\ell) \right) \cdot   e^{- \sum_{c=2}^q \gamma^{c}(\ell)} \\
 =&\ \left(e^{-\gamma^1(\ell)} - 1 + \gamma^{1}(\ell) \right) \cdot   e^{- x(B(\ell)) + \gamma^1(\ell)}\enspace.
\end{align*}
Because $\gamma^1(\ell) \ge 0$ and the right-hand side is monotonically increasing in $\gamma^1(\ell)$ for $\gamma^1(\ell)\ge 0$, it is minimized for $\gamma^1(\ell) \ge 0$ being as small as possible.
In particular, this implies $p_t^1(\ell) - p_t^{0}(\ell) \ge 0$.
Together with 
\[
\gamma^1(\ell)\ =\  x(B(\ell)) - \sum_{c=2}^q \gamma^{c}(\ell) \ \ge\   x(B(\ell)) - \sum_{c=2}^q \eta^c_t \ =\ x(B(\ell)) - (x(S_t) - \eta^1_t)\enspace,
\]
we thus obtain
\begin{align*}
   p_t^1(\ell) - p_t^{0}(\ell)  \ \ge&\ \left(e^{-x(B(\ell))+x(S_t)- \eta^1_t } - 1 + x(B(\ell)) - x(S_t) + \eta^1_t \right)\cdot \mathds{1}\left\{ x(B(\ell)) \ge x(S_t) - \eta^1_t\right\} \\
   & \ \ \cdot   e^{- (x(S_t) - \eta^1_t)} \\
             \ =&\ \left(e^{-x(B(\ell))+x(S_t)- \eta^1_t } - 1 + x(B(\ell)) - x(S_t) + \eta^1_t \right)\cdot \mathds{1}\left\{ x(B(\ell)) \ge x(S_t) - \eta^1_t\right\} \\
             & \ \  \cdot h_1(\lambda_t^0, \mu^0_t, \lambda_t^1, \mu^1_t, \eta^1_t) \enspace,
\end{align*}
where, given $x_1,x_2 \in \mathbb{R}$, the indicator function $\mathds{1}$ is defined as follows.
\begin{equation*}
    \mathds{1}\{x_1 \ge x_2\} = \begin{cases}
        1 \quad \text{ if } x_1 \ge x_2 \\
        0 \quad \text{ otherwise} \enspace .
    \end{cases}
\end{equation*}
Before we show that this implies \eqref{eq:bound_to_show_profit},
we prove that the above lower bound on $ p_t^i(\ell) - p_t^{i-1}(\ell) $ holds also in the case $\eta^i_t \leq \frac{1}{2} \cdot x(S_t)$.
Using $e^{-\gamma^c } \ge 1 - \gamma^c$, we have
\begin{align*}
   p_t^i(\ell) - p_t^{i-1}(\ell)  \ =&\   \left(e^{- \gamma^{i}(\ell)}  - 1 + \gamma^{i}(\ell) \right) \cdot \prod_{c=1}^{i-1} \left( 1- \gamma^{c}(\ell) \right) \cdot \prod_{c=i+1}^q  e^{- \gamma^{c}(\ell)} \\
\ge&\ \left(e^{- \gamma^{i}(\ell)}  - 1 + \gamma^{i}(\ell) \right)  \cdot \prod_{c\in \{1,\dots,q\}\setminus\{i\}} (1-\gamma^{c}(\ell)) \\
\ge&\ \left(e^{- \gamma^{i}(\ell)}  - 1 + \gamma^{i}(\ell) \right)  \cdot \Big(1-\sum_{c\in \{1,\dots,q\}\setminus\{i\}} \gamma^{c}(\ell)\Big) \\
=&\ \left(e^{- \gamma^{i}(\ell)}  - 1 + \gamma^{i}(\ell) \right)  \cdot \left(1-x(B(\ell)) +  \gamma^{i}(\ell) \right) \enspace.
\end{align*}
Because  $\gamma^i(\ell) \ge 0$ and the right-hand side is monotonically increasing in $\gamma^i(\ell)$ for  $\gamma^i(\ell) \ge 0$, it is minimized for $\gamma^i(\ell)\ge 0$ being as small as possible.
In particular, this implies $p_t^i(\ell) - p_t^{i-1}(\ell) \ge 0$.
Together with $\gamma^i(\ell) \ge   x(B(\ell)) - \sum_{c\in \{1,\dots,q\} \setminus \{i\}} \eta^c_t  = x(B(\ell)) - (x(S_t) - \eta^i_t)$,
we thus obtain
\begin{align*}
   p_t^i(\ell) - p_t^{i-1}(\ell)  \ \ge&\ \left( e^{-x(B(\ell))+x(S_t)- \eta^i_t } - 1+  x(B(\ell))- x(S_t) +  \eta^i_t  \right) \cdot \mathds{1}\left\{ x(B(\ell)) \ge x(S_t) - \eta^i_t\right\} \\
   & \ \ \cdot   \left(1 - (x(S_t) - \eta^i_t)\right) \\
      \ \ge&\ \left( e^{-x(B(\ell))+x(S_t)- \eta^i_t } - 1+  x(B(\ell))- x(S_t) +  \eta^i_t  \right) \cdot \mathds{1}\left\{ x(B(\ell)) \ge x(S_t) - \eta^i_t\right\}\\
      & \ \ \cdot   h_1(\lambda_t^0, \mu^0_t, \lambda_t^1, \mu^1_t, \eta^i_t)\enspace.
\end{align*}
Recall that we have shown this lower bound on $p_t^i(\ell) - p_t^{i-1}(\ell)$ also for the case $\eta^i_t > \frac{1}{2} x(S_t)$.
Hence,  in both cases we have
\begin{equation*}
\begin{aligned}
&\sum_{\ell \in S_t} x_{\ell} \cdot \left( p_t^i(\ell) - p_t^{i-1}(\ell) \right)\ =\
\int_{0}^{x(S_t)}  \left( p_t^i(\ell_z) - p_t^{i-1}(\ell_z)\right) \ dz  \\
\ge&\ \int_{0}^{x(S_t)} \left(e^{-x(B(\ell_z))+x(S_t)- \eta^i_t}  - 1 + x(B(\ell_z)) -x(S_t)+  \eta^i_t \right) \cdot \mathds{1}\left\{ x(B(\ell_z)) \ge x(S_t) - \eta^i_t\right\} dz \\
           & \ \  \cdot h_1(\lambda_t^0, \mu^0_t, \lambda_t^1, \mu^1_t, \eta^i_t)\enspace.
\end{aligned}
\end{equation*}
In order to show \eqref{eq:bound_to_show_profit} and complete the proof of the overall lemma, we show the following:
\begin{equation}\label{eq:final_gain_calculation}
\begin{aligned}
&\int_{0}^{x(S_t)} \!\!\left(e^{-x(B(\ell_z))+x(S_t)- \eta^i_t}  - 1 + x(B(\ell_z)) -x(S_t)+  \eta^i_t \right)\cdot \mathds{1}\left\{x(B(\ell_z)) \ge x(S_t)-  \eta^i_t  \right\} \ dz \\
\geq &
\int_{0}^{x(S_t)} \!\!\left(e^{-\beta(\ell_z)+x(S_t)- \eta^i_t}  - 1 + \beta(\ell_z) -x(S_t)+  \eta^i_t \right) \cdot \mathds{1}\left\{ \beta(\ell_z) \ge x(S_t)-  \eta^i_t  \right\} dz \\[1em]
=& \  h_2(\lambda_t^0, \mu^0_t, \lambda_t^1, \mu^1_t, \eta^i_t)\enspace.
\end{aligned}
\end{equation}
The equation of~\eqref{eq:final_gain_calculation} results by computing the integral using the definition of $\beta(\ell_z)$ (see~\eqref{eq:size_below_sets} and \eqref{eq:stack_partition}).
For the inequality in~\eqref{eq:final_gain_calculation}, we observe that replacing $x(B(\ell_z))$ by some lower bound on $x(B(\ell_z))$ yields a lower bound on the left-hand side of this equation
because $( e^{-y}-1+y)\cdot\mathds{1}\{y \ge 0\}$ is increasing in $y$. 
By \eqref{eq:size_below_sets}, we have $\beta(\ell_z)\le x(B(\ell_z))$ and this implies the inequality in~\eqref{eq:final_gain_calculation}.
Because $h_1(\lambda_t^0, \mu^0_t, \lambda_t^1, \mu^1_t, \eta^i_t)$ is clearly non-negative by definition and because the second integral in~\eqref{eq:final_gain_calculation} is non-negative for any values of $\beta(\ell_z)$, $x(S_t)$, and $\eta^i_t$, we can conclude that $\mathrm{gain}(\lambda_t^0, \mu^0_t, \lambda_t^1, \mu^1_t, \eta^i_t) \ge 0$.
\end{proof}

To finish the proof of Lemma~\ref{lem:bound_stack_problem}, it remains to prove that the function $\mathrm{gain}$ fulfills \eqref{eq:worst_case_no_short_odd_cycles}.
Using Lemma~\ref{lem:bound_conditional_prob} this leads to the following condition. 
We need to prove that the expression \eqref{eq:b_condition_checked_numerically} below is non-negative for all $v,w\in T$:

\begin{equation}\label{eq:b_condition_checked_numerically}
\begin{aligned}
  &\ \frac{b}{x(S_w) -\eta_w^{c_v}} \cdot  \mathrm{gain}(\lambda_w^0,\mu^0_w, \lambda_w^1, \mu^1_w, \eta^{c_v}_w)  \\
  & + s_{vw} \cdot z_1  + (\eta_w^{c_v} - s_{vw}) \cdot z_2 \\
 &+ \max\{0,\ x(S_v) - \eta_w^{c_v}\} \cdot z_3 + \max\{0,\ 1 - x(S_v) -\eta_w^{c_v} +s_{vw}\} \cdot z_4 \\[3mm]
=&\ \frac{b}{\lambda^0_w + \mu^0_w + \lambda^1_w + \mu^1_w - \eta_w^{c_v}} 
         \cdot  \mathrm{gain}(\lambda_w^0,\mu^0_w, \lambda_w^1, \mu^1_w, \eta^{c_v}_w)  \\
&\ - s_{vw} \cdot  \left(b-\tfrac{1}{3}\right) - (\eta_w^{c_v} - s_{vw}) \cdot \left(2\left(b-\tfrac{2}{5}\right)-\tfrac{1}{30}\right) \\
&+ \max\{0,\ x(S_v) - \eta_w^{c_v}\} \cdot \left( \tfrac{1}{2}-b \right) + \max\{0,\ 1 - x(S_v) -\eta_w^{c_v} +s_{vw}\} \cdot \left(1-b\right)\enspace.
\end{aligned}
\end{equation}
We verify this condition numerically for $b=0.42$ using a computer program that does the following.
First, we observe that for a fixed value of $\eta_w^{c_v}$ one can easily compute the worst possible value
\begin{align*}
\min \Big\{&   - s_{vw} \cdot  \left(b-\tfrac{1}{3}\right) - (\eta_w^{c_v} - s_{vw}) \cdot \left(2\left(b-\tfrac{2}{5}\right)-\tfrac{1}{30}\right) \\
&+ \max\{0,\ x(S_v) - \eta_w^{c_v}\} \cdot \left( \tfrac{1}{2}-b \right) + \max\{0,\ 1 - x(S_v) -\eta_w^{c_v} +s_{vw}\} \cdot \left(1-b\right)  :\\
&  0 \le s_{vw} \le \eta_w^{c_v},\ s_{vw} \le x(S_v) \le 1 \Big\}
\end{align*}
of the last part of the expression \eqref{eq:b_condition_checked_numerically} because, except for the two maxima, 
the expression is linear in the two variables $s_{vw}$ and $x(S_v)$.
(Here, the conditions on $s_{vw}$ and $x(S_v)$ follow from the definition of $s_{vw}$ and $\eta_w^{c_v}$.)
To prove non-negativity of the overall expression, we use a fine grid to partition the area
\begin{equation*}
\left\{ (\lambda^0_w, \mu^0_w, \lambda_w^1, \mu^1_w, \eta^{c_v}_w) \in [0,1]^5 : 
            \eta^{c_v}_w \le  \lambda^0_w + \mu^0_w +\lambda_w^1 + \mu^1_w \le 1 \right\}
\end{equation*}
into many small subareas. Then we verify non-negativity of \eqref{eq:b_condition_checked_numerically} separately for each of the grid cells.
For this we use pessimistic estimates on the value of the function $\mathrm{gain}$.
When making the grid cells small enough, simple pessimistic estimates are sufficient to verify that \eqref{eq:b_condition_checked_numerically} is non-negative for $b=0.42$.\footnote{
We also did a second and independent calculation using Mathematica to minimize \eqref{eq:b_condition_checked_numerically}.
This confirmed non-negativity.}

\subsection{Combining the different rounding procedures}

In this section we combine the result of our stack analysis with  rounding procedure discussed in Section~\ref{sec:cg} (using Chv\'atal-Gomory cuts) in order to obtain 
a polynomial-time randomized algorithm for $O(1)$-wide cactus augmentation that achieves an approximation ratio of \apxfac in expectation.
(We will explain how to turn the algorithm into a deterministic one in the next section.)
Our procedure first computes an optimal solution $x$ to the linear program 
$\min\{1^T y : y \in P_{\mathrm{bundle}}^{\mathrm{min}} \cap P_{\mathrm{cross}}\}$, where $P_{\mathrm{cross}}$ is the polytope from Lemma~\ref{lem:cross-link_rounding}, which is needed for the other backbone procedure. We then return the better of the results of two rounding procedures, namely Algorithm~\ref{algo:stack_rounding} and the rounding procedure guaranteed by Lemma~\ref{lem:cross-link_rounding}.

\smallskip

To analyze the approximation ratio of this algorithm, we exploit the following property of the function~$g$, which was defined in \eqref{eq:defg}.

\begin{lemma}\label{lem:convexity_of_g}
The function $g : [0,1]^4 \to \mathbb{R}$ is convex in each of its parameters.
\end{lemma}
\begin{proof}
We compute the second partial derivative of $g$ with respect to each of the parameters $\lambda^0$, $\mu^0$, $\lambda^1$, and $\mu^1$.
They can all easily be seen to be non-negative, using the inequality $1+y \le e^y$ for all $y\in \mathbb{R}$.
\end{proof}

Given a $k$-wide instance $(G=(V,E), L)$ of cactus augmentation, we define $\alpha := \frac{x(L{\mathrm{cross}})}{x(L)}$.
Then we have 
\begin{equation}\label{eq:relate_alpha_stack_load}
2\cdot\alpha \cdot x(L) =  \sum_{v\in T} (\uplambda_v^0 + \upmu_v^0+ \uplambda_v^1 + \upmu^1_v)\enspace.
\end{equation}
By Lemma~\ref{lem:stack_algo}, Algorithm~\ref{algo:stack_rounding} yields a solution with at most the following number of links in expectation:
\begin{align}
&x(L_{\mathrm{in}}) + 2 \cdot x(L_{\mathrm{cross}}) - b \cdot \sum_{v\in T} g(\uplambda_v^0, \upmu^0_v,\uplambda_v^1, \upmu^1_v) \notag \\
=&\ \left(1 + \alpha - \frac{b}{x(L)} \sum_{v\in T} g(\uplambda_v^0, \upmu^0_v,\uplambda_v^1, \upmu^1_v)  \right) \cdot x(L) \notag \\
=&\ \left(1 + \alpha - \frac{b \cdot 2\alpha}{\sum_{v\in T} (\uplambda_v^0 + \upmu^0_v + \uplambda_v^1 +\upmu^1_v) } \sum_{v\in T} g(\uplambda_v^0, \upmu^0_v,\uplambda_v^1, \upmu^1_v)  \right) \cdot x(L)\enspace . \label{eq:result_stack_algo}
\end{align}
Now we use the rounding procedure discussed in Section~\ref{sec:cg} (using Chv\'atal-Gomory cuts).
By Lemma~\ref{lem:cross-link_rounding} we obtain a solution with at most $x(L)  + x(L_{\mathrm{in}}) - x(L_{\mathrm{up}})$ many links.
Lemma~\ref{lem:generating_up_links} and \eqref{eq:relate_alpha_stack_load} imply 
\begin{align}
x(L)  + x(L_{\mathrm{in}}) - x(L_{\mathrm{up}}) \leq&\ x(L) + (1- \alpha) \cdot x(L) - \sum_{v\in T} \upmu^0_v  \notag \\
=&\ \left(2-\alpha - \frac{2 \cdot \alpha}{\sum_{v\in T} (\uplambda_v^0 + \upmu_v^0+ \uplambda_v^1 + \upmu^1_v)} \cdot\sum_{v\in T} \upmu^0_v \right) \cdot x(L)\enspace .
\label{eq:result_cg_algorithm}
\end{align}
We output the better of the results of the two rounding procedures and, hence, the expected number of links will be at most the minimum of~\eqref{eq:result_stack_algo} and~\eqref{eq:result_cg_algorithm}.
In order to get an upper bound on the approximation ratio, we consider the worst case of this minimum over all possible values of $\alpha$, $\uplambda_v^0$, $\upmu^0_v$, $\uplambda_v^1$, $\upmu^1_v$. 
In order to restrict the set of ``possible values'', we now derive some conditions that these values need to fulfill.
First, because $x(\delta_L(t)) \ge 1 $ for all $t\in T$ and because all cross-links incident to a terminal $t\in T$ are contained in $\Uplambda^0_t$, we have 
$2\cdot  x(L) \ge |T| + \sum_{v\in T} (\upmu^0_v +\uplambda^1_v +\upmu^1_v)$, implying the inequality below (the equation follows from~\eqref{eq:relate_alpha_stack_load}):
\[
 \sum_{v\in T} \left(\uplambda_v^0 + \upmu^0_v +\uplambda_v^1 +\upmu^1_v \right) = 2 \cdot \alpha \cdot x(L) \ge \alpha\cdot \sum_{v\in T} \left(1 +\upmu^0_v + \uplambda^1_v +\upmu^1_v\right)\enspace.
\]
Moreover, by Lemma~\ref{lem:upper_bound_mu1_links}, we have
\[
 \sum_{v\in T} \left(\uplambda_v^0 + \upmu^0_v +\uplambda_v^1 + 2 \cdot \upmu^1_v\right) \le |T| \enspace.
\]
Therefore, we obtain an approximation ratio of
\begin{align*}
\max \Big\{&\ \min\Big\{\ 2-\alpha - \frac{2 \cdot \alpha}{\sum_{v\in T} (\uplambda_v^0 + \upmu_v^0+ \uplambda_v^1 + \upmu^1_v)} \cdot\sum_{v\in T} \upmu^0_v,\\
               &\ \phantom{\min\Big\{}\  1 + \alpha - \frac{b \cdot 2 \alpha}{\sum_{v\in T} (\uplambda_v^0 + \upmu^0_v +\uplambda_v^1 +\upmu^1_v) }\cdot  \sum_{v\in T} g(\uplambda_v^0, \upmu^0_v, \uplambda_v^1, \upmu^1_v)\ \Big\}\  : \\
&\quad \alpha \in [0,1], \\
&\quad  \uplambda_v^0, \upmu^0_v, \uplambda_v^1, \upmu^1_v \in [0,1] \text{ for  all } v\in T, \\
&\quad \sum_{v\in T} \left(\uplambda_v^0 + \upmu^0_v +\uplambda_v^1 +\mu^1_v \right) \ge \alpha\cdot \sum_{v\in T} \left(1 +\upmu^0_v + \uplambda^1_v +\upmu^1_v\right),\\
&\quad   \sum_{v\in T} \left(\uplambda_v^0 + \upmu^0_v +\uplambda_v^1 + 2 \cdot \mu^1_v\right) \le |T| \\
 \Big\}&\enspace.
\end{align*}
Because, due to Lemma~\ref{lem:convexity_of_g}, the function $g$ is convex in each of its parameters, replacing each of the variables $\uplambda_v^0, \upmu^0_v, \uplambda_v^1, \upmu^1_v$ by their average value over all of the terminals $v\in T$,
does not increase the value of $\sum_{v\in T} g(\uplambda_v^0, \upmu^0_v, \uplambda_v^1, \upmu^1_v)$.
Moreover, this replacement does not change $\sum_{v\in T} (\uplambda_v^0 + \upmu_v^0+ \uplambda_v^1 + \upmu^1_v)$.
We also note that the average values of $\uplambda_v^0, \upmu^0_v, \uplambda_v^1, \upmu^1_v$ will be feasible for the above optimization problem given that $\uplambda_v^0, \upmu^0_v, \uplambda_v^1, \upmu^1_v$ are feasible.
Thus, we can simplify the optimization problem and conclude that we obtain an approximation ratio of 
\begin{equation}\label{eq:final_optimization_problem}
\begin{aligned}
\max \Big\{&\ \min\Big\{\ 2-\alpha - \frac{2 \cdot \alpha}{\uplambda^0 + \upmu^0+ \uplambda^1 + \upmu^1} \cdot  \upmu^0,\\
               &\ \phantom{\min\Big\{}\  1 + \alpha - \frac{b \cdot 2 \alpha}{\uplambda^0 + \upmu^0 +\uplambda^1 +\upmu^1}\cdot  g(\uplambda^0, \upmu^0, \uplambda^1, \upmu^1)\ \Big\}\  : \\
&\quad \alpha \in [0,1], \\
&\quad  \uplambda^0, \upmu^0, \uplambda^1, \upmu^1 \in [0,1], \\
&\quad \uplambda^0 + \upmu^0 +\uplambda^1 +\mu^1 \ge \alpha\cdot \left(1 +\upmu^0 + \uplambda^1 +\upmu^1 \right),\\
&\quad  \uplambda^0 + \upmu^0 +\uplambda^1 + 2 \cdot \mu^1 \le 1 \\
 \Big\}&\enspace.
\end{aligned}
\end{equation}
To determine the value of this optimization problem, we again use a computer program, similar to the one we used to verify \eqref{eq:b_condition_checked_numerically}.
This yields that the optimum value of \eqref{eq:final_optimization_problem} is at most $\apxfac$.

\begin{lemma}
Let $k$ be a constant.
There is a randomized polynomial-time algorithm that, given a $k$-wide instance of CacAP, returns a CacAP solution with at most $\apxfac \cdot |\OPT|$ many links in expectation.
\end{lemma}
\begin{proof}
Lemma~\ref{lem:cross-link_rounding} and Lemma~\ref{lem:optimize_efficiently_min_bundle} imply that we can compute an optimum solution
$x$ to the linear program $\min\{1^T y : y \in P_{\mathrm{bundle}}^{\mathrm{min}} \cap P_{\mathrm{cross}}\}$ in polynomial time. 
We have shown above that there is an efficient randomized algorithm computing a CacAP solution with at most $\apxfac \cdot x(L)$ many links in expectation.
Using that $P_{\mathrm{cross}}$ is a relaxation of the convex hull of incidence vectors of  CacAP solutions and that 
$ P_{\mathrm{bundle}}^{\mathrm{min}}$ is a relaxation of the convex hull of incidence vectors of $L_{\mathrm{cross}}$-minimal CacAP solutions,
Lemma~\ref{lem:existence_optimum_minimal_solution} implies $x(L) \le |\OPT|$.
\end{proof}

\subsection{Derandomizing}\label{sec:derandomizing}

Finally, we explain how we can turn our randomized algorithm into a deterministic algorithm that achieves an approximation ratio of $\apxfac$,
which is the approximation ratio we already proved for our randomized algorithm above.
Recall that the rounding procedure from Section~\ref{sec:cg} is deterministic.
The only randomized part of our algorithm is the sampling of a solution $F_i$ for every principal subcactus $G_i$ in Algorithm~\ref{algo:stack_rounding}.

To analyze Algorithm~\ref{algo:stack_rounding}, we proved that in the corresponding instance of the stack problem there exists a removable set of links of size at least $\sum_{t\in T} X_t$.
From this we concluded that the randomized algorithm outputs a solution with at most $\sum_{c=1}^q |F_c| -\sum_{t\in T} X_t$ many links.
Our proof shows that if we can find solutions $F_i$ for every principal subcactus $G_i$ such that
$\sum_{c=1}^q |F_c| - \sum_{t\in T} X_t$ is at most as large as its expected value, then we achieve the claimed approximation ratio of $\apxfac$.

We use the method of conditional expectations, i.e., the following procedure.
We consider the principal subcacti $G_i$ for $i\in[q]$ sequentially to fix the random link set $F_i$ to one of the possible sets $F_i^j$ with $j\in [h]$.
Therefore, when we consider subcactus $G_i$ we have already fixed the value of the random variable $F_d$ to some value $F^{j_d}_d$  for each $d<i$.
We have
\begin{align*}
 &\mathbb{E}\left[\sum_{c=1}^q |F_c| -\sum_{t\in T} X_t \ \Big|\  F_d = F^{j_d}\text{ for }d< i\right] \\
&=\ \sum_{j=1}^h p_j \cdot  \mathbb{E}\left[\sum_{c=1}^q |F_c| -\sum_{t\in T} X_t \ \Big|\  F_d = F^{j_d}_d\text{ for }d< i \text{ and }F_i = F_i^j \right]\enspace,
\end{align*}
where $x|_{L_i} = \sum_{j=1}^h p_j \cdot \chi^{F_i^j}$, and we have $F_i = F_i^j$ with probability $p_j$.
This implies that for some $j_i\in [h]$ we have 
\begin{equation}\label{eq:condition_needed_for_derandomizing}
\mathbb{E}\left[\sum_{c=1}^q |F_c| -\sum_{t\in T} X_t \ \Big|\  F_d = F^{j_d}_d\text{ for }d< i+1\right]
\ \le\ \mathbb{E}\left[\sum_{c=1}^q |F_c| -\sum_{t\in T} X_t \ \Big|\  F_d = F^{j_d}_d\text{ for }d< i\right]\enspace.
\end{equation}
Then we fix $F_i$ to be $F^{j_i}_i$ for this $j_i\in [h]$.
Doing this for all $i=1,\dots, q$, yields a solution $F_1,\dots,F_q$ such that
\[
\sum_{c=1}^c |F_c| -\sum_{t\in T} X_t\ \le\ \mathbb{E}\left[\sum_{c=1}^q |F_c| -\sum_{t\in T} X_t \right]\enspace.
\]
In order to achieve a polynomial runtime we need to compute the conditional expectations appearing in condition~\eqref{eq:condition_needed_for_derandomizing} in polynomial time.
The only difficulty is compute the conditional expected values of the random variables $X_t = b \cdot \mathbbm{1}[v\text{ has an entering dominated link in }F] + Z_v$.

First, we note that we can compute the probability 
that a terminal $t\in T$ has an incoming dominated link in $F$ in polynomial time. 
Having a vector $x \in P^{\mathrm{min}}_{\mathrm{bundle}}(G,L)$ and fixed decompositions  $x|_{L_d} = \sum_{j=1}^h p_j \cdot \chi^{F_d^j}$ for all $d\in [q]$ we can compute this probability exactly. (This can be seen from the proof of Lemma~\ref{lem:simple_domination_prob}.)

It remains to compute the conditional expectation of the random variables $Z_t$.
To this end we need to compute the conditional probabilities of the (polynomially many) events $\Delta^{vw}_i$ for $i\in \{1,\dots,4\}$ and $v,w\in T$.
Recall that the event $\Delta^{vw}_i$ happens if the path $P_v$ ends in $w$, its length is even, positive and at most $\frac{1}{1-2b}$, and certain conditions (depending on $i \in\{1,\dots,4\}$) on links incident to $v$ and $w$ in $(T,F)$ are fulfilled.
Because $P_v$ has constant length in this case, there are only polynomially many possibilities for the path $P_v$ in the event $\Delta^{vw}_i$ and hence we can enumerate over all of them.
Let us now consider a fixed path $v$-$w$ path $P_v$ of even positive length of at most $\frac{1}{1-2b}$.
In particular, the set $C \subseteq [q]$ of colors of vertices of $P_v$ has constant size.
Therefore, we can enumerate all possible outcomes for $(F_c)_{c\in C}$ in polynomial time.
The only impact that $F_l$ for other colors $l\in [q]\setminus C$ has on the fact of whether the event $\Delta^{vw}_i$ happens for the fixed path $P_v$, is the following.
If for any of the vertices in $P_v$ there is a link $\ell=(v,w) \in F_c$ such that $\ell$ is below the incoming link of $v$ in
$\bigcup_{c\in C} F_c$, then the path $P_v$ is not the $v$-$w$ path we fixed.
(In case of ties, i.e., when there are several outgoing links of $v$ such that each of them is below the others, we fix an arbitrary order on these links to make the choice of $P_v$ unique.)
For each solution $F_c^j$ with $c\in [q]\setminus C$ we can determine whether it contains a link $\ell$ that would lead to a different choice of $P_v$ and we can do this independently for each color.
Therefore, we can compute the probability that any such link $\ell$ gets sampled for any of these colors in polynomial time.

This shows that we can achieve the approximation ratio $\apxfac$ for $k$-wide connectivity augmentation also with a deterministic polynomial-time algorithm.
Together with Theorem~\ref{thm:main_reduction}, this concludes the proof of our main result, i.e., Theorem~\ref{thm:main}.

\section{Extensions to Weighted Connectivity Augmentation}\label{sec:weighted}
In this section, we discuss how our results carry over to weighted CacAP with bounded costs, in a similar spirit as this has been done for TAP (see~\cite{adjiashvili_2018_beating,fiorini_2018_approximating,grandoni_2018_improved,nutov_2017_tree}), to obtain the following result.
\begin{theorem}\label{thm:main_weighted}
For $B > 0$, there is a polynomial-time algorithm that for a weighted CacAP instance $(G=(V,E),L,c)$, where $c:L\to [1,B]$ are the link costs, returns a solution of approximation guarantee $\sfrac{3}{2} - f(B)$, where $f:[1,\infty)\to \mathbb{R}_{>0}$ is a positive function that tends to $0$ as $B$ tends to infinity.
\end{theorem}
Only minor changes are necessary to our procedure to obtain Theorem~\ref{thm:main_weighted}, which we briefly discuss now. In the following, when using Landau notations like $O$, $\Theta$, and $\Omega$, we will always be explicit about the dependence on $B$ and $\epsilon$, even though $B$ and $\epsilon$ are (independent) constants.

Instead of reducing a general CacAP instance to $O(\sfrac{1}{\epsilon^2})$-wide ones as stated in Theorem~\ref{thm:main_reduction}, we obtain a reduction to $O(\sfrac{B^2}{\epsilon^2})$-wide instances.
This can be achieved with the precise same reduction technique as in the unweighted case, but by using an error parameter of $\epsilon = \sfrac{\bar{\epsilon}}{B}$ if we aim at obtaining an $\alpha (1+\bar{\epsilon})$-approximation for unrestricted CacAP given an $\alpha$-approximation for $64 (8+3 \frac{\bar{\epsilon}}{B})\frac{B^2}{\bar{\epsilon}^2}$-wide CacAP.
Indeed, given a point $x\in [0,1]^L$ and assuming that the procedure does not return a separating hyperplane for $x$, then the extra links incurred through the reduction are due to the $x$-heavy cut covering and the merging of solutions. The reduction shows that the number of these extra links is bounded by $O(\epsilon x(L))$. Because our weighted instance has link costs within $[1,B]$, this could lead to a total extra cost of up to $O(B \epsilon x(L)) = O(\bar{\epsilon} x(L)) = O(\bar{\epsilon}) \sum_{\ell\in L} c_\ell x_\ell$, which is no more than an $\bar{\epsilon}$-fraction of the total cost of $x$, as desired.

Because our backbone procedures work for weighted $O(1)$-wide instances, this already implies a $(1.5+\epsilon)$-approximation algorithm for weighted CacAP with bounded costs.\footnote{We note that analogous to the approach presented by Nutov~\cite{nutov_2017_tree} for TAP, this result can be extended to slightly super-constant values of $B$.} Finally, to obtain factors below $1.5$, we employ the same improvement of the bundle-based algorithm as for the unweighted case. The analysis shows that either the backbone procedures already provide a factor below $1.5$---and this immediately carries over to the weighted case as the backbone procedures work also for the weighted case---or there is a set of cross-links that can be deleted and which comprises a constant fraction of all links in the solution. However, also in the latter case we win a constant factor, which is possibly by a factor of $B$ lower than in the unweighted case because the deleted cross-links may only have a cost of $1$ whereas the remaining links may all have a cost of up to $B$. This leads to the claimed result.
\section{Fixed-parameter tractability of weighted CacAP}\label{sec:fpt}

In this section we show why Lemma~\ref{lem:fpt_terminals} holds, which implies that weighted CacAP is fixed-parameter tractable when parameterized by the number of terminals. This result follows readily by combining a reduction of Basavaraju, Fomin, Golovach, Misra, Ramanujan, and Saurabh~\cite{basavaraju_2014_parameterized} from CacAP to Steiner Tree and a slight generalization of the well-known Dreyfus-Wagner algorithm for Steiner Tree~\cite{dreyfus_1971_steiner} to Node-Weighted Steiner Tree, which was presented by Buchanen, Wang, and Butenko~\cite{buchanan_2017_algorithms}.

We start by discussing the reduction from Cactus Augmentation to Node-Weighted Steiner Tree shown in~\cite{basavaraju_2014_parameterized}, which is based on an elegant characterization of feasible solutions to a CacAP instance.
\begin{lemma}[\cite{basavaraju_2014_parameterized}]\label{lem:red_to_steiner}
Let $(G=(V,E), L)$ be a CacAP instance, and let $T\subseteq V$ be the degree $2$ vertices (terminals) of $G$. Then one can efficiently construct a graph $H=(W,U)$ with vertex set $W=T\cup L$ such that a set $F\subseteq L$ is a CacAP solution if and only if $H[F\cup T]$ is connected.%
\footnote{We recall that $H[F\cup T]$ denotes the subgraph of $H$ induced by the vertices $F\cup T$ of $H$.}
\end{lemma}
Hence, Lemma~\ref{lem:red_to_steiner} reduces a weighted CacAP problem on an instance $(G=(V,E),L,c)$, where $c\in \mathbb{R}_{\geq 0}^L$ are the link costs, to the problem of finding a set $F\subseteq L$ of minimum cost $c(F)$ such that $H[F\cup T]$ is connected. This problem is known as the Node-Weighted Steiner Tree problem, and an adaptation of the well-known Dreyfus-Wagner dynamic programming algorithm~\cite{dreyfus_1971_steiner}, as presented in~\cite{buchanan_2017_algorithms}, allows for solving it in the same time as the Dreyfus-Wagner algorithm needs to solve Steiner Tree, which is $3^{|T|} \poly(|W|) = 3^{|T|} \poly(|V|)$. This implies Lemma~\ref{lem:fpt_terminals}.

\bibliographystyle{alpha}
\newcommand{\etalchar}[1]{$^{#1}$}

\end{document}